\documentclass[11pt,a4paper,reqno]{book}
\usepackage[utf8]{inputenc}
\usepackage[T1]{fontenc}
\usepackage{ae}
\usepackage{aecompl}
\usepackage{aeguill}
\usepackage[english]{babel}
\usepackage{amsmath}
\usepackage{amsthm}
\usepackage{amsfonts}
\usepackage{amssymb}
\usepackage{amscd}
\usepackage{bbm}
\usepackage{url}
\usepackage[pdftex,bookmarks=true,bookmarksnumbered=false]{hyperref} 
\usepackage[all]{xy}
\usepackage[left=3cm,top=3cm,bottom=3cm,right=3cm,head=1cm,foot=1cm]{geometry}
\input xy
\xyoption{all}

\newcommand{\R}{\mathbb{R}}
\newcommand{\C}{\mathbb{C}}
\newcommand{\Q}{\mathbb{Q}}
\newcommand{\N}{\mathbb{N}}
\newcommand{\Z}{\mathbb{Z}}
\newcommand{\F}{\mathcal{F}}
\newcommand{\E}{\mathcal{E}}

\newcommand{\Hilb}{\mathcal{H}}
\renewcommand{\L}{\mathcal{L}}

\newcommand{\norm}[1]{\lVert#1\rVert}
\newcommand{\ip}[1]{\left\langle#1\right\rangle}
\newcommand{\floor}[1]{\lfloor#1\rfloor}

\newcommand{\Hom}{\operatorname{Hom}}
\newcommand{\Tr}{\operatorname{Tr}}

\newcommand{\Pd}{\operatorname{Pd}}
\newcommand{\Cliff}{\operatorname{Cliff}}
\newcommand{\End}{\operatorname{End}}
\newcommand{\ch}{\operatorname{ch}}
\newcommand{\Ch}{\operatorname{Ch}}
\newcommand{\Todd}{\operatorname{Todd}}
\newcommand{\Tor}{\operatorname{Tor}}
\newcommand{\Map}{\operatorname{Map}}
\newcommand{\pfaff}{\operatorname{pfaff}}
\newcommand{\hol}{\operatorname{hol}}
\newcommand{\Dirac}{D\!\!\!\!/}
\newcommand{\im}{\operatorname{im}}
\newcommand{\id}{\operatorname{id}}
\newcommand{\Mod}{\operatorname{Mod}}

\newcommand{\rk}{\operatorname{rk}}
\newcommand{\ev}{\operatorname{ev}}
\newcommand{\ind}{\operatorname{index}}
\newcommand{\diag}{\operatorname{diag}}
\newcommand{\HP}{\operatorname{HP}}
\newcommand{\HL}{\operatorname{HL}}

\newcommand{\Ext}{\operatorname{Ext}}
\newcommand{\GL}{\operatorname{GL}}

\newcommand{\Aut}{\operatorname{Aut}}
\newcommand{\CP}{\C\!\operatorname{P}}

\newcommand{\pr}{\operatorname{pr}}
\newcommand{\spin}{\text{Spin}}
\newcommand{\spinc}{\text{Spin$^c$}}

\theoremstyle{plain}
\newtheorem{theorem}{Theorem}[section]
\newtheorem{proposition}[theorem]{Proposition}
\newtheorem{lemma}[theorem]{Lemma}
\newtheorem*{corollary}{Corollary}

\theoremstyle{definition}
\newtheorem{definition}[theorem]{Definition}

\newtheorem*{claim}{Claim}

\theoremstyle{remark}

\SelectTips{eu}{10}

\title{\textbf{Geometric and topological aspects of \linebreak Type IIB D-branes}}
\date{\today}
\author{Kim Laine \vspace{5ex} \\ Thesis Advisor: Jouko Mickelsson \vspace{10ex} \\ }

\begin{document}
\thispagestyle{empty}
\maketitle
\pagestyle{empty}
\tableofcontents
\thispagestyle{empty}
\newpage
\pagestyle{empty}

\chapter*{Introduction}
\pagestyle{plain}
\setcounter{page}{1}
The discovery of general relativity and quantum theory were remarkable events in the history of theoretical physics. They completely revolutionized our understanding of space, time, matter and energy at very large and very small scales. The simple and intuitive Newtonian pictures of space, time, gravity and pointlike particles were obviously fundamentally flawed and had to be replaced by constructions so astonishingly unintuitive that even the best modern populistic books on these subjects are almost certainly lacking at some fundamental level. The problem being, of course, that there is simply nothing in everyday life that would resemble these phenomena closely enough to be very useful as an explanatory analogue. Thus, if one wishes to truly understand such theories, the only possibility is through mathematics. An ideal plan of work would be to first perform experiments as far as possible, then try to construct a suitably general mathematical structure befitting the physical theory and proceed to investigate it. If the mathematical structure used is general and elaborate enough, one can hope to run into new substructures, not obvious from the original physical theory. Finally, one should attempt to perform an experiment to verify if the new substructure is indeed realized in nature. 

String theory is an ongoing attempt to find a consistent mathematical structure unifying both general relativity and quantum theory as substructures, which replaces pointlike elementary particles with small $1$-dimensional supersymmetric strings, oscillating in a $10$-dimensional spacetime. In addition to strings, superstring theory contains higher-dimensional dynamical extended objects, called D-branes, to which open string endpoints are restricted. An open string propagating in spacetime determines a $2$-dimensional superconformal field theory and the boundary conditions the D-brane determines must preserve the superconformal invariance of the theory. This is why D-branes are inherently more complicated objects than simple submanifolds of spacetime.

In the low-energy limit superstring theory is described by an effective field theory, which is $10$-dimensional supergravity. It contains several generalized higher-dimensional massless gauge fields, most importantly the RR-fields and the NS-NS $B$-field. D-branes are the extended objects coupling to the RR-fields. Dirac quantization for the RR-charge (of a D-brane) implies that the charge lifts to a class in some generalized integral cohomology theory. There does not seem to be any absolute way of determining what this generalized cohomology theory should be, but there are strong arguments suggesting that for RR-charges the correct cohomology theory is $K$-theory, at least when there is no background $B$-field present.

Turning on a background $B$-field makes the situation considerably more difficult. The $B$-field and D-branes are related to each other, because the $B$-field couples to the $2$-dimensional open string worldsheet by higher-dimensional parallel transport and the boundary of the worldsheet couples to the D-brane. There is an anomaly associated to these couplings: the Freed-Witten anomaly. This anomaly can be cancelled when a certain cohomology equation is satisfied. When the $B$-field is set to be nontrivial, the Freed-Witten anomaly cancellation condition implies that D-brane charges are classified by twisted $K$-theory, rather than by $K$-theory. Finally, twisted $K$-theory can be defined purely algebraically using operator $K$-theory, opening up a possibility for extending the definition of RR-charge to noncommutative D-brane worldvolumes in a noncommutative spacetime.

The problem with string theory is that the mathematical structure is incredibly complicated and so far there are no known experiments to verify its reality. One can only continue to investigate the mathematical structure in hopes of running into a possibility for experimental verification. Even if the the theory of D-branes would never be verified experimentally, it has already contributed an immense amount of motivation and research topics to pure mathematics, especially to modern geometry, topology and algebra.

\section*{The plan of the work}
Chapter 1 consists of a brief introduction to the underlying physical theory, followed by a long and comprehensive overview of the rest of the text. This is to make the thesis accessible to a wider audience.

Chapter 2 is the beginning of the main body of the text. It starts with a review of the theory of characteristic classes and $\spin$ structures, followed by the introduction of $K$-theory (of vector bundles). After that, we explain how the $K$-theoretic classification of D-brane charge arises. We also take a look at an alternative, and perhaps more intuitive, classification scheme of D-branes in terms of $K$-homology, the homological dual theory of $K$-theory. More precisely, we use the geometric model of Baum and Douglas for $K$-homology.

In chapter 3 we discuss the Freed-Witten anomaly. Most of the chapter deals with introducing the required mathematical machinery, including Deligne cohomology, bundle gerbes, bundle gerbe modules and their connections, higher-dimensional holonomy and higher-dimensional parallel transports. Having done that, we can give a proper mathematical description for all the physical objects contributing to the Freed-Witten anomaly. The chapter concludes with a derivation of the anomaly cancellation condition.

Chapter 4 starts with an introduction to operator $K$-theory and twisted $K$-theory, which we use to classify the RR-charge of twisted D-branes, that is, D-branes in a nontrivial background $B$-field. Our ultimate goal in this chapter is a noncommutative generalization of the D-brane charge formula. Several concepts from geometry and topology need to be translated to the of noncommutative geometry first. This requires the use of $KK$-theory of Kasparov, a unification of $K$-theory and $K$-homology, and the bivariant local cyclic homology theory of Puschnigg, a noncommutative generalization of ordinary homology and cohomology.

The range of mathematical topics treated in this thesis is undeniably quite wide. We provide introductions to the more advanced topics, but the reader is still expected to be familiar with basic topics in geometry, topology and algebra.

\chapter{Overview}\label{Chapter: Overview}
\section{Strings and D-branes}\label{Section: Overview, Strings and D-branes}
In \emph{Type IIB superstring theory}\footnote{We will only discuss Type IIB superstring theory, since it is, in some sense, the easiest.} pointlike elementary particles are replaced with supersymmetric $1$-dimensional strings, open or closed, of typical length $10^{-35}\,\text{m}$. Thus, they are way beyond what can be seen in particle accelerators. The spacetime, denoted by $X$ throughout the text, is assumed to be a \emph{smooth, $10$-dimensional\footnote{This is the so-called \emph{critical dimension}, which is required for the superconformal anomaly of the worldsheet superconformal field theory to vanish \cite{Pol1,Joh1}.} oriented Riemannian\footnote{A more physical assumption would be to take $X$ to be pseudo-Riemannian, but we assume that one can always perform Wick rotations in the path integrals when necessary.} $\spin$ manifold\footnote{A manifold $X$ is $\text{Spin}$ if and only if the \emph{second Stiefel-Whitney class} of its cotangent bundle, $w_2(X)\in H^2(X,\Z_2)$, vanishes, and $\text{Spin$^c$}$ if and only if the \emph{third integral Stiefel-Whitney class} of its cotangent bundle $W_3(X)\in H^ 3(X,\Z)$ vanishes. A reader unfamiliar with these concepts can, for now, regard them simply as technical conditions defined by the vanishing of the above-mentioned Stiefel-Whitney classes. The assumption that the spacetime should be $\text{Spin}$ is, of course, motivated by the physical fact that spinor particles exist.}}. All manifolds are locally compact and paracompact. Locally compact spaces can be compactified by adding a ''point at infinity''. It is a deep result that compact manifolds have the homotopy type of a finite CW-complex. Paracompactness is important for the existence of smooth partitions of unity and good open covers, which makes \v Cech cohomology particularly well-behaving. We denote by $M$ the worldsheets of both open and closed strings. Thus, $M$ is a $2$-dimensional compact oriented manifold with an embedding $\xi:M\rightarrow X$. A closed string worldsheet is a closed manifold\footnote{That is, compact and without boundary.}. An open string worldsheet has a boundary, corresponding to the endpoints of the open string.

The dynamics of the classical string are governed by the famous \emph{Polyakov action} \cite{Pol1}. Canonical quantization can be carried out, at least in flat spacetime, yielding an infinite tower of oscillatory modes, each corresponding to a massless or massive particle. Of course, open and closed strings yield very different sets particles.

The closed string spectrum contains a locally defined symmetric rank $2$ tensor field, which can be identified with the \emph{graviton field} (it satisfies the linearized Einstein equation), a locally defined $2$-form field known as the \emph{NS-NS (Neveu-Schwarz--Neveu-Schwarz) $B$-field}, and a scalar field known as the \emph{dilaton}, whose vacuum expectation value determines the string coupling constant. Of these, the $B$-field is the only one we are interested in. They are higher-dimensional gauge fields, with gauge transformations
\[ B\mapsto B + d\Lambda\,, \]
where $\Lambda$ is a $1$-form, leaving physical states invariant. Instead, the $3$-form field $H=dB$, called the \emph{$H$-flux}, is globally defined.

The closed string spectrum contains also other higher $p$-form fields $G_p$, the \emph{RR (Ramond-Ramond) fields}. They are globally defined fields of odd degree (in Type IIB superstring theory). Locally, one can find \emph{RR-potentials} $C_p$, where $p=0,2,4,6,8$, such that $G_{p+1}=dC_p$ (in the absence of sources)\footnote{The RR-fields are actually not all independent. They are related to each other by the duality condition $*G_p=G_{10-p}$. This causes problems in quantization \cite{MW,FH,Val}, but since we shall not discuss quantized RR-fields, the reader is advised to simply forget about the duality condition.}. Let $G$ denote the polyform of all RR-fields. In the absence of sources, the equations of motion for the RR-fields are
\[ d*G=0 \,. \]
Analogously to classical electromagnetism, we can introduce \emph{sources} for the RR-fields as violations of these equations by an \emph{RR-charge density} polyform $j_e$:
\[ d*G = j_e \,. \]
Such an equation of motion can be obtained from an action of the form
\[ \int_X G\wedge *G + \int_X C\wedge j_e \,. \]
Each homogeneous component of $j_e$ determines a nontrivial class $[j_e^p]$ in \emph{compactly supported de Rham cohomology} $H^p_{c,dR}(X)$ \cite{BT} of spacetime. The \emph{Poincar\'e duals} of these classes are homology cycles, representing certain even-dimensional extended objects.

Open strings are more complicated. First of all, boundary conditions have to be set for the string in such a way that the superconformal invariance of the worldsheet theory is preserved. The open string worldsheet $M$, embedded into $X$ by $\xi:M\rightarrow X$, has a boundary $\partial M$. Naively, the boundary conditions can be given as a subset of $X$ onto which $\xi(\partial M)$ should be restricted. This is achieved by specifying a $p+1$-dimensional closed oriented manifold $\Sigma$, which we call a \emph{D$p$-brane worldvolume} or simply a \emph{D-brane worldvolume}, together with a continuous or smooth map $\phi:\Sigma\rightarrow X$. It is said that the D-brane \emph{wraps} the worldvolume $\Sigma$. We then require that $\xi(\partial M)\subset \phi(\Sigma)$.

The open string spectrum yields a massless locally defined $\mathfrak{u}(n)$-valued $1$-form field $A$, living on D-brane worldvolumes. The usual physical interpretation is that $n$ denotes the number of coincident D$p$-branes wrapping the same worldvolume. Thus, a single D-brane supports a simple $1$-form field. Now, open strings are attached to D-branes, not merely to their worldvolumes. Thus, given a stack of $n$ coincident D-branes, all wrapping $\Sigma$, the associated open string states are enriched with new non-dynamical degrees of freedom, since the endpoints are free to attach to any of the $n$ D-branes. These new degrees of freedom, called \emph{Chan-Paton factors}, are antihermitian $n\times n$ matrices, elements of $\mathfrak{u}(n)$. There is an anomaly associated to the open string path integral, the \emph{Freed-Witten anomaly}, which must be cancelled for the theory to be consistent. Let $A_\alpha$ denote the $A$-field on an open set $U_\alpha$ of $\Sigma$. To cancel the anomaly, on intersections $U_\alpha \cap U_\beta$ the local representatives $A_\alpha$,$A_\beta$ must satisfy
\begin{equation}\label{Eq: A-field transformation, phys} A_\beta = g_{\alpha\beta}^{-1} A_\alpha g_{\alpha\beta} - i g_{\alpha\beta}^{-1}\,dg_{\alpha\beta} - \Lambda_{\alpha\beta}\,, \end{equation}
where $g_{\alpha\beta}$ are functions $U_{\alpha\beta}\rightarrow U(n)$ and $\Lambda_{\alpha\beta}$ are (real-valued) $1$-forms. This resembles the transformation rule of an ordinary $U(n)$ gauge field, which is a connection on a $U(n)$ principal bundle. On triple intersections the ''transition functions'' $\{g_{\alpha\beta}\}$ satisfy
\[ g_{\alpha\beta}g_{\beta\gamma}g_{\alpha\gamma}^{-1} = \zeta_{\alpha\beta\gamma} \,, \]
where $\{\zeta_{\alpha\beta\gamma}\}$ are continuous $U(1)$-valued functions on $U_\alpha\cap U_\beta \cap U_\gamma$. The right-hand side should obviously be interpreted as $\zeta_{\alpha\beta\gamma}\mathbbm{1}_{n\times n}$. This means that $\{g_{\alpha\beta}\}$ provide the transition data for a principal $PU(n):=U(n)/U(1)$ bundle, the \emph{(projective) Chan-Paton bundle}. The obstruction for lifting a $PU(n)$ bundle to a $U(n)$ bundle is measured by an $n$-torsion cohomology class in $H^3(\Sigma,\Z)$. For the Chan-Paton bundle, this class is denoted by $\beta(\underline{\zeta})$, due to the following simple construction. Notice first, that $\{\zeta_{\alpha\beta\gamma}\}$ determines a \v Cech $2$-cocycle, since
\[ \zeta_{\alpha\beta\gamma}\zeta_{\alpha\gamma\delta} = g_{\alpha\beta}g_{\beta\gamma}g^{-1}_{\alpha\gamma}g_{\alpha\gamma}g_{\gamma\delta}g_{\alpha\delta}^{-1} = g_{\alpha\beta}g_{\beta\gamma}g_{\gamma\delta}g_{\alpha\delta}^{-1} = \zeta_{\beta\gamma\delta}g_{\alpha\beta}g_{\beta\delta}g_{\alpha\delta}^{-1} = \zeta_{\beta\gamma\delta}\zeta_{\alpha\beta\delta} \,, \]
which is the cocycle condition
\[ \zeta_{\beta\gamma\delta}\zeta_{\alpha\gamma\delta}^{-1}\zeta_{\alpha\beta\delta}\zeta_{\alpha\beta\gamma}^{-1} = 1 \,. \]
Thus, it defines a \v Cech cohomology class $[\underline{\zeta}]\in H^2(\Sigma,\underline{U(1)})$. Here $\underline{U(1)}$ denotes the \emph{sheaf} \cite{Bry} of continuous functions taking values in $U(1)$. More generally, $\underline{G}$ denotes the sheaf of continuous functions taking values in the Abelian group $G$. We use the notation $\underline{\zeta}$ to denote the \v Cech cocycle associated to $\{\zeta_{\alpha\beta\gamma}\}$. In this notation the \v Cech cocycle condition is written as $\delta(\underline{\zeta})=1$. The exact sequence of sheaves
\[ \xymatrix{ 0 \ar[r] & \underline{\Z} \ar[r] & \underline{\R} \ar[r] & \underline{U(1)} \ar[r] & 0 } \]
induces a long exact sequence
\[ \xymatrix{ \ldots \ar[r] & \check{H}^2(\Sigma,\underline{\R}) \ar[r] & \check{H}^2(\Sigma,\underline{U(1)}) \ar[r]^-\beta & H^3(\Sigma,\underline{\Z}) \cong H^3(\Sigma,\Z) \ar[r] & \check{H}^3(\Sigma,\underline{\R}) \ar[r] & \ldots } \]
of \v Cech cohomology groups \cite{Bry}. It is easy to show, by a partition of unity argument, that $H^p(\Sigma,\underline{\R})$ vanish for $p>0$. Thus, the \emph{Bockstein homomorphism} $\beta$ is an isomorphism, mapping $[\underline{\zeta}]$ to $\beta([\underline{\zeta}])\in H^3(\Sigma,\underline{\Z})\cong H^3(\Sigma,\Z)$\footnote{A more concrete description is given as follows. Writing $\underline{\zeta} = \exp(i\underline{\rho})$, where $\underline{\rho}$ is a real-valued \v Cech $2$-cochain, and using the cocycle condition $\delta(\underline{\zeta})=1$, yields
\[ \exp(i\delta(\underline{\rho})) = \delta(\exp(i\underline{\rho})) = \delta(\underline{\zeta})=1 \,. \]
But this is possible if and only if
\[ \delta(\underline{\rho}) = 2\pi \underline{n} \,, \]
for some integer-valued \v Cech $3$-cochain $\underline{n}$. It is obviously a cocycle due to $\delta^2=0$. The image $\beta([\underline{\zeta}])\in H^3(X,\Z)$ is precisely the class $[\underline{n}]\in H^3(X,\Z)$.}. It is easy to see that the vanishing of this class is equivalent to the existence of $U(n)$ lifts $\{\widetilde{g}_{\alpha\beta}$ of $\{g_{\alpha\beta}\}$, satisfying
\[ \widetilde{g}_{\alpha\beta}\widetilde{g}_{\beta\gamma}\widetilde{g}_{\alpha\gamma}^{-1} = 1\,. \]
Thus, they define a $U(n)$ bundle, the \emph{lift} of the $PU(n)$ bundle. The class $\beta([\underline{\zeta}])$ plays a fundamental role in Freed-Witten anomaly cancellation. In conclusion: a D-brane (or rather, a stack of $n$ coincident D-branes) consists of a closed oriented worldvolume manifold $\Sigma$, a map $\phi:\Sigma\rightarrow X$ and a $PU(n)$ bundle over $\Sigma$.

The string spectrum described above can actually be fully obtained only in flat spacetime. We have to assume that the low-energy field content remains the same also in more general spacetimes. Nevertheless, a quantum gravity theory, such as string theory, should be able to produce perturbations to the background metric. Adding a linear perturbation to the metric in the Polyakov action corresponds to turning on an infinite background, or a \emph{coherent state}, of gravitons \cite{Pol1}. One can similarly consider backgrounds of the other massless string states, in particular, a \emph{background $B$-field}\footnote{A background $B$-field makes the theory of D-branes much more complicated. This is because the $B$-field couples to the open string worldsheet and the boundary of the worldsheet couples to the $A$-field, and thus to the D-brane. Therefore, Freed-Witten anomaly cancellation imposes certain geometric and topological requirements for the D-brane.}. Suppose that various massless background fields have been turned on. By studying string perturbation theory in the low-energy limit, one obtains equations of motion for these background fields. It is then possible to construct a \emph{low-energy effective (worldvolume) action}, corresponding to these equations of motion. The low-energy effective field theory is the \emph{$10$-dimensional Type IIB supergravity} \cite{Pol1}.

\section{The low-energy effective field theory}\label{Section: Overview, The low-energy effective field theory}
We now introduce couplings and sources to the effective field theory. For example, the $B$-field should couple to a $2$-dimensional object, which is the string worldsheet. The coupling should be something like
\[ \int_M \xi^* B \,, \]
where $\xi:M\rightarrow X$ is the embedding of the open or closed string worldsheet. For an open string worldsheet this integral should be taken with a grain of salt, because when $\partial M \neq \varnothing$, it is not gauge invariant:
\[ \int_M \xi^*(B+d\Lambda) = \int_M \xi^* B + \int_{\partial M} \Lambda \,. \]
For a closed string worldsheet this is obviously not a problem, since $\partial M = \varnothing$. This failure of gauge invariance is one of the main aspects of the Freed-Witten anomaly.

Similarly, if $C_{p+1}$ is a $p+1$-form RR-potential, it should couple to a $p+1$-dimensional extended object. It was a remarkable discovery by Polchinski \cite{Pol2} that these are precisely the D-branes we described earlier\footnote{As a side note, recall that the worldvolume $\Sigma$ was not generally assumed to be a submanifold of $X$. Indeed, it was demonstrated in \cite{ES} that also nonrepresentable homology cycles, that is, homology cycles which are not represented by any nonsingular submanifolds, may still carry RR-charge.}. Now, one might naively anticipate that a D$p$-brane should couple to the RR-potentials by
\[ \int_\Sigma \phi^* C_{p+1} \,, \]
but this turns out to be incorrect: a D$p$-brane carries also lower-dimensional RR-charge and thus should couple also to lower-dimensional RR-fields. We now assume that all projective Chan-Paton bundles can all be lifted to $U(n)$ bundles. We call these lifts \emph{$U(n)$ Chan-Paton bundles} and their associated vector bundles \emph{$U(n)$ Chan-Paton vector bundles}. This assumption is related to the Freed-Witten anomaly and will later be dropped. Anomaly cancellation arguments can be used to deduce the coupling \cite{CY,El-S}
\[ \int_\Sigma \phi^* C \wedge \ch(E) \wedge \Todd(\Sigma) \wedge \frac{1}{\phi^*\sqrt{\Todd(X)}} = \int_X C \wedge \phi_*(\ch(E) \wedge \Todd(\Sigma)) \wedge \frac{1}{\sqrt{\Todd(X)}} \,, \]
where $E$ is the $U(n)$ Chan-Paton vector bundle of the D-brane, $\phi_*$ the cohomological \emph{Gysin ''wrong way'' homomorphism}, defined by $\phi_* := \Pd^{-1}_X \circ \phi_* \circ \Pd_\Sigma$, where $\Pd$ denotes the Poincar\'e duality isomorphism and $\phi_*$ (on the right-hand side) the natural push-forward in homology. The double meaning of $\phi_*$ should not cause problems. Finally, $\ch(E)\in H^\text{even}_\text{dR}(\Sigma)$ is the \emph{Chern character} of $E$, and $\Todd(\Sigma)\in H^\text{even}_\text{dR}(\Sigma)$ and $\Todd(X)\in H^\text{even}_\text{dR}(X)$ the \emph{Todd classes} of the tangent bundles of $\Sigma$ and $X$, respectively. These are examples of so-called \emph{characteristic classes} of vector bundles. It is necessary to assume $\Sigma$ to be a $\spinc$ manifold for $\Todd(\Sigma)$ to be defined. Now, while this coupling may look complicated and unnatural, it turns out to be quite natural from the point of view of \emph{$K$-theory}. In any case, this coupling suggests that at a classical level the D-brane charge is represented by the compactly supported de Rham cohomology class
\[ Q_\R(\Sigma,E) := \phi_*(\ch(E) \wedge \Todd(\Sigma)) \wedge \frac{1}{\sqrt{\Todd(X)}} \in H^\text{even}_{c,dR}(X) \,. \]
In fact, $\Todd$ and $\ch$ are classes in $H^\text{even}(-,\Q)$. Thus we can as well work with rational cohomology instead of de Rham cohomology and define
\begin{equation}\label{Eq: Rational D-brane charge, 1} Q_\Q(\Sigma,E) := \phi_*(\ch(E)\smile \Todd(\Sigma))\smile \frac{1}{\sqrt{\Todd(X)}} \in H^\text{even}_c(X,\Q) \,, \end{equation}
where
\[ \smile:H^p(-,\Q)\otimes H^q(-,\Q)\rightarrow H^{p+q}(-,\Q) \]
is the \emph{cup product} and $H^\text{even}_c(X,\Q)$ the \emph{compactly supported cohomology} \cite{May}.

\section{From tachyon condensation to $K$-theory}\label{Section: Overview, From tachyon condensation to K-theory}
Consider a D$p$-brane wrapping a worldvolume $\phi:\Sigma\rightarrow X$. An \emph{anti-D$p$-brane}, or a \emph{D$\overline{p}$-brane} for short, is one with opposite RR-charge. If D$p$-brane and a D$\overline{p}$-brane wrap the same worldvolume, then obviously the configuration carries no total RR-charge and there should be no conservation laws preventing it from decaying to the closed string vacuum. Unfortunately, the dynamical process through which this happens can not be fully described without a complete quantum theory of D-branes, which we do not have.

It was proposed by Sen in \cite{Sen}, that D-brane annihilation can be described by a process called \emph{tachyon condensation}. The idea is that the coincident $p-\overline{p}$ D-brane pair supports a tachyonic field (with negative mass squared), which indicates an instability in the configuration: the pair is not in a stable potential minimum. The annihilation process can then be described as \emph{rolling} of the tachyon potential to a stable local minimum. Suppose first that both the D-brane and the anti-D-brane carry topologically equivalent $U(n)$ Chan-Paton bundles vector bundles $E$. In this case, the tachyon field is a section of the trivial bundle $E\otimes E^*$ (which obviously admits global nonvanishing sections). Thus, the tachyon field can assume a constant nonzero value everywhere, corresponding to the potential minimum. The resulting energy density precisely cancels that of the original D-brane pair.

Recall that a D$p$-brane can carry also lower-dimensional RR-charge. Let us see how this is presented in the tachyon condensation scenario. Consider a configuration of a D$p$-brane with $U(1)$ Chan-Paton bundle $E$ and a D$\overline{p}$-brane with a topologically \emph{different} $U(1)$ Chan-Paton bundle $F$, both wrapping the same worldvolume $\phi:\Sigma\rightarrow X$. In this case the tachyon is a section of $E\otimes F^*$, which is nontrivial and thus may \emph{not} admit any global nonvanishing sections. This causes a possible topological obstruction, measured by the nontriviality of the tensor product bundle, for the tachyon field to assume its vacuum expectation value everywhere. As a result, the energy density of the original D-brane pair is \emph{not} cancelled everywhere. Remarkably, the resulting nonvanishing energy density can be identified with that of a lower-dimensional D-brane.

The D-brane annihilation process is more naturally described in terms of RR-charges. Recall that for a D-brane with worldvolume $\phi:\Sigma\rightarrow X$ and $U(n)$ Chan-Paton vector bundle $E$, the (rational) RR-charge was given by the formula \eqref{Eq: Rational D-brane charge, 1}. The charge on an anti-D-brane wrapping the same worldvolume and carrying a $U(n)$ Chan-Paton bundle vector bundle $F$ should be
\[ \overline{Q}_\Q(\Sigma,F) := -\phi_*(\ch(F) \smile \Todd(\Sigma)) \smile \frac{1}{\sqrt{\Todd(X)}} \,. \]
Let us denote the pair by $(E,F)$ and define
\[ \ch((E,F)) := \ch(E) - \ch(F) \in H^\text{even}(\Sigma,\Q) \,, \]
so that the total RR-charge of the pair $(E,F)$ can be expressed as
\begin{equation}\label{Eq: Rational D-brane charge, 2} Q_\Q(\Sigma,(E,F)) := Q_\Q(\Sigma,E) - \overline{Q}_\Q(\Sigma,F) = \phi_*(\ch((E,F)) \smile \Todd(\Sigma)) \smile \frac{1}{\sqrt{\Todd(X)}} \,. \end{equation}
Complex line bundles are completely classified by their first Chern class $c_1$, which is a class in $H^2(X,\Z)$. For line bundles the Chern character reduces to $\ch(E) = 1 + c_1(E)$. Thus, a coincident $p-\overline{p}$ pair with $U(1)$ Chan-Paton bundles $E$ and $F$, the RR-charge is
\[ Q_\Q(\Sigma,(E,F)) = \phi_*((c_1(E)-c_1(F)) \smile \Todd(\Sigma)) \smile \frac{1}{\sqrt{\Todd(X)}} \,, \]
which shows exactly how the lower-dimensional RR-charge of the remaining lower-dimensional D-brane depends on the Chan-Paton bundles. Since the class $c_1(E)-c_1(F)$ is of degree $2$, the charge $Q_\Q(\Sigma,(E,F))$ couples to a D$(p-2)$-brane. In other words, while there is no D$p$-brane charge left after the annihilation process, there might be D$(p-2)$-brane charge, carried by a D$(p-2)$-brane whose worldvolume is associated to the solitonic tachyon field. One can similarly consider a coincident $p-\overline{p}$ pair with $U(n)$ and $U(m)$ Chan-Paton vector bundles $E$ and $F$. The Chern characters of higher rank bundles are polynomials of higher Chern classes, which are all even degree integral cohomology classes, but the lowest degree part is always the rank of the bundle, so that $\ch(E)-\ch(F) = n-m + \ldots$, where the ellipsis denotes classes in $H_c^{2k}(X,\Q)$, where $k=1,2,\ldots$. Inserting this into the charge formula, we notice that the resulting configuration carries D$p$-brane charge if and only if $n\neq m$. It may carry lower-dimensional charge depending on the Chan-Paton bundles.

The above discussion leads directly to the following simple observation \cite{Wit1}. Consider a $p-\overline{p}$ pair with $U(n)$ and $U(m)$ Chan-Paton vector bundles $E$ and $F$, wrapping the worldvolume $\phi:\Sigma\rightarrow X$. If we add to this configuration a pair $p-\overline{p}$ with $U(k)$ Chan-Paton vector bundles $(G,G)$, the Chan-Paton bundles of the original configuration are enhanced to $U(m+k)$ and $U(n+k)$ bundles $E\oplus G$ and $F\oplus G$, respectively. Since the added pair carries no total RR-charge, one would expect the RR-charge of the original configuration to be left invariant. Indeed, since the Chern character satisfies $\ch(E\oplus G) = \ch(E)+\ch(G)$, it follows that
\begin{align*} \ch((E\oplus G,F\oplus G)) &= \ch(E\oplus G)-\ch(F\oplus G) \\ &= \ch(E)+\ch(G)-(\ch(F)+\ch(G)) \\ &=\ch((E,F)) \,. \end{align*}
Thus, we consider the configurations $(E,F)$ and $(E\oplus G,F\oplus G)$ to be \emph{physically equivalent} and write
\begin{equation}\label{Eq: Physical equivalence relation} (E,F) \sim (E\oplus G,F\oplus G) \,. \end{equation}
The equivalence class of a pair $(E,F)$ is denoted by $[E]-[F]$.

Let $\mathrm{Vect}_\C(\Sigma)$ denote the Abelian monoid of isomorphism classes of complex vector bundles over the compact manifold $\Sigma$. An isomorphism class represented by a vector bundle $E$ is denoted by $[E]$. Addition in $\mathrm{Vect}_\C(\Sigma)$ is defined by the direct sum:
\[ [E]+[F] = [E\oplus F]\,. \]
The \emph{$0$-th $K$-theory group} $K^0(\Sigma)$, often denoted simply by $K(\Sigma)$, is the free Abelian group generated by formal differences of classes of $\mathrm{Vect}_\C(\Sigma)$. Thus, the elements of $K^0(\Sigma)$ are formal differences $[E]-[F]$. Of course, $E$ and $F$ may very well have different rank. Now, comparing
\[ [E\oplus G]-[F\oplus G] = [E]+[G]-[F]-[G] = [E]-[F] \]
to the physical equivalence relation \eqref{Eq: Physical equivalence relation} reveals immediately that physically equivalent D-brane configurations correspond precisely to $K$-theory classes of $K^0(\Sigma)$, the correspondence being given by
\[ (E,F)/\sim \quad \Longleftrightarrow \quad [E]-[F]\,. \]
If $\Sigma$ is wrapped only by D-branes with Chan-Paton bundle $E$ and no anti-D-branes, the corresponding $K$-theory class is $[E]-[0]$, where $0$ denotes the zero bundle. We denote this simply by $[E]$.

For noncompact manifolds\footnote{Recall that every manifold is locally compact and, thus, admit a one-point compactification by adding a ''point at infinity''.} the definition of $K^0$ is a little different. Suppose that $X$ is noncompact and
\[ X^\infty:= X\amalg \{\infty\} \]
its one-point compactification. The $K$-theory group $K^0(X)$ is defined as the subgroup $\widetilde{K}^0(X^\infty)\subset K^0(X^\infty)$, generated by pairs of vector bundles isomorphic at the point $\{\infty\}$. In other words, for a noncompact manifold $X$, $K^0(X)$ contains only differences of vector bundles vanishing asymptotically.

\section{D-branes and $K$-theory}\label{Section: Overview, RR-charge and K-theory}
As we explained above, there is good reason to believe that worldvolume $K$-theory classifies physically equivalent configurations of D-branes, at least when there is no background $B$-field. It turns out, that the seemingly obscure rational RR-charge formula \eqref{Eq: Rational D-brane charge, 2} is extremely natural from the $K$-theoretic perspective, as was pointed out by Minasian and Moore in the famous paper \cite{MM}.

Consider a $K$-theory class $x\in K^0(\Sigma)$ corresponding to a D-brane configuration wrapping $\phi:\Sigma\rightarrow X$. Earlier we extended the Chern character to a map from pairs $(E,F)$ to $H^\text{even}(\Sigma,\Q)$. It is easy to see that it descends to a well-defined group homomorphism
\[ \ch:K^0(\Sigma)\rightarrow H^\text{even}(\Sigma,\Q) \,, \]
mapping
\[ \ch([E]-[F]) := \ch(E)-\ch(F) \,. \]
The rational RR-charge of $x$ is thus given by
\begin{equation}\label{Eq: Rational D-brane charge, 3} Q_\Q(\Sigma,x) = \phi_*(\ch(x) \smile \Todd(\Sigma)) \smile \frac{1}{\sqrt{\Todd(X)}} \,. \end{equation}
By the \emph{Atiyah and Hirzebruch version of the Grothendieck-Riemann-Roch theorem} \cite{AH1,AH2}, if $\Sigma$ is $\spinc$, there exists a \emph{Gysin ''wrong way'' homomorphism}
\[ \phi_!:K^0(\Sigma)\rightarrow K^0(X)\,, \]
satisfying
\begin{equation}\label{Eq: A-H-G-R-R, 1} \ch(\phi_!(x))\smile \Todd(X) = \phi_*(\ch(x)\smile \Todd(\Sigma)) \,. \end{equation}
By applying \eqref{Eq: A-H-G-R-R, 1} to \eqref{Eq: Rational D-brane charge, 3}, we obtain the \emph{Minasian-Moore formula} \cite{MM}
\begin{equation}\label{Eq: Overview, Minasian-Moore formula}  Q_\Q(\Sigma,x) = \ch(\phi_!(x))\smile \sqrt{\Todd(X)} \in H^\text{even}_c(X,\Q) \,. \end{equation}
Dirac quantization requires the RR-charge to lift into some \emph{generalized integral cohomology theory} \cite{Whi,Bro1,Bro2}, which in this case seems to be $K$-theory \cite{Fre1,Fre2}. Thus, the quantized (integral) RR-charge should be a $K$-theory class of spacetime. From the point of view of \eqref{Eq: Overview, Minasian-Moore formula}, a natural choice would be
\[ Q_\Z(\Sigma,x) := \phi_!(x)\,. \]
The quantized charge differs fundamentally from the corresponding classical charge in that the quantized charge may contain torsion. Real and rational cohomology are obviously insensitive to torsion due to $\Q\otimes_\Z \Z_n = 0$.

We have seen how D-branes can be classified by their RR-charge, which is a class in the $K$-theory of spacetime. This is not a very concrete description, though, considering that D-branes are associated to worldvolume manifolds and $U(n)$ vector bundles\footnote{We are still assuming that the $PU(n)$ Chan-Paton bundles lift to $U(n)$ bundles.} living on them. It was argued by Periwal in \cite{Per}, that a more ''natural'' description may be given using \emph{$K$-homology}, the \emph{dual homology theory of $K$-theory}, of spacetime. It admits a purely geometric definition \cite{Jak1,Jak2,BD,BHS}, encoding precisely by the information we have used to define a D-brane: the $\spinc$ worldvolume manifold $\Sigma$, the continuous map $\phi:\Sigma\rightarrow X$ and the $U(n)$ Chan-Paton vector bundle $E\rightarrow \Sigma$. It turns out, that there is a \emph{$K$-theoretic Poincar\'e duality isomorphism} between $K$-theory and $K$-homology, which makes the two classification schemes equally valid. They merely emphasize different physical aspect of D-branes.

\section{Cancelling the Freed-Witten anomaly}
Recall that the $B$-field is a locally defined $2$-form field and that its derivative $H:=dB$, known as the $H$-flux, is a globally defined $3$-form field. By Dirac quantization, the de Rham cohomology class $[H]\in H^3_\text{dR}(X)$ lifts to a class $[H]\in H^3(X,\Z)$, which may also have a torsion part, unlike its image in $H^3_\text{dR}(X)$. The correct mathematical framework for Dirac quantized fields is \emph{differential cohomology} \cite{Fre1,Fre2,HS}. For example, RR-fields are described by something called \emph{differential $K$-theory} \cite{Val}. As was already implicitly stated above, the integral cohomology theory associated to the $B$-field seems to be the ordinary integral cohomology $H^3(X,\Z)$. So, to properly understand the $B$-field (and the $H$-flux), we need to know what the associated differential cohomology theory is. It turns out to be something called \emph{Deligne cohomology} \cite{Bry,HS}, whose cocycles are triples $(\underline{g}, -\underline{\Lambda},\underline{B})$, where $\underline{g}$ is a $\underline{U(1)}$-valued \v Cech $2$-cocycle, $\underline{\Lambda}$ a $1$-form-valued \v Cech $1$-cochain and $\underline{B}$ a $2$-form-valued $0$-cochain on $X$, satisfying the following relations:
\[ d\log \underline{g} = \delta(\underline{\Lambda})\,, \qquad d\underline{\Lambda} + \delta(\underline{B}) = 0 \,. \]
The second equation implies $\delta(d\underline{B})=d\delta(\underline{B})=0$, which means that $d\underline{B}$ lifts to a global $3$-form $H$.

The notation above was chosen so as to make the correspondence with D-brane physics obvious: $\underline{B}$ corresponds to the physical $B$-field, $H$ to the physical $H$-flux and $\underline{\Lambda}$ to that in \eqref{Eq: A-field transformation, phys}, with the minor modification that the quantities associated to the Deligne cocycle are the corresponding physical quantities multiplied by the complex number $i$, which is left from taking $d\log$ of $\underline{g}$. Thus, from now on, we shall use the notation
\[ B := iB_\text{phys}\,, \quad H := iH_\text{phys} \,, \quad \Lambda := i\Lambda_\text{phys} \,, \quad A := iA_\text{phys} \,. \]
This changes the physical $A$-field transformation equation \eqref{Eq: A-field transformation, phys} into
\begin{equation} \label{Eq: A-field transformation} A_\beta = g_{\alpha\beta}^{-1}A_\alpha g_{\alpha\beta} + g_{\alpha\beta}^{-1}\,dg_{\alpha\beta} - \Lambda_{\alpha\beta} \,. \end{equation}
It would be convenient to have a geometric interpretation of the $B$-field, just like ordinary gauge fields are geometrically represented by a connection and curvature on a principal bundle. Such a geometric model for the $B$-field, or rather the Deligne cocycle, is given by the theory of \emph{bundle gerbes with connection and curving} \cite{Mur,BCMMS,MS1}.

The relation between the $A$-field and the $B$-field can be found by studying the exponential of the open string worldsheet action. It contains the terms
\[ \exp(iS_\text{open})\supset \pfaff(\Dirac_\xi)\cdot \exp\left(\int_M \xi^*B \right)\cdot \Tr\exp\left(\oint_{\partial M} \xi^*A \right) \,, \]
where $\xi:M\rightarrow X$ is the embedding of the open string worldsheet and $\pfaff(\Dirac_\xi)$ the \emph{Pfaffian}, the square root of the determinant, of the worldsheet Dirac operator \cite{FW}. It turns out, that $\pfaff(\Dirac_\xi)$ is not a well-defined function on the space of embeddings $M\rightarrow X$, but rather a section of a certain line bundle. There is also the other problem we mentioned earlier, namely, that the integral of the $B$-field is not well-defined due to the boundary of $M$. 
This ill-definedness of the open string path integral is known as the Freed-Witten anomaly \cite{FW}. To cancel this anomaly, the gauge transformation of the $A$-field must be such that the trace of its holonomy around $\partial M$ is ill-defined in just a particular way to cancel the ill-definedness of the other two terms. It was shown in \cite{FW,Kap,CJM} that this happens if the $A$-field transforms as in \eqref{Eq: A-field transformation} and if the cohomology equation
\begin{equation} \label{Eq: FW anomaly cancellation} [H]|_\Sigma = W_3(\Sigma) + \beta([\underline{\zeta}]) \end{equation}
holds. Geometrically the $A$-field can be interpreted, not as a connection on a bundle, but a connection on an object called a \emph{bundle gerbe module} \cite{BCMMS,CJM}. In the main body of the text we shall explain what bundle gerbe modules and their connections are and go carefully through the calculations which show how the anomaly is cancelled when the above conditions hold.

Let us now explain the important equation \eqref{Eq: FW anomaly cancellation}. Let $\phi:\Sigma\rightarrow X$ denote the D-brane worldvolume. First, $[H]|_\Sigma:= \phi^*[H]\in H^3(\Sigma,\Z)$ is the restriction of the $H$-flux $[H]\in H^3(X,\Z)$ onto the worldvolume. Second, the characteristic class $W_3(\Sigma)\in H^3(\Sigma,\Z)$ is the so-called \emph{third integral Stiefel-Whitney class}, which is the obstruction for the existence of a $\spinc$ structure on $\Sigma$\footnote{Again, a reader unfamiliar with $\spinc$ structures can regard this simply as a technical condition, slightly stronger than orientability.}. Finally, $\beta([\underline{\zeta}])\in H^3(\Sigma,\Z)$ is the obstruction for lifting the Chan-Paton bundle from a $PU(n)$ bundle to a $U(n)$ bundle. Recall now, that while developing the $K$-theoretic classification of D-branes, we assumed $[H]|_\Sigma=0$, $W_3(\Sigma)=0$\footnote{$\Sigma$ had to be $\spinc$ for the Riemann-Roch theorem.} and that the Chan-Paton bundle was a $U(n)$ vector bundle for its $K$-theory class to be defined. Obviously \eqref{Eq: FW anomaly cancellation} is satisfied under those conditions.

What happens if, say, $[H]|_\Sigma \neq 0$ and there is a single D-brane wrapping $\Sigma$? The Chan-Paton bundle is then automatically a $U(1)$ bundle and thus $\beta([\underline{\zeta}])=0$. But now the anomaly cancellation condition implies that such a configuration is not physically possible unless $\Sigma$ satisfies $W_3(\Sigma)=[H]|_\Sigma$\footnote{Since $W_3(\Sigma)$ is $2$-torsion, also $[H]|_\Sigma$ must be $2$-torsion. Hence, the image of $[H]$ in de Rham cohomology must necessarily be zero.}. In this case we can form a $K$-theory class $[E]\in K^0(\Sigma)$ from the Chan-Paton bundle vector bundle $E$, but it can not be pushed into $K^0(X)$.

Consider next $n$ D-branes wrapping $\Sigma$ with a $PU(n)$ Chan-Paton bundle. Suppose that the obstruction class for the $U(n)$ lift is $\beta([\underline{\zeta}])\neq 0$\footnote{It is easy to see that $\beta([\underline{\zeta}])$ is necessarily $n$-torsion, implying that $[H]|_\Sigma$ must also be $n$-torsion for the anomaly to be cancelled.}. In this case, not even a worldvolume $K$-theory class can be constructed from the (projective) Chan-Paton bundle. 

Finally, consider the case of nontorsion $[H]|_\Sigma$. In this case something quite radical must happen for the anomaly to cancel, namely, the torsion degree of $\beta([\underline{\zeta}])$ should somehow be taken to infinity. From the physical point of view, this corresponds to taking the number of D-branes wrapping $\Sigma$ to infinity in some sense \cite{Wit1,Wit2}. The precise mathematical interpretation is to allow the Chan-Paton bundle to be a principal $PU(\Hilb)$ bundle \cite{BM,AS1}, where $\Hilb$ is the standard infinite-dimensional separable Hilbert space and $PU(\Hilb)$ the quotient of $U(\Hilb)$ by its center:
\[ PU(\Hilb) := U(\Hilb)/U(1) \,. \]
Again, the $K$-theoretic classification fails miserably, since we can not even construct the worldvolume $K$-theory class.

\section{Classifying twisted D-branes}
We have seen how the $K$-theoretic classification fails when $\Sigma$ is not $\spinc$ or when the Chan-Paton bundle can not be lifted to a $U(n)$ bundle. The key to solving these problems is to use \emph{twisted $K$-theory} \cite{Ros,Bla,AS1,BM,CW,BCMMS}, a modification of $K$-theory, where $U(n)$ vector bundles are replaced with $PU(n)$ bundles, whose transition functions $\{g_{\alpha\beta}\}$ satisfy a \emph{$\zeta$-twisted cocycle condition}
\[ g_{\beta\gamma}g_{\alpha\gamma}^{-1}g_{\alpha\beta} = \zeta_{\alpha\beta\gamma} \,. \]
The direct sum of two such \emph{$\zeta$-twisted bundles} is clearly again $\zeta$-twisted. Thus, a $\zeta$-twisted $PU(n)$ vector bundle defines a class in the \emph{$\zeta$-twisted $K$-theory} $K^0(\Sigma,\zeta)$. The group $K^0(\Sigma,\zeta)$ depends only on the \v Cech cohomology class $[\zeta]\in H^2(\Sigma,\underline{U(1)})\cong H^3(\Sigma,\Z)$. Therefore, the natural object to twist $K$-theory with, is not a cocycle but a cohomology class. More generally, $K$-theory can be twisted by any class in $H^3(\Sigma,\Z)$, torsion or nontorsion. For $\sigma\in H^3(\Sigma,\Z)$, we denote $\sigma$-twisted $K$-theory by $K^0(\Sigma,\sigma)$. When the twisting class is trivial, twisted $K$-theory reduces to ordinary ''untwisted'' $K$-theory. The definition can be given in several different ways, each useful for certain purposes. The most convenient for us being that of Rosenberg \cite{Ros}. 

The Freed-Witten anomaly cancellation condition yields\footnote{$W_3(\Sigma)$ is $2$-torsion.}
\[ \beta([\underline{\zeta}]) = [H]|_\Sigma + W_3(\Sigma) \,. \]
This means that the $PU(n)$ (or possibly even $PU(\Hilb)$) Chan-Paton vector bundle defines a class in $K^0(\Sigma,[H]|_\Sigma+W_3(\Sigma))$. For a continuous map $\phi:\Sigma\rightarrow X$, there exists a \emph{twisted Gysin homomorphism} \cite{CW}
\[ \phi_!:K^0(\Sigma,[H]|_\Sigma+W_3(\Sigma)) \rightarrow K^0(X,[H]) \,. \]
This indicates that, in the presence of a nontrivial $H$-flux, quantized D-brane charge should take values in the $[H]$-twisted $K$-theory group of spacetime. In particular, if $\beta([\underline{\zeta}])=0$ but $[H]|_\Sigma = W_3(\Sigma)\neq 0$, the class of the Chan-Paton bundle is in (untwisted) worldvolume $K$-theory, but becomes $[H]$-twisted when pushed into spacetime.

A twisted version of $K$-homology and its relation to D-brane theory are developed in \cite{Wan,Sza1}. It correctly extends the $K$-homological approach to nontrivial $B$-field backgrounds and non-$\spinc$ worldvolumes.

D-branes satisfying a nontrivial Freed-Witten anomaly cancellation condition are referred to as \emph{twisted D-branes}. We already explained how (Dirac-quantized physical equivalence classes of) configurations of twisted D-branes can be classified by twisted $K$-theory or twisted $K$-homology, instead of ordinary untwisted $K$-theory or $K$-homology. In the untwisted case we had a Chern character from $K$-theory to cohomology, which was used to construct the classical approximation of the quantized RR-charge. In the twisted case there is a \emph{twisted Chern character} from twisted $K$-theory to \emph{twisted (de Rham) cohomology} \cite{MS2,MS3,BCMMS,AS2}. Let $\zeta \in H^3_\text{dR}(X)$ be a closed de Rham $3$-form. The \emph{$\zeta$-twisted de Rham differential} is
\[ d_\zeta := d - \zeta \,, \]
which operates on differential forms $\Omega^\text{even/odd}(X)$ by $d_\zeta(\omega) = d\omega + \zeta\wedge \omega$\footnote{Clearly, $d_\zeta$ maps $\Omega^\text{even}(X)$ to $\Omega^\text{odd}(X)$ and \emph{vice versa}.}. A direct calculation shows that $d_\zeta^2 = 0$. The (periodized) $\zeta$-twisted cohomology groups are defined by
\[ H^\bullet(X,\zeta) := \ker(\delta_\zeta)/\im(\delta_\zeta) \,, \]
where $\bullet=0,1$. Moreover, it turns out that 
\[ H^\bullet(X,\zeta)\cong H^\bullet(X,\zeta+d\eta) \]
for any $2$-form $\eta$, where, on the level of cocycles, the isomorphism is given by
\[ \omega \mapsto e^\eta\wedge \omega \,. \]
Thus, de Rham cohomology is most naturally twisted by the cohomology class $[\zeta]\in H^3_\text{dR}(X)$, not the cocycle itself. More generally, if $\sigma\in H^3(X,\Z)$ and $\sigma_\text{dR}\in H^3_\text{dR}(X)$ is image under the natural inclusion, \emph{$\sigma$-twisted (de Rham) cohomology} is defined by
\[ H^\bullet(X,\sigma) := H^\bullet(X,\sigma_\text{dR}) \,, \]
where $\sigma_\text{dR}$ is the image of $\sigma$ in $H^3_\text{dR}(X)$. Now, given a twisting class $\sigma\in H^3(X,\Z)$, there exists a $\sigma$-twisted Chern character homomorphism
\[ \ch_\sigma:K^0(X,\sigma)\rightarrow H^\text{even}(X,\sigma) \,. \]
In our physical example $\sigma=[H]$. The twisted Chern character maps the push-forward of the twisted $K$-theory class of the Chan-Paton vector bundle to (compactly supported) twisted cohomology $H^0_c(X,[H])$. If $[H]$ is pure torsion, then $[H]_\text{dR}=0$ and $[H]$-twisted cohomology $H^0(X,[H])$ becomes isomorphic to untwisted de Rham cohomology $H^\text{even}_\text{dR}(X)$. Thus, when $[H]$ is pure torsion, the charge is again a class in untwisted cohomology, given by a generalization of \eqref{Eq: Overview, Minasian-Moore formula}, where $\ch$ is simply replaced with $\ch_{[H]}$\footnote{Remark, that $\omega \wedge \mu$ is $d_{\sigma_\text{dR}}$-closed for any $d_{\sigma_\text{dR}}$-closed $\omega$ and $d$-closed $\mu\in \Omega^\bullet(X)$.}.

\section{Operator $K$-theory}
Twisted $K$-theory can be defined in several ways. The most convenient for us is the definitions through \emph{operator $K$-theory} (or \emph{$K$-theory of $C^*$-algebras}) \cite{Bla}, because it automatically opens up the possibility to generalize the theory of D-branes and their RR-charges to D-branes whose worldvolumes and the spacetime are \emph{noncommutative spaces} \cite{Con,GVF,Mad,Lan2}.

Let $A$ be a $C^*$-algebra \cite{GVF,Ped}, that is, a Banach algebra with involution, satisfying the $C^*$-norm property
\[ ||a^*a|| = ||a||^2 \,, \]
for all $a\in A$. $C^*$-algebras have extremely nice spectral properties and their $K$-theory is particularly well-behaving. Morphisms between $C^*$-algebras respect the involution and are called $*$-homomorphisms. $C^*$-algebras are also at the basis of noncommutative topology and geometry: any locally compact Hausdorff space $X$ is completely characterized by the commutative $C^*$-algebra of its  continuous complex-valued functions vanishing at infinity, $C_0(X)$\footnote{If the space $X$ is compact, the algebra $C_0(X)$ coincides with $C(X)$. Thus, for a compact space, the corresponding $C^*$-algebra is unital and for a noncompact it is nonunital.}. Many important topological and geometric concepts can be expressed purely algebraically, for example vector bundles, differential calculus, $\spinc$ structures, Dirac operators and de Rham (co)homology, after which the $C^*$-algebra can be replaced by any, possibly noncommutative, $C^*$-algebra. Thus, noncommutative spaces \emph{are} $C^*$-algebras. Noncommutative ($\spin$) geometry is a vast subject and treating it in any reasonable detail would take easily a thesis or two. Nevertheless, even readers unfamiliar with the intricate details can easily appreciate the generalization of D-brane theory to the noncommutative ($C^*$-algebraic) case.

\emph{Operator $K$-theory} is a noncommutative generalization of $K$-theory of vector bundles, defined for a broad class of (topological) algebras, including the so-called \emph{local $C^*$-algebras} \cite{Bla} and, naturally, $C^*$-algebras. The most important examples of local $C^*$-algebras are the smooth function algebras of manifolds. For any manifold $X$, there are isomorphisms
\[ K_0(C_0^\infty(X)) \cong K_0(C_0(X))\cong K^0(X) \,, \]
where $K_0$ denotes operator $K$-theory.


Twisted $K$-theory can now be defined using operator $K$-theory as follows \cite{Ros}. Let $\sigma \in H^3(X,\Z)$ be a twisting cohomology class. There is bijective correspondence between $H^3(X,\Z)$ and isomorphism classes of principal $PU(\Hilb)$ bundles over $X$\footnote{In other words, the classifying space $BPU(\Hilb)$ is a model for the Eilenberg-MacLane space $K(\Z,3)$ \cite{AS1}.}. The cohomology class corresponding to a principal $PU(\Hilb)$ bundle is called its \emph{Dixmier-Douady class}. Let $P$ be a principal $PU(\Hilb)$ bundle with Dixmier-Douady class $\sigma$. We denote the $C^*$-algebra of compact operators on the standard infinite-dimensional separable Hilbert space by $\mathcal{K}$. Then, since $PU(\Hilb)\cong \Aut(\mathcal{K})$, we can form the associated $PU(\Hilb)$ vector bundle over $X$ with fibres isomorphic to $\mathcal{K}$:
\[ \E_{\sigma}:= P_\sigma\times_{PU(\Hilb)} \mathcal{K} \rightarrow X \,. \]
Rosenberg defined $\sigma$-twisted $K$-theory as
\[ K^0(X,\sigma) := K_0(C_0(X,\E_\sigma)) \,. \]
A principal $PU(\Hilb)$ bundle with trivial Dixmier-Douady class is the projectivization of a principal $U(\Hilb)$ bundle \cite{AS1}. The famous theorem of Kuiper \cite{Kui} states that $U(\Hilb)$ is contractible (in strong operator topology). Thus, the fibres are all contractible, which means that the principal bundle must trivial. The associated bundle $\E_0\rightarrow X$ is simply the trivial bundle $X\times \mathcal{K}$. The twisted function algebra is then
\[ C_0(X,\E_0) = C_0(X,\mathcal{K}) \cong C_0(X)\otimes \mathcal{K}\,. \]
$K$-theory is \emph{stable}, or \emph{Morita invariant}, in the sense that
\[ K_0(A\otimes \mathcal{K}) \cong K_0(A) \,, \]
for any $C^*$-algebra $A$. This shows that twisted $K$-theory reduces to ordinary untwisted $K$-theory when the twisting class vanishes:
\[ K^0(X,0)=K_0(C_0(X,\E_0))=K_0(C_0(X)\otimes \mathcal{K})\cong K_0(C_0(X)) \cong K^0(X) \,, \]
as would be expected.

There exists also a \emph{noncommutative Chern character}, which coincides with the twisted Chern character mentioned earlier, when the $C^*$-algebra is a twisted function algebra. In noncommutative geometry (periodized) de Rham cohomology of a smooth manifold $X$ is given by \emph{periodic cyclic homology} $\HP_\bullet$ \cite{Con,GVF} of the smooth function algebra $C_0^\infty(X)$:
\[ H^\bullet_\text{dR}(X) \cong \HP_\bullet(C_0^\infty(X)) \,,\quad \bullet=\text{even/odd} \,. \]
Similarly, (periodized) de Rham homology is given \emph{periodic cyclic cohomology} $\HP^\bullet(C_0^\infty(X))$. Let now $\mathcal{A}$ be any local $C^*$-algebra. The noncommutative version of the Chern character is the \emph{Chern-Connes character} \cite{GVF}:
\[ \ch:K_0(\mathcal{A})\rightarrow \HP_0(\mathcal{A}) \,. \]
We denote by $A$ the $C^*$-algebra completition of $\mathcal{A}$. For example $A=C_0(X)$ for $\mathcal{A}=C_0^\infty(X)$ and $A=C_0(P\times_{PU(\Hilb)} \mathcal{K})$ for $\mathcal{A}=C_0^\infty(P\times_{PU(\Hilb)} \mathcal{L}^1)$, where $P$ is a principal $PU(\Hilb)$ bundle with Dixmier-Douady class $\sigma$ and $\mathcal{L}^1\subset \mathcal{K}$ is the ideal of trace class operators. Finally, $\sigma$-twisted cohomology is isomorphic to periodic cyclic homology of the twisted smooth function algebra \cite{MS3}. Thus, the diagram
\[ \xymatrix{ K^0(X,\sigma) \ar[r]^-{\text{def.}} \ar[d]_{\ch_{\sigma}} & K_0(A) \ar[r]^\cong & K_0(\mathcal{A}) \ar[d]^{\ch} \\ H^0(X,\sigma) \ar[rr]^\cong & & \HP_0(\mathcal{A}) } \]
commutes.

There exists also noncommutative generalization of $K$-homology, denoted by $K^0(A)$, for a $C^*$-algebra $A$ \cite{Kas3,HR}. Twisted $K$-homology of a topological space can be defined as $K$-homology of the twisted function algebra.

\section{Noncommutative generalizations}
The $C^*$-algebraic definitions given above open up hopes for generalizing the Minasian-Moore formula \eqref{Eq: Overview, Minasian-Moore formula} to more general \emph{noncommutative D-branes} in a noncommutative spacetime \cite{BMRS1,BMRS2,Sza2}. Let the $C^*$-algebra $B$ be a noncommutative D-brane worldvolume and $A$ a noncommutative spacetime. At the level of algebras the map $\phi:\Sigma\rightarrow X$ is replaced by a $*$-homomorphism $\phi:A\rightarrow B$. The quantized RR-charge of the noncommutative D-brane $B$ is an operator $K$-theory class of $B$, pulled back into the $K$-theory of $A$ by a noncommutative version of the Gysin homomorphism. The existence of the Gysin homomorphism can be thought of as a noncommutative analogue of the Freed-Witten anomaly cancellation condition. To be able to understand noncommutative Gysin homomorphisms and the relation between $K$-theory and $K$-homology (of $C^*$-algebras), we need to introduce \emph{Kasparov's $KK$-theory}, a bivariant unification of $K$-theory and $K$-homology \cite{Kas1,Kas2,Bla,Hig1,Hig2,KT}.

$KK$-theory can be defined both axiomatically and explicitly\footnote{In fact, several different explicit constructions for it are known.}. For now, it is enough for the reader to think of $KK_\bullet$, $\bullet=0,1$ as bivariant bifunctors from pairs of separable $C^*$-algebras to the category of Abelian groups, contravariant in the first variable and covariant in the second. $KK$-theory unifies $K$-theory and $K$-homology in the following sense:
\begin{align*} K_0(A) &\cong KK_0(\C,A) \qquad \text{($K$-theory)} \\ K^0(A) &\cong KK_0(A,\C) \qquad \text{($K$-homology)} \end{align*}
Possibly the most important property of $KK$-theory is the existence of a powerful product structure, the \emph{Kasparov product}: for separable $C^*$-algebras $A_1,A_2,B_1,B_2,D$, the Kasparov product is an associative, bilinear map
\[ \otimes_D:KK_i(A_1,B_1\otimes D)\times KK_j(D\otimes A_2,B_2) \rightarrow KK_{i+j}(A_1\otimes A_2,B_1\otimes B_2) \,. \]
It is convenient to restrict the discussion to \emph{nuclear} $C^*$-algebras, for which there is no ambiguity in the topology of the tensor product $C^*$-algebras \cite{Lan1}.

We are interested in pairs of separable $C^*$-algebras $(A,\widetilde{A})$, for which there exists a \emph{fundamental class} $\Delta\in KK_d(A\otimes \widetilde{A},\C)$, ''invertible'', in a certain sense, with respect to the Kasparov product. The pair $(A,\widetilde{A})$ is called a \emph{Poincar\'e duality (PD) pair}. If a $C^*$-algebra $B$ and its opposite algebra $B^\circ$ form a PD pair $(B,B^\circ)$, we say that $B$ is a \emph{Poincar\'e duality (PD) algebra}. Poincar\'e duality for a PD pair is defined by taking ''cap product'' with the fundamental class:
\[ KK_\bullet(\C,A) \xrightarrow{(-)\otimes_A \Delta} KK_{\bullet+d}(\widetilde{A},\C) \,. \]
The fundamental class is not unique. In fact, the space of fundamental classes is isomorphic to the group of invertible classes in the ring $KK_0(A,A)\cong KK_0(B,B)$. For example, $(C_0(X),C_0(X,\Cliff(T^*X))$ is a PD pair. Moreover, if $X$ is $\spinc$, then $C(X)$ is a PD algebra. 

Next we need to introduce a \emph{bivariant cyclic homology theory}, which is compatible with $KK$-theory. It is clear that periodic cyclic theory can not be used here, since it does not work well for $C^*$-algebras and (Kasparov's) $KK$-theory works \emph{only} for $C^*$-algebras. Instead, we need to use the \emph{bivariant local cyclic homology} of Puschnigg \cite{Pus1,Pus2}, denoted by $\HL_\bullet(-,-)$, $\bullet=1,2$, which works for both $C^*$-algebras and for local $C^*$-algebras. Moreover, it satisfies
\[ \HL_\bullet(\C,C_0(X,\E_\sigma) \cong \HL_\bullet(\C,C_0^\infty(X,\E_\sigma)) \cong \HP_\bullet(C_0^\infty(X,\E_\sigma)) \cong H^\bullet_c(X,\sigma) \]
and an analogous result in homology. There exists a \emph{bivariant Chern character}, a homomorphism
\[ \ch:KK_\bullet(A,B) \rightarrow \HL_\bullet(A,B) \,. \]
Restricting to $KK_0(\C,C_0(X,\E_\sigma)$ yields the usual twisted Chern character. Given a pair $(A,\widetilde{A})$ of separable $C^*$-algebras, for which there exists a certain \emph{cyclic fundamental class} $\Xi\in \HL_d(A\otimes \widetilde{A},\C)$, we have an analogous notion of Poincar\'e duality in $\HL$-theory as we had in $KK$-theory. The pair $(A,\widetilde{A})$ is called a \emph{cyclic Poincar\'e duality (C-PD) pair}. Again, the cyclic fundamental class of a C-PD pair is by no means unique. Any PD pair $(A,\widetilde{A})$, with fundamental class $\Delta$, is automatically also a C-PD pair, since $\ch(\Delta)$ is a cyclic fundamental class. However, in general there is no reason to assume that $\Xi = \ch(\Delta)$. The difference between $\Xi$ and $\ch(\Delta)$ is the \emph{noncommutative $\Todd$ class} $\Todd(A)\in \HL_0(A,A)$. A noncommutative Minasian-Moore formula would require us to take the square root of the $\Todd$ class, which is not always possible. We assume for now that $\sqrt{\Todd(A)}\in \HL_0(A,A)$ exists for the spacetime $C^*$-algebra $A$.

The final ingredient needed for a noncommutative Minasian-Moore formula is a generalization of the Gysin map. A $*$-homomorphism $f:A\rightarrow B$ is said to be \emph{$K$-oriented}, if there exists a class $f!\in KK_d(B,A)$, satisfying certain functorial properties. It can be used to define the \emph{noncommutative Gysin homomorphisms}
\[ f_! := (-)\otimes_B f!:KK_\bullet(\C,B)\rightarrow K_{\bullet+d}(\C,A) \]
and
\[ f^! := f!\otimes_A (-):K_\bullet(A,\C)\rightarrow K_{\bullet+d}(B,\C) \,. \]
One can construct similar Gysin maps in $\HL$-theory, if there exists a class $f*\in \HL_d(B,A)$ with the same functorial properties. Given a suitable noncommutative D-brane $B$, a suitable noncommutative spacetime $A$ and a $*$-homomorphism $\phi:A\rightarrow B$, the Minasian-Moore formula \eqref{Eq: Overview, Minasian-Moore formula} is generalized by
\[ Q_\C(B,\xi) := \ch(\phi_!(\xi))\otimes_A \sqrt{\Todd(A)} \in \HL_\bullet(\C,A) \,, \]
where $\xi\in K_\bullet(B)$ is a noncommutative Chan-Paton bundle.

\chapter{D-Branes And $K$-Theory}\label{Chapter: D-Branes And $K$-Theory}
In this chapter we shall explore the connection between $K$-theory and D-branes. The spacetime $X$ was assumed to be a smooth oriented Riemannian $\spin$ manifold and the D-brane worldvolume $\Sigma$ an even-dimensional oriented manifold, together with a continuous map $\phi:\Sigma\rightarrow X$. For now, we assume that $\Sigma$ is $\spinc$, that the quantized $H$-flux $[H]\in H^3(X,\Z)$ vanishes and that $\Sigma$ carries a $U(n)$ Chan-Paton bundle. These assumptions will be dropped (much) later, after we have discussed the Freed-Witten anomaly.

Before delving into $K$-theory, we shall take some time to introduce a few mathematical concepts which will be of fundamental importance throughout the text: characteristic classes and $\spin$/$\spinc$ structures.

\section{Universal bundles and characteristic classes}
Throughout the text we shall make extensive use of so-called \emph{characteristic classes}, for example the \emph{Chern character}, the \emph{Chern classes}, the \emph{Atiyah-Hirzebruch class $\widehat{A}$}, the \emph{Stiefel-Whitney classes} and the \emph{Todd class}. It is, therefore, worthwhile to devote a few pages to a discussion of the basic ideas and definitions. More extensive treatments can be found in numerous books, for example \cite{MS4,BT,May,LM,Hus}, which are our main references.

Characteristic classes of a vector bundle $E\rightarrow X$, real or complex, are cohomology classes of $X$, which measure the nontriviality of $E$. There are important differences between real and complex vector bundles, as we shall see. Let us start by introducing \emph{classifying spaces}.

\begin{definition}\label{Definition: Classifying space} A \emph{classifying space} for a Lie group $G$ is a topological space $BG$ equipped with a principal $G$ bundle $EG\rightarrow BG$, the \emph{universal bundle}, whose total space $EG$ is contractible. The space of isomorphism classes of principal $G$ bundles over any manifold $X$ is isomorphic to $[X,BG]$, the space of homotopy classes of continuous maps $X\rightarrow BG$. More precisely, every principal $G$ bundle over $X$ is isomorphic to a pullback of $EG\rightarrow BG$ by a \emph{classifying map} $f:X\rightarrow BG$. It is a standard result in bundle theory that the isomorphism class of the pullback bundle $f^{-1}EG\rightarrow X$ depends only on the homotopy class of $f$.
\end{definition}
\begin{theorem}\label{Theorem: Classifying spaces exist}
Classifying spaces, as defined above, always exist.
\end{theorem}
\begin{proof}
Category theoretically the theorem is motivated as follows. The cofunctor
\[ \mathrm{Prin}_G(-):\mathbf{hCW}\rightarrow \mathbf{Set} \]
from the homotopy category of pointed CW-complexes to the category of sets, sending a CW-complex $X$ to the set $\mathrm{Prin}_G(X)$ of isomorphism classes of principal $G$ bundles over $X$, is representable by the \emph{Brown's representability theorem} \cite{Bro1,Bro2}. An application of the \emph{Yoneda lemma} \cite{Mac} implies the existence of a \emph{universal element} for the functor $\mathrm{Prin}_G(-)$, which establishes the representation. This universal element is precisely the universal bundle $EG\rightarrow BG$. The Yoneda lemma yields the bijective correspondence $\mathrm{Prin}_G(X)\cong [X,BG]$ and the pullback construction of definition \ref{Definition: Classifying space}.

The categorial approach is not very convenient for practical calculations. Milnor \cite{Mil} discovered the following explicit construction. The \emph{join} of two topological spaces $X$ and $Y$, denoted by $X\star Y$, is the quotient space $X\times Y\times [0,1]/\sim$, where $\sim$ collapses $\{x\}\times Y\times \{0\}$ and $X\times \{y\} \times \{1\}$, for each $x\in X$, $y\in Y$, to points. We denote the $(n+1)$-fold join $G\star \ldots \star G$ by $EG_n$. Its points are specified by $n$-tuples $(t_1,\ldots,t_n)\in [0,1]\times\ldots \times [0,1]$, with $\sum_i t_i = 1$, together with elements $g_i\in G$, $i=1\ldots n$, for which $t_i\neq 0$. Denoting the points by $(t_1 g_1,\ldots,t_n g_n)$ makes the construction more clear. The topology of $EG_n$ is chosen to be the strongest topology, such that the coordinate functions
\[ t_i:EG_n \rightarrow [0,1] \quad \text{and} \quad g_i:t_i^{-1}(0,1]\rightarrow G \]
are continuous. The \emph{infinite join} $G\star G\star\ldots$ is the direct limit
\[ EG := \varinjlim_n E_n \]
of topological spaces.

One can show that each $EG_n$ is $(n-1)$-connected, that is, $\pi_k(EG_n)=0$ for all $0\le k \le n-1$. It follows that $EG$ is \emph{weakly contractible} \cite{May}:
\[ \pi_k(EG) = \pi_k(\varinjlim_n EG_n) \cong \varinjlim_n \pi_k(EG_n) = 0 \,, \]
which implies, by \emph{Whitehead's theorem} \cite{May}, that $EG$ is contractible.

The universal bundle is constructed as follows. There is a continuous right action of $G$ on $EG$, given by
\[ (t_1 g_1,\ldots,t_n g_n)g := (t_1 g_1 g,\ldots t_n g_n g)\,. \]
The base space $BG$ is the space of orbits $EG/G$ and the bundle projection the canonical projection $EG\rightarrow EG/G$. Local trivializations and transition functions are given in \cite{Mil}, along with the required proofs for continuity. An explicit verification of the universal property is given in \cite{Hus}, but it turns out that contractibility is already enough to prove universality \cite{Ste1}\footnote{Compare this to the theory of Eilenberg-MacLane spaces, where it is enough to find a space with the correct homotopy properties to immediately conclude that it provides a representation of the cohomology functor \cite{May}. In particular, the $n$-th Eilenberg-MacLane space must be $(n-1)$-connected.}.
\end{proof}
Let us now focus on the cases $G=U(n)$ and $G=O(n)$. Let $K$ denote either $\C$ or $\R$ and $\mathrm{Gr}_n(K^N)$ the set of $n$-dimensional $K$-linear subspaces, that is, $n$-planes passing through the origin of $K^N$. A little thought reveals that
\begin{eqnarray*} \mathrm{Gr}_n(\C^N) \cong \frac{U(N)}{U(N-n)\times U(n)} \\ \mathrm{Gr}_n(\R^N) \cong \frac{O(N)}{O(N-n)\times O(n)} \,, \end{eqnarray*}
because $U(n)$ (resp. $O(n)$) acts transitively on $\mathrm{Gr}_n(\C^N)$ (resp. $\mathrm{Gr}_n(\R^N)$) and a particular $n$-plane is left invariant precisely under the subgroups of rotations $U(n)$ (resp. $O(n)$) of the plane itself and rotations $U(N-r)$ (resp. $O(N-r)$) of the complementary subspace. The natural inclusions
\[ \mathrm{Gr}_n(K^N)\subset \mathrm{Gr}_n(K^{N+N'}) \]
allow us to define the direct limit
\[ \mathrm{Gr}_n(K^\infty) := \varinjlim_N \mathrm{Gr}_n(K^N) \,. \]
These are the famous complex and real \emph{infinite Grassmann manifolds}. There exist certain \emph{tautological bundles} $\mathrm{E}_n(K^N)\rightarrow \mathrm{Gr}_n(K^N)$, whose fibres at a particular $n$-plane are the orthonormal frames of the plane itself, acted upon by $U(n)$ or $O(n)$ in the obvious manner. Taking the direct limit over $N$ yields corresponding bundles over the infinite Grassmannians:
\[ \vcenter{ \xymatrix{ \mathrm{E}_n(K^\infty) \ar[d] \\ \mathrm{Gr}_n(K^\infty) } } := \varinjlim_N \vcenter{ \xymatrix{ \mathrm{E}_n(K^N) \ar[d] \\ \mathrm{Gr}_n(K^N) } } \,. \]
\begin{proposition}
\[ \vcenter{ \xymatrix{ \mathrm{E}_n(\C^\infty) \ar[d] \\ \mathrm{Gr}_n(\C^\infty) } } \cong \vcenter{ \xymatrix{ EU(n) \ar[d] \\ BU(n) } } \quad \text{and} \quad  \vcenter{ \xymatrix{ \mathrm{E}_n(\R^\infty) \ar[d] \\ \mathrm{Gr}_n(\R^\infty) } } \cong \vcenter{ \xymatrix{ EO(n) \ar[d] \\ BO(n) } } \,. \]
\end{proposition}
\begin{proof}
The total spaces $\mathrm{E}_n(\C^N)$ and $\mathrm{E}_n(\R^N)$, consisting of orthonormal frames of $n$-planes, are diffeomorphic to the \emph{Stiefel manifolds} $U(N)/U(N-n)$ and $O(N)/O(N-n)$. This is reasonably clear, since the subgroups $U(N-n)\subset U(N)$ and $O(N-n)\subset O(N)$ of rotations in the complementary subspace of a particular $n$-plane are precisely those that leave a frame of the $n$-plane invariant. The fibrations
\begin{eqnarray*} \xymatrix{ U(N-1) \ar@{^{(}->}[r] & U(N) \ar[r] & S^{2N+1} } \\ \xymatrix{ O(N-1) \ar@{^{(}->}[r] & O(N) \ar[r] & S^N } \end{eqnarray*}
yield long exact sequences
\begin{eqnarray*} \xymatrix{ \ldots \ar[r] & \pi_{k+1}(S^{2N+1}) \ar[r] & \pi_k(U(N-1)) \ar[r] & \pi_k(U(N)) \ar[r] & \pi_k(S^{2N+1}) \ar[r] & \ldots } \\ \xymatrix{ \ldots \ar[r] & \pi_{k+1}(S^N) \ar[r] & \pi_k(O(N-1)) \ar[r] & \pi_k(O(N)) \ar[r] & \pi_k(S^N) \ar[r] & \ldots } \end{eqnarray*}
in homotopy. Since $\pi_i(S^m)=0$ for $i\le m-1$, we obtain the isomorphisms
\begin{eqnarray*} \pi_k(U(N-1))\cong \pi_k(U(N)) \cong \ldots \cong \pi_k(U(N-n)) \,,\quad k\le 2(N-n)+1 \\ \pi_k(O(N-1))\cong \pi_k(O(N)) \cong \ldots \cong \pi_k(O(N-n)) \,,\quad k\le N-n-1 \,. \end{eqnarray*}
On the other hand, the fibrations
\begin{eqnarray*} \xymatrix{ U(N-n) \ar@{^{(}->}[r] & U(N) \ar[r] & U(N)/U(N-n) \cong \mathrm{E}_n(\C^N) } \\ \xymatrix{ O(N-n) \ar@{^{(}->}[r] & O(N) \ar[r] & O(N)/O(N-n) \cong \mathrm{E}_n(\R^N) } \end{eqnarray*}
induce long exact sequences
\begin{eqnarray*} \xymatrix{ \ldots \ar[r]^-{i_k} & \pi_k(U(N)) \ar[r] & \pi_k(\mathrm{E}_n(\C^N) \ar[r] & \pi_{k-1}(U(N-n)) \ar[r]^-{i_{k-1}} & \ldots } \\ \xymatrix{ \ldots \ar[r]^-{j_k} & \pi_k(O(N)) \ar[r] & \pi_k(\mathrm{E}_n(\R^N) \ar[r] & \pi_{k-1}(O(N-n)) \ar[r]^-{j_{k-1}} & \ldots } \end{eqnarray*}
But we already know, that the inclusion maps
\begin{eqnarray*} i_k:\pi_k(U(N-n))\rightarrow \pi_k(U(N)) \\ j_k:\pi_k(O(N-n))\rightarrow \pi_k(O(N)) \end{eqnarray*}
are isomorphisms for $k$ small enough. Due to exactness, this is possible only if
\begin{eqnarray*} \pi_k(\mathrm{E}_n(\C^N)) = 0\,,\quad k\le 2(N-n)+1 \\
\pi_k(\mathrm{E}_n(\R^N)) = 0 \,,\quad k\le N-n-1 \,. \end{eqnarray*}
The proof is concluded by
\begin{eqnarray*} \pi_k(\mathrm{E}_n(\C^\infty)) = \pi_k(\varinjlim_N \mathrm{E}_n(\C^N))\cong \varinjlim_N \pi_k(\mathrm{E}_n(\C^N)) = 0 \\ \pi_k(\mathrm{E}_n(\R^\infty)) = \pi_k(\varinjlim_N \mathrm{E}_n(\R^N))\cong \varinjlim_N \pi_k(\mathrm{E}_n(\R^N)) = 0  \,, \end{eqnarray*}
an application of Whitehead's theorem to show that $\mathrm{E}_n(\C^\infty)$ and $\mathrm{E}_n(\R^\infty)$ are contractible, and finally an application of the result in \cite{Ste1}, which states that contractibility is enough to imply universality.
\end{proof}

Using the ''associated vector bundle'' construction, the representation theorem for $U(n)$ principal bundles can be translated to the language of $U(n)$ vector bundles: every $U(n)$ vector bundle over $X$ is obtained, up to isomorphism, as the pullback of the associated vector bundle of $EU(n)$ by a map $f:X\rightarrow BU(n)$. Obviously, a similar result holds for $O(n)$ vector bundles.

The cohomology groups of classifying spaces are extremely interesting and are the starting point for constructing characteristic classes of bundles.
\begin{definition}
The cohomology classes in $H^\bullet(BG,\Lambda)$ are called \emph{universal characteristic classes}\footnote{Here $\Lambda$ denotes the coefficient ring.}. By the representation theorem, any principal $G$ bundle $P\rightarrow X$ is isomorphic to the pullback bundle $f^{-1}EG\rightarrow X$ by a continuous map $f:X\rightarrow BG$. The \emph{characteristic classes} of $P$ are the pullbacks of the universal characteristic classes by $f$ into $H^\bullet(X,\Lambda)$. They are well-defined, since $f$ is unique up to homotopy.
\end{definition}

Let us look at some explicit examples.
\begin{proposition}[\cite{MS4}]
The cohomology ring $H^\bullet(BO(n),\Z_2)$ is a freely generated polynomial ring $\Z_2[w_1,\ldots,w_n]$, whose generators 
\[ w_k \in H^k(BO(n),\Z_2) \]
are called the \emph{universal Stiefel-Whitney classes}.
\end{proposition}
For a $O(n)$ vector bundle $E$ over $X$, represented by a continuous map $f:X\rightarrow BO(n)$, the \emph{$k$-th Stiefel-Whitney class} of $E$ is
\[ w_k(E) := f^*w_k \in H^k(X,\Z_2) \,. \]
We also define $w_0:=1$, so that $w_0(E)=1$ for any vector bundle $E$. The Stiefel-Whitney classes have some remarkable properties. For example, $w_1(E)=0$ if and only if the vector bundle is orientable\footnote{Expressed homotopically, a vector bundle is orientable if and only if the classifying map factors through the simply connected space $BSO(n)$.}. The second Stiefel-Whitney class $w_2(E)$ is an obstruction for the existence of a $\spin$ structure: the vector bundle $E$ admits a $\spin$ structure if and only if $w_1(E)=w_2(E)=0$. Finally, the short exact sequence
\[ \xymatrix@!=2.5pc{ 0 \ar[r] & \Z \ar[r]^{\times 2} & \Z \ar[r]^-{\mod 2} & \Z_2 \ar[r] & 0 } \]
induces a long exact sequence
\begin{equation}\label{Eq: S-W exact sequence} \xymatrix{ \ldots \ar[r] & H^k(X,\Z) \ar[r] & H^k(X,\Z_2) \ar[r]^\beta & H^{k+1}(X,\Z) \ar[r] & \ldots } \end{equation}
in cohomology, where $\beta$ is the Bockstein homomorphism. The \emph{integral Stiefel-Whitney classes} are defined by
\begin{equation}\label{Eq: Integral S-W class} W_k(E) := \beta(w_{k-1}(E)) \,. \end{equation}
Of particular importance is $W_3(E)$ for the reason that $E$ admits a $\spinc$ structure if and only if $w_1(E)=W_3(E)=0$. Readers unfamiliar with $\spinc$ and $\spin$ structures can think of them as slightly stronger ''orientability'' conditions, defined as the vanishing of $W_3(E)$ and $w_3(E)$, respectively. The \emph{total universal Stiefel-Whitney class} is the sum
\[ w := w_0+w_1 + w_2 + \ldots \in H^\bullet(BO(n),\Z_2) \]
and the \emph{total Stiefel-Whitney class} of the $O(n)$ bundle $E$ is
\[ w(E) := f^*w \in H^\bullet(X,\Z_2) \,. \]
Given two real vector bundles $E$ and $F$, the total Stiefel-Whitney class is additive in the sense that \cite{MS4}
\begin{equation}\label{Eq: S-W summation formula} w(E\oplus F) = w(E)\smile w(F) \,. \end{equation}

The second simplest case is that of integer coefficients.
\begin{proposition}[\cite{MS4}]\label{Proposition: Universal Chern-classes}
The cohomology ring $H^\bullet(BU(n),\Z)$ is a freely generated polynomial ring $\Z[c_1,\ldots,c_n]$, whose generators
\[ c_k\in H^{2k}(BU(n),\Z) \]
are called the \emph{universal Chern classes}.
\end{proposition}
For a $U(n)$ vector bundle $E$ represented by a continuous map $f:X\rightarrow BU(n)$, the \emph{$k$-th Chern class} is
\[ c_k(E) := f^*c_k \in H^{2k}(X,\Z) \,. \]
We also define $c_0:=1$, so that $c_0(E)=1$ for any vector bundle $E$. The \emph{total universal Chern class} is the sum
\[ c := c_0+ c_1 + c_2 + \ldots \in H^\text{even}(BU(n),\Z) \]
and the \emph{total Chern class} of the bundle $E$ is
\[ c(E) := f^*c \in H^\text{even}(X,\Z) \,. \]
Given two complex vector bundles $E$ and $F$, the total Chern class is additive in the sense that \cite{MS4}
\begin{equation}\label{Eq: Chern class summation formula} c(E\oplus F) = c(E)\smile c(F) \,. \end{equation}
For a complex vector bundle the Chern classes and the Stiefel-Whitney classes of the underlying real vector bundle are related.
\begin{proposition}[\cite{LM}] \label{Proposition: Chern class, S-W class relation}
There is a ring isomorphism
\[ H^\bullet(BU(n),\Z_2) \cong \Z_2[\widetilde{c}_1,\ldots,\widetilde{c}_n] \,, \]
where $\widetilde{c}_k$ is the mod $2$ reduction of $c_k$. The natural inclusion $BU(n)\rightarrow BO(2n)$ induces a homomorphism
\[ H^\bullet(BO(2n),\Z_2)\rightarrow H^\bullet(BU(n),\Z_2) \,, \]
which maps
\[ w_{2k}\mapsto \widetilde{c}_k \quad \text{and} \quad w_{2k+1}\mapsto 0 \,. \]
Hence, the underlying $O(2n)$ vector bundle $E_\R$ of a $U(n)$ vector bundle $E$ satisfies $w_{2k}(E_\R) \equiv c_k(E)\,(\operatorname{mod}\,2)$ and $w_{2k+1}(E_\R)=0$. In particular, the underlying real vector bundle of any complex vector is orientable\footnote{In fact, it is canonically orientable. This is evident from the following alternative proof. Choose a complex basis $\{e_1,\ldots,e_n\}$ for each fibre $E_x$ of $E$. Then $\{e_1,ie_1,\ldots,e_n,ie_n\}$ is a basis for the $2n$-dimensional fibre of the underlying real vector space and defines an orientation for each fibre. The orientation is independent of the complex basis used, because $\GL_n(\C)$ is connected, allowing for continuous orientation-preserving transformations between any two complex bases.} $\spinc$ structure.
\end{proposition}

The definition of a characteristic class can be rephrased category theoretically as follows.
\begin{definition} A characteristic class $C$ is a natural transformation \cite{Mac}
\[ C:\mathrm{Prin}_G(-)\rightarrow H^\bullet(-,\Lambda)\,. \]
Again, the definition can equally well be given in terms of associated vector bundles. \end{definition}
This definition emphasizes the \emph{naturality} property of characteristic classes: if $Y\rightarrow X$ is a continuous map and $E\rightarrow X$ a bundle with characteristic class $C(E)$, then $C(f^{-1}E)=f^*C(E)$.

Let $f:X\rightarrow BU(n)$ be the classifying map for a vector bundle $E\rightarrow X$ and $C$ a characteristic class $\mathrm{Vect}_\C^n(-)\rightarrow H^\bullet(-,\Z)$. Naturality implies
\[ C(E) = C(f^{-1}EU(n)) = f^*C(EU(n)) \,. \]
Since $H^\bullet(BU(n),\Z)$ is generated by $\{c_1,\ldots,c_n\}$, $C(EU(n))$ is a polynomial $P_C(c_1,\ldots ,c_n)$. Thus, $C(E)$ is of the form
\begin{align*} C(E) &= f^*P_C(c_1,\ldots ,c_n) \\ &= P_C(f^*c_1,\ldots ,f^*c_n) \\ &= P_C(c_1(E),\ldots ,c_n(E)) \,, \end{align*}
a polynomial in the Chern classes of $E$. This proves the following important result.
\begin{proposition}\label{Proposition: Characteristic classes are polynomials of Chern classes} Every characteristic class $C:\mathrm{Vect}_\C^k(-)\rightarrow H^\bullet(-,\Z)$ is a polynomial in the Chern classes. Moreover, any characteristic class taking values in either $H^\bullet(-,\Q)$ or $H^\bullet(-,\R)\cong H^\bullet_\text{dR}(-)$ can be expressed as a polynomial in the images of the Chern classes mapped into the respective cohomology theory. \qed \end{proposition}

\begin{proposition}\label{Proposition: Chern class classifies line bundles}
The first Chern class $c_1(L)\in H^2(X,\Z)$ of a complex line bundle $L\rightarrow X$ uniquely classifies the bundle up to isomorphism. Moreover, $c_k(L)=0$ for $k\ge 2$.
\end{proposition}
To prove this, we need to know the homotopy groups of $S^1$ and $S^\infty:=\varinjlim S^n$.
\renewcommand{\labelenumi}{\textbf{\alph{enumi})}}
\begin{lemma}\label{Lemma: Homotopies of spheres}
\textbf{a)} $\pi_k(S^1)=0$ for $k\ge 2$. \quad \textbf{b)} $\pi_k(S^\infty)=0$ for all $k$.
\end{lemma}
\begin{proof}
\begin{enumerate}
\item Recall from the theory of covering spaces \cite{May} that a map $f:S^k\rightarrow S^1$ 
\[ \xymatrix{ & & \R \ar[d]^p \\ S^k \ar[rr]_f \ar@{-->}[urr]^{\widetilde{f}} & & S^1 } \]
admits a lift $\widetilde{f}:S^k \rightarrow \R$, if and only if $f_*\pi_1(S^k)\subset p_*\pi_1(\R)$. This is obviously the case when $k\ge 2$, since $\pi_1(\R)=0$ and $\pi_1(S^k)=0$ for all $k\ge 2$. Therefore, $\pi_k(S^1)$ is the space of homotopy classes of continuous maps $S^k\rightarrow S^1$, all of which lift to maps $S^k\rightarrow \R$ (for $k\ge 2)$, yielding elements in $\pi_k(\R)$. The induced map $p_*:\pi_k(\R)\rightarrow \pi_k(S^1)$ is surjective (again, for $k\ge 2$). Injectivity follows from a basic result in the theory of covering spaces that the lifts of homotopic maps $S^k\rightarrow S^1$ are homotopic \cite{May}. Thus, $\pi_k(S^1)\cong \pi_k(\R) = 0$ for $k\ge 2$.
\item Using $\pi_k(S^n)=0$ for $k<n$, we get
\[ \pi_k(S^\infty) \cong \varinjlim_n \pi_k(S^n) = 0 \,. \]
\end{enumerate}
\end{proof}
\begin{proof}[Proof of \ref{Proposition: Chern class classifies line bundles}]
By proposition \ref{Proposition: Universal Chern-classes}, the cohomology ring is freely generated by $c_1\in H^2(BU(1),\Z)$. The classifying space is
\[ BU(1)=\mathrm{Gr}_1(\C^\infty)=\varinjlim_N \frac{U(N)}{U(N-1)\times U(1)} \,. \]
The \emph{$N$-dimensional complex projective space} is the quotient space $\CP^N := S^{2N+1}/S^1$. Using
\[ U(N)/U(N-1)\cong S^{2N+1} \]
gives
\[ \CP^N \cong \frac{S^{2N+1}}{S^1} \cong \frac{U(N)}{U(N-1)\times U(1)} \,. \]
Thus, the classifying space $BU(1)$ is simply
\[ BU(1) \cong \varinjlim_N \CP^N \,, \]
the \emph{infinite complex projective space}, which we denote by $\CP^\infty$. On the other hand, the direct limit of the fibrations
\[ \xymatrix{ S^1 \ar[r] & S^{2N+1} \ar[r] & \CP^N } \]
is the fibration
\[ \xymatrix{ S^1 \ar[r] & S^\infty \ar[r] & \CP^\infty } \,, \]
which induces a long exact sequence of homotopy groups
\[ \xymatrix{ \ldots \ar[r] & \pi_k(S^1) \ar[r] & \pi_k(S^\infty) \ar[r] & \pi_k(\CP^\infty) \ar[r] & \pi_{k-1}(S^1) \ar[r] & \ldots } \]
This, together with lemma \ref{Lemma: Homotopies of spheres}, yields
\begin{equation} \label{Eq: Homotopy of CP^infty} \pi_k(\CP^\infty) \cong  \begin{cases} \Z & \text{for $k=2$\,,} \\ 0 & \text{otherwise.} \end{cases} \end{equation}
In other words, $\CP^\infty$ is a model for the Eilenberg-MacLane space $K(\Z,2)$. This concludes the proof, since both isomorphism classes of complex line bundles over $X$ and classes of $H^2(X,\Z)$ are in bijective correspondence with continuous maps $X\rightarrow \CP^\infty= K(\Z,2)$. More precisely, for any two maps $f_1,f_2:X\rightarrow BU(1)$ the pullbacks of the first universal Chern class are equal, $f_1^*c_1 = f_2^*c_1$, if and only if $f_1$ and $f_2$ are homotopic, if and only if the bundles $f_1^{-1}EU(1)$ and $f_2^{-1}EU(1)$ are isomorphic.
\end{proof}

We now restrict specifically to smooth manifolds. All smooth manifolds admit \emph{good open covers} \cite{BT}, open covers whose elements and their finite intersections are all contractible. Given the transition data for any complex line bundle $L\rightarrow X$, it is easy to construct a cohomology class $e(L)\in H^2(X,\Z)$, which uniquely determines $L$ up to isomorphism. This is known as the \emph{Euler class} \cite{BT} and is constructed as follows. Let $\{g_{\alpha\beta}\}$ be the transition functions of $L$ with respect to a good open cover $\{U_\alpha\}$. They can be written as
\[ g_{\alpha\beta} = \exp\left(i\phi_{\alpha\beta}\right)\,, \]
where $\{\phi_{\alpha\beta}\}$ are real-valued functions. The cocycle condition for $\{g_{\alpha\beta}\}$ becomes
\[ g_{\beta\gamma}g^{-1}_{\alpha\gamma}g_{\alpha\beta} = \exp\left(i(\phi_{\beta\gamma}-\phi_{\alpha\gamma}+\phi_{\alpha\beta})\right) = 1\,, \]
which is satisfied if and only if
\[ \phi_{\beta\gamma}-\phi_{\alpha\gamma}+\phi_{\alpha\beta} \in 2\pi \Z \,. \]
We denote
\[ n_{\alpha\beta\gamma} := -\frac{1}{2\pi}\left(\phi_{\beta\gamma}-\phi_{\alpha\gamma}+\phi_{\alpha\beta}\right) \,, \]
which is obviously closed as a \v Cech cocycle and thus determines a \v Cech cohomology class $[\underline{n}]\in \check{H}^2(X,\Z)$\footnote{Since we are working on a manifold, both \v Cech cohomology and de Rham cohomology coincide with singular cohomology with the appropriate coefficients. Hence, the distinction between $H^\bullet(X,-)$ and $\check{H}^\bullet(X,-)$, and $H^\bullet(X,\R)$ and $H^\bullet_\text{dR}(X)$ is purely notational and is merely used to emphasize the method of calculation.}. Note that, even though $n_{\alpha\beta\gamma} = -\delta(\phi)_{\alpha\beta\gamma}/(2\pi)$, the class $[\underline{n}] \in \check{H}^2(X,\Z)$ is generally nontrivial, since $\phi_{\alpha\beta}/2\pi$ are generally only real numbers, not integers. The Euler class of $L$ is\footnote{More precisely, we mean the Euler class of the underlying real vector bundle of rank $2$.}
\[ e(L):=[\underline{n}] \,. \]
\begin{theorem}\label{Theorem: e(L)=c_1(L)} For a complex line bundle $L\rightarrow X$ the Euler class and the first Chern class coincide:
\[ e(L) = c_1(L) \in H^2(X,\Z)\,. \]
\end{theorem}
\begin{proof}
It is easy to see that $e(L)$ uniquely determines $L$ up to isomorphism and that it behaves naturally under pullbacks. It is thus a characteristic class and as such, by proposition \ref{Proposition: Characteristic classes are polynomials of Chern classes}, is a polynomial in the first Chern class. Both are classes in $H^2(X,\Z)$ and hence equal up to a constant multiplier, which turns out to be $1$ \cite{BT}.
\end{proof}

We have now established a simple method for computing Chern classes of complex line bundles from transition data. Remarkably, it turns out that many statements concerning Chern classes of higher rank bundles can be proven by merely proving the statements for line bundles, due to the following famous theorem, known as the \emph{splitting principle}.
\begin{theorem}[Splitting principle \cite{May,BT}]\label{Theorem: Splitting principle}
Let $E\rightarrow X$ be a $U(n)$ vector bundle. There exists a manifold $F(E)$, the \emph{flag manifold}, and a map $f:F(E)\rightarrow X$, such that the pullback bundle $f^{-1}E\rightarrow F(E)$ is isomorphic to a sum of complex line bundles
\[ f^{-1}E \cong L_1 \oplus \ldots \oplus L_n \longrightarrow F(E) \]
and the induced map $f^*:H^\bullet(X,\Z)\rightarrow H^\bullet(F(E),\Z)$ is an injection.
\end{theorem}

The splitting principle can be used to compute the Chern classes of dual bundles. Denote the dual bundle of a complex line bundle $L$ by $L^*$. We then have
\[ c_1(L)+c_1(L^*)=c_1(L\otimes L^*) = 0 \,, \]
which can be easily seen from the explicit construction of $c_1$ for line bundles in terms of the transition functions. Let now $E$ be a $U(n)$ vector bundle. By the splitting principle there exists a flag bundle $F(E)$, such that
\begin{align*} f^{-1}E &\cong L_1\oplus \ldots \oplus L_n\rightarrow F(E) \\ f^{-1}E^* &\cong L_1^*\oplus \ldots \oplus L_n^*\rightarrow F(E) \,. \end{align*}
By naturality of the Chern classes and the summation formula \eqref{Eq: Chern class summation formula}, we get
\begin{align*} f^*c(E^*) &= c(f^{-1}E^*) \\ &= (1+c_1(L_1^*))\smile \ldots \smile (1+c_1(L_n^*)) \\ &= (1-c_1(L_1))\smile \ldots \smile (1-c_1(L_n)) \\ &= \sum_i (-1)^i c_i(f^{-1}E) \\ &= f^*\Big(\sum_i (-1)^i c_i(E) \Big) \,. \end{align*}
The map $f^*$ is injective according to the splitting principle. Thus,
\[ c_i(E^*) = (-1)^i c_i(E)\,. \]

The Chern classes can also be computed purely analytically as we shall next show. Let $E\rightarrow X$ be a $U(n)$ vector bundle over a smooth manifold $X$. Let $A$ be a $\mathfrak{u}(n)$-valued gauge connection on $E$ and $F=dA+A\wedge A$ the corresponding curvature.
\begin{theorem}\label{Theorem: Analytic Chern classes}
The de Rham polyform
\begin{equation}\label{Eq: Analytic Chern class} c(F) := \det\left(1 - \frac{F}{2\pi i}\right)\,. \end{equation}
is closed and its de Rham cohomology class is independent of the chosen connection. Moreover,
the image of the total Chern character $c(E)\in H^\text{even}(X,\Z)$ in $H^\text{even}_\text{dR}(X)$ coincides with $[c(F)]$.
\end{theorem}
\begin{proof}
Let us start by proving the theorem for complex line bundles. If $F$ is the curvature of a $U(1)$ vector bundle $L$, then
\[ c(F)=c_1(F)=-\frac{F}{2\pi i} \,. \]
Recall that the \v Cech--de Rham isomorphism $H^\bullet_\text{dR}(X)\cong \check{H}^\bullet(X,\R)$ between de Rham and \v Cech cohomologies follows from the \v Cech-de Rham double complex by a ''tic-tac-toe'' argument \cite{BT}. In particular, the computation of the representative of $[-F/(2\pi i)]\in H^2_\text{dR}(X)$ in \v Cech cohomology $\check{H}^2(X,\R)$ is illustrated by the diagram
\[ \xymatrix{ 0 \ar[r]^-{\text{res}} & 0 & & & & \\ -\frac{F}{2\pi i} \ar[u]^-d \ar[r]^-{\text{res}} & \{-\frac{F_\alpha}{2\pi i} \} \ar[r]^-{\delta} \ar[u]^-{d} & 0 & & \\ & \{ -\frac{A_\alpha}{2\pi i} \} \ar[u]^-d \ar[r]^-{\delta} & \{ -\frac{\delta(A)_{\alpha\beta}}{2\pi i} \} \ar[u]^-d \ar[r]^-{\delta} & 0 & \\ & & \{ -\frac{\phi_{\alpha\beta}}{2\pi} \} \ar[u]^-d \ar[r]^-{\delta} & \{ -\frac{\delta(\phi)_{\alpha\beta\gamma}}{2\pi} \} = \{ n_{\alpha\beta\gamma} \} \ar[u]^-d \ar[r]^-{\delta} & 0 } \]
Let us clarify this a bit. First, $F$ is globally defined and closed under $d$. Restriction to open sets of the good cover is denoted by $\text{res}$. Of course, $\delta\circ \text{res}=0$. Next, $dF_\alpha=0$, since the diagram commutes. Open sets of a good cover are contractible and thus their cohomology vanishes (this is the Poincar\'e lemma), which implies that $F_\alpha = dA_\alpha$ for $1$-forms $A_\alpha$. Again, commutativity of the diagram implies $d\delta(A)_{\alpha\beta}=0$. By the Poincar\'e lemma we have $\delta(A)_{\alpha\beta}=d(i\phi_{\alpha\beta})$, where $\phi_{\alpha\beta}$ are $0$-forms (functions). Finally, by commutativity  $d\delta(\phi)_{\alpha\beta\gamma}=0$, which means that $\delta(\phi)_{\alpha\beta\gamma}$ are all constants. We simplify the notation by writing $n_{\alpha\beta\gamma}=-\delta(\phi)_{\alpha\beta\gamma}/(2\pi)$. Clearly, $\{n_{\alpha\beta\gamma}\}$ determines a nontrivial \v Cech cohomology class with constant coefficients. When $F$ is the curvature of a connection on a line bundle, remarkable things happen. Namely, $\{A_{\alpha}\}$ are then the local connection $1$-forms, satisfying $\delta(A)_{\alpha\beta}=d\log g_{\alpha\beta}$, where $g_{\alpha\beta}$ are the transition functions. Then $\phi_{\alpha\beta}=-i\log g_{\alpha\beta}$ and $\delta(\phi)_{\alpha\beta\gamma} = -i\log \delta(g)_{\alpha\beta\gamma}$. But since $\{g_{\alpha\beta}\}$ is closed under $\delta$, we get $\delta(\phi)_{\alpha\beta\gamma} \in 2\pi \Z$. Therefore, $n_{\alpha\beta\gamma} = -\delta(\phi)_{\alpha\beta\gamma}/(2\pi)$ are integers and they determine a class $[\underline{n}]\in \check{H}^2(X,\Z)$\footnote{In the de Rham picture this property is expressed by fact that the integral of $-F/(2\pi i)$ over any integral homology cycle evaluates to an integer.}. The class $[\underline{n}]$ is precisely the Euler class of the line bundle which we discussed earlier and, by theorem \ref{Theorem: e(L)=c_1(L)}, it coincides with the first Chern class $c_1(L)$. This proves theorem for line bundles.

Consider then a general $U(n)$ vector bundle $E\rightarrow X$ with classification map $f:X\rightarrow BU(n)$ and curvature $F$. The determinant in \eqref{Eq: Analytic Chern class} is tricky to deal with, unless $F$ is diagonal. It can always be diagonalized by a unitary matrix $A$, since $\mathfrak{u}(n)$ consists of complex anti-hermitian matrices. Because
\[ \det\left(1-\frac{AFA^{-1}}{2\pi i}\right) = \det\left(1-\frac{F}{2\pi i}\right) \,, \]
we lose nothing by assuming $F$ to be diagonal. Let us continue under this assumption.

By the splitting principle \ref{Theorem: Splitting principle} and naturality of the universal bundle construction, there exists a flag bundle $g:F(E)\rightarrow X$, such that
\[ f^{-1}E\cong (g\circ f)^{-1}\cong L_1\oplus \ldots \oplus L_n \,, \]
where $L_k\rightarrow F(E)$ are line bundles. By naturality and the summation formula \eqref{Eq: Chern class summation formula}, we get
\[ f^*c(E) = c(f^{-1}(E)) = c(L_1\oplus \ldots \oplus L_n) = c(L_1)\smile \ldots \smile c(L_n) \,. \]
We already proved that
\[ c(L_k) = 1+c_1(L_k) = 1 + [c_1(F_k)] = [1+c_1(F_k)] = \left[1-\frac{F_k}{2\pi i}\right] \,, \]
where $F_k$ are curvatures of the line bundles $L_k$, induced by $F$. Plugging these into the previous formula yields
\begin{align*} f^*c(E) &= \left[1-\frac{F_1}{2\pi i}\right]\smile \ldots \smile \left[1-\frac{F_n}{2\pi i}\right] \\ &= \left[\prod_{k=1}^n \left(1-\frac{F_k}{2\pi i}\right)\right] = \left[\det\left(1-\frac{f^*F}{2\pi i}\right)\right] = f^*\left[\det\left(1-\frac{F}{2\pi i}\right)\right] \,, \end{align*}
By the splitting principle, the map $f^*$ is injective, which implies
\[ c(E) = \left[\det\left(1-\frac{F}{2\pi i}\right)\right]\,, \]
concluding the proof.
\end{proof}
The images of $c_k(E)$ in de Rham cohomology are represented by the homogeneous components $c_k(F)$ of $c(F)$,
\[ c(F) = c_0(F) + c_1(F) + \ldots + c_n(F) \,,\quad c_k(F)\in \Omega^{2k}(X) \,. \]
From now on, we use the same notation, $c(E)$ and $c_k(E)$, to denote the Chern classes both in integral and de Rham cohomology. It should always be clear from the context which one is used.

\begin{definition}
Recall from the proof of theorem \ref{Theorem: Analytic Chern classes} that $F$ can always be diagonalized by a unitary matrix $A$. Let
\[ -\frac{AFA^{-1}}{2\pi i} = \diag(x_1,\ldots,x_n) \,, \]
where $\{x_k\}$ are $2$-forms, known as the \emph{Chern roots}. Writing out the determinant in $c(F)$ in terms of $\{x_k\}$ yields
\begin{align*} c(F) &= \det\left(1-\frac{F}{2\pi i}\right) = \prod_{k=1}^n (1+x_k) \\ &= 1 + \underbrace{\left(\sum_{k=1}^n x_k\right)}_{c_1(F)} + \underbrace{\left(\sum_{1\le k < l\le n} x_k \wedge x_l\right)}_{c_2(F)} + \ldots + \underbrace{\left(\prod_{k=1}^n x_k\right)}_{c_n(F)} \,. \end{align*}
Clearly, $c_k(F)$ are simply the \emph{elementary symmetric polynomials} in the Chern roots. To simplify the notation, we shall make no distinction between $x_k$ and the cohomology class $[x_k]$, denoting both by $x_k$ and referring to both with the name \emph{Chern root}.
\end{definition}
It is a well-known result that all symmetric polynomials can be constructed as polynomials in the elementary symmetric polynomials. In this case, all symmetric polynomials in the Chern roots can be expressed in terms of the Chern classes. Therefore, to construct further characteristic classes, all we need to do is to construct symmetric polynomials in the Chern roots.

The \emph{Chern character} of a $U(n)$ vector bundle $E$ is the characteristic class
\begin{equation}\label{Eq: Chern character of a vector bundle} \ch(E) := \sum_{k=1}^n e^{x_k} = n + \sum_{k=1}^n x_k + \frac{1}{2!} \sum_{k=1}^n x_k^2 + \ldots \in H^\text{even}(X,\Q) \,. \end{equation}
The homogeneous components of $\ch(E)$ are denoted $\ch_k(E)\in H^{2k}(X,\Q)$. It is a straightforward task to write them out in terms of the Chern classes. The three first are
\begin{align*} \ch_0(E) &= \rk(E) \\ \ch_1(E) &= c_1(E) \\ \ch_2(E) &= \frac{1}{2}\left(c_1(E)^2 - 2 c_2(E)\right) \,. \end{align*}
In terms of a curvature $2$-form $F$ of $E$, the Chern character is
\begin{equation}\label{Eq: Analytic Chern character} \ch(E) = \sum_{k=1}^n e^{x_k} = \Tr\,\exp(\diag(x_1,\ldots,x_n)) = \left[\Tr\,\exp\left(-\frac{F}{2\pi i}\right)\right] \,. \end{equation}
For a complex line bundle $L$, the definition simplifies to
\[ \ch(L)=1+c_1(L) = \left[1-\frac{F}{2\pi i}\right] \in H^2(X,\Q) \,. \]
Using \eqref{Eq: Analytic Chern character}, it is a simple calculation to show that \cite{LM}
\begin{align} \label{Eq: Chern character, additivity} \ch(E\oplus F) &= \ch(E) + \ch(F) \\ \label{Eq: Chern character, multiplicativity} \ch(E\otimes F) &= \ch(E)\smile \ch(F) \,, \end{align}
for any complex vector bundles $E$ and $F$ over $X$. The importance of the Chern character becomes clear, when we learn that it is compatible with $K$-theory under direct sums and tensor products of vector bundles.

The \emph{complex Todd class} of a $U(n)$ vector bundle $E$, with Chern roots $\{x_k\}$, is the characteristic class
\[ \Todd_\C(E) := \prod_{k=1}^n \frac{x_k}{1-e^{-x_k}} \in H^\text{even}(X,\Q)\,. \]
The \emph{complex Atiyan-Hirzebruch class} is the characteristic class
\[ \widehat{A}_\C(E) := \prod_{k=1}^n \frac{x_k/2}{\sinh(x_k/2)} \in H^{4\bullet}(X,\Q) \,. \]
These two are related by
\begin{align*} \Todd_\C(E) &= \prod_{k=1}^n \frac{x_k}{1-e^{-x_k}} = \prod_{k=1}^n \frac{x_k}{1-e^{-x_k}} = \prod_{k=1}^n \left( \frac{x_k/2}{\sinh(x_k/2)}\right)\smile e^{\sum_k x_k/2} \\ &= \widehat{A}_\C(E) \smile e^{c_1(E)/2} \,. \end{align*} 

Since the Chern classes are only defined for complex vector bundles, so are $\ch$, $\hat{A}_\C$ and $\Todd_\C$. However, given a real vector bundle, we can always complexify the fibres to obtain a complex vector bundles. The Chern roots of the complexified bundle can then be used to define characteristic classes of the original real bundle. For example, let $E\rightarrow X$ be a $O(n)$ vector bundle. Its complexification is the product bundle $E\otimes \C$. The splitting pullback of $E\otimes \C$ by $f:F(E\otimes \C)\rightarrow X$ is of the form \cite{LM}
\[ f^{-1}(E\otimes \C) \cong L_1\otimes L_1^* \otimes \ldots \otimes L_{\floor{n}}\otimes L_{\floor{n}}^* \,, \]
where $L_k$ are complex line bundles and $L_k^*$ their dual bundles. Since the dual line bundles satisfy $c_1(L^*)=-c_1(L)$, the Chern roots of $E\otimes \C$ are $\{\pm x_1,\pm x_2, \ldots,\pm x_{\floor{x}}\}$, where $f^*x_k = c_1(L_k)$. The \emph{Atiyah-Hirzebruch class} of $E$ is the characteristic class
\begin{equation}\label{Eq: A-class expansion} \widehat{A}(E) := \prod_{k=1}^{\floor{n}} \frac{x_k/2}{\sinh(x_k/2)} \in H^{4\bullet}(X,\Q) \,. \end{equation}
It is the square root of the complex Todd class of $E\otimes \C$:
\begin{align*} \Todd_\C(E\otimes \C) &:= \prod_{k=1}^{\floor{n}} \frac{-x_k^2}{(1-e^{-x_k})(1-e^{x_k})} = \prod_{k=1}^{\floor{n}} \frac{-x_k^2}{(e^{x_k/2}-e^{-x_k/2})(e^{-x_k/2}-e^{x_k/2})} \\ &= \prod_{k=1}^{\floor{n}} \frac{x_k^2}{(e^{x_k/2}-e^{-x_k/2})(e^{x_k/2}-e^{-x_k/2})} = \prod_{k=1}^{\floor{n}} \left(\frac{x_k/2}{\sinh(x_k/2)}\right)^2 = \widehat{A}(E)^2 \,. \end{align*}
The real version of the \emph{Todd class} is more difficult to define and the reason behind the definition can not be explained at this point. In fact, it is only defined for $\spinc$ vector bundles. Let $E$ be an $O(n)$ vector bundle over $X$, satisfying $w_1(E)=W_3(E)=0$. Exactness of \eqref{Eq: S-W exact sequence} and the definition \eqref{Eq: Integral S-W class} imply, that $w_2(E)$ must be the mod $2$ reduction of a class $d(E)\in H^2(X,\Z)$. Recall from proposition \ref{Proposition: Chern class, S-W class relation} that if $E$ is a \emph{complex} vector bundle with underlying real vector bundle $E_\R$, then $w_2(E_\R)$ is the mod $2$ reduction of $c_1(E)$. Thus, in the complex case, the first Chern class plays the role of $d$. This motivates the definition
\begin{equation}\label{Eq: Real Todd class} \Todd(E) := \widehat{A}(E) \smile e^{d(E)/2} \,, \end{equation}
for a \emph{real} $\spinc$ vector bundle $E$. 

\section{$\spin$ and $\spinc$ structures}
We have already referred several times to $\spin$ and $\spinc$ vector bundles, which we defined as vector bundles whose Stiefel-Whitney classes satisfy the appropriate vanishing conditions. In this section the reader is acquainted with more intuitive geometric definitions.

Consider first a principal $O(n)$ bundle $P_{O(n)}$ over $X$. The group $O(n)$ has two connected components, one consisting of matrices with determinant $+1$ and the other of matrices with determinant $-1$. Let $\{g_{\alpha\beta}\}$ be the transition functions of $P_{O(n)}$ with respect to a good open cover. We would like to find reductions and lifts of the transition functions to more connected groups than $O(n)$. The first step would be to find a reduction of the structure group to the connected ($0$-connected) group $SO(n)$. For this, we need to find a set of $SO(n)$-valued transition functions $\{g'_{\alpha\beta}\}$, differing from $g_{\alpha\beta}$ by a coboundary. Let
\[ t_{\alpha\beta} := \det(g_{\alpha\beta})\in \{\pm 1\}\cong \Z_2 \,. \]
Then $\underline{t}$ is a $\Z_2$-valued \v Cech cocycle,
\[ \delta(\underline{t})=\det(\delta(\underline{g}))=\det(1)=1 \]
and thus defines a class $[\underline{t}]\in H^1(X,\Z_2)$. If this class is trivial, $t_{\alpha\beta} = s_{\beta}s_{\alpha}^{-1}$ for some $\Z_2$-valued functions $\{s_\alpha\}$, then we can find $O(n)$-valued matrices $\{h_\alpha\}$, such that $\det\,h_\alpha = s_\alpha$. These can be used to define $SO(n)$-valued cocycles (transition functions)
\[ g_{\alpha\beta}' := h_\alpha g_{\alpha\beta} h_\beta^{-1} \,. \]
The bundle defined by $\{g_{\alpha\beta}'\}$ is isomorphic to the bundle defined by $\{g_{\alpha\beta}\}$, since the transition functions differ only by a coboundary. Thus, the structure group can be reduced to $SO(n)$ if and only if $[\underline{t}]\in H^1(X,\Z_2)$ vanishes. The bundle $P_{O(n)}$ is then said to be \emph{orientable} and the bundle defined by $\{g_{\alpha\beta}'\}$ is denoted by $P_{SO(n)}$. The class $[\underline{t}]$ can be shown to coincide with $w_1(P_{O(n)})$ by observing that $[\underline{t}]$ respects naturality and that $[\underline{t}]\in H^1(BO(n),\Z)\cong \Z_2$ for the universal bundle $EO(n)\rightarrow BO(n)$ produces the nontrivial element (otherwise every $n$-plane bundle would be orientable). The orientations of $P_{SO(n)}$ over a (generally disconnected) space $X$ are in one-to-one correspondence with $H^0(X,\Z_2)$. Obviously, the correspondence is given by assigning to each connected component of $X$ one of the two possible orientations.
\begin{proposition}\label{Proposition: Orientability} A principal $O(n)$ bundle $P_{O(n)}\rightarrow X$ is orientable if and only if $w_1(P_{O(n)})=0$. The distinct orientations are in bijective correspondence with $H^0(X,\Z_2)$.\qed \end{proposition}

The next step in simplifying the structure group would be to find a lift of the $SO(n)$-valued transition functions to take values in a simply connected ($1$-connected) structure group. The group $SO(n)$ has a $2$-sheeted universal covering space $\spin(n)$ \cite{LM}, which fits into the short exact sequence
\[ \xymatrix{ 1 \ar[r] & \Z_2 \ar[r] & \spin(n) \ar[r]^\eta & SO(n) \ar[r] & 1 } \]
where $\eta:\spin(n)\rightarrow SO(n)$ is the fibre map. $\spin(n)$ is a Lie group with Lie algebra $\mathfrak{spin}(n)\cong \mathfrak{so}(n)$.

Under what conditions can the $SO(n)$-valued transition functions $\{g_{\alpha\beta}\}$ of $P_{SO(n)}$ be lifted to $\spin(n)$-valued functions $\{\widetilde{g}_{\alpha\beta}\}$? Of course, the lifting can always be performed locally by simply lifting $SO(n)$ into either one of the two sheets of the covering, but there might be a topological obstruction for patching the local lifts together on triple overlaps. To this end, suppose that $\{\widetilde{g}_{\alpha\beta}\}$ constitutes a \emph{local} $\spin$ lift. The cocycle condition for $\{g_{\alpha\beta}\}$ requires
\[ \eta(\widetilde{g}_{\alpha\beta}\widetilde{g}_{\beta\gamma}\widetilde{g}^{-1}_{\alpha\gamma}) = \eta(\widetilde{g}_{\alpha\beta})\eta(\widetilde{g}_{\beta\gamma})\eta(\widetilde{g}^{-1}_{\alpha\gamma}) = 1\,. \]
However, $\ker(\eta) \cong \{\pm 1\}\subset \text{Spin}(n)$, where $-1\in \text{Spin}(n)$ is a ''nontrivial'' element not in $SO(n)$. Thus, the elements
\[ u_{\alpha\beta\gamma} := \widetilde{g}_{\alpha\beta}\widetilde{g}_{\beta\gamma}\widetilde{g}^{-1}_{\alpha\gamma} \in \{\pm 1\} \cong \Z_2  \]
form a \v Cech cocycle $[\underline{u}]\in H^2(X,\Z_2)$. If this cocycle is a coboundary, that is, if
\[ u_{\alpha\beta\gamma} = \widetilde{h}_{\alpha\beta}\widetilde{h}_{\beta\gamma}\widetilde{h}^{-1}_{\alpha\gamma} \]
for some $\spin(n)\ni \{\pm 1\}$-valued functions $\{\widetilde{h}_{\alpha\beta}\}$, then $\{\widetilde{g}_{\alpha\beta}\}$ (not a cocycle) can be modified into the cocycle $\{\widetilde{g}_{\alpha\beta}\widetilde{h}^{-1}_{\alpha\beta}\}$, determining the transition data for a principal $\spin(n)$ bundle. Conversely, a global $\spin$ lift obviously yields a trivial class $[\underline{u}]\in H^2(X,\Z_2)$. If the $\spin$ lift does exist, the principal $\spin(n)$ bundle is denoted by $P_{\spin(n)}$. The proof that $w_2(E)$ coincides with $[\underline{u}]$ is again done in two steps: first noticing that the construction respects naturality and then arguing that $[\underline{u}]$ of the universal bundle $ESO(n)\rightarrow BSO(n)$ produces the nonzero element in $H^2(BSO(n),\Z_2)\cong \Z_2$ (otherwise every $SO(n)$ bundle would have a $\spin$ lift). Consider now another local $\spin$ lift $\{\widetilde{f}_{\alpha\beta}\}$ of $P_{SO(n)}$. The differences $\widetilde{f}_{\alpha\beta}\widetilde{g}_{\alpha\beta}^{-1}$ satisfy
\[ \eta(\widetilde{f}_{\alpha\beta}\widetilde{g}_{\alpha\beta}^{-1})= g_{\alpha\beta}g_{\alpha\beta}^{-1}=1 \,, \]
which means that
\[ \widetilde{f}_{\alpha\beta}\widetilde{g}_{\alpha\beta}^{-1} \in \{\pm 1\} \cong \Z_2 \,. \]
These differences define a class in $H^1(X,\Z_2)$, which is trivial precisely when the two $\spin$ lifts differ by a cocycle. We have established the following.
\begin{proposition} A principal $SO(n)$ bundle $P_{SO(n)}\rightarrow X$ lifts to a $\spin(n)$ bundle $P_{\spin(n)}$ precisely when $w_2(P_{SO(n)})=0$. If $w_2(P_{SO(n)})=0$, a choice of the $\spin$ lift is called a $\spin$ structure. Distinct $\spin$ structures are in bijective correspondence with $H^1(X,\Z_2)$.\qed \end{proposition}

The existence of a $\spinc$ structure is a slightly stronger condition than orientability, but weaker than the existence of a $\spin$ structure. Let us start again from the principal $SO(n)$ bundle $P_{SO(n)}$ with transition functions $\{g_{\alpha\beta}\}$ and a generally nontrivial Stiefel-Whitney class $w_2(P_{SO(n)})$. The problem was that the local $\spin$ lifts of the $SO(n)$-valued transition functions could not be made to agree on triple overlaps. However, it is possible to add a suitable extra structure to the problem, which can be tuned to precisely cancel the obstruction for a global lift into a nontrivial $2$-sheeted covering. More precisely, instead of trying to lift the $SO(n)$-valued transition functions of $P_{SO(n)}$ to $\spin(n)$, we augment the $SO(n)$ bundle with a $U(1)$ bundle and try to lift the $SO(n)\times U(1)$-valued transition functions into a respective nontrivial $2$-sheeted covering, which is the $\spinc$ group
\[ \spinc(n) := \spin(n)\times_{\Z_2} U(1) \,, \]
where $\times_{\Z_2}$ identifies $(-1,-1)$ and $(1,1)$. Elements of $\spinc(n)$ are thus equivalence classes $[(g,h)]$, where $g\in \spin(n)$ and $h\in U(1)$. The group $\spinc(n)$ fits into the exact sequence
\[ \xymatrix{ 1 \ar[r] & \Z_2 \ar[r] & \spinc(n) \ar[r]^\xi & SO(n)\times U(1) \ar[r] & 1 } \,, \]
where the nontrivial element of $\Z_2$ is mapped into $\spinc(n)$ as $[(-1,1)]=[(1,-1)]$. Let us try to construct a $\spinc$ lift for a set of transition functions $(\{g_{\alpha\beta}\},\{h_{\alpha\beta}\})$ of a $SO(n)\times U(1)$ bundle $P_{SO(n)}\times P_{U(1)}$. Locally, again, such a lift obviously exists, but globally there might be a topological obstruction for patching the local liftings together on triple overlaps. We denote a set of local lifts by $[(\{\widetilde{g}_{\alpha\beta}\},\{\widetilde{h}_{\alpha\beta}\})]$. The map $\xi$ should map these to the original transition functions:
\[ \xi([(\widetilde{g}_{\alpha\beta},\widetilde{h}_{\alpha\beta})]\cdot [(\widetilde{g}_{\beta\gamma},\widetilde{h}_{\beta\gamma})]\cdot[(\widetilde{g}^{-1}_{\alpha\gamma},\widetilde{h}^{-1}_{\alpha\gamma})]) = \xi([(\widetilde{g}_{\alpha\beta}\widetilde{g}_{\beta\gamma}\widetilde{g}^{-1}_{\alpha\gamma},\widetilde{h}_{\alpha\beta}\widetilde{h}_{\beta\gamma}\widetilde{h}^{-1}_{\alpha\gamma})])=1 \,. \]
Thus,
\[ [(\widetilde{g}_{\alpha\beta}\widetilde{g}_{\beta\gamma}\widetilde{g}^{-1}_{\alpha\gamma},\widetilde{h}_{\alpha\beta}\widetilde{h}_{\beta\gamma}\widetilde{h}^{-1}_{\alpha\gamma})] \in \{[(1,1)],[(-1,1)]\} \cong \Z_2 \,, \]
determining a class in $H^2(X,\Z_2)$. Clearly, the $\spinc$ lift exists if and only if this class vanishes (compare to the case of the $\spin$ lift). Remarkably, we now have extra freedom with regards to the $\spin$ lifts $\{\widetilde{g}_{\alpha\beta}\}$, because the equivalence class $[(1,1)]\in \spinc(n)$ consists of two points: $(1,1)$ and $(-1,-1)$. Thus, it is enough to have
\[ \widetilde{g}_{\alpha\beta}\widetilde{g}_{\beta\gamma}\widetilde{g}^{-1}_{\alpha\gamma} = u_{\alpha\beta\gamma} = \widetilde{h}_{\alpha\beta}\widetilde{h}_{\beta\gamma}\widetilde{h}^{-1}_{\alpha\gamma} \]
for the $\spinc$ lift to exist. It is not obvious, though, that such local lifts of the $P_{U(1)}$ bundle exists. The functions $\{\widetilde{g}_{\alpha\beta}\}$ differ from $\{g_{\alpha\beta}\}$ only by signs. This can be expresses in terms of integers $\{n_{\alpha\beta}\}$ as
\[ h_{\alpha\beta} = (-1)^{n_{\alpha\beta}}\widetilde{h}_{\alpha\beta} \,. \]
The cocycle condition yields
\[ h_{\alpha\beta}h_{\beta\gamma}h^{-1}_{\alpha\gamma} = (-1)^{n_{\alpha\beta}+n_{\beta\gamma}-n_{\alpha\gamma}}\widetilde{h}_{\alpha\beta}\widetilde{h}_{\beta\gamma}\widetilde{h}^{-1}_{\alpha\gamma} = (-1)^{\delta(n)_{\alpha\beta\gamma}}u_{\alpha\beta\gamma} = 1\,. \]
Now, $[\delta(\underline{n})]$ is a class in $H^2(X,\Z)$ and $[\underline{u}]$ a class in $H^2(X,\Z_2)$. The equation
\[ (-1)^{\delta(n)_{\alpha\beta\gamma}}u_{\alpha\beta\gamma} = 1 \]
is satisfied if and only if $[\underline{u}] = [\delta(\underline{n})] \text{ mod } 2$. Finally, observe that
\[ h_{\alpha\beta}h_{\beta\gamma}h^{-1}_{\alpha\gamma} = e^{2\pi i(\delta(n)_{\alpha\beta\gamma}/2)}u_{\alpha\beta\gamma} = e^{2\pi i((\delta(n)_{\alpha\beta\gamma}/2) + (\delta(n)_{\alpha\beta\gamma} \text{ mod } 2)/2)} \,. \]
The Chern class $c_1(P_{U(1)})$ is thus given by applying the \v Cech coboundary operator to the expression in the exponential. However, only the half-integer parts of the exponential contribute to the nontriviality of the Chern class. Hence, we can write
\[ c_1(P_{U(1)})=[(\delta(\underline{n})/2) + (\delta(\underline{n})/2)]=[\delta(\underline{n})] \,. \]
This concludes the proof that the $\spinc$ lift of $P_{SO(n)}\times P_{U(1)}$ exists if and only if
\begin{equation}\label{Eq: Spin^c existence 1} w_2(P_{SO(n)}) = [\underline{u}] = c_1(P_{U(1)}) \text{ mod } 2 \,. \end{equation}
By the exactness of \eqref{Eq: S-W exact sequence}, \eqref{Eq: Spin^c existence 1} is equivalent to
\[ W_3(P_{SO(n)}) = \beta(w_2(P_{SO(n)})) = 0 \,. \]
\begin{proposition}\label{Proposition: Spin^c structure} A principal $SO(n)$ bundle $P_{SO(n)}\rightarrow X$ admits a $\spinc(n)$ lift $P_{\spinc(n)}$ precisely when $W_3(P_{SO(n)})=0$. If $W_3(P_{SO(n)})=0$, a choice of the $\spinc$ lift is called a $\spinc$ structure. Distinct $\spinc$ structures of an $SO(n)$ principal bundle are in bijective correspondence with $2H^2(X,\Z)\oplus H^1(X,\Z_2)$. The first term comes from the freedom of tensoring the $P_{U(1)}$ bundle (satisfying $c_1(P_{U(1)})=w_2(P_{SO(n)}) \text{ mod } 2$) by $\widetilde{P}_{U(1)}\otimes \widetilde{P}_{U(1)}$, where $\widetilde{P}_{U(1)}$ is any principal $U(1)$ bundle\footnote{Namely, then $c_1(P_{U(1)}\otimes \widetilde{P}_{U(1)}\otimes \widetilde{P}_{U(1)}) \mod 2 = w_2(P_{SO(n)})$.}.
The second term comes from the freedom of choosing the lifts as in the $\spin$ case. \qed \end{proposition}
Note that any bundle with a $\spin$ structure carries a canonical $\spinc$ structure, given by the trivial class in $H^2(X,\Z)$, corresponding to a trivial $U(1)$ bundle. Also, the underlying $SO(2n)$ bundle of any $SU(n)$ bundle is canonically $\spinc$ by proposition \ref{Proposition: Chern class, S-W class relation}.

We have discussed only principal bundles so far, but there are analogous concepts for vector bundles through the associated bundle construction. More precisely, we call an $O(n)$ or $SO(n)$ vector bundle orientable, $\spin$ or $\spinc$, if their frame bundles are orientable, $\spin$ or $\spinc$, respectively.

\begin{definition} A manifold is said to be orientable/$\spin$/$\spinc$, if its tangent bundle is orientable/$\spin$/$\spinc$. \end{definition}

\section{$K$-theory of vector bundles}
In this section we discuss the basic concepts of $K$-theory (of vector bundles). A reader interested in more comprehensive expositions is referred to \cite{Ati,Kar,LM,GVF}. Our treatment is quite minimalistic, containing only the basics needed to understand the rest of the thesis.

$K$-theory can be defined for a wide class of topological spaces. We have no need for very general kinds of spaces as we are only interested in $K$-theory of the spacetime and the D-brane worldvolume, which are both manifolds. Thus, it is more than enough to suppose that the spaces are locally compact Hausdorff topological spaces of the homotopy type of a CW-complex. However, we shall start by restricting to compact Hausdorff spaces with the homotopy type of a finite CW-complex (for example, compact manifolds).

Let $\mathrm{Vect}_\C(X)$ denote the set of isomorphism classes of complex\footnote{$K$-theory can as well be defined for real vector bundles. $K$-theory for real vector bundles is typically referred to as \emph{$KR$-theory}. Just as complex $K$-theory classifies D-branes in Type II superstring theory, $KR$-theory classifies D-branes in Type I superstring theory. $KR$-theory is also more complicated than $K$-theory and thus the $K$-theoretic aspects are more pronounced in Type I theory. We recommend \cite{Val,OS} for readers interested in the real case and its relation to D-brane physics.} vector bundles over the compact space $X$. For vector bundles $E$ and $F$, their classes in $\mathrm{Vect}_\C(X)$ are denoted by $[E]$ and $[F]$. Defining an operation
\[ [E]+[F] := [E\oplus F] \,, \]
turns $\mathrm{Vect}_\C(X)$ into an Abelian monoid, with identity given by the $0$-dimensional bundle.
\begin{definition} \label{Definition: K-theory} The $K$-theory group $K(X)$ is the \emph{Grothendieck group} of $\mathrm{Vect}_\C(X)$. Explicitly, it is the quotient of the free Abelian group generated by $\mathrm{Vect}_\C(X)$ by the subgroup of elements of the form $[E\oplus F]-[E]-[F]$. In other words, it is the free Abelian group of formal differences $[E]-[F]$ of isomorphism classes of vector bundles over $X$. $K$-theory classes are called \emph{virtual bundles} and $\rk:K(X)\rightarrow \Z$,
\[ \rk([E]-[F]):=\rk(E)-\rk(F) \,, \]
is the \emph{virtual rank} of $[E]-[F]$. The class $[E]-[0]\in K(X)$ is denoted by $[E]$.

The Grothendieck group $K(X)$ can also be axiomatically characterized by the following universal property. If $G$ is any other group and $f:\mathrm{Vect}_\C(X) \rightarrow G$ any monoid homomorphism, then there exists a unique group homomorphism $g:K(X)\rightarrow G$, such that the diagram
\begin{equation}\label{Eq: K-theory universality property} \xymatrix{ \mathrm{Vect}_\C(X) \ar[rr]^{E\mapsto [E]} \ar[dr]_f & & K(X) \ar@{-->}[ld]^g \\ & G & } \end{equation}
commutes. By universality, if such a group $K(X)$ exists (and it does by the above standard construction), then it is necessarily unique.
\end{definition}

\begin{proposition}
If $X$ is a finite disjoint union of compact spaces $X=\amalg_k X_k$, then the $K$-theory group decomposes into a direct sum
\[ K(X) \cong \bigoplus_k K(X_k) \,. \]
There is nothing to prove really, since vector bundles over $X$ are completely described by their restrictions onto $\{X_k\}$.
\qed \end{proposition}

\begin{proposition}[\cite{LM}]\label{Proposition: Complementary bundles exist}
For any vector bundle $E$ over $X$, there exists a complementary bundle $E^\perp$, such that $E\oplus E^\perp \cong \mathbbm{1}_n$, where $\mathbbm{1}_n$ denotes the trivial bundle of rank $n$.
\end{proposition}
By the above proposition, any class $[E]-[F]$ can be written as
\[ [E]-[F] = [E]-[F\oplus F^\perp]+[F^\perp] = [E\oplus F^\perp]-[\mathbbm{1}_n] \,. \]
In other words, every class of $K(X)$ is of the form $[E]-[\mathbbm{1}_n]$, for some vector bundle $E$ and non-negative integer $n$.

\begin{proposition}
$K$-theory defines a cofunctor from the homotopy category of ''suitable spaces'', for example compact Hausdorff spaces, to the category of Abelian groups. A morphism $f:X\rightarrow Y$ gets mapped into a morphism $f^*:K(Y)\rightarrow K(X)$, given explicitly by
\[ f^*([E]-[F]) := [f^{-1}E]-[f^{-1}F] \,. \]
Since the isomorphism classes of pullback bundles depend only on the homotopy class of the map inducing the pullback, the domain of the cofunctor is, indeed, the homotopy category. It follows immediately that homotopy equivalent spaces have isomorphic $K$-theory.
\qed \end{proposition}

\begin{definition}
For vector bundles $E$ and $E'$ over $X$, the tensor product $E\otimes E'$ is again a vector bundle over $X$. This allows us to define a $\Z$-linear map, the \emph{cup product},
\begin{equation}\label{Eq: K-theory cup product} \smile:K(X)\otimes K(X) \rightarrow K(X) \end{equation}
by
\begin{align*} ([E]-[F])\smile([E']-[F']) :&= [E\otimes E']-[E\otimes F']-[F\otimes E']+[F\otimes F'] \\ &= [(E\otimes E')\oplus (F\otimes F')]-[(E\otimes F')\oplus (F\otimes E')] \,. \end{align*}
The \emph{exterior product} is a map
\begin{equation} \label{Eq: K-theory generalized cup product} \otimes:K(X)\otimes K(Y) \rightarrow K(X\times Y) \,, \end{equation}
defined as follows. Note first, that for vector bundles $E\rightarrow X$ and $F\rightarrow Y$ there is a tensor product bundle $E\otimes F\rightarrow X\times Y$. Let $E$ and $F$ be vector bundles over $X$, and $E'$ and $F'$ over $Y$. The projections  $\pi_X:X\times Y \rightarrow X$ and $\pi_Y: X \times Y \rightarrow Y$ induce maps $\pi^*_X$ and $\pi^*_Y$ in $K$-theory. To simplify the definition, we write
\[ E\boxtimes E' := \pi^{-1}_X E \otimes \pi^{-1}_Y E' \]
and
\[ [E]\boxtimes [E'] := \pi^*_X[E] \otimes \pi^*_Y[E'] = [E\boxtimes E'] \in K(X\times Y) \,. \]
Using this notation, the exterior product is
\begin{align*} ([E]-[F])\otimes([E']-[F']) :&= [E]\boxtimes [F']-[E]\boxtimes [F']-[F]\boxtimes [E']+[F]\boxtimes [F'] \\ &= [(E\boxtimes E')\oplus (F\boxtimes F')]-[(E\boxtimes F')\oplus (F\boxtimes E')] \,. \end{align*}
For $Y=X$, the combined map $K(X)\otimes K(X) \rightarrow K(X\times X) \rightarrow K(X)$, where the second arrow is induced by the diagonal inclusion $X\hookrightarrow X\times X$, coincides with \eqref{Eq: K-theory cup product}.
\end{definition}
The cup product turns $K(X)$ into a ring, with multiplicative identity given by the class $[\mathbbm{1}]$, the class of the trivial line bundle.

The $K$-group of a one-point space $\{\infty\}$ is $K(\{\infty\})\cong \Z$, where the integer corresponding to a virtual bundle is simply its virtual rank. Let $X^\infty$ be a compact pointed space, with basepoint $\infty$. The inclusion $i:\{\infty\} \hookrightarrow X$ induces a map
\[ i^*:K(X^\infty)\rightarrow K(\{\infty\})\cong \Z \,, \]
which maps a virtual bundle to the virtual rank of its restriction onto the connected component of the basepoint.
\begin{definition} The subset
\[ \widetilde{K}(X^\infty) := \ker(i) \subset K(X^\infty) \]
is an ideal without identity. It is called the \emph{reduced $K$-theory} group of $X^\infty$. It fits into the naturally split exact sequence
\[ \xymatrix{ 0 \ar[r] & \widetilde{K}(X^\infty) \ar[r] & K(X) \ar[r]^-{i^*} & K(\{\infty\}) \cong \Z \ar[r] & 0 } \]
The splitting yields the isomorphism
\[ K(X^\infty) \cong \widetilde{K}(X^\infty)\oplus \Z \,. \]
\end{definition}
For a connected pointed space\footnote{If $X$ has $n <\infty$ connected components, the virtual rank function is
\[ \rk:K(X)\rightarrow \bigoplus_1^n \Z \,, \]
which assigns an integer, the virtual rank of the restriction of the virtual bundle, to each component.} $X^\infty$, $\widetilde{K}(X^\infty)$ consists of virtual bundles $[E]-[F]$, with $\rk([E]-[F])=0$. If $X$ is a non-pointed compact space, we can always augment a disjoint basepoint $\infty$. The resulting pointed space is
\[ X^\infty := X\amalg \{\infty\} \,. \]
Clearly, if $X$ is non-pointed, we have
\[ \widetilde{K}(X^\infty) \cong K(X) \,. \]

Let $(X,Y)$ be a compact pair of spaces, $\varnothing\neq Y\subset X$. Collapsing $Y$ to a point $\infty$ yields the canonically pointed space $X/Y$ (with basepoint $\infty$). For $Y=\varnothing$, we set $X/\varnothing := X^\infty = X\amalg \{\infty\}$.
\begin{definition}
The \emph{relative $K$-theory} group of the pair $(X,Y)$ is
\[ K(X,Y) := \widetilde{K}(X/Y) \,. \]
\end{definition}
Intuitively, relative $K$-theory is the $K$-theory of $X$, with all contribution from $Y$ removed. It is built from virtual bundles over $X$ that vanish over $Y$. The need to use reduced $K$-theory can, thus, be understood as the requirement of having bundles of equal rank forming the virtual bundle, for otherwise it would be impossible for them to cancel each other over $Y$. For a pointed space $X^\infty$ with basepoint $\infty$, we recover reduced $K$-theory from relative $K$-theory as\footnote{Remark, that these are all standard constructions in (generalized) (co)homology.}
\[ \widetilde{K}(X^\infty) = K(X,\{\infty\}) \,. \]
Remark also, that
\[ K(X,\varnothing) = \widetilde{K}(X^\infty) \cong K(X) \,, \]

$K$-theory actually provides a multitude of functors, $K^{-n}$, where $n\in \N$\footnote{The purpose of the minus sign in $K^{-n}$ is to emphasize the cohomological nature of $K$-theory.}. Before defining the higher $K$-functors, let us quickly recall a few standard constructions from topology of pointed spaces \cite{Whi,May}. Let $X$ and $Y$ be pointed spaces with basepoints $x_0$ and $y_0$. Their \emph{wedge sum} $X\vee Y$ is obtained by forming the disjoint union and collapsing the basepoints together\footnote{This is the coproduct in the category of pointed spaces.}:
\[ X \vee Y := X\amalg Y/\{x_0,y_0\} \,. \]
There is a natural inclusion $X\vee Y\hookrightarrow X\times Y$ whose image is $X\times \{y_0\}\cup \{x_0\}\times Y$. If we collapse $X\vee Y\subset X\times Y$ into a point, we obtain the \emph{smash product},
\[ X\wedge Y := X\times Y/X\vee Y \,. \]
\begin{lemma}
\[ S^n \cong S^1\wedge \ldots \wedge S^1 \]
\end{lemma}
\begin{proof}
First, rewrite $X\wedge S^1$ using $S^1\cong I/\{0,1\}$ as
\begin{equation}\label{Eq: Reduced suspension, clear form} X\wedge S^1 = \frac{X\times I}{X\times \{0,1\} \cup \{x_0\}\times I} \,. \end{equation}
Next, consider the diagram
\begin{equation}\label{Eq: Suspension diagram, S^n wedge S^1} \xymatrix{ S^n\times I \ar[dd]^\pi \ar[r]^\iota & S^{n+1} \\ \\ (S^n\times I)/(S^n\times \{0,1\}\cup \{x_0\}\times I) \ar[uur]_{\widetilde{\iota}} } \end{equation}
where $\iota$ is the obvious (continuous) inclusion $S^n\times I \rightarrow S^{n+1}$ combined with a continuous deformation of all the basepoints into a single point. It is obviously surjective but not injective, since the subspace $S^n\times \{0,1\}\cup \{x_0\}\times I$ gets mapped to a single (base)point. Thus, passing to the quotient space \eqref{Eq: Reduced suspension, clear form} yields the natural bijection $\widetilde{\iota}$. The diagram \eqref{Eq: Suspension diagram, S^n wedge S^1} obviously commutes.

Let $V\subset S^{n+1}$ be open. Then $\iota^{-1}(V) = \pi^{-1}\circ \widetilde{\iota}^{-1}(V)\subset S^n\times I$ is open. A set $U\subset(S^n\times I)/(S^n\times \{0,1\}\cup \{x_0\}\times I)$ is open in the quotient topology precisely when $\pi^{-1}(U)$ is open in $S^n\times I$. Therefore, $\widetilde{\iota}^{-1}(V)$ is open, which proves the continuity of $\widetilde{\iota}$.

Let then $W\subset (S^n\times I)/(S^n\times \{0,1\}\cup \{x_0\}\times I)$ be closed. Its image under $\widetilde{\iota}$ is
\[ \widetilde{\iota}(W) = \widetilde{\iota}\circ \pi \circ \pi^{-1}(W) = \iota \circ \pi^{-1}(W) \,, \]
where $\pi^{-1}(W)$ is closed by the quotient topology. Since $S^n\times I$ is compact, $\pi^{-1}(W)$ is compact. Its image under the continuous map $\iota$ is again compact. Finally, since $S^{n+1}$ is Hausdorff, all compact subsets of $S^{n+1}$ are closed. Thus, $\widetilde{\iota}$ is closed, which concludes the proof.
\end{proof}
The smash product of any pointed space $X$ with $S^1$ is called the \emph{reduced suspension} of $X$ and is denoted by
\[ \Sigma X := S^1 \wedge X \,. \]
The $n$-fold reduced suspension is then
\[ \Sigma^n X := S^1 \wedge \ldots \wedge S^1 \wedge X \cong S^n \wedge X \,. \]

\begin{theorem}[\cite{Ati,Kar}]
The $K$-functor we have discussed above is actually the $0$-th $K$-functor $K^0$. The higher $K$-functors $K^{-n}$, $n\in \N$, are also contravariant functors from the homotopy category of spaces to the category of Abelian groups. For a compact pointed space $X^\infty$, we set
\[ \widetilde{K}^{-n}(X^\infty) := \widetilde{K}(\Sigma^n X^\infty) \,, \]
for a compact pair $(X,Y)$ the higher relative $K$-theory groups are
\[ K^{-n}(X,Y) := \widetilde{K}^{-n}(X/Y) = \widetilde{K}(\Sigma^n(X/Y)) \,. \]
For a generally non-pointed space $X$ we define
\begin{equation}\label{Eq: Higher K-theory} K^{-n}(X) := K^{-n}(X,\varnothing) = \widetilde{K}^{-n}(X^\infty) \,. \end{equation}
These are all standard constructions in (generalized) cohomology theory.

The exterior product \eqref{Eq: K-theory generalized cup product} restricts to a product \cite{Ati}
\begin{equation}\label{Eq: Exterior product in reduced K-theory} \otimes:\widetilde{K}(X)\otimes \widetilde{K}(Y) \rightarrow \widetilde{K}(X\wedge Y) \,, \end{equation}
which is enough to generalize the product structures to relative and higher $K$-theory. Furthermore, choosing $Y=X$ and taking the pullback by the diagonal inclusion $X\hookrightarrow X\times X$, yields corresponding generalizations of \eqref{Eq: K-theory cup product}. For example,
\[ \smile:K^{-n}(X)\otimes K^{-m}(X) \rightarrow K^{-n-m}(X) \]
is defined by first using \eqref{Eq: Higher K-theory}, then \eqref{Eq: Exterior product in reduced K-theory} to obtain a class in $\widetilde{K}(\Sigma^n(X^\infty)\wedge \Sigma^m(X^\infty))$, then using the fact that $\wedge$ is commutative to obtain a class in $\widetilde{K}(\Sigma^{n+m}(X^\infty \wedge X^\infty))$, then taking the pullback by the diagonal map to obtain a class in $\widetilde{K}(\Sigma^{n+m} X^\infty)$ and finally using \eqref{Eq: Higher K-theory} again to obtain a class in $K^{-n-m}(X)$.

It should not come as a surprise that $K$-theory satisfies all of the Eilenberg-Steenrod axioms except the dimension axiom, thus determining a \emph{generalized cohomology theory}.
\end{theorem}

The failure of the dimension axiom is seen as follows. If $\{\infty\}$ is considered as a non-pointed space, then $\{\infty\}^\infty = S^0$, the $0$-sphere, a set of two points. It is easy to see that $\Sigma S^0 \cong S^1$. Thus, we have
\[ K^{-1}(\{\infty\}) = \widetilde{K}(\Sigma S^0) = \widetilde{K}(S^1) \,, \]
and more generally,
\[ K^{-n}(\{\infty\}) = \widetilde{K}(S^n) \,. \]
It turns out, that $\widetilde{K}(S^\text{even}) \cong \Z$, which is contrary to the dimension axiom, that all higher cohomology groups of a point should vanish. However, all the other Eilenberg-Steenrod axioms do hold. For example, associated to a compact pair $(X,Y)$ there is a long exact sequence
\begin{equation}\label{Eq: K-theory long exact sequence} \xymatrix{ \ldots \ar[r] & K^{-2}(Y) \ar[r] & K^{-1}(X,Y) \ar[r] & K^{-1}(X) \ar `r_l[dll] `[dlll] `^r[dll] [dll] \\  & K^{-1}(Y) \ar[r] & K^0(X,Y) \ar[r] & K^0(X) \ar[r] & K^0(Y) } \end{equation}
induced by the \emph{Puppe sequence} \cite{May}. If $X^\infty$ has a basepoint $\{\infty\}$, \eqref{Eq: K-theory long exact sequence} yields the naturally split exact sequences
\[ \xymatrix{ 0 \ar[r] & K^{-n}(X^\infty,\{\infty\}) \cong \widetilde{K}^{-n}(X^\infty) \ar[r] & K^{-n}(X^\infty) \ar@<+0.5ex>[r] & \ar@<+0.5ex>[l] K^{-n}(\{\infty\}) \ar[r] & 0 } \]
which imply the isomorphisms
\begin{equation}\label{Eq: K-theory vs reduced K-theory} K^{-n}(X^\infty) \cong \widetilde{K}^{-n}(X^\infty) \oplus K^{-n}(\{\infty\}) \cong \widetilde{K}^{-n}(X^\infty) \oplus \widetilde{K}(S^n) \,. \end{equation}
It is a general result in cohomology theory, that a reduced cohomology theory for pointed spaces determines a relative cohomology theory for pairs of non-pointed spaces and \emph{vice versa} \cite{May}. In the case of $K$-theory, $\widetilde{K}^{-\bullet}$ is the reduced cohomology theory corresponding to $\widetilde{K}^{-\bullet}$. It is convenient to restrict the spaces to be CW-complexes. Otherwise, we would have to make sure that the inclusions of the basepoints are cofibrations \cite{May}. In particular, this is the case for manifolds.

Generalized cohomology and homology theories, such as $K$-theory, admit always a homotopic definition in terms of an \emph{$\Omega$-spectrum (a loop spectrum)}\footnote{Just like ordinary cohomology and homology can be defined using the Eilenberg-MacLane spectrum.} \cite{Bro1,Whi}. Unfortunately, we can not discuss the general principles here, as it would take us too far from the topic. Instead, we treat only the special case of $K$-theory. It should not be surprising, considering the homotopic classification of vector bundles, that (the $\Omega$-spectrum of) $K$-theory should be somehow closely related to $\{BU(n)\}$. However, $K$-theory involves considering vector bundles of \emph{all} ranks, which suggests that all the $\{BU(n)\}$ spaces should be taken into account. Thus, the limit
\[ BU := \varinjlim_n BU(n) \,, \]
taken over the natural inclusions $BU(n)\subset BU(n+1)$ is expected to play a role. We shall not prove the following, even though the proof is fairly short and easy.
\begin{theorem}[\cite{May}]
Give $\Z$ the discrete topology and choose the basepoint of $BU\times \Z$ to reside in the component $BU\times \{0\}$. All homotopies are understood to be in the category of pointed spaces (they preserve the basepoint). Then, for any compact pointed space $X^\infty$, there exists an isomorphism
\begin{equation}\label{Eq: Homotopic reduced K-theory} \widetilde{K}(X^\infty) \cong [X^\infty,BU\times \Z] \,. \end{equation}
If $X^\infty$ is connected, this is equivalent to
\begin{equation*}\label{Eq: Homotopic reduced K-theory, connected case} \widetilde{K}(X^\infty) \cong [X^\infty,BU] \,. \end{equation*}
For a non-pointed compact space we first augment a basepoint and then apply the pointed definition. If $X$ is non-pointed, compact and $X^\infty$ denotes the space with a disjoint basepoint added, then there is an isomorphism
\begin{equation}\label{Eq: Homotopic K-theory} K(X) \cong [X^\infty,BU\times \Z] \,. \end{equation}
The ring structure of $K$-theory can also be expressed in terms of the homotopic definition, but we shall not go into that. The reader is referred to \cite{Whi,May} for related results in a more general context.
\end{theorem}

Remarkably, it turns out that only two of all the $K$-theory groups are actually distinct. This is the famous \emph{Bott periodicity theorem}. Detailed proofs can be found in \cite{Ati,Kar} or in any standard treatise on $K$-theory. There are several more or less equivalent formulation of this important and classical result. We shall present some of them, in an attempt to clarify how the $K$-theoretic version is related to the others. We start from the classical result of Bott \cite{Bot}, concerning homotopy groups of $U(n)$, and work our way towards the $K$-theoretic version. 

Using the natural inclusions $U(n)\hookrightarrow U(n+1)$, we can define the limit
\[ U := \varinjlim_n U(n) \,, \]
the ''infinite union'' of unitary groups. The result of Bott concerned the higher homotopy groups of $U$\footnote{Recall that for big enough $n$, $\pi_k(U(n))$ depends only on $k$.}.
\begin{theorem}[Bott periodicity \cite{Bot}] \label{Theorem: Bott periodicity, version 1}
\[ \pi_k(U) \cong \pi_{k+2}(U) \,. \]
\end{theorem}
The long exact sequence of homotopy groups, induced by the fibration
\[ \xymatrix{ U(n) \ar[r] & EU(n) \ar[r] & BU(n) } \]
is
\[ \xymatrix{ \ldots \ar[r] & \pi_k(EU(n)) \ar[r] & \pi_k(BU(n)) \ar[r] & \pi_{k-1}(U(n)) \ar[r] & \pi_{k-1}(EU(n)) \ar[r] & \ldots } \,, \]
By exactness and $\pi_k(EU(n))=\pi_{k-1}(EU(n))=0$, there are isomorphisms\footnote{This can also be understood more straightforwardly as follows: $\pi_k(BU(n)) \cong \mathrm{Vect}^n_\C(S^k)$, but on the other hand, such vector bundles are uniquely classified by the homotopy class of the \emph{clutching function} at the equator, which is homotopy equivalent to $S^{k-1}$. Thus, we have
\[ \pi_k(BU(n)) \cong \mathrm{Vect}^n_\C(S^k) \cong [S^{k-1},U(n)] = \pi_{k-1}(U(n)) \,. \]}
\[ \pi_k(BU(n)) \cong \pi_{k-1}(U(n)) \,. \]
Passing to the limit $n\rightarrow \infty$ and applying Theorem \ref{Theorem: Bott periodicity, version 1} yields
\[ \pi_k(BU) \cong \pi_{k-1}(U) \cong \pi_{k+1}(U) \cong \pi_{k+2}(BU) \,. \]
For based CW-complexes $X,Y$ the suspension and loop space functors are adjoint in the sense that
\begin{equation}\label{Eq: Suspension and loop are adjoint} [\Sigma X,Y]\cong [X,\Omega Y] \,. \end{equation}
Thus,
\[ \pi_k(BU) \cong [S^{k+2},BU] = [\Sigma S^{k+1},BU]\cong [S^k,\Omega^2 BU] = \pi_k(\Omega^2 BU) \,, \]
for all $k\ge 1$. Still, $\pi_0(BU)=0$ and $\pi_0(\Omega^2 BU) \cong \pi_2(BU) \cong \Z$, but replacing $\pi_k(BU)$ with $\pi_k(BU\times \Z)$ does not affect any of the higher homotopy groups, since we are working with basepoint preserving homotopies and we have
\[ \pi_k(BU\times \Z) \cong \pi_k(\Omega^2 BU) \,. \]
Furthermore, an application of \eqref{Eq: Suspension and loop are adjoint} yields
\[ \pi_k(\Omega^2 BU) = [S^k,\Omega^2 BU] \cong [S^{k+2},BU] \cong [S^{k+2},BU\times \Z] \cong \pi_k(\Omega^2(BU\times \Z)) \,. \]
We have obtained a weak homotopy equivalence $BU\times \Z \cong \Omega^2(BU\times \Z)$, which implies a homotopy equivalence by Whitehead's theorem\footnote{$BU$ is a CW-complex, since the direct limit of CW-complexes is again a CW-complex \cite{May}.}.
\begin{theorem}[Bott periodicity]\label{Theorem: Bott periodicity, version 2}
There are homotopy equivalences
\[ BU\times \Z \cong \Omega^2 BU \cong \Omega^2 (BU\times \Z) \,. \]
\end{theorem}
The $K$-theoretic version follows now easily. If $X^\infty$ is a compact pointed space,
\begin{align*} \widetilde{K}^{-n}(X^\infty) &= \widetilde{K}(\Sigma^n X^\infty) \cong [\Sigma^n X^\infty,BU\times \Z] \\ &\cong [\Sigma^{n-2} X^\infty, \Omega^2(BU\times \Z)] \cong [\Sigma^{n-2} X^\infty,BU\times \Z] \\ &\cong \widetilde{K}^{-n+2}(X^\infty) \,. \end{align*}
By \eqref{Eq: K-theory vs reduced K-theory}, we immediately obtain a similar result in (unreduced) $K$-theory.
\begin{theorem}[Bott periodicity]
There are canonical isomorphisms
\[ K^{-n}(X) \cong K^{-n+2}(X)  \,. \]
\end{theorem}
Using Bott periodicity, the sequence \eqref{Eq: K-theory long exact sequence} can be extended infinitely to the right, by defining (for $n\in \Z$)
\[ K^n := \begin{cases} K^0 & \text{for $n$ even,} \\ K^{-1} & \text{for $n$ odd.} \end{cases} \]
From now on, we shall talk only about $K^0$ and $K^1$, and define
\begin{definition}
\[ K^\bullet(X) := K^0(X)\oplus K^1(X)\,, \]
and similarly for relative and reduced $K$-theory.
\end{definition}

An immediate consequence of the Bott periodicity theorem is that the long exact sequence \eqref{Eq: K-theory long exact sequence} truncates to a $6$-term cyclic exact sequence
\[ \xymatrix{ K^0(X,Y) \ar[r] & K^0(X) \ar[r] & K^0(Y) \ar[d] \\ K^1(Y) \ar[u] & K^1(X) \ar[l] & K^1(X,Y) \ar[l] } \]

Let us take another look at \eqref{Eq: K-theory vs reduced K-theory}. It would be very nice to know what the groups $\widetilde{K}(S^n)$ are.
\begin{proposition}
\begin{equation}\label{Eq: K-theory of a point} K^n(\{\infty\}) \cong \widetilde{K}(S^n) \cong \begin{cases} \Z & \text{for $n$ even,} \\ 0 & \text{for $n$ odd.} \end{cases} \end{equation}
\end{proposition}
\begin{proof}
By Bott periodicity, it is clear that $\widetilde{K}(S^n)$ depends only on whether $n$ is even or $odd$. Also, since $\widetilde{K}(S^2)\cong K^{-2}(\{\infty\}) \cong K(\{\infty\})\cong \Z$, the only thing we need to do is to find $\widetilde{K}(S^1)$.

We use the various results that came up when discussing Bott periodicity:
\[ \widetilde{K}(S^1) = [S^1,BU\times \Z] \cong [S^1,BU] = \pi_1(BU) \,. \]
There is a fibration
\[ \xymatrix{ S^{2n-1} \ar[r] & BU(n-1) \ar[r] & BU(n) } \,, \]
with associated long exact sequence of homotopy groups
\[ \xymatrix{ \ldots \ar[r] & \pi_1(S^{2n-1}) \ar[r] & \pi_1(BU(n-1)) \ar[r] & \pi_1(BU(n)) \ar[r] & \pi_0(S^{2n-1}) \ar[r] & \ldots } \,. \]
But if $n\ge 2$, then $\pi_1(S^{2n-1})=\pi_0(S^{2n-1})=0$ and there are isomorphisms
\[ \pi_1(BU(n)) \cong \pi_1(BU(n-1)) \cong \ldots \cong \pi_1(BU(1)) \,. \]
Since $BU(1)=\CP^\infty$, we see from \eqref{Eq: Homotopy of CP^infty} that
\[ \pi_1(BU(n)) = 0 \,. \]
Therefore,
\[ \pi_1(BU) \cong \varinjlim_n \pi_1(BU(n)) = 0 \,, \]
concluding the proof.

A very different kind of proof can be given by a clever application of the theory of Clifford modules \cite{ABS}, where the periodicity can be understood to follows from the periodicity of complex Clifford algebras. We shall not go into that, however.
\end{proof}
By \eqref{Eq: K-theory vs reduced K-theory}
\begin{equation}\label{Eq: K-theory vs reduced K-theory, 2} K^n(X^\infty) \cong \begin{cases} \widetilde{K}^n(X^\infty)\oplus \Z &\text{for $n=0$\,,} \\ \widetilde{K}^n(X^\infty) &\text{for $n=1$\,,} \end{cases} \end{equation}
for any compact pointed space $X^\infty$.

Recall the definition of the Chern character of a vector bundle \eqref{Eq: Chern character of a vector bundle} and its simple behavior under direct sums and tensor products of vector bundles, equations \eqref{Eq: Chern character, additivity}, \eqref{Eq: Chern character, multiplicativity}. In particular, \eqref{Eq: Chern character, additivity} implies that
\[ \ch:\mathrm{Vect}_\C \rightarrow H^\text{even}(X,\Q) \]
is a monoid homomorphism. Thus, by the universal property of $K$-theory mentioned in definition \ref{Definition: K-theory}, it has a unique extension into a group homomorphism
\begin{equation} \label{Eq: Chern character, K-theory} \ch:K(X) \rightarrow H^\text{even}(X,\Q) \,, \end{equation}
which maps a virtual bundle $[E]-[F]$ to $\ch(E)-\ch(F)$. It is easy to see from \eqref{Eq: Chern character, multiplicativity} that \eqref{Eq: Chern character, K-theory} is actually a ring homomorphism\footnote{This is the fundamental reason why the Chern character will be so important to us: it is compatible with the algebraic structure of $K$-theory. The appearance of the Chern character in the D-brane charge formula is already, by itself, strong evidence that $K$-theory should be expected to play a role in classifying D-branes.}. For a pointed space $X^\infty$, restricting $\ch$ to reduced $K$-theory classes yields classes in \emph{reduced cohomology} $\widetilde{H}(X^\infty,\Q)$:
\[ \ch(\widetilde{K}(X^\infty)) \subset \widetilde{H}^\text{even}(X^\infty,\Q) := \ker(H^\text{even}(X^\infty,\Q)\rightarrow H^\text{even}(\{\infty\},\Q) \,. \]
For classes in $K^1(X)$ the Chern character maps
\[ \ch:K^1(X) = \widetilde{K}(\Sigma X^\infty) \rightarrow \widetilde{H}^\text{even}(\Sigma X^\infty,\Q) \cong \widetilde{H}^\text{odd}(X^\infty,\Q) \cong H^\text{odd}(X,\Q) \,. \]
It is a famous result, that the Chern character becomes an isomorphism between rationalized $K$-theory and cohomology:
\[ \ch:K^\bullet(X)\otimes \Q \xrightarrow{\cong} H^\bullet(X,\Q) \,. \]
This means that the difference between $K$-theory and cohomology with rational coefficients is contained in the torsion subgroup.

Let us now get rid of the requirement for compactness and generalize to locally compact spaces. Of course, we are primarily interested in manifolds, which are always locally compact spaces of the homotopy type of a CW-complex. Locally compact spaces always admit one-point compactifications. In particular, if the original space is a manifold, then the compactified space is of the homotopy type of a finite CW-complex. For locally compact spaces $X$ and $Y$ with one-point compactifications $X^\infty$ and $Y^\infty$, morphisms $X\rightarrow Y$ are continuous maps that are restrictions of continuous basepoint preserving maps $X^\infty \rightarrow Y^\infty$. For example, \emph{proper maps}\footnote{A map is \emph{proper}, if inverse images of compact sets are compact.} are morphisms between locally compact spaces, but every morphism between locally compact spaces is not a proper map. However, considering proper maps is enough for our applications, because any continuous map from a compact space to a Hausdorff space, such as the inclusion of the D-brane $\phi:\Sigma\rightarrow X$, is a proper.

Suppose that $X$ is a locally compact space with one-point compactification $X^\infty$. We define $K$-theory for $X$ as in \eqref{Eq: Higher K-theory}:
\begin{equation} \label{Eq: K-theory, noncompact}  K^\bullet(X) := \widetilde{K}^\bullet(X^\infty) \,. \end{equation}
This definition extends $K$-theory into the category of locally compact spaces as a \emph{compactly supported} (generalized) cohomology theory, in the following sense.
\begin{proposition}[\cite{Kar}]\label{Proposition: K-theory compactly supported}
Let $X$ be a locally compact and
\[ \ldots \subset X_k \subset X_{k+1} \subset \ldots \]
an inductive family of open subspaces of $X$, such that every compact subset of $X$ is contained in at least one $X_k$. Then
\[ K^\bullet(X) \cong \varinjlim_k K^\bullet(X_k) \,. \]
\end{proposition}
From the physical point of view, the compactly supported definitions can be understood as ''finite energy assumptions'' for physical objects. More precisely, if the energy or charge density of some physical object is described by a (generalized) cohomology theory, then the restriction to the compactly supported subgroup is equivalent to excluding classes which are ''infinitely nontrivial'', corresponding to infinitely extended nonzero energy or charge densities.

The Chern character generalizes without problems to the noncompact case. Its image lies in \emph{compactly supported cohomology} (compare with \eqref{Eq: K-theory, noncompact} and proposition \ref{Proposition: K-theory compactly supported}) \cite{May}:
\[ H^\bullet_c(X,G) := \ker(H^\bullet(X^\infty)\rightarrow H^\bullet(\{\infty\}) \cong G) \,, \]
where $G$ is some Abelian coefficient group. Thus, the Chern character is a homomorphism
\[ \ch:K^\bullet(X) \rightarrow H^\bullet_c(X,\Q) \,, \]
where, on the right-hand side, $H^\bullet_c(X,\Q)$ denotes $\Z_2$-periodized compactly supported cohomology.

We conclude this section by presenting the \emph{Atiyah-Hirzebruch version of the Grothendieck-Riemann-Roch theorem} \cite{AH1,AH2}. Let $\Sigma$ and $X$ be smooth oriented manifolds, $\Sigma$ compact and $X$ generally noncompact. By the contravariant functoriality of cohomology theories, a smooth map $\phi:\Sigma\rightarrow X$, which is necessarily also proper, induces natural maps
\[ \phi^*:H^\bullet_c(X,\Z)\rightarrow H^\bullet(\Sigma,\Z)\,, \qquad \phi^*:K^\bullet(X)\rightarrow K^\bullet(\Sigma) \,. \]
However, in D-brane theory we need to be able to send $K$-theory and cohomology classes from $\Sigma$ to $X$. There is a general recipe for obtaining such ''wrong way'' functoriality, at least when the spaces are of a particular kind. Let us recall some basic constructions of algebraic topology.

First of all, we restrict to manifolds with an additional structure known as an \emph{orientation}. For a closed\footnote{That is, compact without boundary.} $n$-dimensional manifold $\Sigma$, an orientation (more precisely, we are referring to a \emph{$\Z$-orientation}) is a homology class $[\Sigma]\in H_n(\Sigma,\Z)$, whose restriction to $H_n(\Sigma,\Sigma\\\{x\})\cong \Z$ for any $x\in \Sigma$, is a generator. Taking the cap product with the fundamental class determines an isomorphism between homology and cohomology groups:
\[ \Pd_{\Sigma}(-) := (-)\frown [\Sigma]:H^k(\Sigma,\Z) \xrightarrow{\cong} H_{n-k}(\Sigma,\Z) \,, \]
the \emph{Poincar\'e duality isomorphism}. For a noncompact manifold $X$, fundamental classes can be defined only locally. A choice of local fundamental classes, consistent at overlaps\footnote{This is nontrivial, since $H_n(\Sigma,\Sigma\setminus\{x\})\cong \Z$ has two generators}, determines an \emph{orientation} for the noncompact manifold. There is a Poincar\'e duality isomorphism in the noncompact case also, but it is between homology and compactly supported cohomology:
\[ \Pd_X(-): H^k_c(X,\Z) \xrightarrow{\cong} H_{n-k}(X,\Z) \,, \]
where
\[ H^k_c(X,\Z) := \ker\left(H^k(X^\infty,\Z) \rightarrow H^k(\{\infty\},\Z)\right) \,. \]
Compactly supported cohomology can be defined also for more general (than locally compact) spaces, but the above definition suffices for our purposes. These definitions of orientability coincide with orientability of the tangent bundle, which, again, was equivalent to vanishing of the Stiefel-Whitney class $w_2(\Sigma)$ \cite{Hus}. The definitions also generalize straightforwardly to (co)homology with coefficients in an Abelian ring $R$, or an $R$-module\footnote{The (local) fundamental classes are always classes in homology with coefficients in $R$, not in an $R$-module.} \cite{May}. If a manifold is $\Z$-orientable, then it is also orientable for other coefficients, most importantly $\R$ and $\Q$. If a manifold is $R$-orientable and \emph{not} $\Z$-orientable (defined completely analogously to $\Z$-orientability), then $R=\Z_2$. Thus, every manifold is $\Z_2$-orientable. If a manifold is $R$-orientable, then there is a Poincar\'e duality isomorphism between cohomology and homology with coefficients in an $R$-module $\pi$:
\[ \Pd_X(-): H^k_c(X,\pi) \xrightarrow{\cong} H_{n-k}(X,\pi) \,. \]
\begin{definition}\label{Definition: Homology Gysin map}
Let $X$ and $Y$ be orientable locally compact manifolds and $f:X\rightarrow Y$ a morphism. The \emph{Gysin ''wrong way'' homomorphism},
\begin{equation}\label{Eq: Homology Gysin map} f_*:H_c^\bullet(X,\Z)\rightarrow H_c^{\bullet+\dim(Y)-\dim(X)}(Y,\Z)\,, \end{equation}
is defined by
\begin{equation}\label{Eq: Homology Gysin map, explicit} f_*:H_c^\bullet(X,\Z)\xrightarrow{\Pd_X} H_{\dim(X)-\bullet}(X) \xrightarrow{f_*} H_{\dim(X)-\bullet}(Y) \xrightarrow{\Pd^{-1}_Y} H_c^{\bullet+\dim(Y)-\dim(X)}(X,\Z) \,. \end{equation}
The map $f_*$ is the natural push-forward in homology. Functoriality follows immediately from the definition.
\end{definition}

One immediately starts to wonder if some sort of generalized orientations could be defined for generalized cohomology and homology theories, which would imply the existence of corresponding Poincar\'e duality type isomorphisms. Turns out that the answer is ''yes'' \cite{Whi,HW}. We shall focus on the special case of $K$-theory and \emph{$K$-orientability}.
\begin{theorem}[\cite{ABS}]
A manifold $X$ is \emph{$K$-orientable}, if and only if it is $\spinc$.
\end{theorem}
We would next like to present a Poincar\'e duality isomorphism between $K$-theory and a ''dual'' generalized homology theory. It is not at all obvious that such a homology theory should even exist. However, turns out that all generalized cohomology theories do come with a natural dual homology theory, which admits a simple homotopic definition in terms of the $\Omega$-spectrum of the cohomology theory \cite{Whi}. For $K$-theory, the dual homology theory is called \emph{(spectral) $K$-homology} and is denoted by $K_\bullet(-)$, $\bullet=0,1$. We shall soon discuss $K$-homology in great detail and present a geometric definition, equivalent to the homotopic one. For now, it is enough to know that it exists, it has natural covariant functoriality, and is dual to $K$-theory in the same sense as ordinary cohomology is dual to ordinary homology.
\begin{theorem}[\cite{Whi}]
For a $K$-orientable manifold $X$ of dimension $n$, there exists a $K$-theoretic Poincar\'e duality isomorphism
\begin{equation}\label{Eq: K-theory Poincare duality} \Pd^K_X:K^\bullet(X) \xrightarrow{\cong} K_{n-\bullet}(X) \,. \end{equation}
\end{theorem}
$K$-theoretic Poincar\'e duality immediately yields a $K$-theoretic Gysin homomorphism, completely analogously to definition \ref{Definition: Homology Gysin map}.
\begin{definition}\label{Definition: Gysin map in K-theory}
Let $X$ and $Y$ be $K$-orientable locally compact manifolds and $f:X\rightarrow Y$ a morphism. The \emph{$K$-theoretic Gysin ''wrong way'' homomorphism},
\begin{equation}\label{Eq: K-theory Gysin map} f_!:K^\bullet(X)\rightarrow K^{\bullet+\dim(Y)-\dim(X)}(Y)\,, \end{equation}
is defined by
\begin{equation}\label{Eq: K-theory Gysin map, explicit} f_!:K^\bullet(X)\xrightarrow{\Pd^K_X} K_{\dim(X)-\bullet}(X) \xrightarrow{f_*} K_{\dim(X)-\bullet}(Y) \xrightarrow{(\Pd_Y^K)^{-1}} K^{\bullet+\dim(Y)-\dim(X)}(X) \,. \end{equation}
The map  $f_*$ is the natural push-forward in $K$-homology. Again, functoriality follows immediately from the definition.
\end{definition}

Orientations are usually not unique. For example, a connected manifold always admits exactly two distinct $\Z$-orientations. Let $X$ be a manifold with fixed $\Z$- and $K$-orientations. The $K$-orientation induces a $\Z$-orientation through a \emph{homological Chern character} homomorphism \cite{Jak1}
\[ \ch:K_\bullet(X) \rightarrow H_\bullet(X,\Q) \,, \]
where the right-hand side denotes $\Z_2$-periodized homology. This map is analogous to the cohomological Chern character, for example, it is an isomorphism modulo the torsion subgroup. The $K$-orientation is, analogously to the $\Z$-orientation, given by local $K$-homology classes. Taking the homological Chern character yields local classes in rational cohomology, which constitute a $\Q$-orientation. Since they are in the image of $\ch$, they lift to a $\Z$-orientation. However, the orientation induced by the $K$-orientation might be different from the original orientation. This causes a \emph{commutativity defect} between the Gysin maps and Chern characters. If $X$ and $Y$ are $K$-oriented locally compact manifolds, $d=\dim(Y)-\dim(X)$ and $f:X\rightarrow Y$ a morphism, the diagram
\begin{equation}\label{Eq: Commutativity defect} \xymatrix{K^\bullet(X) \ar[rr]^{f_!} \ar[d]_{\ch} & & K^{\bullet+d}(Y) \ar[d]^{\ch} \\ H_c^\bullet(X,\Q) \ar[rr]_{f_*} & & H_c^{\bullet+d}(Y,\Q) } \end{equation}
does \emph{not} commute. The $\Todd$-class measures the difference between the original $\Z$-orientation and the $\Z$-orientation induced by the $K$-orientation\footnote{This statement will be made precise later.}. The Grothendieck-Riemann-Roch theorem tells us precisely how \eqref{Eq: Commutativity defect} needs to be modified to restore commutativity.
\begin{theorem}[(Atiyah-Hirzebruch-)Grothendieck-Riemann-Roch \cite{AH1,AH2}]\label{Theorem: A-H-G-R-R}
When $X$ and $Y$ are as above, the diagram
\[ \xymatrix{K^\bullet(X) \ar[rr]^{f_!} \ar[d]_{\Todd(X)\smile \ch} & & K^{\bullet+d}(Y) \ar[d]^{\Todd(Y)\smile \ch} \\ H_c^\bullet(X,\Q) \ar[rr]_{f_*} & & H_c^{\bullet+d}(Y,\Q) } \]
commutes. More generally, the theorem holds if $W_3(X)\neq 0$ and $W_3(Y)\neq 0$ as long as
\[ W_3(f) := \beta(w_2(X)-f^*w_2(Y)) = 0 \,, \]
which is equivalent to the vector bundle $TX\oplus f^{-1}TY$ being $\spinc$\footnote{This is easy to show. From \eqref{Eq: S-W summation formula}, we have
\[ 1+w_1(TX\oplus f^{-1}TY)+w_2(TX\oplus f^{-1}TY)+\ldots = (1+w_2(TX)+\ldots)\smile (1+w_2(f^{-1}TY)+\ldots) \,, \]
where we have used the assumption that $X$ and $Y$ are orientable, together with proposition \ref{Proposition: Orientability}. Comparing classes of equal degree on both sides reveals
\[ w_2(TX\oplus f^{-1}TY) = w_2(TX) + w_2(f^{-1}TY) = w_2(TX) + f^*w_2(TY)=w_2(X)+f^*w_2(Y)=w_2(X)-f^*w_2(Y) \,. \]}. The map $f$ is said to be \emph{$K$-oriented}. Of course, if $f$ is $K$-oriented, but $X$ and $Y$ are not, the Gysin maps require more careful thought. However, in our applications such a situation does not arise.
\end{theorem}

\section{D-branes and $K$-theory}\label{Section: RR-charge and $K$-theory}
We have now enough mathematical machinery at our disposal to start exploring the D-brane charge. The necessary physical background was explained in chapter \ref{Chapter: Overview} along with two strong arguments in favor of the $K$-theoretic classification scheme. In this section we go through these arguments again, putting our new understanding of characteristic classes and $K$-theory to use. The role of $K$-theory in D-brane theory was first realized by Minasian and Moore \cite{MM} and the basics were worked out in \cite{Wit1}. A nice overview is also given in \cite{Wit2}.

It was explained in section \ref{Section: Overview, Strings and D-branes} that quasi-classically\footnote{The complete quantum theory of D-branes is, as of yet, unknown. However, a quasi-classical model (of the D-brane charge) can be obtained through Dirac quantization.} a D-brane in Type IIB superstring theory is described by an even-dimensional closed worldvolume manifold $\Sigma$ carrying a principal $PU(n)$ bundle\footnote{Locally, such a bundle arises as the projectivization of a principal $U(n)$ bundle, but globally there may be an obstruction to gluing the $U(n)$ lifts together. Thus, not every $PU(n)$ bundle arises as the projectivization of a $U(n)$ bundle \cite{AS1}. We shall return to this later in more detail.}, the Chan-Paton bundle, together with a continuous map $\phi:\Sigma\rightarrow X$, where $X$ is the $10$-dimensional Riemannian spacetime manifold. Both $\Sigma$ and $X$ are oriented in Type II theories, unoriented in Type I theory, and $X$ is assumed to be $\spin$, since spinor particles obviously exist in reality\footnote{To be absolutely precise, the existence of a $\spin$ structure is inferred from the existence of \emph{uncharged} spinor particles, since only a $\spinc$ structure is required for \emph{$U(1)$-charged} spinors.}. 

A few assumptions must be made at this point. First, we assume that the D-brane worldvolume $\Sigma$ is $K$-oriented (a $\spinc$ manifold). Second, we assume that $PU(n)$ Chan-Paton bundles are always projectivizations of $U(n)$ bundles. Their associated $U(n)$ vector bundles are typically also referred to as Chan-Paton bundles. Recall from section \ref{Section: Overview, The low-energy effective field theory} that the coupling of a D-brane with worldvolume $\Sigma$ and $U(n)$ Chan-Paton (vector) bundle $E\rightarrow \Sigma$ to the RR-fields is of the form\footnote{Square roots and inverse square roots of $\Todd(X)$ (which coincides with $\widehat{A}(X)$ by \eqref{Eq: Real Todd class} since $X$ is $\spin$) are defined by a straightforward expansion:
\[ \frac{1}{\sqrt{\widehat{A}(X)}} \stackrel{\eqref{Eq: A-class expansion}}{:=} \prod_k \sqrt{\frac{\sinh(x_k/2)}{x_k/2}} = \prod_k \sqrt{1+\frac{(x_k/2)^2}{3!}+\ldots} = \prod_k \left( 1 + \frac{1}{2}\Big(\frac{(x_k/2)^2}{3!} + \ldots\Big) + \ldots\right) \,. \] } 
\begin{equation}\label{Eq: D-brane-RR-field coupling} \int_X C\wedge \phi_*(\ch(E)\wedge \Todd(\Sigma))\wedge \frac{1}{\sqrt{\Todd(X)}} \,, \end{equation}
where $C=C_0+C_2+\ldots$ is the total RR-potential polyform. The integral is understood to evaluate the $10$-form part of the even-degree polyform and discard everything else. Such a coupling implies that the total RR-charge (density) of the D-brane is the cohomology class
\begin{equation}\label{Eq: Real RR-charge} Q_\R(\Sigma,E) := \phi_*(\ch(E)\wedge \Todd(\Sigma))\wedge \frac{1}{\sqrt{\Todd(X)}} \in H_{c,dR}^\text{even}(X) \,, \end{equation}
where $\phi_*$ is the Gysin map in cohomology. Since both $\Sigma$ and $X$ are even-dimensional, $\phi_*$ does not change the degree of the class $\ch(E)\wedge \Todd(\Sigma)$ in periodized cohomology. All the quantities appearing in \eqref{Eq: Real RR-charge} take naturally values in rational cohomology, which allows us to work as well with
\begin{equation}\label{Eq: Minasian-Moore formula, intermediate 1} Q_\Q(\Sigma,E) := \phi_*(\ch(E)\smile \Todd(\Sigma))\smile \frac{1}{\sqrt{\Todd(X)}} \in H_c^\text{even}(X,\Q) \,. \end{equation}

Remark, that the Gysin map shifts the degree of $\ch(E)$ by $\dim(X)-\dim(\Sigma)=10-(p+1)=9-p$. Since $\sqrt{\Todd(X)}= 1 +\ldots$ and $\ch(E) = \rk(E) + \ldots$, where the ellipsis denotes classes of higher (even) degree, we can write \eqref{Eq: Minasian-Moore formula, intermediate 1} out as
\[ Q_\Q(\Sigma,E) := \underbrace{\rk(E)\phi_*(1)}_{\in H^{9-p}_c(X,\Q)} + \ldots \,. \]
The first term $\rk(E)\phi_*(1)$ pairs up with the RR-potential $C_{p+1}$ in \eqref{Eq: D-brane-RR-field coupling} to yield a class in $H^{10}_\text{c,dR}(X)$, which can then be integrated. The result is the \emph{D$p$-brane charge} of the D-brane, which is the higher degree charge the D-branes carries. Hence, the RR-charge actually dictates the dimensionality of the D-brane: one can think of a D$p$-brane as a D-brane with highest nonzero RR-charge of degree $p$ (coupling to $C_{p+1}$). The class $\rk(E)\phi_*(1)$ represents the image in spacetime of the Poincar\'e dual the fundamental class of $\Sigma$. The rest of the cohomology classes in $Q_\Q(\Sigma,E)$ are of lower degree and thus they pair up with RR-potentials of higher degree to yield integrable $10$-forms.

Minasian and Moore remarked in \cite{MM}, that by considering $E$ rather as a $K$-theory class $[E]\in K^0(\Sigma)$ than simply a vector bundle allows us to apply the Grothendieck-Riemann-Roch theorem \ref{Theorem: A-H-G-R-R} to \eqref{Eq: Minasian-Moore formula, intermediate 1}, to get
\begin{equation}\label{Eq: Minasian-Moore formula, intermediate 2} Q_\Q(\Sigma,[E]) := \ch(\phi_!([E]))\smile \sqrt{\Todd(X)} \in H_c^\text{even}(X,\Q) \,. \end{equation}
Recall from section \ref{Section: Overview, From tachyon condensation to K-theory}, that an anti-D-brane for a D-brane with RR-charge $Q_\Q(\Sigma,[F])$ is a D-brane carrying the opposite RR-charge
\[ \overline{Q}_\Q(\Sigma,[F]):=-Q_\Q(\Sigma,[F]) \,. \]
Using $\ch(-[F])=-\ch(F)$ in the Minasian-Moore formula \eqref{Eq: Minasian-Moore formula, intermediate 2} yields
\[ \overline{Q}_\Q(\Sigma,[F])=Q_\Q(\Sigma,-[F]) \,. \]
Thus, a configuration of $n$ D-branes with $U(n)$ Chan-Paton bundle $E$ and $m$ anti-D-branes with $U(m)$ Chan-Paton bundle $F$, wrapping the same worldvolume $\Sigma$, carries total RR-charge
\[ Q_\Q(\Sigma,[E])+\overline{Q}_\Q(\Sigma,[F]) = Q_\Q(\Sigma,[E])-Q_\Q(\Sigma,[F]) = \ch(\phi_!([E]-[F]))\smile \sqrt{\Todd(X)} \,. \]
Hence, the $K$-theoretic definition for RR-charge easily generalizes to configurations of coincident D-branes and anti-D-branes. It makes sense to denote the total RR-charge of such a configuration by $Q_\Q(\Sigma,[E]-[F])$. Concluding, the total RR-charge of a stack of D-branes with Chan-Paton bundle $E$ and a stack of anti-D-branes with Chan-Paton bundle $F$, all wrapping the same worldvolume $\Sigma$, is given by the \emph{Minasian-Moore formula}
\begin{definition}[Minasian-Moore \cite{MM}]
\begin{equation}\label{Eq: Minasian-Moore formula} Q_\Q(\Sigma,x) := \ch(\phi_!(x))\smile \sqrt{\Todd(X)} \in H_c^\text{even}(X,\Q) \,, \end{equation}
where $x=[E]-[F]\in K^0(\Sigma)$.
\end{definition}

The above discussion was purely classical. In quantum theory we need to impose a \emph{Dirac quantization condition} on the charges. The general principles are explained in \cite{Fre1,Fre2}. The idea is that Dirac quantization corresponds to lifting the charges from rational (or real) cohomology to some generalized cohomology theory with integral coefficients. Unfortunately, there seems to be no way of determining absolutely what the generalized cohomology theory should be, but one can try to guess it by investigating the structure of the classical theory. Suppose $h^\bullet(-)$ is a generalized cohomology theory with integral coefficients. The Dirac quantization condition is the requirement for the existence of a natural transformation
\begin{equation}\label{Eq: Chern character natural transformation} \ch:h^\bullet(-)\rightarrow H^\bullet(-,\Q)\,, \end{equation}
the ''Chern character'', such that the classical charge lies in the image lattice of $\ch$. Finally, we postulate that the quantized charge actually lives in the integral cohomology theory, not merely in the image in $H^\bullet(-,\Q)$. The distinction being that the integral cohomology theory is sensitive to torsion effects, whereas rational cohomology is obviously not. Thus, the Dirac quantized theory might contain charge configurations which are unstable in the classical theory, but are nevertheless stable in the quantum theory due to topological torsion effects. In classical electromagnetism the integral cohomology theory is simply $H^\bullet(-,\Z)$ with the ''Chern character'' being the natural inclusion. The hypothesis is that for RR-charges the integral cohomology theory is $K$-theory and the ''Chern character'' is the \emph{modified Chern character}
\[ \Ch(-) := \ch(-)\smile \sqrt{\Todd(X)} \,. \]

The classical fields should also be Dirac quantized by lifting them to some kind of integral cohomology theory. Recall, again, the case of electromagnetism, where classically the electromagnetic field is given by a field strength $2$-form $F$. Dirac quantization then implies that $F$ is really the curvature of a connection on a principal $U(1)$ bundle. However, this geometric description somewhat obscures the general principle, namely, principal $U(1)$ bundles with connection provide a geometric model for a certain cohomology class in \emph{Deligne cohomology} \cite{Bry,Joh2}. Deligne cohomology, on the other hand, turns out to be an example of a \emph{differential cohomology theory}, which seems to be the correct mathematical framework for Dirac quantized fields \cite{Fre1,Fre2,HS,Val}. The idea with differential cohomology is that it captures both global topological data and local differential geometric data of the field and combines them into a single package. Both are needed, since it is the differential geometric representative of the field that enters the path integral, but on the other hand, it ignores global topological torsion effects. Another example of differential cohomology is given by the $B$-field, which, purely classically, is described by the $3$-form field strength $H$. Dirac quantization lifts $H$ to a class in a differential cohomology theory, which turns out to be a certain class in Deligne cohomology. However, the Deligne class of the $B$-field is of higher rank than that of the electromagnetic field and is slightly more complicated. In particular, it can not be geometrically represented by a principal bundle with connection, but by a much more complicated object: a \emph{bundle gerbe with connection and curving} \cite{Mur,BCMMS,Joh2,Ste2,Gaj,MS1}. For RR-fields the correct differential cohomology theory seems to be \emph{differential $K$-theory}. A nice exposition of differential $K$-theory, along with applications to D-brane theory, is given in \cite{Val}.

From the point of view of the Minasian-Moore formula \eqref{Eq: Minasian-Moore formula} it seems natural that the generalized cohomology theory for RR-charges should be $K$-theory. More precisely, that the \emph{quantized RR-charge} of a D-brane configuration corresponding to a $K$-theory class $x\in K^0(\Sigma)$\footnote{The class $x$ being constructed as in \eqref{Eq: Minasian-Moore formula}.} should be
\begin{equation}\label{Eq: Quantized RR-charge} Q_\Z(\Sigma,x) := \phi_!(x) \in K^0(X) \,. \end{equation}
The natural transformation \eqref{Eq: Chern character natural transformation} from $K$-theory to rational cohomology, which physically associates to a quantized RR-charge its classical approximation, is clearly given by the \emph{modified Chern character}
\[ \Ch(-) := \ch(-)\smile \sqrt{\Todd(X)} \,. \]

Recall that the role of the class $\Todd(X)$ was to cure the commutativity defect in \eqref{Eq: Commutativity defect}, arising from the difference between the fixed orientation of $X$ and the orientation induced by the $K$-orientation. This is further emphasized by the following \emph{isometric pairing formula}. Suppose $X$ is compact. There is a natural bilinear pairing $(-,-)_K:K^0(X)\otimes_\Z K^0(X) \rightarrow \Z$, given by
\[ ([E],[F])_K := \ind(\Dirac_{E\otimes F}) \,, \]
with the obvious bilinear extension to classes of a more general kind. Similarly, there is a natural bilinear pairing $(-,-)_H:H_c^\bullet(X,\Q)\otimes_\Q H_c^{10-\bullet}(X,\Q)\rightarrow \Q$, defined as the obvious evaluation pairing of the cup product of the cup product of the two classes with the fundamental class:
\[ (\eta,\theta)_H := \ip{\eta\smile \theta , [X]} \,. \]
\begin{proposition}[Isometric pairing formula]\label{Proposition: Isometric pairing formula}
The modified Chern character is an isometry with respect to the pairings in $K$-theory and cohomology.
\end{proposition}
\begin{proof}
The proof is a straightforward application of the Atiyah-Singer index theorem \cite{LM}:
\begin{align*} ([E],[F])_K &= \ind(\Dirac_{E\otimes F}) \stackrel{\text{A-S}}{=} \ip{\ch(E\otimes F)\smile \Todd(X), [X]} \\ &\stackrel{\eqref{Eq: Chern character, multiplicativity}}{=} \ip{\Big(\ch(E)\smile \sqrt{\Todd(X)}\Big)\smile \Big(\ch(F)\smile\sqrt{\Todd(X)}\Big),[X]} \\ &= \Big(\Ch([E]),\Ch([F])\Big)_H \,. \end{align*}
The bilinear extension to more general $K$-theory classes is trivial.
\end{proof}

In conclusion, merely the presence of the Chern character in the coupling of D-branes to RR-fields suggests that $K$-theory might somehow be related to RR-charge. This hypothesis is strongly supported by the Minasian-Moore formula, which expresses the coupling in terms of a spacetime $K$-theory class constructed from the Chan-Paton bundle. These arguments were developed further in \cite{Wit1}, along with an explanation of the relation to an earlier observation by Sen \cite{Sen} that a bound state of $n$ coincident D$p$-branes and $n$ D$\overline{p}$-branes carrying $U(n)$ Chan-Paton bundles $E$ and $F$, respectively, contains a \emph{tachyonic mode} in its open string spectrum. The appearance of a tachyon indicates that the configuration has not yet reached a stable minimum of potential energy and is thus unstable. Unfortunately, the dynamical process of the tachyon field \emph{rolling} to a stable minimum can not be described rigorously without complete quantum theory, which we do not have. However, given a particular tachyon field it is possible to figure out what the final state of the tachyon field is after the rolling process has finished (recall section \ref{Section: Overview, From tachyon condensation to K-theory}). The tachyon field is a section of the bundle $E\otimes F^*$. When $E$ and $F$ are not isomorphic, this bundle is nontrivial, causing the tachyon field to run into a topological obstruction when trying to reach its vacuum expectation value globally. Instead, it collapses into a topologically stable solitonic ''vortex'' configuration, with nonzero energy density and RR-charge smeared over some subspace of $\Sigma$. Remarkably, the remaining energy density and RR-charge correspond precisely to that of a lower dimensional D-brane. This phenomenon is reflected by the Minasian-Moore formula as
\begin{align*} Q_\Q(\Sigma,[E]-[F]) &= \phi_*((\ch(E)-\ch(F))\smile \Todd(\Sigma))\smile\frac{1}{\sqrt{\Todd(X)}} \\ &= \phi_*((\underbrace{\rk(E)-\rk(F)}_{=0} + c_1(E)-c_1(F) + \ldots)\smile \Todd(\Sigma))\smile\frac{1}{\sqrt{\Todd(X)}} \\ &= \underbrace{\phi_*(c_1(E)-c_1(F))}_{\in H_c^{11-p}(X)} + \ldots \,, \end{align*}
where the ellipsis denotes cohomology classes of higher degree. Thus, the highest RR-potential such a configuration can couple to is $C_{p-1}$. The class $\phi_*(c_1(E)-c_1(F))$ obviously represents the image in spacetime homology of the Poincar\'e dual of the fundamental class of the residual D$(p-2)$-brane in $\Sigma$.

The $K$-theoretic classification extends also to Type IIA and Type I D-branes: Type IIA D-branes are classified by $K^1(X)$ and Type I D-branes by \emph{real $K$-theory}. We shall not discuss these, but refer the interested reader to \cite{Wit1,Hor,Val,OS}.

There is still one big problem left, namely, everything above relied heavily on the two assumptions that the Chan-Paton bundle is the projectivization of a $U(n)$ bundle and that the $K$-theoretic Gysin map exists, which required $\phi:\Sigma\rightarrow X$ to be $K$-oriented. Since $X$ is $\spin$ (and hence also $\spinc$), $K$-orientability of $\phi$ is equivalent to $\Sigma$ being $\spinc$. But on the other hand, there does not seem to be any obvious physical reason why either one of these assumptions should hold. Before trying to solve these problems, we take a look at an alternative, perhaps more natural, description of D-branes in terms of \emph{$K$-homology}.

\section{Geometric $K$-homology}
Recall that reduced $K$-theory could be homotopically defined for a compact pointed manifold $X^\infty$, as
\[ \widetilde{K}^n(X^\infty) := [X^\infty,\mathbb{K}_n] \,, \]
where the spectrum $\mathbb{K}_n$ is
\[ \mathbb{K}_n \cong \begin{cases} BU\times \Z & \text{for $n$ even,} \\ U := \varinjlim_k U(k) & \text{for $n$ odd.} \end{cases} \]
The associated reduced homology theory, call it \emph{(reduced) spectral $K$-homology}, is given by \cite{Whi}
\[ \widetilde{K}_n(X^\infty) := \varinjlim_k \pi_{n+k}(X^\infty\wedge \mathbb{K}_k) \,. \]
More generally, for any (locally compact) manifold $X$, we set
\[ K_n(X) := \widetilde{K}_n(X^\infty) \,. \]
There is then a Poincar\'e duality isomorphism between $K$-theory and spectral $K$-homology \cite{Whi}, whenever the manifold is $K$-oriented. A compact $K$-oriented manifold has a (global) fundamental $K$-homology class, which establishes the Poincar\'e duality by a homotopically defined cap product with the $K$-theory classes. The homotopic definitions are not very intuitive, however, and it would be convenient to have a geometric definition for $K$-homology\footnote{The geometric definition (in terms of vector bundles) for $K$-theory is easy to deduce from the spectrum. For $K$-homology the task seems to be much more difficult.}. It would then be interesting to see how the Poincar\'e dual of the quantized RR-charge of a D-brane configuration, which is a spacetime $K$-theory class, is represented geometrically in $K$-homology. The geometric definition for $K$-homology was initially introduced by Baum and Douglas in \cite{BD}. A more general approach to geometric representations for homology theories was developed in \cite{Jak1,Jak2}.

\emph{Geometric $K$-homology cycles}, or \emph{$K$-cycles} for short, on a compact manifold $X$ are represented by triples $(\Sigma,x,\phi)$, where $\Sigma$ is a closed $\spinc$ manifold, $x\in K^0(X)$ and $\phi:\Sigma\rightarrow X$ is a continuous map. Two such triples $(\Sigma_1,x_1,\phi_1)$ and $(\Sigma_1,x_2,\phi_2)$ are isomorphic if there exists a diffeomorphism $f:\Sigma_1\rightarrow \Sigma_2$ such that the $\spinc$ structures are compatible under $f^*$, $x_1 = f^*x_2$, and the diagram
\[ \xymatrix{ \Sigma_1 \ar[r]^f \ar[rd]_{\phi_1} & \Sigma_2 \ar[d]^{\phi_2} \\ & X } \]
commutes. Consider then the free Abelian group generated by the isomorphism classes, quotiented by the subgroup of classes of the form
\begin{equation}\label{Eq: Addition of geometric K-cycles} (\Sigma_1\amalg \Sigma_1,x,\phi) - (\Sigma_1,x|_{\Sigma_1},\phi|_{\Sigma_1}) - (\Sigma_2,x|_{\Sigma_2},\phi|_{\Sigma_2}) \,. \end{equation}
The resulting Abelian group is the group of $K$-cycles, with addition given by the disjoint union. Henceforth, we shall make no distinction between a triple $(\Sigma,x,\phi)$ and the $K$-cycle it represents. The group of $K$-cycles is naturally $\Z_2$-graded, given by the parity of the dimension of $\Sigma$. The sets of $K$-cycles with $\Sigma$ of even and odd dimension are denoted by $\Gamma^0(X)$ and $\Gamma^1(X)$, respectively. 

To obtain $K$-homology, the group of $K$-cycles needs to be quotiented by three equivalence relations.
\begin{description}
\item{\textbf{$\spinc$ bordism:}} Let $(\Sigma_1,x_1,\phi_1)$ and $(\Sigma_2,x_2,\phi_2)$ be two $K$-cycles on $X$. They are said to be \emph{bordant},
\[ (\Sigma_1,x_1,\phi_1) \sim_b (\Sigma_2,x_2,\phi_2)\,, \]
if there exists a ''$K$-cycle'' $(\Sigma,x,\phi)$, where $W$ is now a compact $\spinc$ manifold \emph{with boundary} $\partial \Sigma$, interpolating between the two $K$-cycles in the sense that
\[ (\partial \Sigma,x|_{\partial \Sigma},\phi|_{\partial \Sigma}) \cong (\Sigma_1,x_1,\phi_1)+(-\Sigma_2,x_2,\phi_2)\,. \]
Here $-\Sigma_2$ denotes $\Sigma_2$ with inverse $\spinc$ structure\footnote{Recall that a $\spinc$ structure on $\Sigma$ is described by a class $c_1(P_{U(1)})\in H^2(\Sigma,\Z)$, whose mod $2$ reduction is $w_2(\Sigma)$. The inverse $\spinc$ structure is the one corresponding to the class $c_1(P_{U(1)}^*) = -c_1(P_{U(1)})$.}.
\item{\textbf{Direct sum:}} For $K$-cycles $(\Sigma,x_k,\phi)$, $k=1,2$, we set
\[ (\Sigma,x_1+x_2,\phi) \sim_s (\Sigma,x_1,\phi)+ (\Sigma,x_2,\phi)\,. \]
\item{\textbf{Vector bundle modification:}} Let $(\Sigma,x,\phi)$ be a $K$-cycle and $F \rightarrow \Sigma$ a real $\spinc$ vector bundle of even rank\footnote{Remark, that $\Sigma$ is not necessarily connected and $F$ can very well have different (but still even) rank on different components.}. If $\mathbbm{1}\rightarrow \Sigma$ is the trivial real line bundle, the ($\spinc$) vector bundle $\pi:F\oplus \mathbbm{1}\rightarrow \Sigma$ has odd rank fibres. Choosing a Riemannian metric on $F\oplus \mathbbm{1}$ yields a split exact sequence
\[ \xymatrix{ 0 \ar[r] & \pi^{-1}(F\oplus \mathbbm{1}) \ar[r] & T(F\oplus \mathbbm{1}) \ar@<+0.5ex>[r] & \ar@<+0.5ex>[l] \pi^{-1}(T\Sigma) \ar[r] & 0 } \]
of vector bundles over $F\oplus \mathbbm{1}$ and the direct sum decomposition
\[ T(F\oplus \mathbbm{1}) \cong \pi^{-1}(F\oplus \mathbbm{1})\oplus \pi^{-1}(T\Sigma) \,. \]
Finally, \eqref{Eq: S-W summation formula} implies (using $w_1(F\oplus \mathbbm{1})=w_1(T\Sigma)=0$) that
\[ w_2(T(F\oplus \mathbbm{1})) = w_2(\pi^{-1}(E\oplus \mathbbm{1})) + w_2(\pi^{-1}(T\Sigma)) \,. \]
Since both $w_2(\pi^{-1}(E\oplus \mathbbm{1}))$ and $w_2(\pi^{-1}(T\Sigma))$ are mod $2$ reductions of integral cohomology classes, so is their sum. Thus, $T(F\oplus \mathbbm{1})$ is a $\spinc$ vector bundle with a canonical choice for the $\spinc$ structure. Choosing a smooth metric on $F\oplus \mathbbm{1}$ enables us to define the associated \emph{unit sphere bundle},
\[ \mathbb{S}(F\oplus \mathbbm{1}) \rightarrow \Sigma \,, \]
whose fibres are the (even-dimensional) unit spheres of the fibres of $F\oplus \mathbbm{1}$. The $\spinc$ structure of $T(F\oplus \mathbbm{1})$ restricts to a $\spinc$ structure for $T\mathbb{S}(F\oplus \mathbbm{1})$, turning $\mathbb{S}(F\oplus \mathbbm{1})$ into a $\spinc$ manifold. The bundle $\mathbb{S}(F\oplus \mathbbm{1})\rightarrow \Sigma$ admits a nowhere vanishing constant section $\sigma$, sending each point of $\Sigma$ to $0\oplus 1$. We set
\[ (\Sigma,x,\phi) \sim_v (\mathbb{S}(F\oplus,\mathbbm{1}),\sigma_!(x),\phi\circ \pi) \,. \]
\end{description}
This concludes the construction of the necessary equivalence relations. All three are obviously compatible with the grading $\Gamma^\bullet(X)$. \emph{Geometric $K$-homology} is now defined as follows\footnote{Our definition is the one given in \cite{Jak1}, but it is isomorphic to that of \cite{BD,BHS,RS}. The proof is not difficult \cite{Val}.}.
\begin{definition}[\cite{Jak1,BD,BHS,RS}]
The geometric $K$-homology groups of $X$ are the Abelian groups
\[ K^g_\bullet(X)=\Gamma^\bullet(X)/\sim \,, \]
where $\sim$ is the equivalence relation generated by $\sim_b$, $\sim_s$ and $\sim_v$. The equivalence class of a $K$-cycle $(\Sigma,x,\phi)$ is denoted $[\Sigma,x,\phi]$.
\end{definition}
The inverse of a class $[\Sigma,E,\phi]$ is $[-\Sigma,E,\phi]$, where $-\Sigma$ denotes $\Sigma$ with inverse $\spinc$ structure. The manifold $\Sigma\amalg -\Sigma$ is the entire boundary of a higher-dimensional manifold and, thus, bordant to the empty manifold $\varnothing$. The zero element is the class $[\varnothing,0,0]$.

A continuous map $f:X\rightarrow Y$ induces a homomorphism of Abelian groups
\[ f_*:K^g_\bullet(X)\rightarrow K^g_\bullet(Y) \]
by
\[ f_*[\Sigma,x,\phi] := [\Sigma,x,f\circ \phi]\,. \]
Homotopic maps induce the same map in $K^g_\bullet$. This follows easily from bordism, since vector bundles, and hence $K$-theory classes of $\Sigma$ extend to $K$-theory classes of $\Sigma\times [0,1]$. Hence, $K^g_\bullet$ are covariant functors from the homotopy category of finite CW-complexes, or some other suitable category containing compact manifolds, to the category of Abelian groups. 

It is not at all easy to verify that $K^g_\bullet$ actually satisfy the axioms of a $2$-periodic (generalized) homology theory. It is, however, an immediate consequence of the following result.
\begin{theorem}[\cite{Jak1,Val}, see also \cite{BHS,Kas3}]
For any smooth manifold $X$ of the homotopy type of a finite CW-complex, spectral $K$-homology is isomorphic to geometric $K$-homology:
\[ K^g_\bullet(X) \cong K_\bullet(X) \,. \]
In particular, this holds for compact manifolds.
\end{theorem}

For a compact pointed manifold $X^\infty$ the collapsing map $X^\infty \rightarrow \{\infty\}$ induces a map $K_\bullet^g(X^\infty)\rightarrow K^\bullet(\{\infty\})$. 
\begin{definition}
The \emph{reduced geometric $K$-homology groups} are
\[ \widetilde{K}_\bullet^g(X^\infty) := \ker(K_\bullet^g(X^\infty)\rightarrow K^\bullet(\{\infty\})) \,. \]
\end{definition}
Like the $K$-theory groups, the $K$-homology groups of $X^\infty$ decompose as
\[ K^g_\bullet(X^\infty) \cong \widetilde{K}^g_\bullet(X^\infty) \oplus K^g_\bullet(\{\infty\}) \,. \]
For the one-point space $\{\infty\}$ we have (for example, by Poincar\'e duality and \eqref{Eq: K-theory of a point})
\[ K_n^g(\{\infty\}) \cong \begin{cases} \Z & \text{for $n=0$\,,} \\ 0 & \text{for $n=1$\,.} \end{cases} \]
Thus (compare to \eqref{Eq: K-theory vs reduced K-theory, 2}), we have
\begin{equation*}\label{Eq: K-homology vs reduced K-homology} K^g_n(X^\infty) \cong \begin{cases} \widetilde{K}^g_n(X^\infty)\oplus \Z &\text{for $n=0$\,,} \\ \widetilde{K}^g_n(X^\infty) &\text{for $n=1$\,.} \end{cases} \end{equation*}
\begin{definition}
For any (locally compact) manifold $X$, define
\[ K_\bullet^g(X) := \widetilde{K}_\bullet^g(X^\infty) \,. \]
\end{definition}

The cap product, defined homotopically in spectral $K$-homology \cite{Whi}, translates in the geometric picture to a map
\[ \frown:K^n(X)\otimes K^g_\bullet(X) \rightarrow K^g_{\bullet-n}(X) \]
with a surprisingly simple form \cite{Jak1,RS}.
\begin{definition}\label{Definition: Geometric K-homology cap product}
For $y\in K^0(X)$ and $[\Sigma,x,\phi]\in K^g_\bullet(X)$, set
\[ y\frown [\Sigma,x,\phi] := [\Sigma,\phi^*y \smile x,\phi] \in K^g_\bullet(X) \,, \]
where $\smile$ is the cup product \eqref{Eq: K-theory cup product} in $K$-theory.

For $z\in K^1(X)$, the cap product $z\frown [\Sigma,x,\phi]$ is given as follows. First, take the cup product $\phi^*z\smile x \in K^1(\Sigma)$. The product manifold $\Sigma\times S^1$ is $\spinc$, which allows us to push $\phi^*z \smile x$ into $K^0(\Sigma\times S^1)$ using the Gysin map $\sigma_!$ induced by the map $\sigma(p) = (p,1)$. We set
\[ z\frown [\Sigma,x,\phi] := [\Sigma\times S^1,\sigma_!(\phi^*z\smile x),\phi\circ \pr_\Sigma] \in K^g_{\bullet+1}(X) \,. \]
\end{definition}

For a closed $K$-oriented $n$-dimensional manifold $X$ the Poincar\'e duality isomorphism \eqref{Eq: K-theory Poincare duality} is given by taking the cap product with the $K$-homology fundamental class. In the geometric picture this class is particularly easy to describe \cite{Jak2}. It is
\[ [X] := [X,[\mathbbm{1}],\id_X] \,. \]
Poincar\'e duality is then given explicitly by
\begin{equation}\label{Eq: K-theory K-homology Poincare duality, explicit} K^\bullet(X) \ni y \stackrel{\Pd^K_X}{\longmapsto} y\frown [X] = \begin{cases} [X,y,\id_X] \in K^g_n(X) & \text{for $\bullet=0$\,,} \\ [X\times S^1,\sigma_!(y),\pr_X] \in K^g_{n-1}(X) & \text{for $\bullet=1$\,.} \end{cases} \end{equation}
For noncompact boundaryless manifolds a $K$-orientation determines only local fundamental classes, yielding local Poincar\'e duality isomorphisms. These can then be patched together into a global isomorphism\footnote{Compare to the similar result in ordinary cohomology \cite{May}.}.

We already mentioned earlier that there exists a homological Chern character, a natural transformation from spectral $K$-homology to periodized rational homology. It has a simple description in the geometric picture. For a class $[\Sigma,x,\phi]\in K^g_\bullet(X)$ it is given by \cite{Jak1,RS}
\begin{equation}\label{Eq: Homological Chern character} \ch([\Sigma,x,\phi]) := \Pd_X(\ch(\phi_!(x))\smile \Todd(X)) \in H_\bullet(X,\Q) \,. \end{equation}
Chern characters in $K$-theory and $K$-homology are compatible with cap products in such a way that the diagram
\[ \xymatrix{K^\bullet(X)\otimes K^g_n(X) \ar[rr]^-{\frown} \ar[d]_{\ch\otimes \ch} && K^g_{n-\bullet}(X) \ar[d]^{\ch} \\ H^\bullet_c(X,\Q)\otimes H_n(X,\Q) \ar[rr]^-{\frown} && H_{n-\bullet}(X,\Q) } \]
commutes. The ordinary cohomology and homology rings are considered to be $\Z_2$-periodized.

\section{D-branes and $K$-homology}
It was first suggested by Periwal in \cite{Per} that D-branes may be more ''naturally'' described by $K$-homology than by $K$-theory. In this section we explain why this is the case. It turns out that geometric $K$-homology provides an intuitive ''worldvolume'' description of D-branes, as opposed to the somewhat abstruse $K$-theoretic ''RR-charge'' description.

We claim that a $p-\overline{p}$ D-brane configuration\footnote{We are still working under the assumptions that the Chan-Paton bundles are projectivizations of $U(n)$ bundles, that D-brane worldvolumes are $\spinc$ and that there is no background $B$-field.} corresponding to a $K$-theory class $x\in K^0(\Sigma)$ (wrapping $\phi:\Sigma\rightarrow X$) is described by the geometric $K$-homology class $[\Sigma,x,\phi]\in K^g_0(X)$\footnote{Recall that since we are dealing with Type IIB theory, $\Sigma$ is even-dimensional.}. Or more precisely, the $K$-homology class corresponds to an equivalence class of physically equivalent D-brane configurations, which can be obtained from each other by a smooth physical processes. For example, nucleating a $p-\overline{p}$ pair with isomorphic Chan-Paton bundles should not change the physical equivalence class. For a noncompact spacetime manifold we defined $K$-homology using relative $K$-homology of the one-point compactification. This is, again, the requirement that physical D-brane configurations should not be ''infinitely nontrivial'', which ensures that no infinite charges or energies are taken into account\footnote{We are measuring the D-brane charge with respect to that of the vacuum.}.

Consider a D-brane wrapping $\phi:\Sigma\rightarrow X$ with $U(n)$ Chan-Paton bundle $E$. The corresponding $K$-homology class is then $[\Sigma,[E],\phi]$. To motivate this, we need to find physical interpretations for the various equivalence relations in the definition of $K$-homology \cite{RS,Sza1}. Addition of $K$-homology classes (as in \eqref{Eq: Addition of geometric K-cycles}) means simply that several disconnected D-branes can be treated by a single $K$-homology class.
\begin{description}
\item{\textbf{Direct sum:}} If two D-branes with $U(n)$ and $U(m)$ Chan-Paton bundles are set to be coincident, we would expect to see the Chan-Paton bundles merge into a $U(n+m)$ bundle. This is precisely what the direct sum equivalence relation does. For D-branes $[\Sigma,[E],\phi]$ and $[\Sigma,[F],\phi]$ it implies
\[ [\Sigma,[E],\phi]+[\Sigma,[F],\phi] = [\Sigma,[E\oplus F],\phi] \,. \]
\item{\textbf{$\spinc$ bordism:}} The $\spinc$ bordism relation is also very easy to interpret physically. It states that a smooth deformation of a D-brane and its Chan-Paton bundle does not change the $K$-homology class. For example, if a D-brane can be smoothly deformed to the vacuum, its $K$-homology class, and thus also its RR-charge, vanishes.
\item{\textbf{Vector bundle modification:}} This one is more complicated to interpret. Physically, it represents the \emph{Myers' dielectric effect} \cite{Mye,Sza1,RS}. The following simple example is given in \cite{RS}. Consider a D$(-1)$-brane\footnote{A D$(-1)$-brane, also known as a \emph{D-instanton}, is simply a $0$-dimensional pointlike D-brane localized in spacetime.} $[\{\infty\},[\mathbbm{1}],\iota]\in K^g_0(X)$. Choose $F$ in the definition of vector bundle modification to be the trivial bundle $\{\infty\}\times \R^{2n}$, with $n\ge 1$. The unit sphere bundle $\mathbb{S}(F\oplus \mathbbm{1})$ is simply $\{\infty\}\times S^{2n}\cong S^{2n}$. Therefore, by vector bundle modification,
\[ [\{\infty\},[\mathbbm{1}],\iota] = [S^{2n},\sigma_!([\mathbbm{1}]),\epsilon] \,, \]
where $\epsilon$ collapses $S^{2n}$ into the point in $X$ and $\sigma$ is the constant section given in the definition of vector bundle modification. This is the dielectric effect: the D$(-1)$-brane ''blows up'' into a collection of spherical D$(2n-1)$-branes.
\end{description}

The $K$-theoretic and $K$-homological descriptions of D-branes should obviously be related to each other by Poincar\'e duality. Let us consider, for simplicity the case of a closed spacetime manifold $X$. Recall the explicit form of the Poincar\'e duality isomorphism from \eqref{Eq: K-theory K-homology Poincare duality, explicit}. If $\phi_!(x)\in K^0(X)$ is the $K$-theory class describing the D-brane configuration, then the corresponding $K$-homology class should be $\Pd^K_X(\phi_!(x)) = [X,\phi_!(x),\id_X]$. However, our physical interpretation associated to the D-brane the class $[\Sigma,x,\phi]$. Luckily, these two are in fact equivalent \cite{Jak2}. This is easy to verify in the special case when $\phi:\Sigma\rightarrow X$ is an embedding and $X$ is connected. The normal bundle $N\Sigma$ of $\Sigma$ in $X$ has rank $n=10-\dim(\Sigma)$. Since $n$ is even\footnote{Remember that $\dim(\Sigma)$ is even in Type IIB theory.} and $\spinc$\footnote{This follows immediately from $TX\cong T\Sigma\oplus N\Sigma$, \eqref{Eq: S-W summation formula} and that $X$ and $\Sigma$ are both $\spinc$.}, we can choose $F=N\Sigma$ in the definition of vector bundle modification. The total space of the unit sphere bundle $\mathbb{S}(F\oplus \mathbbm{1})$ is $10$-dimensional and embeds into $X$. Such an embedding is actually a diffeomorphism \cite{RS}. Using the notation in the definition of vector bundle modification, we immediately get (by vector bundle modification and bordism)
\[ [\Sigma,x,\phi] = [\mathbb{S}(F\oplus \mathbbm{1}),\sigma_!(x),\phi\circ \pi] = [X,\phi_!(x),\id_X] \,. \]

The Minasian-Moore formula \eqref{Eq: Minasian-Moore formula} can also be expressed purely in terms of homological quantities by employing the homological Chern character \eqref{Eq: Homological Chern character}:
\begin{align*} Q_\Q(\Sigma,x) &= \ch(\phi_!(x))\smile \sqrt{\Todd(X)} = \ch(\phi_!(x))\smile \Todd(X) \smile \frac{1}{\sqrt{\Todd(X)}} \\ &= \Pd_X^{-1}\circ \ch([\Sigma,x,\phi])\smile \frac{1}{\sqrt{\Todd(X)}} \,. \end{align*}
\begin{definition}[Minasian-Moore]\label{Definition: Homological Minasian-Moore formula}
The homological classical RR-charge of a D-brane configuration $[\Sigma,x,\phi]\in K^g_0(X)$ is
\begin{align*} \widehat{Q}_\Q(\Sigma,x) :&= \Pd(Q_\Q(\Sigma,x)) \\ &= \Pd_X\bigg[\Pd_X^{-1}\circ \ch([\Sigma,x,\phi])\smile \frac{1}{\sqrt{\Todd(X)}}\bigg] \in H_\text{even}(X,\Q) \,. \end{align*}
\end{definition}

\chapter{The Freed-Witten Anomaly}\label{Chapter: Geometric Aspects Of The B-field}
The goal of this chapter is to drop the assumptions used in the previous chapter, that
\begin{enumerate}
\item Chan-Paton bundles lift to $U(n)$ bundles,
\item D-brane worldvolumes are $\spinc$,
\item there is no background $B$-field.
\end{enumerate}
Obviously, the first two assumptions are crucial for the $K$-theoretic and $K$-homological classification schemes discussed above. If a) is not satisfied, we can not form $K$-theory classes from the Chan-Paton bundles. If a) is satisfied but b) is not, we would not be able to describe the D-branes in terms of spacetime quantities, since Gysin maps would fail to exist. The last condition is more mysterious and is required by \emph{Freed-Witten anomaly cancellation}.

Suppose that we have a nonzero locally defined background $B$-field. Let us first take a look at closed strings. The worldsheet of a closed string is a closed surface $M$, embedded into spacetime by $\xi:M\rightarrow X$. The $B$-field couples to the worldsheet as follows \cite{Pol1}:
\[ \exp(iS_\text{closed}) \supset \exp\left(\int_M \xi^*B\right) \,. \]
The exponential of the integral is called the \emph{holonomy} of $B$ over $M$. We shall employ the shorthand notation
\begin{equation}\label{Eq: B-field holonomy, 1} \hol(M,B) := \exp\left(\int_M \xi^*B\right) \end{equation}
for holonomy. Since any two local representatives of the $B$-field differ by a de Rham coboundary $d\Lambda$, the integrals of $\xi^*B$ over $M$ are invariant under gauge transformations $B\mapsto B + d\Lambda$:
\[ \exp\left(\int_M \xi^*(B+d\Lambda)\right) = \exp\left(\int_M \xi^*B\right) \cdot \exp(\left(\int_M d\xi^*\Lambda\right) \stackrel{\text{Stokes}}{=} \exp\left(\int_M \xi^*B\right) \,. \]
Still, $\hol(M,B)$ is not automatically a well-defined function on the space of maps, $\Map(M,X)$. We demand that
\[ \int_M \xi^*B \]
is well-defined on $\Map(M,X)$ up to an integral multiple of $2\pi i$, which leads to the usual Dirac quantization condition that the de Rham class $[H]_\text{dR}$, evaluated on any closed homology cycle, must be integral \cite{Alv}. As was briefly explained in section \ref{Section: RR-charge and $K$-theory}, we postulate that the Dirac quantized $H$-flux should, more correctly, take values in $H^3(X,\Z)$ \cite{Fre1,Fre2}. The distinction is that a class $[H]\in H^3(X,\Z)$ may contain torsion parts, whereas its image in $H^3_\text{dR}(X)$ is necessary nontorsion (for the obvious reason that $\Z_n \otimes \R = 0$)\footnote{Recall also from section \ref{Section: RR-charge and $K$-theory} that the correct mathematical formalism to describe the Dirac quantized $B$-field should be a differential cohomology theory, in this case one associated to $H^\bullet(-,\Z)$ \cite{Fre1,Fre2}. It happens to be isomorphic to something called the (diagonal) \emph{Deligne cohomology} \cite{HS}, which we introduce soon.}.

We are actually more interested in open strings than closed strings. The worldsheet of an open string is a compact oriented surface, again denoted by $M$, but this time with boundary $\partial M$. Again, we denote its embedding into spacetime by $\xi:M\rightarrow X$. The $B$-field couples again to $M$, just like in \eqref{Eq: B-field holonomy, 1}, but now the integral is no longer gauge independent:
\[ \int_M \xi^*(B+d\Lambda)=\int_M \xi^*B + \oint_{\partial M} \xi^*\Lambda \neq \int_M \xi^*B \,. \]
This is a major problem, because it means that the path integral is ill-defined.

Recall from \emph{Introduction And Overview}, that D-branes could be thought of as boundary conditions for open strings. This manifests in the open string path integral as a coupling of the boundary, $\partial M$ (corresponding to the open string endpoints), to the D-brane. More precisely, recall that the D-brane carries a Lie algebra valued locally defined $1$-form ''gauge'' field $A$. The boundary $\partial M$ couples to $A$ by \cite{Pol1}
\[ \Tr\hol(\partial M,A) := \Tr \exp\left(\oint_{\partial M} \xi^*(A)\right) \,. \]
Now, if $A$ is a gauge connection on a principal $U(n)$ bundle, $\Tr\hol(\partial M,A)$ is the well-defined parallel transport of $A$ around $\partial M$ \cite{Nak}. However, these are the kind of assumptions we want to get rid of, but it is unclear if the path integral is well-defined otherwise.

There is a third difficulty associated to the fermionic part of the worldsheet path integral, namely the \emph{Pfaffian} of the worldsheet Dirac operator \cite{FW},
\[ \pfaff({\Dirac}_\xi)=\exp\left(\int d\psi\,\bar{\psi}i\Dirac_\xi \psi\right)\,, \]
is ill-defined as a function on the space $\Map(M,X)$.

Concluding, the open string path integral contains the terms
\begin{equation}\label{Eq: Open string path integral, problematic} \exp(iS_\text{open}) \supset \pfaff(\Dirac_\xi) \cdot \hol(M,B) \cdot \Tr\hol(\partial M,A) \,, \end{equation}
each of which is possibly ill-defined. This is called the \emph{Freed-Witten anomaly} \cite{FW}. In the next sections, we shall explain how each of these terms can be given a rigorous definition and how the ill-defined parts precisely cancel each other, when a certain cohomological condition,
\[ [H]|_\Sigma = W_3(\Sigma) + \beta([\underline{\zeta}]) \,, \]
is satisfied \cite{FW,Kap,CJM}. Here $[H]|_\Sigma$ is the pullback of $[H]$ onto $\Sigma$ and $\beta([\underline{\zeta}])\in H^3(X,\Z)$ is the obstruction for lifting the $PU(n)$ Chan-Paton bundle into a $U(n)$ bundle. For example, if $[H]|_\Sigma=0$ and the Chan-Paton bundle is a $U(n)$ bundle, then the worldvolume $\Sigma$ must necessarily be $\spinc$, just like we assumed above. Because $W_3(\Sigma)$ is $2$-torsion and $\beta([\underline{\zeta}])$ $n$-torsion, $[H]|_\Sigma$ must also be pure torsion. This means that the classical $H$-flux $3$-form must vanish, when restricted onto $\Sigma$. The $K$-theoretic classification of D-branes in the presence of nontorsion $[H]|_\Sigma$ was presented as a puzzle in \cite{Wit1}. To cancel the anomaly, we would somehow have to take the torsion degree of $\beta([\underline{\zeta}])$ to infinity. This would correspond to taking the number of D-branes to infinity. It was finally explained in \cite{BM} that the correct way to take this limit is to allow the Chan-Paton bundles to be an infinite-dimensional projective bundle, a $PU(\Hilb):=U(\Hilb)/U(1)$ bundle, where $\Hilb$ denotes the usual infinite-dimensional separable Hilbert space and $U(1)$ the center of $U(\Hilb)$.

\section{Deligne cohomology}
We shall start with a brief review of the basics of \emph{sheaf theory}, \emph{sheaf cohomology} and \emph{sheaf hypercohomology}, which are used to define Deligne cohomology. We then explain, how the obscure sheaf hypercohomology groups can be computed relatively easily using \emph{\v Cech hypercohomology}. Our treatment is based on \cite{Bry}, but we tend to make more use of the language of category theory than \cite{Bry} does. We omit the proofs, since they would take too many pages to present and also, because the material in this section is not used much in the rest of the text.

Let $\mathbf{Set^o}(X)$ denote the category of open sets of $X$, with morphisms the inclusions, and $\mathbf{Ab}$ the category of Abelian groups.
\begin{definition}[\cite{Bry}]
A \emph{presheaf} $A$ of Abelian groups on $X$ is a contravariant functor\footnote{We are only interested in presheaves and sheaves of Abelian groups. However, in most of the following definitions, the category $\mathbf{Ab}$ can easily be replaced by any other category.}
\[ A:\mathbf{Set^o}(X) \rightarrow \mathbf{Ab} \,. \]
For an inclusion $V\subset U$ of open sets of $X$, we denote the homomorphism $A(V\subset U):A(U)\rightarrow A(V)$ by
\[ A(U)\ni \alpha_U \mapsto \alpha_U|_V \in A(V) \]
and call it the \emph{restriction map}. We also assume that $A(\emptyset)=0$.
\end{definition}
Let $\mathcal{V}_U = \{V_i\}$ be an open covering of $U\in \mathbf{Set^o}(X)$. A family of elements $\{\alpha_i\in A(V_i)\}$ is said to be \emph{compatible}, if $\alpha_i|_{V_i\cap V_j} = \alpha_j|_{V_i\cap V_j}$ for all $i,j$. A compatible family is said to satisfy the \emph{glueing condition}, if there exists a unique $\alpha_U \in U$, such that $\alpha_i = \alpha_U|_{V_i}$ \mbox{for all $i$}.
\begin{definition}
The presheaf $A$ is a \emph{sheaf}, if every compatible family of elements satisfies the glueing condition.
\end{definition}
If $A$ and $B$ are (pre)sheaves $\mathbf{Set^o}(X)\rightarrow \mathbf{Ab}$, a \emph{morphism of (pre)sheaves} is a natural transformation $\phi:A\rightarrow B$. In other words, it associates to every $U\in \mathbf{Set^o}(X)$ a homomorphism $\phi(U):A(B)\rightarrow B(U)$, such that the diagram
\[ \xymatrix{ A(U) \ar[d]_{|_V} \ar[r]^{\phi(U)} & B(U) \ar[d]^{|_V} \\ A(V) \ar[r]^{\phi(V)} & B(V) } \]
commutes. (Pre)sheaves (of Abelian groups) on $X$, together with their morphisms, form a category $\mathbf{(Pre)Sheaf}(X)$. For $Y\in \mathbf{Set^o}(X)$, there is a cofunctor
\[ |_Y:\mathbf{(Pre)Sheaf}(X) \rightarrow \mathbf{(Pre)Sheaf}(Y) \,, \]
the \emph{restriction functor}, which assigns to a (pre)sheaf $A\in \mathbf{(Pre)Sheaf}(X)$ the (pre)sheaf $A|_Y$, defined by $A|_Y(V) = A(V)$, for $V\in \mathbf{Set^o}(Y)$.

The \emph{stalk} of a (pre)sheaf $A$ at $x\in X$ is the direct limit
\[ A_x := \varinjlim_{U\ni x} A(U) \,. \]
For any open set $V \ni x$ of $X$, the direct limit provides a natural restriction map
\[ |_x : A(V) \rightarrow A_x \,. \]

Elements of $A(U)$ are called \emph{(local) sections} of $A$ over $U$.  Elements of $A(X)$ are the \emph{global sections} of $A$. The \emph{local section functor} is the bifunctor
\[ \Gamma(-,-):\mathbf{Set^o}(X) \times \mathbf{(Pre)Sheaf}(X) \rightarrow \mathbf{Ab} \,, \]
defined by
\[ \Gamma(U,A) := A(U) \,. \]
Restricting to $U=X$ yields the \emph{global section functor}, $\Gamma(X,-)$. A section $\alpha_U \in A(U)$ is \emph{continuous at $x\in U$}, if there exists an open set $V\subset U$, such that $x\in V$ and $\alpha_U|_V = \alpha_U|_x$. A section of $A(U)$ is \emph{continuous}, if it is continuous at every $x\in U$. The \emph{continuous section functor},
\[ \widetilde{\Gamma}(-,-):\mathbf{Set^o}(X) \times \mathbf{(Pre)Sheaf}(X) \rightarrow \mathbf{Ab} \,, \]
assigns to the pair $(U,A)\in \mathbf{Set^o}(X) \times \mathbf{(Pre)Sheaf}(X)$ the subgroup of continuous sections of $A(U)$.
\begin{definition}[\cite{Bry}]
Let $A$ be a presheaf. Its \emph{associated sheaf}, or \emph{sheafification}, is a sheaf $\widetilde{A}$, given by
\[ \widetilde{A}(U) := \widetilde{\Gamma}(U,A) \,. \]
A sheaf $A$ is also a presheaf and its sheafification $\widetilde{A}$ coincides $A$ \cite{Bry}.
\end{definition}

Let $\phi:A\rightarrow B$ be a morphism of sheaves (of Abelian groups) on $X$. Its \emph{kernel}, $\ker(\phi)$, is a sheaf on $X$, given by
\[ \ker(\phi)(U) := \ker\big(\phi(U):A(U)\rightarrow B(U)\big) \,, \]
for $U\in \mathbf{Set^o}(X)$. Let $P$ be the assignment
\[ P(U) := \im\big(\phi(U):A(U)\rightarrow B(U)\big) \,. \]
Then $P$ defines a presheaf on $X$. The \emph{image} of $\phi$ is
\[ \im(\phi) := \widetilde{P} \,. \]
These definitions enable us to talk about \emph{exact sequences of sheaves}, \emph{complexes of sheaves} and \emph{cohomology sheaves}. These are defined as one would expect. A sequence
\[ \xymatrix{ \ldots \ar[r] & A^{n-1} \ar[r]^{\phi^{n-1}} & A^n \ar[r]^{\phi^n} & A^{n+1} \ar[r] & \ldots } \]
of sheaves (and morphisms of sheaves) is exact, if $\ker(\phi^{k+1}) = \im(\phi^k)$ at each step. A complex of sheaves, $(K^\bullet,d^\bullet)$, is a sequence of sheaves
\[ \xymatrix{ \ldots \ar[r] & K^{n-1} \ar[r]^{d^{n-1}} & K^n \ar[r]^{d^n} & K^{n+1} \ar[r] & \ldots } \]
where $d^{j+1} \circ d^j = 0$, for all $j$. 

A sheaf $K$ on $X$ is \emph{injective}, if for any morphism of sheaves $\phi:A\rightarrow K$ and an injective morphism\footnote{That is, a morphism with $\ker=0$.} $\iota:A\rightarrow B$, there exists a morphism $\psi:B\rightarrow K$, such that
\[ \xymatrix{ A \ar[r]^{\iota}\ar[d]_\phi & B \ar@{-->}[ld]^\psi \\ K & } \]
commutes. An \emph{injective resolution} of a sheaf $A$ is a complex $(I^\bullet,d^\bullet)$ of injective sheaves, together with a monomorphism $i_A:A\rightarrow I^0$, such that the sequence
\[ \xymatrix{ 0 \ar[r] & A \ar[r]^i & I^0 \ar[r]^{d^0} & I^1 \ar[r]^{d^1} & I^2 \ar[r] & \ldots } \]
is exact. The category of sheaves on $X$ has \emph{enough injectives} in the sense that every sheaf $A$ on $X$ admits an injective resolution \cite{Bry}.

The global section functor $\Gamma(X,-)$ is left-exact: for an exact sequence
\[ \xymatrix{ 0 \ar[r] & A \ar[r] & B \ar[r] & C \ar[r] & 0 } \]
of sheaves, the sequence
\begin{equation}\label{Eq: Global section sequence} \xymatrix{ 0 \ar[r] & \Gamma(X,A) \ar[r] & \Gamma(X,B) \ar[r] & \Gamma(X,C) } \end{equation}
is exact. Since $\mathbf{Sheaf}(X)$ has enough injectives, \eqref{Eq: Global section sequence} can be extended canonically to the right into a long exact sequence
\[ \xymatrix{ 0 \ar[r] & \Gamma(X,A) \ar[r] & \Gamma(X,B) \ar[r] & \Gamma(X,C) \ar `r_l[dll] `[dlll] `^r[dll] [dll] \\ & R^1\Gamma(X,A) \ar[r] & R^1\Gamma(X,B) \ar[r] & R^1\Gamma(X,C) \ar `r_l[dll] `[dlll] `^r[dll] [dll] \\ & R^2\Gamma(X,A) \ar[r] & R^2\Gamma(X,B) \ar[r] & R^2\Gamma(X,C) \ar[r] & \ldots } \]
where
\[ R^j\Gamma(X,-):\mathbf{Sheaf}(X)\rightarrow \mathbf{Ab} \]
are the \emph{right derived functors} of $\Gamma(X,-)$, defined as follows. For a sheaf $F$ on $X$, take any injective resolution $F\xrightarrow{i_F} I^\bullet$. Applying $\Gamma(X,-)$, yields a complex $\Gamma(X,I^\bullet)$ of Abelian groups. The right derived functors $R^\bullet\Gamma(X,-)$ send the sheaf $F$ to the cohomology groups of this complex:
\[ R^\bullet\Gamma(X,F) := \frac{\ker\big(\Gamma(X,I^\bullet) \rightarrow \Gamma(X,I^{\bullet+1})\big)}{\im\big(\Gamma(X,I^{\bullet-1})\rightarrow \Gamma(X,I^\bullet)\big)} \,. \]
\begin{definition}[\cite{Bry}]
The \emph{sheaf cohomology} groups of a sheaf $A$ on $X$, are the Abelian groups
\[ H^\bullet(X,A) := R^\bullet\Gamma(X,A) \,. \]
\end{definition}
It is far from clear that the sheaf cohomology groups are well-defined, for example, that they are independent of the injective resolutions used. The reader is referred to \cite{Bry} for a proof.

Next we want to define \emph{sheaf hypercohomology}, which is a generalization of sheaf cohomology for sequences of sheaves. We need first the concept of an \emph{injective resolution of a complex of sheaves} on $X$. We assume, that all complexes of sheaves are bounded below, meaning that they are of the form
\[ \xymatrix{ K^0 \ar[r]^{d_K} & K^1 \ar[r]^{d_K} & K^2 \ar[r]^{d_K} & \ldots } \]
We refer the reader to \cite{Bry} for a proof of the following. Given a bounded below complex of sheaves $(K^\bullet,d_K)$\footnote{We drop the upper indices from the maps $d_K^\bullet$ to simplify notation.}, there exists a double complex
\[ \xymatrix{ \vdots & \vdots & \vdots &  \\
I^{0,2} \ar[u]_\delta \ar[r]^d & I^{1,2} \ar[u]_\delta \ar[r]^d & I^{2,2} \ar[u]_\delta \ar[r]^d & \ldots \\
I^{0,1} \ar[u]_\delta \ar[r]^d & I^{1,1} \ar[u]_\delta \ar[r]^d & I^{2,1} \ar[u]_\delta \ar[r]^d & \ldots \\
I^{0,0} \ar[u]_\delta \ar[r]^d & I^{1,0} \ar[u]_\delta \ar[r]^d & I^{2,0} \ar[u]_\delta \ar[r]^d & \ldots \\
K^0 \ar[u]_{i_0} \ar[r]^{d_K} & K^1 \ar[u]_{i_1} \ar[r]^{d_K} & K^2 \ar[u]_{i_2} \ar[r]^{d_K} & \ldots } \]
called an injective resolution of $(K^\bullet,d_K)$, where
\begin{enumerate}
\item the complex $(I^{p,\bullet},\delta)$ is an injective resolution of $K^p$,
\item the complex $\im(d) \cap I^{p,\bullet}$ is an injective resolution of $\im(d_K) \cap K^p$,
\item the complex $\ker(d) \cap I^{p,\bullet}$ is an injective resolution of $\ker(d_K) \cap K^p$, and finally
\item the complex $\underline{H}^{p,\bullet}$ of horizontal cohomology sheaves is an injective resolution of $\underline{H}^p(K^\bullet)$.
\end{enumerate}
\begin{definition}[\cite{Bry}]
Let $(K^\bullet,d_K)$ be as above with an injective resolution $I^{\bullet\bullet}$. The sheaf hypercohomology group $H^j(X,K^\bullet)$ is the $j$-th cohomology group of the total complex
\[ \bigg(\bigoplus_{p+q=j} \Gamma(X,I^{p,q}),D:=\delta+(-1)^q\,d\bigg) \,. \]
\end{definition}
Again, it is far from clear that this definition is unambiguous. The proof that it indeed is, is given in \cite{Bry}.

Using sheaf hypercohomology, we can finally define \emph{(smooth) Deligne cohomology}. Let $X$ be a smooth manifold. We denote by $\underline{U(1)}$ the sheaf of $U(1)$-valued functions on $X$ and by $\underline{\Omega}^p$ the sheaf of differential $p$-forms on $X$.
\begin{definition}\label{Definition: Deligne cohomology (sheaf cohomology)}
Let $\mathcal{D}^q$ be the complex of sheaves
\[ \mathcal{D}^q := \xymatrix{ \underline{U(1)}\ar[r]^{d\log} & \underline{\Omega}^1 \ar[r]^d & \ldots \ar[r]^d & \underline{\Omega}^q } \]
Deligne cohomology groups are the sheaf hypercohomology groups $H^\bullet(X,\mathcal{D}^q)$.
\end{definition}
In fact, the original definition of Deligne cohomology is given by a different complex of sheaves \cite{Bry}. However, the definition given above yields the same groups and is much more tractable. The original definition would be required, for example, to define the cup product, but we have no need for such advanced cohomological properties of Deligne cohomology. It is enough for us to have a model for the cocycles and be able to compute the groups, as easily as possible. Unfortunately, computing Deligne cohomology groups still seems to be an extremely difficult task. Even normal sheaf cohomology groups are very difficult to compute, since one has to deal with abstruse injective resolutions. Luckily, for such well-behaving spaces as smooth manifolds, there is a beautiful way around these difficulties.

Let $A$ be a sheaf (of Abelian groups) on a smooth manifold $X$, with a good open cover $\mathcal{U}=\{U_i\}$. To simplify the notation, we write
\[ U_{i_0\ldots i_p} := U_{i_0} \cap \ldots \cap U_{i_p} \,. \]
As usual, the \v Cech complex $(C^\bullet(\mathcal{U},A),\delta)$ is the complex of Abelian groups
\[ C^p(\mathcal{U},A)=\prod_{i_0 \ldots i_p} A(U_{i_0 \ldots i_p})\,, \]
together with the \v Cech coboundary operator $\delta:C^p(\mathcal{U},A)\rightarrow C^{p+1}(\mathcal{U},A)$ given by
\[ \delta(\underline{\alpha})_{i_0\ldots i_{p+1}} = \sum_{j=0}^{p+1} (-1)^j (\alpha_{i_0\ldots i_{j-1} i_{j+1} \ldots i_{p+1}})|_{U_{i_0\ldots i_{p+1}}} \,, \]
for $\underline{\alpha}\in C^p(\mathcal{U},A)$. 
\begin{definition}
The \v Cech cohomology groups of $\mathcal{U}$ with coefficients in the sheaf $A$, are the cohomology groups $\check{H}^\bullet(\mathcal{U},A)$ of the \v Cech complex.
They certainly seem to depend on the open cover $\mathcal{U}$. However, it is a classical result that on such a well-behaving space as a smooth manifold, \v Cech cohomology is independent of the good open cover used \cite{Bry}.
\end{definition}

Let $\mathcal{C}^p(\mathcal{U},A)$ be the sheaf that associates to an open set $V\in \mathbf{Set^o}(X)$ the group $C^p(\mathcal{U}|_V,A|_V)$ These form, in an obvious way, a complex of sheaves $(\mathcal{C}^\bullet(\mathcal{U},A),\delta)$. Moreover, the sequence
\[ \xymatrix{ 0 \ar[r] & A \ar[r]^i & \mathcal{C}^0(\mathcal{U},A) \ar[r]^\delta & \mathcal{C}^1(\mathcal{U},A) \ar[r]^\delta & \mathcal{C}^2(\mathcal{U},A) \ar[r] & \ldots } \]
of sheaves, where $i:A\rightarrow \mathcal{C}^0(\mathcal{U},A)$ is the natural morphism, is exact. It is called the \emph{\v Cech resolution} of $A$. If the sheaves $\mathcal{C}^\bullet(\mathcal{U},A)$ were injective, the \v Cech resolution would be an injective resolution and would provide an extremely convenient way of computing sheaf cohomology. Unfortunately, they are \emph{not} injective. In any case, we can form the complex $\Gamma(X,\mathcal{C}^\bullet(\mathcal{U},A))$ of groups. Its cohomology group are simply the \v Cech cohomology groups $\check{H}^\bullet(\mathcal{U},A)$. There is no \emph{a priori} reason to suspect that \v Cech cohomology would have anything to do with sheaf cohomology. However, the best possible thing happens anyway.
\begin{proposition}[\cite{Bry}]\label{Proposition: Cech cohomology = sheaf cohomology}
\v Cech cohomology groups are isomorphic to sheaf cohomology groups:
\[ \check{H}^\bullet(\mathcal{U},A) \cong H^\bullet(X,A) \,. \]
\end{proposition}
This changes the very difficult problem of computing sheaf cohomology straight from the definition into the relatively simple combinatorial problem of computing \v Cech cohomology. To be able to compute Deligne cohomology, we need an analogous result for sheaf hypercohomology.

Simply put, \emph{\v Cech hypercohomology} is the total cohomology of the natural extension of the \v Cech resolution to a complex of sheaves, just as sheaf hypercohomology was an extension of sheaf cohomology to a complex of sheaves. Let $(K^\bullet,d_K)$ be a sequence of sheaves on $X$. Each $K^p$ admits a \v Cech resolution $\mathcal{C}^\bullet(\mathcal{U},K^p)$, which can be patched together to form a double complex of sheaves. Applying the global sector functor yields a double complex of groups
\begin{equation}\label{Eq: Cech double resolution complex} \xymatrix{ \vdots & \vdots & \vdots & \\
C^2(\mathcal{U},K^0) \ar[u]_\delta \ar[r]^{d_K} & C^2(\mathcal{U},K^1) \ar[u]_\delta \ar[r]^{d_K} & C^2(\mathcal{U},K^2) \ar[u]_\delta \ar[r]^-{d_K} & \ldots \\
C^1(\mathcal{U},K^0) \ar[u]_\delta \ar[r]^{d_K} & C^1(\mathcal{U},K^1) \ar[u]_\delta \ar[r]^{d_K} & C^1(\mathcal{U},K^2) \ar[u]_\delta \ar[r]^-{d_K} & \ldots \\
C^0(\mathcal{U},K^0) \ar[u]_\delta \ar[r]^{d_K} & C^0(\mathcal{U},K^1) \ar[u]_\delta \ar[r]^{d_K} & C^0(\mathcal{U},K^2) \ar[u]_\delta \ar[r]^-{d_K} & \ldots } \end{equation}
\begin{definition}[\cite{Bry}]
The \v Cech hypercohomology groups $\breve{H}^\bullet(\mathcal{U},K^\bullet)$ of $(K^\bullet,d_K)$, are the total cohomology groups of \eqref{Eq: Cech double resolution complex}, with respect to the coboundary operator
\begin{equation}\label{Eq: Double complex coboundary} D:=\delta + (-1)^q\,d_K \,, \end{equation}
where $q$ denotes the \v Cech index. Again, a result in \cite{Bry} states that, at least for manifolds, \v Cech hypercohomology groups are independent of the good open cover used.
\end{definition}
Finally, there is the generalization of proposition \ref{Proposition: Cech cohomology = sheaf cohomology} to hypercohomology.
\begin{proposition}[\cite{Bry}]\label{Proposition: Cech hypercohomology = sheaf hypercohomology}
\v Cech hypercohomology groups are isomorphic to sheaf hypercohomology groups:
\[ \check{H}^\bullet(\mathcal{U},K^\bullet) \cong H^\bullet(X,K^\bullet) \,. \]
\end{proposition}

Applying proposition \ref{Proposition: Cech hypercohomology = sheaf hypercohomology} to the definition \ref{Definition: Deligne cohomology (sheaf cohomology)} immediately yields the following alternative definition.
\begin{definition}
The Deligne cohomology groups $H^\bullet(X,\mathcal{D}^q)$ are isomorphic to the total cohomology groups of the double complex
\[ \xymatrix{ \vdots & & \vdots & & \vdots \\
C^3(\mathcal{U},\underline{U(1)}) \ar[u]_\delta \ar[rr]^{-d\log} & & C^3(\mathcal{U},\underline{\Omega}^1) \ar[u]_\delta \ar[r]^-{-d} & \ldots \ar[r]^-{-d} & C^3(\mathcal{U},\underline{\Omega}^q) \ar[u]_\delta \\
C^2(\mathcal{U},\underline{U(1)}) \ar[u]_\delta \ar[rr]^{d\log} & & C^2(\mathcal{U},\underline{\Omega}^1) \ar[u]_\delta \ar[r]^-d & \ldots \ar[r]^-d & C^2(\mathcal{U},\underline{\Omega}^q) \ar[u]_\delta \\
C^1(\mathcal{U},\underline{U(1)}) \ar[u]_\delta \ar[rr]^{-d\log} & & C^1(\mathcal{U},\underline{\Omega}^1) \ar[u]_\delta \ar[r]^-{-d} & \ldots \ar[r]^-{-d} & C^1(\mathcal{U},\underline{\Omega}^q) \ar[u]_\delta \\
C^0(\mathcal{U},\underline{U(1)}) \ar[u]_\delta \ar[rr]^{d\log} & & C^0(\mathcal{U},\underline{\Omega}^1) \ar[u]_\delta \ar[r]^-d & \ldots \ar[r]^-d & C^0(\mathcal{U},\underline{\Omega}^q) \ar[u]_\delta } \]
with respect to the differentials written therein. The alternating signs arise from \eqref{Eq: Double complex coboundary}. For example, a part of a cocycle might read
\[ \xymatrix{ \ar[r]^-d & 0 \\ & \underline{\alpha} \ar[u]^-\delta \ar[rr]^-{-d} & & \delta(\underline{\beta})-d\underline{\alpha}=0 \\ & & & \underline{\beta} \ar[r]^-d \ar[u]^-\delta & d\underline{\beta} + \delta(\underline{\gamma}) = 0 \\ & & & & \underline{\gamma} \ar[u]^-\delta \ar[r]^-{-d} & 0 \\ & & & & & \ar[u]^-\delta } \]
\end{definition}

We are particularly interested in \emph{diagonal Deligne cohomology} groups $H^p(X,\mathcal{D}^p)$. For example, consider the complex
\[ \xymatrix{ \vdots & \vdots \\
C^2(\mathcal{U},\underline{U(1)}) \ar[u]_\delta \ar[r]^{d\log} & C^2(\mathcal{U},\underline{\Omega}^1) \ar[u]_\delta \\
C^1(\mathcal{U},\underline{U(1)}) \ar[u]_\delta \ar[r]^{d\log} & C^1(\mathcal{U},\underline{\Omega}^1) \ar[u]_\delta \\
C^0(\mathcal{U},\underline{U(1)}) \ar[u]_\delta \ar[r]^{d\log} & C^0(\mathcal{U},\underline{\Omega}^1) \ar[u]_\delta } \]
We denote the set of $p$-fold intersections of open sets of the cover $\mathcal{U}$ by $X_\mathcal{U}^{[p]}$. Note, that $X_\mathcal{U}^{[1]}$ is simply the set $\mathcal{U}$. We shall use the notation $X_\mathcal{U} := X_\mathcal{U}^{[1]}$. Now, a cohomology class in $H^1(X,\mathcal{D}^1)$ is represented by a pair
\[ (\underline{g},\underline{A}) \in \underline{U(1)}(X_\mathcal{U}^{[2]})\oplus \underline{\Omega}^1(X_\mathcal{U}) \,, \]
satisfying
\[ \delta(\underline{g})_{\alpha\beta\gamma} = g_{\alpha\beta}g_{\alpha\gamma}^{-1}g_{\beta\gamma} = 1 \quad \text{and} \quad \delta(\underline{A})_{\alpha\beta} = A_\beta-A_\alpha=d\log g_{\alpha\beta}\,. \]
In other words, a cocycle $(\underline{g},\underline{A})$ can be geometrically interpreted as a principal $U(1)$ bundle $P\rightarrow X$, with transition functions $\{g_{\alpha\beta}\}$ and connection $\{A_\alpha\}$, or equivalently as the associated complex line bundle with connection. The coboundaries are pairs
\[ (\delta(\underline{f})_{\alpha\beta},d\log f_\alpha)=(f_\beta f_\alpha^{-1},d\log f_\alpha)\,. \]
Adding such terms to $(\underline{g},\underline{A})$ corresponds precisely to performing a bundle isomorphism and a corresponding gauge transformation. Hence, elements of $H^1(X,\mathcal{D}^1)$ correspond to isomorphism classes of complex hermitian lines bundles with connection. Recall that the first Chern class of the line bundle is the image of $\underline{g}$ under the Bockstein homomorphism
\[ \beta:H^1(X,\underline{U(1)})\rightarrow H^2(X,\Z) \]
and that its image under the inclusion $H^2(X,\Z)\rightarrow H^2_\text{dR}(X)$ coincides with the cohomology class of $d\underline{A}$. The first Chern class of this line bundle is called the \emph{Dixmier-Douady class} of the Deligne class.

\section{Bundle gerbes}
It will become clear during this chapter that a Dirac quantized $B$-field is described by a class in $H^2(X,\mathcal{D}^2)$\footnote{More correctly, it is a class in a differential cohomology theory, which happens to coincide with $H^2(X,\mathcal{D}^2)$, but we have no need for such a general point of view.}. Dirac quantized electromagnetic fields, which are cohomologically classes in $^1(X,\mathcal{D}^1)$, correspond geometrically to principal $U(1)$ bundles with connection. Likewise, having a geometric model for classes of $H^2(X,\mathcal{D}^2)$ provides new insight to the case of the $B$-field.

There exist several useful realizations of $H^2(X,\mathcal{D}^2)$. One of them is the theory of \emph{gerbes} developed by Giraud \cite{Bry,Moe,Hit}. However, Giraud's gerbes are not very well suited for our purpose and we shall instead employ a bundle theoretic realization by Murray \cite{Mur}, called \emph{bundle gerbes}\footnote{More precisely, \emph{$U(1)$ bundle gerbes}. However, we call them simply bundle gerbes, since all bundle gerbes we are interested in are of this type.} with \emph{bundle gerbe connection and curving}. Our references on bundle gerbes are \cite{Mur,BCMMS,CJM,Joh2,Ste2}.

Let $X$ and $Y$ be smooth manifolds, $\mathcal{U}=\{U_\alpha\}$ be a good open cover of $X$, $\pi:Y\rightarrow X$ a smooth, locally split map, with local sections $s_\alpha:U_\alpha \rightarrow X$. The $p$-fold fibre product $Y^{[p]}$ is defined as
\[ Y^{[p]} = \{(y_1,\ldots,y_p)\in Y\times \ldots \times Y : \pi(y_1)=\ldots=\pi(y_p)\} \,. \]
We also need the maps $\pi_i:Y^{[p]}\rightarrow Y^{[p-1]}$ which simply omit the $i$-th element:
\[ \pi_i(y_1,\ldots,y_p)=(y_1,\ldots,y_{i-1},y_{i+1},\ldots,y_p)\,. \]
\begin{definition}[\cite{Mur}] Let $\pi:Y\rightarrow X$ be a locally split map and $L\rightarrow Y^{[2]}$ a complex line bundle, with a bundle isomorphism
\[ L_{(y_1,y_2)}\otimes L_{(y_2,y_3)}\xrightarrow{\cong} L_{(y_1,y_3)}\,, \]
called the \emph{bundle gerbe product}. Note, that the isomorphism is between $L$ and the restriction of $L\otimes L\rightarrow Y^{[2]}\times Y^{[2]}$ over
\[ Y^{[2]}\circ Y^{[2]}=\{(y_1,y_2,y_3,y_4)\in Y^{[2]}\times Y^{[2]}|y_2=y_3\}\,. \]
The product is required to be associative, whenever triple product are defined. It is shown in \cite{Mur} that the existence of the product necessarily implies $L_{(y,y)}\cong \C$ and $L_{(y_1,y_2)}\cong L^*_{(y_2,y_1)}$, where $L^*$ is denotes dual bundle. The triple $(L,Y,X)$ is called a $U(1)$ bundle gerbe and is represented diagrammatically by
\[ \xymatrix{ L \ar[d] & \\
Y^{[2]} \ar@<+0.5ex>[r]^{\pi_1} \ar@<-0.5ex>[r]_{\pi_2} & Y \ar[d]_\pi \\
& X } \]
To keep the notations simple, we denote $(L,Y,X)$ by $(L,Y)$ or $L$, when the omitted objects are clear from the context.
\end{definition}

Let $(L,Y,X)$ be a bundle gerbe, $Y'\rightarrow X'$ locally split and $\phi:X'\rightarrow X$ a map with a lift
\[ \xymatrix{ Y' \ar[r]^{\widehat{\phi}} \ar[d] & Y \ar[d] \\
X' \ar[r]^{\phi} & X } \,. \]
Now, $\widehat{\phi}$ induces a map $\widehat{\phi}^{[2]}:{Y'}^{[2]}\rightarrow Y^{[2]}$, which can be used to pull back $L\rightarrow Y^{[2]}$ into $(\widehat{\phi}^{[2]})^{-1}L\rightarrow {Y'}^{[2]}$. The bundle gerbe $\phi^{-1}L := ((\widehat{\phi}^{[2]})^{-1}(L),Y',X')$ is called the \emph{pullback of the bundle gerbe $L$}. There are two important special cases of this construction. The first is that we are given a map $\phi:X'\rightarrow X$ and take as $Y'$ the pullback $\phi^{-1}(Y)$. The second is the case where $\phi:X\rightarrow X$ is the identity map. Let $(L,Y)$ be a bundle gerbe. We can then define another bundle gerbe $(L^*,Y)$, called the \emph{dual of $(L,Y)$}, where $L^*$ is the dual of $L$. If $(L,Y,X)$ and $(J,Z,X)$ are two bundle gerbes, there is a \emph{product bundle gerbe} $(L\otimes J,Y\times_{\pi} Z)$, where $Y\times_{\pi} Z := \{(y,z)\in Y\times Z : \pi_Y(y)=\pi_Z(z)\}$.

Let $J\rightarrow Y$ be a complex line bundle. We denote by $(\delta(J),Y)$ the bundle gerbe $(\pi_1^{-1}(J)\otimes\pi_2^{-1}(J)^*,Y)$, where $\delta(J)_{(y_1,y_2)} = J_{y_2}\otimes J_{y_1}^*$, with the natural pairing
\[ \delta(J)_{(y_1,y_2)}\otimes \delta(J)_{(y_2,y_3)}=J_{y_2}\otimes J_{y_1}^*\otimes J_{y_3}\otimes J_{y_2}^* \rightarrow J_{y_3}\otimes J_{y_1}^*=\delta(J)_{(y_1,y_3)}\,, \]
as the bundle gerbe product.

Let $f:Y\rightarrow Z$ be a smooth map, $f^{[2]}:Y^{[2]}\rightarrow Z^{[2]}$ the induced map on the fibre products and $g:L\rightarrow J$ a bundle morphism, covering $f^{[2]}$. The pair $(g,f)$ defines a \emph{bundle gerbe morphism} $(L,Y)\rightarrow (J,Z)$, if it commutes with the bundle gerbe products\footnote{More generally, a bundle gerbe morphism can be defined between bundle gerbes over different base spaces, in which case we would need a third map between the base spaces, with fibre preserving lift $f$. However, we have no need for such general morphisms.}. If $f$ and $g$ are isomorphisms, $(g,f)$ is called a \emph{bundle gerbe isomorphism}. A \emph{trivial bundle gerbe} is a bundle gerbe $(L,Y)$, which is isomorphic to some $(\delta(J),Y)$. The line bundle $J$ together with the isomorphism $(L,Y)\cong (\delta(J),Y)$ is called a \emph{trivialization} of $(L,Y)$. 

Bundle gerbes are completely classified, up to a certain equivalence relation, by $H^3(X,\Z)$, just like isomorphism classes of complex line bundles are classified by $H^2(X,\Z)$. Unfortunately, bundle gerbe isomorphism is not the correct equivalence relation. Instead, we require the concept of \emph{stable isomorphism}. Two bundle gerbes, say $L$ and $J$, are said to be stably isomorphic, if there exists a complex line bundle $K$, such that $L\cong J\otimes \delta(K)$. The bundle $K$ is referred to as the stable isomorphism. It is easy to see that any two stable isomorphisms, $K$ and $H$, between bundle gerbes $L$ and $J$, must be related by $K = H\otimes \pi^{-1}E$, where $E\rightarrow X$ is a complex line bundle.
\begin{proposition}[\cite{Mur}] The group of stable isomorphism classes of bundle gerbes over $X$, with group operator the product of bundle gerbes, is isomorphic to $H^3(X,\Z)$.
The cohomology class corresponding to a particular bundle gerbe, or its stable equivalence class, is referred to as the \emph{Dixmier-Douady class}.
\end{proposition}

The Dixmier-Douady class of a bundle gerbe $(L,Y,X)$ is constructed as follows. Let $\mathcal{U}=\{U_\alpha\}$ be a good open cover of $X$, with local sections $s_\alpha:U_\alpha\rightarrow Y$ of $Y\rightarrow X$. On double intersections $U_{\alpha\beta}$, these can be combined into sections $(s_\alpha,s_\beta):U_{\alpha\beta}\rightarrow Y^{[2]}$, which can then be used to pull $L\rightarrow Y^{[2]}$ back onto each $U_{\alpha\beta}$, yielding a set of $U(1)$ vector bundles
\[ L_{\alpha\beta} := (s_\alpha,s_\beta)^{-1}L\rightarrow U_{\alpha\beta}\,. \]
Because $\mathcal{U}$ is a good cover, $U_{\alpha\beta}$ are all contractible and hence $L_{\alpha\beta}$ are trivial bundles. They admit global sections $\sigma_{\alpha\beta}:U_{\alpha\beta}\rightarrow L_{\alpha\beta}$. The bundle gerbe product $L_{(\alpha,\beta})\otimes L_{(\beta,\gamma)}\rightarrow L_{(\alpha,\gamma)}$ induces a bundle isomorphism $L_{\alpha\beta}\otimes L_{\beta\gamma}\rightarrow L_{\alpha\gamma}$, which can be used to patch $\sigma_{\alpha\beta}$ together on triple overlaps:
\[ \sigma_{\alpha\beta}\otimes \sigma_{\beta\gamma} = \sigma_{\alpha\gamma}g_{\alpha\beta\gamma}\,, \]
where $g_{\alpha\beta\gamma}$ are functions $U_{\alpha\beta\gamma}\rightarrow U(1)$. On quadruple overlaps
\[ \sigma_{\alpha\beta}\otimes\sigma_{\beta\gamma}\otimes\sigma_{\gamma\delta}=\sigma_{\alpha\gamma}\otimes\sigma_{\gamma\delta}g_{\alpha\beta\gamma} =\sigma_{\alpha\gamma}g_{\alpha\beta\gamma}g_{\alpha\gamma\delta} \,, \]
but on the other hand
\[ \sigma_{\alpha\beta}\otimes\sigma_{\beta\gamma}\otimes\sigma_{\gamma\delta}=\sigma_{\alpha\beta}\otimes\sigma_{\beta\delta}g_{\beta\gamma\delta} =\sigma_{\alpha\gamma}g_{\alpha\beta\delta}g_{\beta\gamma\delta} \,. \]
Thus, we conclude that $\{g_{\alpha\beta\gamma}\}$ satisfy the \v Cech cocycle condition
\[ g_{\alpha\beta\gamma}^{-1}g_{\alpha\beta\delta}g_{\alpha\gamma\delta}^{-1}g_{\beta\gamma\delta}=1\,, \]
and thus define an element $[\underline{g}]\in H^2(X,\underline{U(1)})$. 

The exact sequence of sheaves
\[ \xymatrix{ 0 \ar[r] & \underline{\Z} \ar[r] & \underline{\R} \ar[r] & \underline{U(1)} \ar[r] & 0 } \]
induces a long exact sequence
\[ \xymatrix{ \ldots \ar[r] & H^2(X,\underline{\R}) \ar[r] & H^2(X,\underline{U(1)}) \ar[r]^-\beta & H^3(X,\underline{\Z}) \cong H^3(X,\Z) \ar[r] & H^3(X,\underline{\R}) \ar[r] & \ldots } \]
of cohomology groups. It is easy to show using a partition of unity argument that $H^p(X,\underline{\R})$ vanish for $p>0$ \cite{Bry} and, hence, that the Bockstein homomorphism
\begin{equation}\label{Eq: Bockstein H^2(X,U(1)) to H^3(X,Z)} \beta:H^2(X,\underline{U(1)}\rightarrow H^3(X,\Z)\end{equation}
is actually an isomorphism\footnote{A more concrete description is given as follows. Writing $\underline{g} = \exp(i\underline{\rho})$, where $\underline{\rho}$ is a real-valued \v Cech $2$-cochain, and using the cocycle condition $\delta(\underline{g})=1$, yields
\[ \exp(i\delta(\underline{\rho})) = \delta(\exp(i\underline{\rho})) = \delta(\underline{g})=1 \,, \]
but this is possible if and only if
\[ \delta(\underline{\rho}) = 2\pi \underline{n} \,, \]
for an integer-valued \v Cech $3$-cochain $\underline{n}$. It is obviously a cocycle due to $\delta^2=0$. The image $\beta([\underline{g}])\in H^3(X,\Z)$ is precisely the class $[\underline{n}]\in H^3(X,\Z)$.}.

The Dixmier-Douady class of $(L,Y)$, which we denote by $d(L)$ or $d(L,Y)$, is the image of $[\underline{g}]$ under \eqref{Eq: Bockstein H^2(X,U(1)) to H^3(X,Z)}. We sometimes also call $[\underline{g}]$ the Dixmier-Douady class, hoping this will not cause any confusion. The Dixier-Douady class satisfies \cite{Mur}
\[ d(L^*)=-d(L) \quad \text{and} \quad d(L\otimes J)=d(L)+d(J)\,. \]
It also behaves naturally under pullbacks of bundle gerbes, namely
\[ d(\phi^{-1}L,Z)=\phi^*d(L,Y)\,. \]

\emph{Lifting bundle gerbes} are one of the most important examples of bundle gerbes. Consider the central extension
\begin{equation}\label{Eq: U(1) central extension} \xymatrix{ U(1) \ar[r] & \widehat{G} \ar[r]^\pi & G } \end{equation}
of Lie groups. We ask, what the obstruction for lifting a principal $G$ bundle $P\rightarrow X$ into a principal $\widehat{G}$ bundle is? In other words, given a principal $G$ bundle $P$, with transition functions $\{g_{\alpha\beta}\}$, exactly when is it possible to find a principal $\widehat{G}$ bundle, with transition functions $\{\widehat{g}_{\alpha\beta}\}$, such that $\pi(\widehat{g}_{\alpha\beta})=g_{\alpha\beta}$? First, consider the central extension \eqref{Eq: U(1) central extension} as a principal $U(1)$ bundle $\widehat{G}\rightarrow G$. Next, define a map $g:P^{[2]}\rightarrow G$ by $p_1 g(p_1,p_2)=p_2$ and pull $\widehat{G}\rightarrow G$ back into a principal $U(1)$ bundle over $P^{[2]}$. The fibre of the pullback over a point $(p_1,p_2)$ is $\{\widehat{g}\in \widehat{G} : p_1\pi(\widehat{g})=p_2\}$. Finally, form the associated complex line bundle of the $U(1)$ principal bundle. The product structure on $\widehat{G}$ induces the bundle gerbe product for the associated line bundle. The resulting bundle gerbe is called the lifting bundle gerbe of $P\rightarrow X$. It is trivial precisely when the aforementioned lift exists.

We are particularly interested in lifting bundle gerbes of central extensions
\[ \xymatrix{ U(1) \ar[r] & U(n) \ar[r] & PU(n) } \]
and
\[ \xymatrix{ U(1) \ar[r] & U(\Hilb) \ar[r] & PU(\Hilb) } \,, \]
where $\Hilb$ is the infinite-dimensional separable Hilbert space, $U(\Hilb)$ the unitary group and $PU(\Hilb):=U(\Hilb)/U(1)$ the projective unitary group. In D-brane theory, a stack of $n$ D-branes carries a $PU(n)$ Chan-Paton bundle. As we have already demonstrated, the possibility of lifting the $PU(n)$ bundle into a $U(n)$ bundle is crucial for the $K$-theoretic classification. As we already briefly mentioned in the beginning of this chapter, the infinite-dimensional case arises from the Freed-Witten anomaly cancellation condition, when the $H$-flux restricts to a nontorsion element on the D-brane worldvolume.

\section{Bundle gerbe connections}
The reason for the introduction of bundles gerbes, was to obtain a purely geometric description of $H^2(X,\mathcal{D}^2)$, classes of which are represented by cocycles
\[ (\underline{g},\underline{\Lambda},\underline{B})=\left(\{g_{\alpha\beta\gamma}\},\{\Lambda_{\alpha\beta}\},\{B_\alpha\}\right)\in \big(C^2(\mathcal{U},\underline{U(1)}),C^1(\mathcal{U},\underline{\Omega}^1),C^0(\mathcal{U},\underline{\Omega}^2)\big)\,. \]
As the notation suggests, the first entry on the left-hand side is the Dixmier-Douady class. We need to add more structure to bundle gerbes to obtain a model for the full Deligne cohomology class.

\begin{definition}[\cite{Mur}] A \emph{bundle gerbe connection} is a connection $\nabla$ on $L\rightarrow Y^{[2]}$, with the additional requirement that it commutes with the bundle gerbe product in the following sense. For any section $t_{12}\otimes t_{23}\in \Gamma(Y^{[2]}\circ Y^{[2]},L\otimes L)$, the connection satisfies
\begin{equation}\label{Eq: BG connection commutativity} \nabla(t_{12}\otimes t_{23}) = \nabla t_{12} \otimes 1 + 1\otimes \nabla t_{23} = \nabla t_{13}\,, \end{equation}
where $t_{13}=t_{12}\otimes t_{23}$ is given by the bundle gerbe product.
\end{definition}

It is not immediately obvious that bundle gerbe connections exist, but \cite{Mur} provides the following simple argument. Let $J\rightarrow Y$ be a $U(1)$ bundle and $\delta(J)$ the corresponding trivial bundle gerbe. It is a classical result that $J$ admits a connection $\nabla$ inducing the tensor product connection $\nabla\otimes 1 + 1\otimes \nabla^*$ on the bundle $\pi_1^{-1}(J)\otimes\pi_2^{-1}(J)^*\rightarrow Y^{[2]}$, where $\nabla^*$ denotes the dual connection. It is straightforward to verify that it commutes with the bundle gerbe product in the sense of \eqref{Eq: BG connection commutativity}. This can be extended to nontrivial bundle gerbes by a simple partition of unity argument. For each open set $U_\alpha$ of $X$ with a local section $s_\alpha$ of $Y$, we denote by $Y_\alpha$ the open subset $\pi^{-1}(U_\alpha)\subset Y$. Any bundle gerbe, say $(L,Y,X)$, restricted to $Y_\alpha$, is trivial. We already argued that there exist bundle gerbe connections for trivial bundle gerbes. By choosing a partition of unity subordinate to the open covering $\{U_\alpha\}$ and pulling it back onto $Y$, we can form a partition of unity for $Y^{[2]}$ subordinate to $\{U_\alpha^{[2]}\}$, which can finally be used to patch together the local bundle gerbe connections to form a bundle gerbe connection for $L$.

The curvature $2$-form of a bundle gerbe connection $\nabla$ is denoted by $F_\nabla$.
\begin{proposition}[\cite{Mur}] 
The complex
\begin{equation}\label{Eq: Murray complex} \xymatrix{ 0 \ar[r] & \Omega^q(X) \ar[r]^{\pi^*} & \Omega^q(Y) \ar[r]^\delta & \Omega^q(Y^{[2]}) \ar[r]^\delta & \Omega^q(Y^{[3]}) \ar[r]^\delta & \ldots } \end{equation}
where $\delta(\eta) := \sum_{i=1}^p (-1)^i \pi_i^*(\eta)$, is exact everywhere.
\end{proposition}
Since $\delta(F_\nabla)=0$, exactness of \eqref{Eq: Murray complex}, there exists a $2$-form $f\in \Omega^2(Y)$, such that
\[ F_\nabla = \delta(f) = \pi_2^*(f)-\pi_1^*(f) \,. \]
The $2$-form $f$ is called the \emph{curving} for the bundle gerbe connection $\nabla$. The choice of a curving is obviously not unique. Any two curvings, $f_1$ and $f_2$, for the same connection $\nabla$, differ by a pullback of a $2$-form on $X$. Choosing a curving $f$ for $\nabla$, we have
\[ \delta(df)=d\delta(f)=dF_\nabla=0 \,. \]
Exactness of \eqref{Eq: Murray complex} implies $df=\pi^*(\omega)$. The $3$-form $\omega\in \Omega^3(X)$ is called the \emph{$3$-curvature} and $\omega/{2\pi i}$ the \emph{Dixmier-Douady form of the pair $(\nabla,f)$}. The Dixmier-Douady form is closed, because $\pi^*(d\omega)=d(\pi^*\omega)=d^2f=0$. In fact, the de Rham cohomology class of $\omega$ coincides with the image of the Dixmier-Douady class under the inclusion $H^3(X,\Z)\rightarrow H^3(X,\R)$ \cite{Mur}.

Using $s_\alpha$ to pull $f\in \Omega^2(Y)$ back onto $X$, yields $1$-forms $B_\alpha:=s_\alpha^*(f)$. Recall that the trivial line bundles $L_{\alpha\beta}\rightarrow U_{\alpha\beta}$ were the pullbacks of $L$ by $(s_\alpha,s_\beta)$. The bundle gerbe connection $\nabla$ is also pulled back by $(s_\alpha,s_\beta)$ to form connections $\nabla_{\alpha\beta}$ on $L_{\alpha\beta}$, which define a set of $1$-forms $\Lambda_{\alpha\beta}\in \Omega^1(Y^{[2]})$ by
\begin{equation}\label{Eq: Definition of Lambda} \nabla_{\alpha\beta}\sigma_{\alpha\beta}=\Lambda_{\alpha\beta}\otimes \sigma_{\alpha\beta}\,. \end{equation}

We then need to extend the concept of stable isomorphism to bundle gerbes with connection and curving. Recall, how a bundle gerbe $(L,Y)$ was said to be trivial if it was induced by a line bundle $J\rightarrow Y$. Now, consider the line bundle $J\rightarrow Y$ equipped with a connection $\nabla$ and curvature $F$. As before, we obtain a trivial bundle gerbe $(\delta(J),Y)$ but this time we also obtain a \emph{trivial bundle gerbe connection} $\delta(\nabla):=\pi_2^*(\nabla)-\pi_1^*(\nabla)$ and its corresponding \emph{trivial curving} $F$. Two bundle gerbes with connection and curving are called stably isomorphic, if they differ by a trivial bundle gerbe with connection and curving. 

The connection between bundle gerbes and Deligne cohomology is the following.
\begin{theorem}[\cite{Mur}] Given a bundle gerbe with connection and curving $(L,Y,X;\nabla,f)$, the triple $(\underline{g},-\underline{\Lambda},\underline{B})$ is a Deligne cocycle:
\[ \xymatrix{ 1 & & & \\ \underline{g} \ar[u]^\delta \ar[r]^{d\log} & 0 & & \\ & -\underline{\Lambda} \ar[u]^\delta \ar[r]^{-d} & 0 & \\ & & \underline{B} \ar[u]^\delta \ar[r]^d & \underline{H} } \]
Moreover, the mapping
\[ (L,Y,X;\nabla,f)\mapsto [\underline{g},-\underline{\Lambda},\underline{B}] \]
is an isomorphism between the group of stable isomorphism classes of bundle gerbes with connection and curving and the Deligne cohomology group $H^2(X,\mathcal{D}^2)$.
\end{theorem}
\begin{proof}
We already showed that $\delta(\underline{g})=1$. According to \eqref{Eq: Definition of Lambda},
\begin{equation}\label{Eq: Lambda transformation pt. 1} (\nabla_{\alpha\beta}\otimes 1 + 1\otimes\nabla_{\beta\gamma})\sigma_{\alpha\beta}\otimes \sigma_{\beta\gamma} = (\Lambda_{\alpha\beta}+\Lambda_{\beta\gamma})\otimes \sigma_{\alpha\gamma}g_{\alpha\beta\gamma}\,. \end{equation}
On the other hand, bundle gerbe connections were required to commute with the product in the sense of \eqref{Eq: BG connection commutativity}, which induces similar commutativity conditions for $\{\nabla_{\alpha\beta}\}$ and thus implies
\begin{equation}\label{Eq: Lambda transformation pt. 2}\begin{split} (\nabla_{\alpha\beta}\otimes 1 + 1\otimes\nabla_{\beta\gamma})\sigma_{\alpha\beta}\otimes \sigma_{\beta\gamma} &= \nabla_{\alpha\gamma} (\sigma_{\alpha\gamma}g_{\alpha\beta\gamma})\\ &= dg_{\alpha\beta\gamma}\otimes \sigma_{\alpha\gamma}+\Lambda_{\alpha\gamma}\otimes \sigma_{\alpha\gamma}g_{\alpha\beta\gamma} \\ &= (d\log g_{\alpha\beta\gamma}+\Lambda_{\alpha\gamma})\otimes \sigma_{\alpha\gamma}g_{\alpha\beta\gamma} \,. \end{split}\end{equation}
Comparing \eqref{Eq: Lambda transformation pt. 1} and \eqref{Eq: Lambda transformation pt. 2} yields the transformation property
\[ \Lambda_{\beta\gamma}-\Lambda_{\alpha\gamma} + \Lambda_{\alpha\beta} = d\log g_{\alpha\beta\gamma}\,. \]
Finally, recall that $F_\nabla = \pi_2^* f - \pi_1^* f$. The sections $(s_\alpha,s_\beta):U_{\alpha\beta}\rightarrow Y$ can be used to pull $F_\nabla$ back to curvature $2$-forms
\[ F_{\alpha\beta}=(s_\alpha,s_\beta)^*F_\nabla = (s_\alpha,s_\beta)^*(\pi_2^*f - \pi_1^* f)=s_\alpha^* f-s_\beta^* f = B_\alpha-B_\beta \,, \]
for the line bundles $L_{\alpha\beta}$. Using $F_{\alpha\beta}=d\Lambda_{\alpha\beta}$, we finally get
\[ d\Lambda_{\alpha\beta} = B_\alpha-B_\beta \,. \]
This proves the first claim, that $(\underline{g},-\underline{\Lambda},\underline{B})$ is a Deligne cocycle. For the second part of the proof, we refer the reader to \cite{MS1}.
\end{proof}

\section{Bundle gerbe modules}\label{Section: Bundle gerbe modules}
In this section we shall describe yet more extra structure for bundle gerbes, namely \emph{bundle gerbe modules} and \emph{bundle gerbe module connections}. These were originally introduced in \cite{BCMMS}, which is our reference for what follows\footnote{The main result in \cite{BCMMS} was the formulation of \emph{twisted $K$-theory} in terms of bundle gerbe modules. We discuss twisted $K$-theory only at the end of this chapter.}.

\begin{definition}[\cite{BCMMS}]
Let $(L,Y,X)$ be a bundle gerbe and $E\rightarrow Y$ a $U(n)$ vector bundle. Suppose there exists a bundle isomorphism
\begin{equation}\label{Eq: BG module isomorphism} \phi:L\otimes \pi_1^{-1} E \xrightarrow{\cong} \pi_2^{-1} E \,, \end{equation}
which is compatible with the bundle gerbe product in the sense of the commutative diagram
\[ \xymatrix{ L_{(y_1,y_2)}\otimes L_{(y_2,y_3)}\otimes E_{y_3} \ar[rr]^-{\id \otimes \phi} \ar[d] & & L_{(y_1,y_2)}\otimes E_{y_2} \ar[d]^\phi \\
L_{(y_1,y_3)}\otimes E_{y_3} \ar[rr]^\phi & & E_{y_1} } \]
Then bundle $E$ is called a (finite rank) bundle gerbe module. The bundle gerbe $(L,Y)$ is said to act on $E$ by \eqref{Eq: BG module isomorphism}.

There is a natural concept of isomorphism for bundle gerbe modules. Two bundle gerbe modules are said to be isomorphic, if they are isomorphic as vector bundles and the isomorphism commutes with the action $\phi$. If $E$ and $F$ are bundle gerbe modules for $L$, the direct sum $E\oplus F$ is also a bundle gerbe module for $L$ and thus defines a monoid structure on the set $\Mod(L)$ of isomorphism classes of (finite rank) bundle gerbe modules for $L$. 
\end{definition}

\begin{proposition}[\cite{BCMMS}]\label{Prop: Stably isomorphic BG have isomorphic Mod}
A stable isomorphism of $(L,Y)$ and $(J,Y)$ induces a semigroup isomorphism $\Mod(L) \cong \Mod(J)$. The isomorphism is not canonical, however, and a change in the stable isomorphism by tensoring with $\pi^{-1}K$\footnote{Recall that any two stable isomorphisms differ necessarily by a tensor product with the pullback of a line bundle $K\rightarrow X$.}, results in a change of the isomorphism by composition with the endomorphism of $\Mod(J)$ induced by tensoring with $\pi^{-1}K$.
\end{proposition}

\begin{proposition}\label{Prop: BG module related to torsion}
If $(L,Y)$ has a bundle gerbe module $E\rightarrow Y$ of finite rank $n$, then the Dixmier-Douady class $d(L)$ is $n$-torsion, that is, $n\,d(L)=0$.
\end{proposition}
\begin{proof}
If $\rk E = 1$, then clearly $E$ is a trivialization of $L$. If $\rk E = n$, then $L^{\otimes n}$ acts on the rank $1$ bundle $\Lambda^n E$. Thus $L^{\otimes n}$ is necessarily trivial, and using the properties of the Dixmier-Douady class we immediately obtain $d(L^{\otimes n})=n\,d(L) = 0$.
\end{proof}

\begin{definition}[\cite{BCMMS}]
A connection $\nabla^E$ on a (finite rank) bundle gerbe module $E\rightarrow Y$ for a bundle gerbe $(L,Y)$ with connection $\nabla$,
is called a \emph{bundle gerbe module connection}, if $\nabla^E$ and $\nabla$ are compatible in such a way that the induced
connections on $L\otimes \pi_1^{-1} E$ and $\pi_2^{-1} E$ are equal under the isomorphism $\phi$, defining the module
structure. It can be shown using a partition of unity argument, somewhat similar to the one used to show that bundle
gerbe connections exist, that also bundle gerbe module connections always exist \cite{BCMMS}.
\end{definition}
It would be interesting to see what bundle gerbe modules look like locally. Let $(L,Y;\nabla,f)$ be a bundle gerbe with
connection and curving, with associated Deligne class $[\underline{g},-\underline{\Lambda},\underline{B}]$. Let
$\{s_\alpha\}$ be local sections $s_\alpha:U_\alpha \rightarrow Y$ of $Y\rightarrow X$. A bundle gerbe
module $E\rightarrow Y$ of rank $n$ can be pulled back into trivial bundles $E_\alpha = s_\alpha^{-1}(E)$. The isomorphism \eqref{Eq: BG module isomorphism} descends into a bundle isomorphism
\[ L_{\alpha\beta}\otimes E_\beta \xrightarrow{\cong} E_\alpha\,. \]
Let $\{\delta_\alpha:U_\alpha\rightarrow E_\alpha\}$ be global sections of $\{E_\alpha\}$. The above isomorphisms then define a set $\{h_{\alpha\beta}\}$ of $n\times n$ matrices by
\[ \sigma_{\alpha\beta}\otimes \delta_\beta = \delta_\alpha h_{\alpha\beta}\,, \]
which satisfy a $\underline{g}$-twisted cocycle condition. This is easily verified by considering a section
\[ \sigma_{\alpha\beta}\otimes \sigma_{\beta\gamma}\otimes \delta_\gamma \in \Gamma(X,L_{\alpha\beta}\otimes L_{\beta\gamma}) \,, \]
which can be simplified using the bundle gerbe product and the bundle gerbe module structure in two ways which, ultimately, must yield equal results. Namely,
\[ \sigma_{\alpha\beta}\otimes \sigma_{\beta\gamma} \otimes \delta_\gamma = \sigma_{\alpha\gamma}g_{\alpha\beta\gamma}\otimes \delta_\gamma = \delta_\alpha h_{\alpha\gamma}g_{\alpha\beta\gamma}\,, \]
but on the other hand
\[ \sigma_{\alpha\beta}\otimes\sigma_{\beta\gamma}\otimes \delta_\gamma = \sigma_{\alpha\beta}\otimes \delta_\beta h_{\beta\gamma}=\delta_\alpha h_{\alpha\beta} h_{\beta\gamma}\,. \]
In order for these to be equal, $\{h_{\alpha\beta}\}$ must satisfy the $\underline{g}$-twisted cocycle condition
\begin{equation}\label{Eq: Twisted cocycle condition} h_{\alpha\beta}h_{\beta\gamma} h_{\alpha\gamma}^{-1} = g_{\alpha\beta\gamma}\,. \end{equation}
We denote the pullback connections $s_\alpha^{-1}(\nabla^E)$ by $\nabla^E_\alpha$ and define
\[ \nabla^E_\alpha \delta_\alpha = A_\alpha\otimes \delta_\alpha\,. \]
A standard calculation using the definition of a bundle gerbe module connection reveals that on double intersections $U_{\alpha\beta}$, $\{A_\alpha\}$ and $\{A_\beta\}$ must be related by
\begin{equation}\label{Eq: A local transformation} A_\beta = h_{\alpha\beta}^{-1} A_\alpha h_{\alpha\beta} + h_{\alpha\beta}^{-1}\,dh_{\alpha\beta} - \Lambda_{\alpha\beta}\,. \end{equation}

While the $\underline{g}$-twisted transition functions $\{h_{\alpha\beta}\}$ do not define a $U(n)$ bundle, their images in $PU(n)$ \emph{do} define a $PU(n)$ bundle $P\rightarrow X$. Let us now form the lifting bundle gerbe for $P\rightarrow X$ and denote it by $(J,P)$. Its Dixmier-Douady class is the obstruction for lifting the $PU(n)$ bundle $P$ into a $U(n)$ bundle, which is simply the failure of $\underline{g}$ being a coboundary. However, the image of $[\underline{g}]$ under the isomorphism $H^2(X,\underline{U(1)})\xrightarrow{\cong} H^3(X,\Z)$ is precisely the Dixmier-Douady class of the original bundle gerbe we started with. Therefore, the Dixmier-Douady class of $(J,P)$ is the same as that of the original bundle gerbe $(L,Y)$. Thus, $(J,P)$ and $(L,Y)$ are stably isomorphic. The lifting bundle gerbe admits a trivial bundle gerbe module $P\times \C^n$\footnote{Recall, how the lifting bundle gerbe was constructed: the fibre of $J$ at a point $(p_1,p_2)\in P^{[2]}$ is $\{g\in U(n) : p_1 \pi(g) = p_2 \}$, where $\pi:U(n)\rightarrow PU(n)$ is the natural projection. The bundle gerbe acts on the module $P\times \C^n$ simply by the natural action of $U(n)$ on $\C^n$.}. The isomorphism of proposition \ref{Prop: Stably isomorphic BG have isomorphic Mod} relates the two bundle gerbe modules. Different stable isomorphisms between $(L,Y)$ and $(J,P)$ yield different bundle gerbe modules for $(L,Y)$, but by proposition \ref{Prop: Stably isomorphic BG have isomorphic Mod}, they differ only by a tensor product with the pullback of a line bundle over $X$. Repeating the construction of the projective bundle for any of these bundle gerbe modules results in $PU(n)$ bundles isomorphic to $P\rightarrow X$. We have proved the following.
\begin{proposition}[\cite{BCMMS}]\label{Proposition: BG modules and projective bundles}
Let $\operatorname{Lin}(X)$ denote the group of isomorphism classes of line bundles over $X$ and $\operatorname{Proj}(X,\sigma)$ the set of isomorphism classes of finite rank projective bundle with Dixmier-Douady class $\sigma$. If $(L,Y)$ is a bundle gerbe over $X$ with $d(L)=\sigma$, then the above construction defines a bijection
\[ \frac{\Mod(L)}{\operatorname{Lin}(X)} \rightarrow \operatorname{Proj}(X,\sigma) \,. \]\qed
\end{proposition}

According to proposition \ref{Prop: BG module related to torsion}, a bundle gerbe with a nontorsion Dixmier-Douady class can not have a finite rank bundle gerbe module. The obvious solution is to allow for bundle gerbe modules to be bundles of infinite-dimensional Hilbert spaces. The natural topology of $U(\Hilb)$ as the structure group of a principal $U(\Hilb)$ bundle is the \emph{strong operator topology} \cite{HJJS}. In principal $PU(\Hilb)$ bundles $PU(\Hilb)$ is endowed with the quotient topology of $U(\Hilb)/U(1)$. Then, $PU(\Hilb)$ acts continuously on $\mathcal{K}$ by conjugation with unitary operators. Moreover, $\Aut(\mathcal{K})\cong PU(\Hilb)$ \cite{CRM}.
\begin{lemma}\label{Lemma: PU(H) bundles = H^3(X,Z)} Isomorphism classes of principal $PU(\Hilb)$ bundles, over $X$ are in bijective correspondence with classes of $H^3(X,\Z)$. \end{lemma}
\begin{proof} A famous theorem by Kuiper \cite{Kui} states that $U(\Hilb)$ is contractible in the strong operator topology. The central extension
\[ \xymatrix{ U(1) \ar[r] & U(\Hilb) \ar[r] & PU(\Hilb) } \]
defines a principal $U(1)$ bundle with contractible total space. The long exact sequence of homotopy groups
\[ \xymatrix { \ldots \ar[r] & \pi_2(U(\Hilb)) \ar[r] & \pi_2(PU(\Hilb)) \ar[r] & \pi_1(U(1)) \ar[r] & \pi_1(U(\Hilb)) \ar[r] & \ldots } \]
yields the isomorphisms $\pi_2(PU(\Hilb))\cong \pi_1(U(1)) \cong \Z$ and $\pi_k(PU(\Hilb))=0$, for all $k\neq 2$. Hence, $PU(\Hilb)$ is a model for the Eilenberg-MacLane space $K(\Z,2)$ \cite{May}. Applying the long exact sequence of homotopy groups to the universal $PU(\Hilb)$ bundle $EPU(\Hilb)\rightarrow BPU(\Hilb)$ and using the contractibility of $EPU(\Hilb)$, yields the isomorphisms
\[ \pi_{k+1}(BPU(\Hilb))\cong \pi_k(PU(\Hilb)) = \pi_k(K(\Z,2))\,. \]
This proves that $BPU(\Hilb)$ is a model for $K(\Z,3)$. Therefore,
\[ H^3(X,\Z)\cong [X,K(\Z,3)] = [X,BPU(\Hilb)]\,, \]
which proves the claim. The class of $H^3(X,\Z)$ corresponding to a $PU(\Hilb)$ bundle is called the \emph{Dixmier-Douady class} of the bundle.
\end{proof}
Let $(L,Y)$ be a bundle gerbe, with possibly nontorsion Dixmier-Douady class $d(L)$. By lemma \ref{Lemma: PU(H) bundles = H^3(X,Z)}, there is a $PU(\Hilb)$ bundle $P\rightarrow X$, with lifting bundle gerbe $(J,P)$, such that $d(J)=d(L)$\footnote{Remark, that this approach works even when $d(L)$ is torsion.}. Now, let $E\rightarrow P$ be a bundle of infinite-dimensional separable Hilbert spaces $\Hilb$. There is an associated principal $U(\Hilb)$ bundle $U(E)\rightarrow P$, whose fibre over $p\in P$ is the group of isomorphisms $E_p \cong \Hilb$. For $f\in U(E)_p$, the right action of $u\in U(\Hilb)$ is given by $fu := f\circ u$.
\begin{definition}[\cite{BCMMS}]
Suppose that the structure group of $U(E)$ reduces to $U_1\subset U(\Hilb)$, the subgroup of unitary operators of the form $1+\text{trace class}$. We denote the principal $U_1$ subbundle of $U(E)$ by $R\rightarrow P$. If $u\in U(\Hilb)$, we denote its projectivization by $[u]\in PU(\Hilb)$. If $p_1,p_2\in P$ and $p_1 [u] = p_2$, then by the construction of the lifting bundle gerbe, $u\in J_{p_1,p_2}$. We then require that the map $U(E)_{p_1}\rightarrow U(E)_{p_2}$, sending $f$ to $ufu^{-1}$, preserves $R$. That is, for any $f\in R_{p_1}$ and $u\in J_{p_1,p_2}$, $ufu^{-1}\in R_{p_2}$. The bundle $E\rightarrow P$ is then called a \emph{$U_1$ bundle gerbe module}. A stable isomorphism between $(L,Y)$ and $(J,P)$ can then be used to carry the $U_1$ bundle gerbe modules for $(J,P)$ to $(L,Y)$. The monoid of $U_1$ bundle gerbe modules for $L$ is denoted by $\Mod_{U_1}(L)$.

Finally, one can define bundle gerbe module connections for $E$ in analogy with the finite dimensional case.
\end{definition}

We shall now briefly comment on how all the mathematical constructions discussed so far relate to the theory of D-branes. Our goal is to make the following claim precise and show that it is consistent with the physical theory of D-branes, as far as we have discussed it.
\begin{claim}\label{Claim: BGs and B-fields} Suppose $[H]|_\Sigma$ is torsion. The Dirac quantized background $B$-field is a bundle gerbe $(L_B,Y,X)$ with connection and curving. The $A$-field, living on the D-brane worldvolume $\Sigma$, is a (finite rank) bundle gerbe module connection for a certain bundle gerbe on $\Sigma$. If $[H]|_\Sigma$ is nontorsion, the $A$-field is a $U_1$ bundle gerbe module connection.
\end{claim}
In fact, one does not even need such complicated objects as bundle gerbes are to treat the case of torsion $[H]|_\Sigma$. In \cite{Kap}, Kapustin solved the anomaly cancellation problem in the torsion case, using the theory of \emph{Azumaya algebras} and \emph{Azumaya algebra modules}, in place of bundle gerbes and bundle gerbe modules. However, the approach of Kapustin can not be extended to treat the case of nontorsion $[H]|_\Sigma$. The bundle gerbal approach was first presented in \cite{CJM}, where the anomaly cancellation problem was finally solved also for the nontorsion case.

We shall now start to make sense of the holonomy terms in \eqref{Eq: Open string path integral, problematic}.

\section{Holonomy}
Let us start by recalling the usual definition of holonomy for $U(1)$ line bundles with connection. We wish to think of isomorphism classes of line bundles with connection as Deligne classes in $H^1(X,\mathcal{D}^1)$ and formulate holonomy as ''Deligne cohomologically'' as possible, in hopes of finding a generalization to higher diagonal Deligne classes, in particular to $H^2(X,\mathcal{D}^2)$. Our references are \cite{Joh2} and \cite{CJM}. 

Let $\gamma:S^1\rightarrow X$ be a closed curve on $X$, for which $\gamma(0)=\gamma(1)=x$ and $P$ a principal $U(1)$ bundle, with connection $1$-form $A$. Then, for any point $p$ of the fibre $P_x$, there exists a unique solution $\widetilde{\gamma}$ of the \emph{parallel transport equation} \cite{Nak}, with $\widetilde{\gamma}(1) = p$, that is, a unique horizontal lift of the curve $\gamma$, starting at $p$. It turns out that
\[ \widetilde{\gamma}(1) = \widetilde{\gamma}(0) \cdot \hol(\gamma) \]
defines a mapping $\hol:LX\rightarrow U(1)$, which depends only on the bundle and the connection. Let us take a closer look at it. The curve $\gamma$ may be broken into parts $\{\gamma([t_i,t_{i+1}])\}$, such that $\gamma([t_i,t_{i+1}])\subset U_i$, where $\{U_i\}$ is a good open cover, with local sections $\{s_i\}$, trivializing the connection. Now, it is well known \cite{Nak} that the parallel transport of $s_i$ along $\gamma|_{[t_i,t_{i+1}]}$ maps
\[ s_i(\gamma(t_i))\mapsto s_i(\gamma(t_{i+1}))\,\exp\left(-\int_{t_i}^{t_{i+1}} \gamma^*(A_i) \right)\,. \]
Combining two such maps yields
\begin{align*} s_i(\gamma(t_i)) &\mapsto s_i(\gamma(t_{i+1}))\,\exp\left(-\int_{t_i}^{t_{i+1}} \gamma^* A_i \right) \\
&= s_{i+1}(\gamma(t_{i+1}))\,\exp\left(-\int_{t_i}^{t_{i+1}} \gamma^* A_i \right)  g_{t_{i+1} t_i}(\gamma(t_{i+1})) \\
&\mapsto s_{i+1}(\gamma(t_{i+2}))\,\exp\left(-\int_{t_i}^{t_{i+1}} \gamma^* A_i - \int_{t_{i+1}}^{t_{i+2}} \gamma^* A_{i+1} \right)  g_{t_{i+1} t_i}(\gamma(t_{i+1})) \,. \end{align*}
It is straightforward to see that the holonomy along the entire closed curve is
\begin{equation}\label{Eq: Holonomy, first attempt} \hol(\gamma) = \prod_i \exp\left(-\int_{t_i}^{t_{i+1}} \gamma^* A_i\right) \cdot \prod_i g_{t_{i+1} t_i}(\gamma(t_{i+1}))\,. \end{equation}

Let us now take a different approach. Denote by $[\underline{g},\underline{A}]\in H^1(X,\mathcal{D}^1)$ the Deligne class of the line bundle with connection. Given the curve $\gamma:S^1\rightarrow X$, we can pull the line bundle back onto $S^1$. The first Chern class of the pullback bundle is a class in $H^2(S^1,\Z)$. However, $H^2(S^1,\Z)=0$, and so the pullback bundle is trivial. This is not to say that the pullback of the Deligne class, $[\gamma^*(\underline{g}),\gamma^*(\underline{A})]$, with respect to the good open cover $\gamma^{-1}(\mathcal{U})=\{\gamma^{-1}(U_\alpha)\}$ of $S^1$, would be zero, but only that it is of the form $[\delta(\underline{h}),\gamma^*(\underline{A})]$. Diagrammatically, this is expressed by
\[ \xymatrix { 1 & \\
\delta(\underline{h}) \ar[u]^\delta \ar[r]^{-d\log} & 0 \\
\underline{h} \ar[u]^\delta & \gamma^* \underline{A} \ar[u]^\delta } \]
where
\[ -\gamma^* A_\beta + d\log h_\beta = -\gamma^* A_\alpha + d\log h_\alpha\,. \]
In other words, $-\gamma^* A_\alpha+d\log h_\alpha$ is a globally defined $1$-form on $S^1$ and, as such, can be integrated. It is easy to see that the difference between two such $1$-forms, $-\gamma^*(A_\alpha)+d\log h_\alpha$ and $-\gamma^*(\widetilde{A}_\alpha)+d\log \widetilde{h}_\alpha$, constructed from the same Deligne class, is $2\pi i$ times an integral $1$-form. Therefore, the exponential
\begin{equation}\label{Eq: Holonomy integral, h dependent} \exp\int_{S^1} \left(-\gamma^* A_\alpha + d\log h_\alpha\right) \end{equation}
takes values in $S^1$. 

To evaluate \eqref{Eq: Holonomy integral, h dependent}, we need to break the integral into a sum of integrals over open sets of $S^1$ in such a way, that the integrals of $A_\alpha$ and $d\log h_\alpha$ become independently defined. Let $\tau$ be a triangulation of $S^1$, consisting of vertices $v$ and edges $e$, such that the image under $\gamma$ of each edge is contained within at least one $U_\alpha$\footnote{Such a triangulation always exists, because $S^1$ is compact and thus has a finite Lebesgue number. We simply choose the edges to be of length less than the Lebesgue number.}. For an edge $e$, we denote by $U_{\rho(e)}$ the open set in which $\gamma(e)$ is contained, that is, $e\subset \gamma^{-1}(U_{\rho(e)})$. Here $\rho:\tau \rightarrow \mathcal{I}$ is an \emph{index map}, mapping elements of the triangulation into the index set of the good open cover $\mathcal{U}$ of $X$. The integral \eqref{Eq: Holonomy integral, h dependent} becomes
\begin{multline}\label{Eq: Holonomy integral, h dependent 2} \prod_e \exp\left(-\int_e \gamma^* A_{\rho(e)} + \int_e d\log h_{\rho(e)} \right) \\ = \prod_e \exp\left(-\int_e \gamma^* A_{\rho(e)} + \log h_{\rho(e)}(v_e^1) - \log h_{\rho(e)}(v_e^0)\right) \,, \end{multline}
where $v_e^0$ and $v_e^1$ are the bounding vertices of $e$. On the other hand, using
\[ \log (\gamma^* g_{\alpha\beta}) = \log h_\beta - \log h_\alpha\,, \]
the difference
\[ \log h_{\rho(e)}(v_e^1) - \log h_{\rho(e)}(v_e^0) \]
becomes
\[ \log(\gamma^* g_{\rho(e)\rho(v_e^0)})(v_e^0)-\log(\gamma^* g_{\rho(e)\rho(v_e^1)})(v_e^1)-\log h_{\rho(v_e^0)}(v_e^0) + \log h_{\rho(v_e^1)}(v_e^1)\,. \]
Summing over all edges makes the last two terms cancel each other. Hence, \eqref{Eq: Holonomy integral, h dependent 2} becomes
\begin{align} \notag &\prod_e \exp\left(-\int_e \gamma^* A_{\rho(e)} + \log(\gamma^* g_{\rho(e)\rho(v_e^0)})(v_e^0)-\log(\gamma^* g_{\rho(e)\rho(v_e^1)})(v_e^1) \right)  \\ \notag &\qquad = \prod_e \exp\left(-\int_e \gamma^* A_{\rho(e)}\right) \cdot \prod_e g_{\rho(e)\rho(v_e^0)}(\gamma(v_e^0)) g_{\rho(e)\rho(v_e^1)}^{-1}(\gamma(v_e^1)) \\ \label{Eq: Holonomy, local expression} &\qquad = \prod_e \exp\left(-\int_e \gamma^* A_{\rho(e)}\right) \cdot \prod_{v\subset e} g^{-1}_{\rho(e)\rho(v)}(\gamma(v)) \,, \end{align}
where the product of $g^{-1}_{\rho(e)\rho(v)}(\gamma(v))$ over vertices should be understood as indicated by the previous form.

More generally, whenever we write a product over edges and vertices, we always assume it to be an \emph{oriented product}, meaning that we automatically take into account both endpoints of the (oriented) edge in such a way, that the value of the function at the initial vertex is taken inverted.

We have obtained an expression for \eqref{Eq: Holonomy integral, h dependent}, which depends only on the Deligne cocycle, $\tau$ and the index map $\rho$. It is proved in \cite{Joh2} that it is also independent of the chosen triangulation and the index map and thus yields, for a given line bundle with connection, a well-defined map from the loop space $LX$ into $S^1$. Remark also, that \eqref{Eq: Holonomy, first attempt} is just a special case of \eqref{Eq: Holonomy, local expression}, which gives us hope of generalizing the concept of holonomy to higher Deligne cohomology, in particular to $H^2(X,\mathcal{D}^2)$. If the connection happened to be globally defined, the holonomy would simplify into
\[ \hol(\gamma) = \exp\left(-\int_{S^1} \gamma^*A\right) = \exp\left(-\int_{\gamma} A\right) \,. \]
Even when the connection is not globally defined, we denote $\hol(\gamma)$ heuristically by
\begin{equation}\label{Eq: Heuristic holonomy integral} \exp\left(-\int_{\gamma} A\right) \,. \end{equation}
In \cite{Alv}, the starting point is to make sense of integrals such as \eqref{Eq: Heuristic holonomy integral}. It might sound surprising that one then necessarily obtains the data for a Deligne class. This is explained by the fact that a suitable holonomy function, together with a generalized gauge field strength define a so-called \emph{Cheeger-Simons differential character}, the set of which forms a group isomorphic to $H^p(X,\mathcal{D}^p)$ \cite{CS2,Bry}. We have nothing more to say about differential characters, but one should still keep in mind the deep meaning of the holonomy function.

Let us recapitulate the above discussion on holonomy along a closed curve $\gamma$. We expressed our line bundle with connection as $[\underline{g},\underline{A}]\in H^1(X,\mathcal{D}^1)$ and pulled it back into a class $[\gamma^* \underline{g},\gamma^* \underline{A}]\in H^1(S^1,\mathcal{D}^1)$. Then we remarked that the pullback bundle is trivial, which was a central aspect of the construction, since the trivialization was used to construct a globally defined $1$-form, which could be then integrated over $S^1$. It was finally argued that the exponential of the integral is a well-defined function on the loop space $LX$, depending only on the Deligne class of the line bundle with connection and that this function coincided with the classical notion of holonomy along a closed curve. Also, there was nothing in the Deligne theoretic construction that would have prevented us from using any closed one-dimensional manifold in place of $S^1$. Only the dimensionality and closedness mattered.

Let now $[\underline{g},-\underline{\Lambda},\underline{B}]\in H^2(X,\mathcal{D}^2)$ be a Deligne class, corresponding to a bundle gerbe with connection and curving. Let $M$ be a $2$-dimensional closed surface and $\xi:M\rightarrow X$ a smooth map. Stable isomorphism classes of bundle gerbes were classified by their Dixmier-Douady class in $H^3(X,\Z)$. Since $H^3(M,\Z)=0$, all bundle gerbes over $M$ are trivial. This means that the pullback $[\xi^*\underline{g},-\xi^*\underline{\Lambda},\xi^*\underline{B}]$ can be written in the form $[\delta(\underline{h}),-\xi^*\underline{\Lambda},\xi^*\underline{B}]$. Diagrammatically, it is
\[ \xymatrix{ 1 & & \\
\delta(\underline{h}) \ar[u]^\delta \ar[r]^{d\log} & 0 & \\
\underline{h} \ar[u]^\delta & -\xi^*\underline{\Lambda} \ar[u]^\delta \ar[r]^{-d} & 0 \\
& & \xi^*\underline{B} \ar[u]^\delta } \]
from which we can immediately read
\[ d\log(h_{\beta\gamma}h^{-1}_{\alpha\gamma}h_{\alpha\beta})=\xi^*\Lambda_{\beta\gamma}-\xi^*\Lambda_{\alpha\gamma}+\xi^*\Lambda_{\alpha\beta}\,, \]
which is just the \v Cech cocycle condition:
\[ \delta(d\log\underline{h} - \xi^*\underline{\Lambda}) = 0 \,. \]
By exactness, there exists $\underline{m}\in C^0(\xi^{-1}(\mathcal{U}),\underline{\Omega}^1)$, such that
\begin{equation} \label{Eq: Holonomy, gauge transformation} d\log h_{\alpha\beta} - \xi^*\Lambda_{\alpha\beta} = m_\beta - m_\alpha\,. \end{equation}
Taking the de Rham differential on both sides, yields
\[ \xi^*d\Lambda_{\alpha\beta} = dm_\alpha - dm_\beta\,. \]
On the other hand, since $d\Lambda_{\alpha\beta} = B_\alpha - B_\beta$, we get
\[ \xi^*B_\beta - dm_\beta = \xi^*B_\alpha - dm_\alpha\,. \]
Hence, the $2$-form $\xi^*B_\alpha-dm_\alpha$ is globally defined and can be integrated over $M$. Indeed,
\begin{equation}\label{Eq: Gerbe holonomy, h dependent} \hol(\xi) = \exp\int_M \left(\xi^* B_\alpha - dm_\alpha\right) \end{equation}
provides the correct generalization of holonomy for bundle gerbes with connection and curving. By choosing a triangulation $\tau$ of $M$, consisting of faces $f$, edges $e$ and vertices $v$, such that the image under $\xi$ of each object in $\tau$ is contained in an element of the open cover $\mathcal{U}$ of $X$, the integral may be computed in local parts. It is straightforward computation to show that
\begin{equation}\label{Eq: Gerbe holonomy, local form} \hol(\xi) = \prod_f \exp\int_f \xi^*B_{\rho(f)} \cdot \prod_{e\subset f} \exp\left(-\int_e \xi^*\Lambda_{\rho(f)\rho(e)}\right) \cdot \prod_{v\subset e \subset f} g_{\rho(f)\rho(e)\rho(v)}(\xi(v))\,, \end{equation}
where the product of $g_{\rho(f)\rho(e)\rho(v)}(\xi(v))$ over vertices is understood the same way as in \eqref{Eq: Holonomy, local expression}. It can also be shown that $\hol(\xi)$ depends only on the Deligne class \cite{Joh2}. If $\underline{B}$ defines a global $2$-form $B$, also $d\underline{m}$ must define a global $2$-form $dm$ and thus, by Stoke's theorem and $\partial M = \varnothing$, the integral \eqref{Eq: Gerbe holonomy, h dependent} simplifies into
\[ \exp\int_M \xi^* B = \exp\int_\xi B\,. \]
In general, the above integral is ill-defined, but we still use it heuristically to denote the holonomy. We could alternatively have started with the above integral and attempt to make sense of it, in the spirit of \cite{Alv}, only to arrive at a Deligne cocycle and \eqref{Eq: Gerbe holonomy, local form}.

An even more general perspective at holonomy is provided by the exact sequence \cite{Gaj}
\[ \xymatrix{ 0 \ar[r] & \Omega^p(X)_{(c,0)} \ar[r] & \Omega^p(X) \ar[r]^\iota & H^p(X,\mathcal{D}^p) \ar[r]^d &
H^{p+1}(X,\Z) \ar[r] & 0 } \]
where $\Omega^p(X)_{(c,0)}$ is the set of closed $p$-forms, with $2\pi i$ times integral periods, $\iota$ maps
$\rho\in \Omega^p(X)$ to the Deligne class $[1,0,\ldots,0,\{\rho|_{U_\alpha}\}]$ and $d$ denotes taking the
Dixmier-Douady class. Consider the pullback of a Deligne class
$[\underline{g},\underline{\omega}^1,\ldots,\underline{\omega}^p]$ by $\xi:M\rightarrow X$, where $M$ is
$p$-dimensional. Because $H^{p+1}(M,\Z)=0$, the exact sequence yields an element $\rho\in \Omega^p(M)$ such that
\[ [1,0,\ldots,0,\{\rho|_{U_\alpha}\}] = [\xi^*\underline{g},\xi^*\underline{\omega}^1,\ldots,\xi^*\underline{\omega}^p]\,. \]
We may now define
\[ \hol(\xi) := \exp\int_M \rho\,. \]
This expression is independent of the chosen $p$-form because, by exactness, any two such $p$-forms $\rho$ and
$\widetilde{\rho}$ representing the same Deligne class, satisfy
\begin{equation}\label{Eq: Holonomy, general} \int_M (\rho-\widetilde{\rho}) \in 2\pi i \Z \,. \end{equation}
It is straightforward to verify that for $H^1(X,\mathcal{D}^1)$ and $H^2(X,\mathcal{D}^2)$, the global forms $\gamma^*A_\alpha - d\log h_\alpha$ and $\xi^*B_\alpha - dm_\alpha$ are candidates for $\rho$. The reader may wonder why there is a sign difference between this definition for the holonomy of $H^1(X,\mathcal{D}^1)$ and the one we developed earlier, starting from the parallel transport equation. This is purely a notational choice really. It turns out that to cancel the Freed-Witten anomaly, we need this sign difference if we want to have $\hol(\partial M,A)$ in the path integral. If we would not insert the extra minus sign here, we would have to do it somewhere else, but this seems to be a convenient choice. Ignoring the extra sign, all the holonomy functions defined above coincide with \eqref{Eq: Holonomy, general}. 

\section{Higher-dimensional parallel transports}
Let $M$ be a $p$-dimensional compact manifold. We denote by $\Map(M,X)$ the infinite-dimensional manifold of smooth maps $M\rightarrow X$, equipped with the compact-open smooth topology \cite{Hir}. For a triangulation $\tau$ of $M$, consisting of simplices $\sigma$ and an index map $\rho$,
\[ V_{(\tau,\rho)}=\{f\in \Map(M,X) : f(\sigma) \subset U_{\rho(\sigma)}\,, \sigma \in \tau\} \]
is an open set in $\Map(M,X)$. The sets $V_{(\tau,\rho)}$ form an open cover $\mathcal{V}$ of $\Map(M,X)$.

In the previous section we explained, how the ill-defined integral
\[ \exp\int_M \xi^* B \,, \]
appearing in the closed string path integral, should be correctly understood in terms of holonomy of a Deligne class. The holonomy integral, as developed in the previous section, works only for closed surfaces. However, the worldsheet of an open string, denote it by $M$, is a compact manifold with boundary $\partial M$. 

Let us start by considering a lower dimensional example. Let $I=[0,1]$ and $\gamma:I\rightarrow X$ an open curve. Given a line bundle $L$ with connection, holonomy around a closed curve could be defined as parallel transport around it. It is natural to expect that parallel transport provides the correct generalization of holonomy for open curves. Unlike the holonomy around a closed curve, parallel transport along an open curve is not a $U(1)$-valued function, because there is no canonical way to identify the fibres over $\gamma(0)$ and $\gamma(1)$, but rather a linear map $L_{\gamma(0)}\rightarrow L_{\gamma(1)}$, or alternatively an element of the vector space $L_{\gamma(0)}^*\otimes L_{\gamma(1)}$. Consider the line bundle $(\ev_0^{-1} L)^*\otimes (\ev_1^{-1} L)\rightarrow \Map(I,X)$, where $\ev_t$ is the evaluation map $\gamma\mapsto \gamma(t)$. The fibre over $\gamma\in \Map(I,X)$ is precisely $L_{\gamma(0)}^* \otimes L_{\gamma(1)}$ and, therefore, we may think of parallel transport as a section of this bundle. It turns out that there is a higher-dimensional generalization of parallel transport and it is a section of a trivial line bundle over $\Map(M,X)$. Moreover, this line bundle is the pullback by the restriction map $r:\Map(M,X)\rightarrow \Map(\partial M,X)$ of a line bundle over $\Map(\partial M,X)$. We shall illustrate, following \cite{Joh2}, how this line bundle can be constructed starting from the local formula for holonomy, which can then be generalized to the higher-dimensional case.


Consider two open curves, $\gamma_0$ and $\gamma_1$, with the same endpoints. We denote by $\gamma_0 \star \gamma_1^{-1}$ the composition of $\gamma_0$ and the inverse of $\gamma_1$, which is then a closed curve. It is clear that the composition of parallel transport along $\gamma_0$ and the inverse of parallel transport along $\gamma_1$ must be equal to $\hol(\gamma_0 \star \gamma_1^{-1})$. In other words, if we denote parallel transport along a curve $\gamma$ by $t(\gamma)$,
\begin{equation}\label{Eq: Parallel transport holonomy relation} t(\gamma_0)t(\gamma_1)^{-1} = \hol(\gamma_0 \star \gamma_1^{-1})\,. \end{equation}
We could attempt to define $t$ using either \eqref{Eq: Holonomy integral, h dependent} or \eqref{Eq: Holonomy, local expression}, by replacing $S^1$ with $I$, for example
\[ t(\gamma) \stackrel{?}{:=} \exp\int_{I} (-\gamma^* A_\alpha + d\log h_\alpha)\,, \]
which obviously satisfies \eqref{Eq: Parallel transport holonomy relation}. However, unlike in the case of holonomy, the integral depends on the choice of the trivialization $\{h_\alpha\}$ and only the combination yielding the holonomy is independent of it. It is explained in \cite{Joh2}, why this approach becomes very inconvenient in the higher-dimensional case. Instead, we shall take \eqref{Eq: Holonomy, local expression} as our starting point and define the \emph{parallel transport map}
\[ t_{(\tau,\rho)}(\gamma) := \prod_e \exp\left(-\int_e \gamma^* A_{\rho(e)}\right) \cdot \prod_{v\subset e} g^{-1}_{\rho(e)\rho(v)}(\gamma(v))\,, \]
where the edges $e$, vertices $v$ and the index map $\rho$ constitute a triangulation $\tau$ of $I$. In this case the definition depends on the triangulation, which is why we have indicated it with a subscript. In other words, it is only well-defined locally on open sets $V_{(\tau,\rho)}\subset \Map(I,X)$. We may restate \eqref{Eq: Parallel transport holonomy relation} using $\{t_{(\tau,\rho)}\}$ as
\begin{equation}\label{Eq: Parallel transport holonomy relation, triangulation dependent} t_{(\tau_0,\rho_0)}(\gamma_0)t_{(\tau_1,\rho_1)}(\gamma_1)^{-1} = \hol(\gamma_0 \star \gamma_1^{-1})\,, \end{equation}
if and only if $\gamma_0\in V_{(\tau_0,\rho)}$, $\gamma_1\in V_{(\tau_1,\rho_1)}$ and the index maps satisfy $\rho_0(v_b)=\rho_1(v_b)$ for the boundary vertices $v_b$, because only then the triangulations $(\tau_0,\rho_0)$ and $(\tau_1,\rho_1)$ combine to form a triangulation $(\tau_0\cup \tau_1,\rho_0\cup \rho_1)$ of $S^1$. Suppose that $\gamma \in V_{(\tau_0,\rho_0)}\cap V_{(\tau_1,\rho_1)}$, where $\rho_0(v_b)=\rho_1(v_b)$ for the boundary vertices. Then we have
\[ t_{(\tau_0,\rho_0)}(\gamma) = \hol(\gamma \star \gamma^{-1})t_{(\tau_0,\rho_0)}(\gamma) = t_{(\tau_1,\rho_1)}(\gamma)t_{(\tau_0,\rho_0)}^{-1}(\gamma)t_{(\tau_0,\rho_0)}(\gamma) = t_{(\tau_1,\rho_1)}(\gamma)\,, \]
which means that we are free to alter the triangulation of $I$ from $(\tau_0,\rho_0)$ to $(\tau_1,\rho_1)$, as long as $\gamma$ is also in $V_{(\tau_1,\rho_1)}$ and the images of the boundary vertices under the index maps stay intact.

The local functions $t_{(\tau,\rho)}:V_{(\tau,\rho)}\rightarrow U(1)$ can be used to construct a trivial line bundle over $\Map(I,X)$, with transition functions $t_{(\tau_i,\rho_i)}^{-1} t_{(\tau_j,\rho_j)}$ on $V_{(\tau_i,\rho_i)}\cap V_{(\tau_j,\rho_j)}$ and natural flat connection, given by the local $1$-forms $\{d\log t_{(\tau,\rho)}\}$. For $\gamma\in V_{(\tau_i,\rho_i)}\cap V_{(\tau_j,\rho_j)}$, the transition functions are
\begin{align*} t_{(\tau_i,\rho_i)}^{-1}(\gamma) t_{(\tau_j,\rho_j)}(\gamma) &= t_{(\tau_i,\rho_i)}^{-1}(\gamma) t_{(\tau_i,\rho_j)}(\gamma) \\
&= \prod_{e\in \tau_i} \exp\int_e \gamma^*(A_{\rho_i(e)}-A_{\rho_j(e)}) \cdot \prod_{v\subset e\in \tau_i} g_{\rho_j(e)\rho_j(v)}^{-1}(\gamma(v))g_{\rho_i(e)\rho_i(v)}(\gamma(v)) \\
&= \prod_{e\in \tau_i} \exp\int_e \gamma^*d\log g_{\rho_j(e)\rho_i(e)} \cdot \prod_{v\subset e\in \tau_i} g_{\rho_j(e)\rho_j(v)}^{-1}(\gamma(v))g_{\rho_i(e)\rho_i(v)}(\gamma(v)) \\
&= \prod_{v\subset e\in \tau_i} g_{\rho_j(e)\rho_i(e)}(\gamma(v)) \cdot \prod_{v\subset e\in \tau_i} g_{\rho_j(v)\rho_j(e)}(\gamma(v))g_{\rho_i(e)\rho_i(v)}(\gamma(v)) \\
&= \prod_{v\subset e\in \tau_i} g_{\rho_j(e)\rho_i(v)}(\gamma(v))g_{\rho_j(v)\rho_j(e)}(\gamma(v)) \\
&= \prod_{v\subset e\in \tau_i} g_{\rho_j(v)\rho_i(v)}(\gamma(v)) \\
&= \prod_{v\subset \partial I} g_{\rho_j(v)\rho_i(v)}(\gamma(v)) \,, \end{align*}
but since these depend only on the boundary of $I$, they may be understood as pullbacks of transition functions $\{\mathfrak{g}_{(\tau_i,\rho_i)(\tau_j,\rho_j)}\}$ of a (generally nontrivial) line bundle over $\Map(\partial I,X)$ by the restriction map $r:\Map(I,X)\rightarrow \Map(\partial I,X)$. In other words, we can write
\[ r^*\mathfrak{g}_{(\tau_i,\rho_i)(\tau_j,\rho_j)} = t_{(\tau_i,\rho_i)}^{-1} t_{(\tau_j,\rho_j)} \,. \]
One can also compute explicit expressions for the local connections $\{d\log t_{(\tau,\rho)}\}$. It turns out that they are of the form $r^*\mathfrak{A}_{(\tau,\rho)}$, where $\{\mathfrak{A}_{(\tau,\rho)}\}$ are local $1$-forms on $\Map(\partial I,X)$, which patch up into a connection for the line bundle given by $\{\mathfrak{g}_{(\tau_i,\rho_i)(\tau_j,\rho_j)}\}$. We may then write these as a Deligne class $\mathfrak{T}[\underline{g},\underline{A}]=[\underline{\mathfrak{g}},\underline{\mathfrak{A}}]\in H^1(\Map(\partial I,X),\mathcal{D}^1)$, called the \emph{transgression} of $[\underline{g},\underline{A}]$. The map $\mathfrak{T}:H^1(X,\mathcal{D}^1)\rightarrow H^1(\partial I,\mathcal{D}^1)$ is called the \emph{transgression homomorphism}.

The above procedure works also in higher-dimensional cases. Let us start with a bundle gerbe with connection and curving, with corresponding Deligne class $[\underline{g},-\underline{\Lambda},\underline{B}]\in H^2(X,\mathcal{D}^2)$. Recall the local formula \eqref{Eq: Gerbe holonomy, local form} for holonomy over a closed surface. Consider a surface $M$, with boundary $\partial M$. The simplest example is, of course, the $2$-dimensional disc $B^2$, with boundary $S^1$. Let $\xi\in \Map(M,X)$ and define the \emph{parallel transport map}
\[ t_{(\tau,\rho)}(\xi) := \prod_f \exp\int_f \xi^*B_{\rho(f)} \cdot \prod_{e\subset f} \exp\left(-\int_e \xi^*\Lambda_{\rho(f)\rho(e)}\right) \cdot \prod_{v\subset e \subset f} g_{\rho(f)\rho(e)\rho(v)}(\xi(v)) \,, \]
where faces $f$, edges $e$ and vertices $v$ constitute a triangulation $\tau$ of $M$ and $\rho$ is the associated index map. As expected, $t_{(\tau,\rho)}$ is not globally defined on $\Map(M,X)$, but only locally on the open sets $V_{(\tau,\rho)}\subset \Map(M,X)$. If $\xi_1,\xi_2\in \Map(M,X)$ are such that the images of the boundaries coincide, they may be glued together to form the \emph{connected sum} $\xi_1\# \xi_2:M\# M \rightarrow X$, oriented by the first factor. For example, if $M=B^2$, then $M\# M  \cong S^2$. This generalizes the composition of paths, since $\gamma_1 \star \gamma_2^{-1} = \gamma_1 \# \gamma_2$. The equivalent of \eqref{Eq: Parallel transport holonomy relation} for surfaces would be
\begin{equation*}\label{Eq: Gerbe parallel transport holonomy relation, triangulation dependent} t_{(\tau_0,\rho_0)}(\xi_0)t_{(\tau_1,\rho_1)}(\xi_1)^{-1} = \hol(\xi_0 \# \xi_1)\,, \end{equation*}
if and only if $\xi_0 \in V_{(\tau_0,\rho_0)}$, $\xi_1\in V_{(\tau_1,\rho_1)}$, $\rho_0(e_b)=\rho_1(e_b)$ and $\rho_0(v_b)=\rho_1(v_b)$, where $v_b$ and $e_b$ denote vertices and edges belonging to the boundary $\partial M$. Analogously to the lower-dimensional case, for such $\rho_0,\rho_1$ and $\xi\in V_{(\tau_0,\rho_0)}\cap V_{(\tau_1,\rho_1)}$, we have
\[ t_{(\tau_0,\rho_0)}(\xi) = t_{(\tau_1,\rho_1)}(\xi) \,. \]

The local functions $\{t_{(\tau,\rho)}\}$ can be used to construct a trivial line bundle over $\Map(M,X)$, with transition functions $t_{(\tau_i,\rho_i)}^{-1} t_{(\tau_j,\rho_j)}$ on $V_{(\tau_i,\rho_i)}\cap V_{(\tau_j,\rho_j)}$ and with a natural flat connection. Computing an explicit expression for the transition functions is a straightforward task, performed in \cite{Joh2}. The result is
\begin{equation}\label{Eq: Explicit expression for gerbe transgression transition functions} \prod_{e\subset \partial M} \exp\left(-\int_e \xi^*\Lambda_{\rho_i(e)\rho_j(e)}\right)\cdot \prod_{v\subset e\subset \partial M} g_{\rho_i(e)\rho_i(v)\rho_j(v)}^{-1}g_{\rho_i(e)\rho_j(e)\rho_j(v)}(\xi(v))\,, \end{equation}
which depends only on the boundary $\partial M$. Therefore, if $r:\Map(M,X)\rightarrow \Map(\partial M,X)$ denotes the restriction map, we can write
\[ r^*\mathfrak{g}_{(\tau_i,\rho_i)(\tau_j,\rho_j)} = t_{(\tau_i,\rho_i)}^{-1} t_{(\tau_j,\rho_j)} \,. \]
Moreover, the natural flat connection for $\{r^*\mathfrak{g}_{(\tau,\rho)}\}$ depends also only on the boundary $\partial M$. We denote it by $r^*\mathfrak{A}$. The \emph{transgression} of $[\underline{\mathfrak{g}},\underline{\mathfrak{A}}]$ is
\[ \mathfrak{T}[\underline{g},-\underline{\Lambda},\underline{B}] := [\underline{\mathfrak{g}},\underline{\mathfrak{A}}]\in H^1(\Map(\partial M,X),\mathcal{D}^1)\,.\]
When $\partial M = S^1$, $\Map(\partial M,X)$ is simply the loop space $LX$.

The \emph{transgression of a cohomology class} $F$ in $H^3(X,\Z)$ is defined by pulling back the class onto $\partial M\times \Map(\partial M,X)$ with the evaluation map
\[ \ev:\partial M\times \Map(\partial M,X) \rightarrow X \]
and integrating over $\partial M$. The result is a class in $H^2(\Map(\partial M,X),\Z)$, denoted by $\mathfrak{T}(F)$.

The transgression map for Deligne cohomology classes is homomorphic in the following sense \cite{GT,Joh2,CJM}. For a class $\chi\in H^2(X,\mathcal{D}^2)$, we denote the line bundle over $\Map(\partial M,X)$, obtained through transgression, by $\mathfrak{T}(\chi)$. For $\chi_1,\chi_2\in H^2(X,\mathcal{D}^2)$, the line bundles $\mathfrak{T}(\chi_1)$, $\mathfrak{T}(\chi_2)$ and $\mathfrak{T}(\chi_1+\chi_2)$ satisfy
\begin{align*} 
\mathfrak{T}(\chi_1+\chi_2) &\cong \mathfrak{T}(\chi_1)\otimes \mathfrak{T}(\chi_2)\,, \\
\mathfrak{T}(-\chi) &\cong \mathfrak{T}(\chi)^*\,, \\
\mathfrak{T}(0) &\cong U(1)\times \Map(\partial M,X)\,.
\end{align*}
Also, for any $\chi\in H^2(X,\mathcal{D}^2)$, the first Chern class $c_1(\mathfrak{T}(\chi))$ coincides with the transgression $\mathfrak{T}(d(L_\chi))$ of the Dixmier-Douady class of the bundle gerbe $L_\chi$, corresponding to the Deligne cohomology class $\chi$. 

Transgression can be performed also for higher Deligne cocycles. Let $M$ be a compact manifold of dimension $m$, with or without boundary. In \cite{GT} the authors gave an explicit expression for a \emph{transgression map} $C^p(X,\mathcal{D}^q)\rightarrow C^{p-m}(X,\mathcal{D}^{q-m})$. When $q=p$ and $M$ is of dimension $p+1$ with a $p$-dimensional boundary $\partial M$, one can construct a line bundle with connection over $\Map(\partial M,X)$ using the transgression map, such that the pullback bundle by the restriction $r:\Map(M,X)\rightarrow \Map(\partial M,X)$ is trivial \cite{GT}. In \cite{Joh2} the construction is explained using similar reasoning as in the cases $p=1,2$. There are also other ways of approaching the higher-dimensional parallel transport construction, but there is no point in discussing them here. Instead, we refer the interested reader to \cite{Joh2}.

\section{Cancelling the Freed-Witten anomaly}
We have now developed enough mathematical machinery to understand the Freed-Witten anomaly. We shall start by recalling how the anomaly arises and then proceed to show, using the above constructions, how it cancels when the D-brane and the $B$-field, or its restriction to the worldvolume, are related in a particular way. Our references for this section are \cite{Joh2}, \cite{CJM} and \cite{Kap}.

Since endpoints of open strings are attached to a D-brane $\phi:\Sigma\rightarrow X$, we are actually interested, not in $\Map(M,X)$, but rather in the submanifold $\Map_\Sigma(M,X)$ of maps, satisfying
\[ \xi(\partial M)\subset \phi(\Sigma) \,. \]
Trivial line bundles over $\Map(M,X)$ restrict to trivial line bundles over $\Map_\Sigma(M,X)$. It follows that the above constructions for holonomy and parallel transport generalize to yield similar constructions for $\Map_\Sigma(M,X)$.

Recall that the fermionic worldsheet path integral \eqref{Eq: Open string path integral, problematic} contained three problematic terms: the Pfaffian $\pfaff({\Dirac}_\xi)$ of the worldsheet Dirac operator, the ''holonomy'' of $B$ over the worldsheet $M$, which we now understand to be a higher-dimensional parallel transport, that is, a section of a trivial line bundle over $\Map_\Sigma(M,X)$, and finally the ''trace of holonomy'' of $A$ around the boundary $\partial M$, which we have not yet discussed.

It was shown in \cite{FW}, that $\pfaff:=\pfaff({\Dirac}_\xi)$ is also not a well-defined function on $\Map_\Sigma(M,X)$, but rather a section of the line bundle
\[ r_\Sigma^{-1}(\mathfrak{T}[\underline{w},0,0])\rightarrow \Map_\Sigma(M,X) \,, \]
where $\underline{w}$ is a representative of the image of the second Stiefel-Whitney class of $\Sigma$ under the inclusion $H^2(\Sigma,\Z_2)\rightarrow H^2(\Sigma,U(1))$
and $r_\Sigma:\Map_\Sigma(M,X)\rightarrow \Map(\partial M,\Sigma)$ is the restriction map. We denote the bundle gerbe corresponding to $[\underline{w},0,0]$ by $L_w$. The Pfaffian is, therefore, a section of $r_\Sigma^{-1}(\mathfrak{T}(L_w))$. We remind the reader that $\mathfrak{T}(L_w)$ is by no means trivial. It is the pullback that is trivial and its trivialization is given by the parallel transports. 

There are basically three cases to consider, of which the first and simplest, treated in \cite{FW}, is when the $B$-field bundle gerbe restricted to $\Sigma$, $L_B|_\Sigma$, satisfies
\begin{equation}\label{Eq: FW cancellation, simplest case} d(L_B|_\Sigma)=d(L_w)\in H^3(\Sigma,\Z)\,. \end{equation}
Using $d(L_{w^{-1}})=d(L_w^*)=-d(L_w)$\footnote{This follows from $d(L_w)$ being $2$-torsion.}, equation \eqref{Eq: FW cancellation, simplest case} becomes
\[ d(L_w\otimes L_B|_\Sigma) = d(L_w^*\otimes L^*_B|_\Sigma)=0 \,, \]
which implies that the transgression $\mathfrak{T}(L_w^*\otimes L^*_B|_\Sigma)$ is a trivial line bundle over $\Map(\partial M,\Sigma)$, whose pullback onto $\Map_\Sigma(M,X)$ is then also trivial. The bundle gerbe $L_w\otimes L_B|_\Sigma$ corresponds to the Deligne cocycle $(\underline{g}\underline{w},-\underline{\Lambda},\underline{B})$\footnote{More correctly, we should write $(\underline{g}|_\Sigma\underline{w},-\underline{\Lambda}|_\Sigma,\underline{B}|_\Sigma)$. To simplify the notation, we leave out the obvious restriction signs and hope this will not cause confusion.}, which fits into 
\[ \xymatrix{ 1 & & \\
\underline{g}\underline{w}=\delta(\underline{h}) \ar[u]^\delta \ar[r]^-{d\log} & 0 & \\
\underline{h} \ar[u]^\delta & -\underline{\Lambda} \ar[u]^\delta \ar[r]^{-d} & 0 \\
& & \underline{B} \ar[u]^\delta } \]
From the diagram we can immediately read that
\[ (gw)_{\alpha\beta\gamma} = h_{\beta\gamma}h_{\alpha\gamma}^{-1}h_{\alpha\beta} \]
and (compare to \eqref{Eq: Holonomy, gauge transformation})
\begin{equation}\label{Eq: Twisted gauge transform for A} A_\beta = A_\alpha + d\log h_{\alpha\beta} - \Lambda_{\alpha\beta} \,, \end{equation}
for some $\underline{A}\in C^0(\mathcal{U}|_\Sigma,\underline{\Omega}^1)$. We already know that the transgression line bundle is trivial. This may be explicitly seen by substituting the above relations into the formula for the transition functions \eqref{Eq: Explicit expression for gerbe transgression transition functions}. Indeed, the transition functions become
\begin{multline*} \prod_{e\subset \partial M}\exp\int_e (-\xi^*d\log h_{\rho_i(e)\rho_j(e)}+\xi^*A_{\rho_j(e)}-\xi^*A_{\rho_i(e)}) \\ \cdot \prod_{v\subset e\subset \partial M} h^{-1}_{\rho_i(e)\rho_i(v)}h_{\rho_i(e)\rho_j(v)}h^{-1}_{\rho_i(v)\rho_j(v)} h_{\rho_j(e)\rho_j(v)}h^{-1}_{\rho_i(e)\rho_j(v)}h_{\rho_i(e)\rho_j(e)}(\xi(v)) \\
= \prod_{e\subset \partial M}\exp\int_e (\xi^*A_{\rho_j(e)}-\xi^*A_{\rho_i(e)}) \cdot \prod_{v\subset e\subset \partial M} h^{-1}_{\rho_i(e)\rho_i(v)}h_{\rho_j(e)\rho_j(v)}(\xi(v)) = \Gamma_i^{-1}(\xi)\Gamma_j(\xi) \,, \end{multline*}
where
\[ \Gamma_i(\xi) = \prod_{e\subset \partial M}\exp\int_e \xi^*A_{\rho_i(e)} \cdot \prod_{v\subset e\subset \partial M} h_{\rho_i(e)\rho_i(v)}(\xi(v)) \]
are local functions trivializing the transgression bundle $\mathfrak{T}(L_w\otimes L_B|_\Sigma)$. Finally, their inverses,
\begin{equation}\label{Eq: Gamma_i^{-1}} \Gamma_i^{-1}(\xi) = \prod_{e\subset \partial M}\exp \left(-\int_e \xi^*A_{\rho_i(e)}\right) \cdot \prod_{v\subset e\subset \partial M} h^{-1}_{\rho_i(e)\rho_i(v)}(\xi(v)) \,, \end{equation}
provide a local trivialization for $\mathfrak{T}(L_w^*\otimes L^*_B|_\Sigma)$.

Let us now analyze $\underline{A}$ a bit. The transformation property \eqref{Eq: Twisted gauge transform for A} is reminiscent of ordinary gauge transformations. Indeed, if $\underline{\Lambda}=0$, it would be of the same form as a gauge transformation of a connection on a bundle, with transition functions $\underline{h}$, although in this case $\underline{h}$ satisfies merely a $\underline{g}\underline{w}$-twisted cocycle condition. Hence, one can heuristically regard $A$ as ''almost'' a $U(1)$ gauge connection on a projective bundle. Ignoring the slightly different meanings of the mathematical objects involved, we notice that \eqref{Eq: Gamma_i^{-1}} and \eqref{Eq: Holonomy, local expression} are quite alike, which motivates the definition
\[ \hol(\partial M,A) :=  r_\Sigma^{-1}(\Gamma^{-1}_i) \,. \]
The holonomy of $B$, or rather the parallel transport, is a trivializing section of $r_\Sigma^{-1}(\mathfrak{T}(L_B|_\Sigma))$. Thus, the product
\[ \pfaff \cdot \hol(M,B) \cdot \hol(\partial M,A) \]
is the difference of two trivializing sections and, as such, a well-defined function on $\Map_\Sigma(M,X)$, taking values in the trivial complex line
\[ r_\Sigma^{-1}(\mathfrak{T}(L_w))\otimes r_\Sigma^{-1}(\mathfrak{T}(L_B|_\Sigma))\otimes r_\Sigma^{-1}(\mathfrak{T}(L_w^* \otimes L_B^*|_\Sigma))\,. \]

We now proceed to make sense of the second case, originally explained by Kapustin \cite{Kap}, where $d(L_B|_\Sigma)\neq d(L_w)$ are such that
\begin{equation}\label{Eq: FW anomaly, condition 2} d(L_B|_\Sigma)-d(L_w) = d(L_B|_\Sigma)+d(L_w) \in \Tor H^3(\Sigma,\Z)\,. \end{equation}
Because $d(L_w)$ is torsion, \eqref{Eq: FW anomaly, condition 2} is equivalent to
\begin{equation}\label{Eq: H is pure torsion} d(L_B|_\Sigma)=[H]|_\Sigma \in \Tor H^3(\Sigma,\Z) \,. \end{equation}
We obviously need some new way to trivialize $\pfaff\cdot \hol(M,B)$. For definiteness, suppose that $d(L_B|_\Sigma)+d(L_w)=d(L_B|_\Sigma \otimes L_w)$ is $n$-torsion. The bundle gerbe $L_B|_\Sigma \otimes L_w$ corresponds to the Deligne cocycle $(\underline{g}\underline{w},-\underline{\Lambda},\underline{B})$. By \eqref{Eq: H is pure torsion}, the global $3$-form $H$ is exact, which means that there exists a global $2$-form $B$, such that $H=dB$. Since $d(B_\alpha-B)=0$, exactness implies $B_\alpha = B+d\mu_\alpha$, where $\mu_\alpha$ is a $1$-form. The difference $B_\alpha-d\mu_\alpha$ is, therefore, a global $2$-form. Combining this with $B_\alpha-B_\beta = d\Lambda_{\alpha\beta}$, yields
\begin{equation}\label{Eq: mu, q, Lambda relation} \mu_\beta = \mu_\alpha + d\log q_{\alpha\beta} - \Lambda_{\alpha\beta} \,, \end{equation}
where $q_{\alpha\beta}$ is a $U(1)$-valued function.

Let $L_\zeta$ be a bundle gerbe, represented by the Deligne cocycle $(\underline{\zeta},-\underline{\lambda},\underline{\beta})$, such that $L_\zeta$ is stably isomorphic to $L_B|_\Sigma \otimes L_w$. The product $L_p := L_B|_\Sigma \otimes L_w \otimes L_\zeta^*$ is then trivial and represented by the trivial Deligne cocycle $(\delta(\underline{p}),\delta(\underline{\nu})-d\log\underline{p},\underline{\nu})$. Expressed diagrammatically, the trivialization looks like
\begin{equation}\label{Eq: FW anomaly cancellation, trivialization diagram} \xymatrix{ 1 \\ \underline{g}\underline{w}\underline{\zeta}^{-1} \ar[u]^-\delta \ar[r]^-d & 0 \\ \underline{p} \ar[u]^-\delta \ar[r]^-{-d\log} & -\underline{\Lambda}+\underline{\lambda} \ar[u]^-\delta \ar[r]^-{-d} & 0 \\ & \underline{\nu} \ar[u]^-\delta \ar[r]^-d & \underline{B}-\underline{\beta} \ar[u]^-\delta } \end{equation}
Since
\begin{equation}\label{Eq: DD-class of L_zeta} d(L_\zeta)=\beta([\underline{\zeta}])=[H]|_\Sigma + W_3(\Sigma)\,, \end{equation}
$L_\zeta$ is an $n$-torsion bundle gerbe. As such, it admits a finite rank bundle gerbe module with connection, locally represented by $\{A_\alpha\}$, satisfying (see \eqref{Eq: A local transformation})
\begin{equation}\label{Eq: FW anomaly cancellation, connection formula 1} A_\beta = h_{\alpha\beta}^{-1} A_\alpha h_{\alpha\beta} + h_{\alpha\beta}^{-1}\,dh_{\alpha\beta} - \lambda_{\alpha\beta} \,, \end{equation}
where
\[ h_{\alpha\beta}h_{\beta\gamma}h^{-1}_{\alpha\gamma} = \zeta_{\alpha\beta\gamma} \,. \]
From \eqref{Eq: FW anomaly cancellation, trivialization diagram} we can read off the relation
\begin{equation}\label{Eq: FW anomaly cancellation, trivialization relation 1} \lambda_{\alpha\beta} = \nu_\beta - \nu_\alpha -d\log p_{\alpha\beta} + \Lambda_{\alpha\beta} \,. \end{equation}
Combining \eqref{Eq: FW anomaly cancellation, trivialization relation 1} with \eqref{Eq: mu, q, Lambda relation} yields
\begin{equation}\label{Eq: FW anomaly cancellation, trivialization relation 2} \lambda_{\alpha\beta} = \nu_\beta - \mu_\beta - \nu_\alpha + \mu_\alpha - d\log (pq^{-1})_{\alpha\beta} \,. \end{equation}
Inserting \eqref{Eq: FW anomaly cancellation, trivialization relation 2} into \eqref{Eq: FW anomaly cancellation, connection formula 1}, defining
\[ \widetilde{A}_\alpha := A_\alpha + \nu_\alpha - \mu_\alpha \qquad \text{and} \qquad \widetilde{h}_{\alpha\beta} := (hpq^{-1})_{\alpha\beta} \,, \]
and using
\[ \widetilde{h}^{-1}_{\alpha\beta}\,d\widetilde{h}_{\alpha\beta} = h^{-1}_{\alpha\beta}\,dh_{\alpha\beta} + d\log (pq^{-1})_{\alpha\beta} \,, \]
yields
\begin{equation}\label{Eq: FW anomaly cancellation, connection formula 2} \widetilde{A}_\beta =  \widetilde{h}_{\alpha\beta}^{-1} \widetilde{A}_\alpha \widetilde{h}_{\alpha\beta} + \widetilde{h}_{\alpha\beta}^{-1}\,d\widetilde{h}_{\alpha\beta} \,, \end{equation}
a formula, which looks precisely like an ordinary gauge transformation. Of course, $\{\widetilde{h}_{\alpha\beta}\}$ is still only a twisted cocycle:
\[ \widetilde{h}_{\alpha\beta}\widetilde{h}_{\beta\gamma}\widetilde{h}_{\alpha\gamma}^{-1} = (\zeta\delta(p)\delta(q)^{-1})_{\alpha\beta\gamma} \,. \]

Let us now briefly comment on the physical interpretations of the above constructions. We interpret $\{A_\alpha\}$ as the locally defined $A$-field and $\{h_{\alpha\beta}\}$ as the $U(n)$ lifts of the transition functions of the Chan-Paton bundle. The obstruction for the Chan-Paton bundle lifting into a $U(n)$ bundle is
\[ [\underline{\zeta}\delta(\underline{p})\delta(\underline{q})^{-1}] = [\underline{\zeta}] \in H^2(\Sigma,\underline{U(1)}) \,, \]
or equivalently, its image $\beta([\underline{\zeta}])\in H^3(\Sigma,\Z)$, where $\beta:H^2(\Sigma,\underline{U(1)})\xrightarrow{\cong} H^3(\Sigma,\Z)$ is the Bockstein isomorphism \eqref{Eq: Bockstein H^2(X,U(1)) to H^3(X,Z)}.

The Freed-Witten anomaly can now be cancelled as follows. The path integral must, of course, contain the usual sections of $r_\Sigma^{-1}(\mathfrak{T}(L_B|_\Sigma))$ and $r_\Sigma^{-1}(\mathfrak{T}(L_w))$, coming from $\hol(M,B)$ and $\pfaff$. Next, we need local sections of $r_\Sigma^{-1}(\mathfrak{T}(L_p^*))=r_\Sigma^{-1}(\mathfrak{T}(L_B^*|_\Sigma \otimes L^*_w \otimes L_\zeta))$, but these are easy to obtain, since the transgressed bundle is trivial. Inserting the data of $L_p$ into \eqref{Eq: Explicit expression for gerbe transgression transition functions}, gives trivializing sections $\Gamma_i^{-1}$ of $\mathfrak{T}(L_p^*)$. We shall drop the index $i$ for notational convenience and simply write $\Gamma^{-1}$ to denote the pullback $r_\Sigma^*(\Gamma_i^{-1})$, keeping in mind that it is a trivializing section, not a function on $\Map_\Sigma(M,X)$. The only ingredient that is missing, is a trivializing section of $r_\Sigma^{-1}(\mathfrak{T}(L_\zeta^*))$. Kapustin showed in \cite{Kap} that it is given by the ''trace of holonomy'' of the $A$-field.

Recall first the Abelian holonomy formula \eqref{Eq: Holonomy, local expression}. In the non-Abelian case there is also a parallel transport equation, but the solution is not as simple to write as in the Abelian case. This is basically because the exponentials
\[ \exp\left(-\int_e \xi^*\widetilde{A}_{\rho(e)}\right) \]
are ambiguous, due to the noncommutativity of $U(n)$. The correct interpretation is to expand the exponential and perform the products of integrals in path-order. This is customarily indicated by placing a $\mathcal{P}$ in front of the exponential. However, it should suffice to think of\footnote{Note that we are using $\widetilde{A}$ here, since it has the correct $U(n)$ gauge transformation property.}
\[ \mathcal{P}\exp\left(-\int_e \xi^*\widetilde{A}_{\rho(e)}\right) \]
as the solution of the non-Abelian parallel transport equation along $e$, if we use the well-known result that these quantities transform as \cite{PS,Gil,Nak}
\[ \mathcal{P}\exp\left(-\int_e \xi^*\widetilde{A}_{\rho_j(e)}\right) = \widetilde{h}^{-1}_{\rho_i(e)\rho_j(e)}(\xi(v_e^1))\cdot \mathcal{P}\exp\left(-\int_e \xi^*\widetilde{A}_{\rho_i(e)}\right) \cdot \widetilde{h}_{\rho_i(e)\rho_j(e)}(\xi(v_e^0)) \]
under
\[ \widetilde{A}_{\rho_j(e)} = \widetilde{h}^{-1}_{\rho_i(e)\rho_j(e)}\widetilde{A}_\alpha \widetilde{h}_{\rho_i(e)\rho_j(e)} + \widetilde{h}^{-1}_{\rho_i(e)\rho_j(e)}\,d\widetilde{h}_{\rho_i(e)\rho_j(e)}\,. \]
To make the following calculations a bit shorter, we shall use the shorthand notation
\[ \mathcal{W}^{\widetilde{A}}_{\rho(e)}(\xi) = \mathcal{P}\exp\left(-\int_e \xi^*\widetilde{A}_{\rho(e)}\right) \,. \]

Let us, for the sake of simplicity, assume that $\partial M \cong S^1$. We define
\begin{multline} \Tr\hol(S^1,A)_{(\tau,\rho)} := \Tr\bigg[\prod_{e\subset S^1} \widetilde{h}^{-1}_{\rho(e)\rho(v_e^1)} \cdot \mathcal{W}^{\widetilde{A}}_{\rho(e)} \cdot \widetilde{h}_{\rho(e)\rho(v_e^0)}\bigg] \\ \cdot \prod_{e\subset S^1} \exp\int_e \xi^*(\nu-\mu)_{\rho(e)} \cdot \prod_{v\subset e\subset S^1} (pq^{-1})_{\rho(e)\rho(v)} \,, \end{multline}
where we have suppressed the arguments of $\widetilde{h}_{\alpha\beta}$, $p$ and $q$, as they are unambiguous, and also left out the argument $\xi$ of $\mathcal{W}^{\widehat{A}}_{\rho(e)}$. Comparing $\Tr\hol(S^1,A)_{(\tau_i,\rho_i)}$ and $\Tr\hol(S^1,A)_{(\tau_j,\rho_j)}$ on overlaps $V_{(\tau_i,\rho_i)}\cap V_{(\tau_j,\rho_j)}$, is a tedious task. The trace is the difficult part:
\begin{align*} &\Tr\bigg[\prod_{e\subset S^1} \widetilde{h}^{-1}_{\rho_j(e)\rho_j(v_e^1)} \cdot \mathcal{W}^{\widetilde{A}}_{\rho_j(e)} \cdot \widetilde{h}_{\rho_j(e)\rho_j(v_e^0)}\bigg] \\
&= \Tr\bigg[\ldots \mathcal{W}^{\widetilde{A}}_{\rho_j(e_{k+1})}\cdot \widetilde{h}_{\rho_j(e_{k+1})\rho_j(v_{e_{k+1}}^0)} \cdot \widetilde{h}^{-1}_{\rho_j(e_k)\rho_j(v_{e_k}^1)} \cdot \mathcal{W}^{\widetilde{A}}_{\rho_j(e_k)}\ldots \bigg] \\
&= \Tr\bigg[\ldots \mathcal{W}^{\widetilde{A}}_{\rho_i(e_{k+1})} \cdot \widetilde{h}_{\rho_i(e_{k+1})\rho_j(e_{k+1})}\widetilde{h}_{\rho_j(e_{k+1})\rho_j(v_{e_{k+1}}^0)} \cdot \widetilde{h}^{-1}_{\rho_j(e_k)\rho_j(v_{e_k}^1)}\widetilde{h}^{-1}_{\rho_i(e_k)\rho_j(e_k)} \cdot \mathcal{W}^{\widetilde{A}}_{\rho_j(e_k)}\ldots \bigg] \displaybreak \\
&= \Tr\bigg[\ldots \mathcal{W}^{\widetilde{A}}_{\rho_i(e_{k+1})} \cdot \widetilde{h}_{\rho_i(e_{k+1})\rho_j(v_{e_{k+1}}^0)} \cdot \widetilde{h}^{-1}_{\rho_i(e_k)\rho_j(v_{e_k}^1)} \cdot \mathcal{W}^{\widetilde{A}}_{\rho_j(e_k)}\ldots \bigg] \\ &\qquad\cdot \prod_k \big(\zeta\delta(p)\delta(q)^{-1}\big)_{\rho_i(e_{k+1})\rho_j(e_{k+1})\rho_j(v_{e_{k+1}}^0)} \big(\zeta\delta(p)\delta(q)^{-1}\big)^{-1}_{\rho_i(e_k)\rho_j(e_k)\rho_j(v_{e_k}^1)}  \\
&= \Tr\bigg[\ldots \mathcal{W}^{\widetilde{A}}_{\rho_i(e_{k+1})} \cdot \widetilde{h}_{\rho_i(e_{k+1})\rho_i(v_{e_{k+1}}^0)}\widetilde{h}_{\rho_i(v_{e_{k+1}}^0)\rho_j(v_{e_{k+1}}^0)} \cdot \widetilde{h}^{-1}_{\rho_i(v_{e_{k+1}}^0)\rho_j(v_{e_k}^1)}\widetilde{h}^{-1}_{\rho_i(e_k)\rho_i(v_{e_{k+1}}^0)} \cdot \mathcal{W}^{\widetilde{A}}_{\rho_j(e_k)}\ldots \bigg] \\ &\qquad \cdot \prod_k \big(\zeta\delta(p)\delta(q)^{-1}\big)_{\rho_i(e_{k+1})\rho_j(e_{k+1})\rho_j(v_{e_{k+1}}^0)}\big(\zeta\delta(p)\delta(q)^{-1}\big)^{-1}_{\rho_i(e_k)\rho_j(e_k)\rho_j(v_{e_k}^1)} \\ &\qquad\qquad \cdot \big(\zeta\delta(p)\delta(q)^{-1}\big)^{-1}_{\rho_i(e_{k+1}) \rho_i(v_{e_{k+1}}^0)\rho_j(v_{e_{k+1}}^0)}\big(\zeta\delta(p)\delta(q)^{-1}\big)_{\rho_i(e_k)\rho_i(v_{e_{k+1}}^0)\rho_j(v_{e_k}^1)} \\
&= \Tr\bigg[\ldots \mathcal{W}^{\widetilde{A}}_{\rho_i(e_{k+1})} \cdot \widetilde{h}_{\rho_i(e_{k+1})\rho_i(v_{e_{k+1}}^0)} \cdot \widetilde{h}^{-1}_{\rho_i(e_k)\rho_i(v_{e_{k+1}}^0)} \cdot \mathcal{W}^{\widetilde{A}}_{\rho_j(e_k)}\ldots \bigg] \\ &\qquad \cdot \prod_k \big(\zeta\delta(p)\delta(q)^{-1}\big)_{\rho_i(e_{k+1})\rho_j(e_{k+1})\rho_j(v_{e_{k+1}}^0)}\big(\zeta\delta(p)\delta(q)^{-1}\big)^{-1}_{\rho_i(e_k)\rho_j(e_k)\rho_j(v_{e_k}^1)} \\ &\qquad\qquad \cdot \big(\zeta\delta(p)\delta(q)^{-1}\big)^{-1}_{\rho_i(e_{k+1}) \rho_i(v_{e_{k+1}}^0)\rho_j(v_{e_{k+1}}^0)}\big(\zeta\delta(p)\delta(q)^{-1}\big)_{\rho_i(e_k)\rho_i(v_{e_k}^1)\rho_j(v_{e_k}^1)}   \\
&= \Tr\bigg[\prod_{e\subset S^1} \widetilde{h}^{-1}_{\rho_i(e)\rho_i(v_e^1)} \cdot \mathcal{W}^{\widetilde{A}}_{\rho_i(e)} \cdot \widetilde{h}_{\rho_i(e)\rho_i(v_e^0)}\bigg] \\ &\qquad \cdot \prod_{v\subset e \subset S^1} \big(\zeta\delta(p)\delta(q)^{-1}\big)_{\rho_i(e)\rho_i(v)\rho_j(v)}\big(\zeta\delta(p)\delta(q)^{-1}\big)^{-1}_{\rho_i(e)\rho_j(e)\rho_j(v)} \,.
\end{align*}
Using this result, we can easily compute
\begin{align*} &\Tr\bigg[\prod_{e\subset S^1} \widetilde{h}^{-1}_{\rho_j(e)\rho_j(v_e^1)} \cdot \mathcal{W}^{\widetilde{A}}_{\rho_j(e)} \cdot \widetilde{h}_{\rho_j(e)\rho_j(v_e^0)}\bigg] \cdot \prod_{e\subset S^1} \exp\int_e \xi^*(\nu-\mu)_{\rho_j(e)} \cdot \prod_{v\subset e\subset S^1} (pq^{-1})_{\rho_j(e)\rho_j(v)}\\
&= \Tr\bigg[\prod_{e\subset S^1} \widetilde{h}^{-1}_{\rho_i(e)\rho_i(v_e^1)} \cdot \mathcal{W}^{\widetilde{A}}_{\rho_i(e)} \cdot \widetilde{h}_{\rho_i(e)\rho_i(v_e^0)}\bigg] \cdot \prod_{e\subset S^1} \exp\int_e \xi^*(\nu-\mu)_{\rho_i(e)} \cdot \prod_{v\subset e\subset S^1} (pq^{-1})_{\rho_i(e)\rho_i(v)} \\ &\qquad \cdot \prod_{e\subset S^1} \exp\int_e \xi^*\big((\nu-\mu)_{\rho_j(e)}-(\nu-\mu)_{\rho_i(e)}\big) \cdot \prod_{v\subset e\subset S^1} (pq^{-1})_{\rho_j(e)\rho_j(v)} (p^{-1}q)_{\rho_i(e)\rho_i(v)} \\ &\qquad \cdot \prod_{v\subset e \subset S^1} \big(\zeta\delta(p)\delta(q)^{-1}\big)_{\rho_i(e)\rho_i(v)\rho_j(v)}\big(\zeta\delta(p)\delta(q)^{-1}\big)^{-1}_{\rho_i(e)\rho_j(e)\rho_j(v)}  \\
&= \Tr\bigg[\prod_{e\subset S^1} \widetilde{h}^{-1}_{\rho_i(e)\rho_i(v_e^1)} \cdot \mathcal{W}^{\widetilde{A}}_{\rho_i(e)} \cdot \widetilde{h}_{\rho_i(e)\rho_i(v_e^0)}\bigg] \cdot \prod_{e\subset S^1} \exp\int_e \xi^*(\nu-\mu)_{\rho_i(e)} \cdot \prod_{v\subset e\subset S^1} (pq^{-1})_{\rho_i(e)\rho_i(v)} \\ &\qquad \cdot \prod_{e\subset S^1} \exp\int_e \xi^*\big((\nu-\mu)_{\rho_j(e)}-(\nu-\mu)_{\rho_i(e)}\big) \cdot \prod_{v\subset e \subset S^1} \zeta_{\rho_i(e)\rho_i(v)\rho_j(v)}\zeta^{-1}_{\rho_i(e)\rho_j(e)\rho_j(v)} \\ &\qquad \cdot \prod_{v\subset e\subset S^1} (pq^{-1})_{\rho_j(e)\rho_j(v)} (p^{-1}q)_{\rho_i(e)\rho_i(v)} \big(\delta(p)\delta(q)^{-1}\big)_{\rho_i(e)\rho_i(v)\rho_j(v)}\big(\delta(p)\delta(q)^{-1}\big)^{-1}_{\rho_i(e)\rho_j(e)\rho_j(v)}  \\
&= \Tr\bigg[\prod_{e\subset S^1} \widetilde{h}^{-1}_{\rho_i(e)\rho_i(v_e^1)} \cdot \mathcal{W}^{\widetilde{A}}_{\rho_i(e)} \cdot \widetilde{h}_{\rho_i(e)\rho_i(v_e^0)}\bigg] \cdot \prod_{e\subset S^1} \exp\int_e \xi^*(\nu-\mu)_{\rho_i(e)} \cdot \prod_{v\subset e\subset S^1} (pq^{-1})_{\rho_i(e)\rho_i(v)} \\ &\qquad \cdot \prod_{e\subset S^1} \exp\int_e \xi^*\big((\nu-\mu)_{\rho_j(e)}-(\nu-\mu)_{\rho_i(e)}\big) \cdot \prod_{v\subset e \subset S^1} \zeta_{\rho_i(e)\rho_i(v)\rho_j(v)}\zeta^{-1}_{\rho_i(e)\rho_j(e)\rho_j(v)} \\ &\qquad \cdot \prod_{v\subset e\subset S^1} (pq^{-1})_{\rho_j(e)\rho_j(v)} (p^{-1}q)_{\rho_i(e)\rho_i(v)} \big(\delta(p)\delta(q)^{-1}\big)_{\rho_i(e)\rho_i(v)\rho_j(v)}\big(\delta(p)^{-1}\delta(q)\big)_{\rho_i(e)\rho_j(e)\rho_j(v)} 
\\
&= \Tr\bigg[\prod_{e\subset S^1} \widetilde{h}^{-1}_{\rho_i(e)\rho_i(v_e^1)} \cdot \mathcal{W}^{\widetilde{A}}_{\rho_i(e)} \cdot \widetilde{h}_{\rho_i(e)\rho_i(v_e^0)}\bigg] \cdot \prod_{e\subset S^1} \exp\int_e \xi^*(\nu-\mu)_{\rho_i(e)} \cdot \prod_{v\subset e\subset S^1} (pq^{-1})_{\rho_i(e)\rho_i(v)} \\ &\qquad \cdot \prod_{e\subset S^1} \exp\int_e \xi^*\big((\nu-\mu)_{\rho_j(e)}-(\nu-\mu)_{\rho_i(e)}\big) \cdot \prod_{v\subset e \subset S^1} \zeta_{\rho_i(e)\rho_i(v)\rho_j(v)}\zeta^{-1}_{\rho_i(e)\rho_j(e)\rho_j(v)} \\ &\qquad \cdot \prod_{v\subset e\subset S^1} (p^{-1}q)_{\rho_i(e)\rho_j(e)} \cdot \underbrace{\prod_{v\subset e\subset S^1} (pq^{-1})_{\rho_i(v)\rho_j(v)}}_{=1}\,.
\end{align*}
Hence, $\Tr\hol(S^1,A)_{(\tau,\rho)}$ transforms as
\begin{align*} \Tr\hol(S^1,A)_{(\tau_j,\rho_j)} &= \Tr\hol(S^1,A)_{(\tau_i,\rho_i)} \cdot \prod_{e\subset S^1} \exp\int_e \xi^*\big((\nu-\mu)_{\rho_j(e)}-(\nu-\mu)_{\rho_i(e)}\big) \\ &\quad \cdot \prod_{v\subset e \subset S^1} \zeta_{\rho_i(e)\rho_i(v)\rho_j(v)}\zeta^{-1}_{\rho_i(e)\rho_j(e)\rho_j(v)} \cdot \prod_{v\subset e\subset S^1} (p^{-1}q)_{\rho_i(e)\rho_j(e)} \,. \end{align*}

On the other hand, inserting the data of $L_\zeta$ into \eqref{Eq: Explicit expression for gerbe transgression transition functions}, gives
\begin{align*} &\prod_{e\subset S^1} \exp\left(-\int_e \xi^*\lambda_{\rho_i(e)\rho_j(e)}\right)\cdot \prod_{v\subset e\subset S^1} \zeta^{-1}_{\rho_i(e)\rho_i(v)\rho_j(v)}\zeta_{\rho_i(e)\rho_j(e)\rho_j(v)} \\
&\stackrel{\eqref{Eq: FW anomaly cancellation, trivialization relation 2}}{=} \exp\int_e \xi^*\big((\nu-\mu)_{\rho_i(e)}-(\nu-\mu)_{\rho_j(e)}\big) \cdot \prod_{v\subset e\subset S^1} (pq^{-1})_{\rho(e_i)\rho(e_j)} \\ &\qquad\qquad \cdot \prod_{v\subset e\subset S^1} \zeta^{-1}_{\rho_i(e)\rho_i(v)\rho_j(v)}\zeta_{\rho_i(e)\rho_j(e)\rho_j(v)}
\end{align*}
as the transition functions of $\mathfrak{T}(L_\zeta)$, but these are simply inverses of the transition functions of $\Tr\hol(S^1,A)_{(\tau,\rho)}$. Therefore, $\Tr\hol(S^1,A)$ defines a trivializing section of the dual $\mathfrak{T}(L_\zeta^*)$, which is precisely what we wanted. We denote the pullback $r_\Sigma^*(\Tr\hol(S^1,A))$ onto $\Map_\Sigma(M,X)$ also by $\Tr\hol(S^1,A)$. The above computations can easily be generalized to the case where $\partial M$ consists of finitely many circles, in which case the transgressed bundles over $\Map_\Sigma(\partial M,X)$ should be understood as tensor products of bundles transgressed over $\Map_\Sigma(S^1,X)$. Then $\Tr\hol(\partial M,A)$ becomes a tensor product of $\Tr\hol(S^1,A)$ over all the boundary circles.

We can now put all the pieces together to construct a well-defined function
\[ \pfaff\cdot \hol(M,B)\cdot \Tr\hol(\partial M,A) \cdot \Gamma^{-1} \]
over the space $\Map_\Sigma(M,X)$. Thus, the Freed-Witten anomaly is cancelled, when $[H]|_\Sigma \in \Tor H^3(\Sigma,\Z)$, whenever the cohomology equation
\begin{equation}\label{Eq: FW anomaly cancellation condition} [H]|_\Sigma = W_3(\Sigma)+\beta([\underline{\zeta}]) \end{equation}
holds.

Even if $[H]|_\Sigma \neq \Tor H^3(\Sigma,\Z)$, by allowing the $A$-field to be a $U_1$ bundle gerbe module connection, the anomaly can still be cancelled when \eqref{Eq: FW anomaly cancellation condition} holds. The interested reader is referred to \cite{CJM} for a treatment of the nontorsion case. We shall not discuss it.

\chapter{Twisted D-Branes And Noncommutative Generalizations}
By a \emph{twisted D-brane}, we mean a D-brane in the presence of a nonvanishing background $B$-field. We start this chapter by extending the $K$-theoretic classification to twisted D-branes. This is done by replacing $K$-theory with \emph{twisted $K$-theory}\cite{Ros,Bla,AS1,BM,CW,BCMMS}. The Chan-Paton bundle of a stack of $n$ coincident D-branes, is generally a $PU(n)$, or even $PU(\Hilb)$, bundle. Just as a configuration of D-branes with $U(n)$ Chan-Paton bundles yields a $K$-theory class $x\in K^0(\Sigma)$, a configuration of D-branes with projective bundles yields a twisted $K$-theory class of $\Sigma$. To express the charge in terms of spacetime quantities, we need a Gysin map in twisted $K$-theory. More precisely, we need a generalization of the concept of $K$-orientability, which is equivalent to the existence of \emph{twisted fundamental classes} in a dual \emph{twisted $K$-homology}\cite{Wan,Sza1} theory. Twisted $K$-theory and twisted $K$-homology can be defined using \emph{$K$-theory and $K$-homology of $C^*$-algebras} \cite{Kas3,Ros,HR}. This takes us automatically to the world of \emph{noncommutative geometry} \cite{Con,GVF,Mad,Lan2}, since $K$-theory and $K$-homology of $C^*$-algebras are the noncommutative generalizations of $K$-theory and $K$-homology of topological spaces.

The philosophy behind noncommutative topology is that any locally compact space $X$ can be completely described by its algebra $C_0(X)$ of continuous functions $X\rightarrow \C$, vanishing at infinity. Clearly $C_0(X)$ is a commutative $C^*$-algebra\footnote{We assume the reader is familiar with the basics of $C^*$-algebras. A friendly introduction is given in \cite{GVF}. A more extensive treatment can be found in \cite{Ped}.}. It turns out that a number of geometric concepts for $X$ can be described as purely algebraic constructions for $C_0(X)$. In particular, $\spinc$-manifolds can be described purely algebraically using objects called \emph{spectral triples} \cite{Con,GVF}. Having obtained an algebraic definition for the geometric objects of interest, we can easily change the $C^*$-algebra to any other, generally noncommutative, $C^*$-algebra, resulting in what is called noncommutative geometry. Unfortunately, we can not discuss noncommutative geometry in detail in this thesis, as it would take a huge number of pages to even get through the basics. 

Having described twisted $K$-theory and $K$-homology in terms of $C^*$-algebras, we can easily replace the $C^*$-algebras of the twisted D-branes (and of the spacetime) with more general noncommutative $C^*$-algebras. To proceed, we need a unification of $K$-theory and $K$-homology of $C^*$-algebras. Such a bivariant $K$-theory is the \emph{$KK$-theory} of Kasparov \cite{Kas1,Kas2,Bla,KT}. Our goal is to make sense of the Minasian-Moore formula \eqref{Eq: Minasian-Moore formula} in the noncommutative case \cite{BMRS1,BMRS2,Sza2}.

\section{Classifying twisted D-branes}
Let $A$ be a $C^*$-algebra. We always assume our $C^*$-algebras to be separable, although this is not required for most of the definitions. From the point of view of noncommutative topology, separability of the function algebra translates to metrizability of the topological space \cite{GVF}. The $C^*$-algebras are \emph{not} assumed to be unital in general. A topological space is compact if and only if its function algebra is unital. Allowing the $C^*$-algebras to be nonunital, translates to allowing the topological spaces to be locally compact\footnote{Of course, if the $C^*$-algebra is noncommutative, it does not represent any actual topological space.}.

The $C^*$-algebra $A$ can be made into a unital $C^*$-algebra $A^+$ by \emph{adjoining a unit} as follows. Define $A^+ := A\oplus \C$, multiplication
\[ (a,\lambda)\cdot (b,\mu) := (ab+\lambda b + \mu a,\lambda\mu) \,. \]
In the commutative picture, this corresponds to one-point compactification of the locally compact space. If $A$ is already unital, $A^+$ is \emph{not} isomorphic to $A$\footnote{Compare this to adjoining a disjoint point to a compact space.}. It is also useful to define
\[ \widetilde{A} := \begin{cases} A & \text{if $A$ is unital,} \\ A^+ & \text{if $A$ is nonunital.} \end{cases} \,. \]

Recall that an \emph{idempotent} of a $C^*$-algebra $A$ is an element $e\in A$, such that $e^2=e$. A \emph{projection} $p\in A$ is a self-adjoint idempotent: $p^*=p^2=p$. Two projections, $p$ and $q$, are said to be \emph{unitarily equivalent}, if $q = upu^*$, for some unitary element $u\in \widetilde{A}$. Unitary equivalence defines an equivalence relation $\sim_u$. The set of unitary equivalence classes of projections of $A$ is denoted by $\mathcal{P}(A)$

Let $M_n(A)$ denote the $C^*$-algebra $A\otimes M_n(\C)$ of $n\times n$ matrices, with entries in $A$. There are natural inclusions, $M_n(A)\hookrightarrow M_{n+1}(A)$, given by diagonal embeddings $a\mapsto \diag(a,0)$, for $a\in M_n(A)$. Taking the limit of $n$ to infinity, yields the $C^*$-algebra of compact operators on $\Hilb$, tensored with $A$:
\[ A\otimes \mathcal{K} := \varinjlim_n M_n(A) \,. \]
Intuitively, it is the $C^*$-algebra of infinite $A$-valued matrices, with a finite number of nonzero entries.

Let $V(A) := \mathcal{P}(A\otimes \mathcal{K})$. For elements $[p],[q]\in V(A)$, define
\[ [p]+[q]:= [p\oplus q] \,, \]
where
\[ p\oplus q := \begin{pmatrix} p & 0 \\ 0 & q \end{pmatrix} \,. \]
This makes $V(A)$ an Abelian monoid, with identity $[0]$. It is Abelian, since
\[ \begin{pmatrix} p & 0 \\ 0 & q \end{pmatrix} \sim_u \begin{pmatrix} 0 & 1 \\ 1 & 0 \end{pmatrix} \begin{pmatrix} p & 0 \\ 0 & q \end{pmatrix} \begin{pmatrix} 0 & 1 \\ 1 & 0 \end{pmatrix} = \begin{pmatrix} q & 0 \\ 0 & p \end{pmatrix} \,. \]
$V(A)$ is \emph{stable}, meaning that
\begin{equation}\label{Eq: V(A) stable} V(A\otimes \mathcal{K}) = \mathcal{P}((A\otimes \mathcal{K}\otimes \mathcal{K}) = \mathcal{P}(A\otimes \mathcal{K}) = V(A) \,, \end{equation}
where we have used $\mathcal{K}\otimes \mathcal{K} \cong \mathcal{K}$.
\begin{definition}
Let $A$ be a unital $C^*$-algebra. The \emph{(operator) $K$-theory} group $K_0(A)$ is the Grothendieck group of $V(A)$. In other words, it is the Abelian group generated by formal differences $[p]-[q]$. 
\end{definition}

The simplest nontrivial $C^*$-algebra is $\C$. If $P$ is a projection in $M_n(\C)$, then, as a self-adjoint matrix, it is diagonalizable by a unitary matrix $U$: $D = UPU^*$, where $D$ is diagonal. Now, $P^2 = U^*DUU^*DU = U^*D^2 U$, but on the other hand $P^2=P=U^*DU$. Thus, $D$ is also a projection and all its entries must necessarily be $1$. Hence, all projections of rank $n$ belong to the same unitary equivalence class. Thus, $V(\C)\cong \N_0$ and $K_0(\C) \cong \Z$.

\begin{lemma} For $C^*$-algebras $A_1$ and $A_2$, 
\begin{equation}\label{Eq: Operator K-theory sum formula} K_0(A_1\oplus A_2) \cong K_0(A_1)\oplus K_0(A_2) \,. \end{equation}
\end{lemma}
\begin{proof}
Let $\pi_k:A_1\oplus A_2 \rightarrow A_k$ denote the projections and $i_k:A_k\rightarrow A_1\oplus A_2$ the inclusions. The split exact sequence
\[ \xymatrix{ 0 \ar[r] & A_1 \ar[r]^-{i_1} & A_1\oplus A_2 \ar@<0.5ex>[r]^-{\pi_2} & A_2 \ar[r]\ar@<0.5ex>[l]^-{i_2} & 0 } \]
induces a split exact sequence \cite{Bla}
\[ \xymatrix{ 0 \ar[r] & K_0(A_1) \ar[r]^-{(i_1)_*} & K_0(A_1\oplus A_2) \ar[r]^-{(\pi_2)_*} & K_0(A_2) \ar[r] & 0 } \]
in $K$-theory. The commutative diagram
\[ \xymatrix{ 0 \ar[d]_\cong \ar[r] & K_0(A_1) \ar[d]_\cong \ar[r]^-{(i_1)_*} & K_0(A_1\oplus A_2) \ar[d]^{(\pi_1)_*\oplus (\pi_2)_*} \ar[r]^-{(\pi_2)_*} & K_0(A_2) \ar[d]_\cong \ar[r] & 0 \ar[d]_\cong \\ 0 \ar[r] & K_0(A_1) \ar[r] & K_0(A_1)\oplus K_0(A_2) \ar[r] & K_0(A_2) \ar[r] & 0 } \]
has exact rows. Since all the vertical arrows, except the one in the middle, are isomorphisms, the \emph{five lemma} \cite{BT} implies that also the middle one is an isomorphism.
\end{proof}

\begin{definition}
Let $A$ be a nonunital $C^*$-algebra. The natural map $\pi:A^+\rightarrow \C$ induces a homomorphism $\pi_*:K_0(A^+)\rightarrow K_0(\C)\cong \Z$. \emph{Reduced operator $K$-theory} of $A^+$ is defined as the subgroup
\[ \widetilde{K}_0(A^+) := \ker(\pi_*) \subset K_0(A^+) \,. \]
For the nonunital $C^*$-algebra $A$ we define
\[ K_0(A) := \widetilde{K}_0(A^+) \,. \]
If $A$ is unital, \eqref{Eq: Operator K-theory sum formula} implies
\[ K_0(A^+) = K_0(A\oplus \C) \cong K_0(A)\oplus K_0(\C) \cong K_0(A) \oplus \Z \,, \]
with
\[ \ker(\pi_*) = K_0(A)\oplus 0 \cong K_0(A) \,. \]
Thus, the definition
\[ K_0(A) := \widetilde{K}_0(A^+) \]
for both unital and nonunital $C^*$-algebras.
\end{definition}
As a Grothendieck group, $K_0$ satisfies a similar universality property as $K^0$ for vector bundles (see \eqref{Eq: K-theory universality property}). Let $A$ and $B$ be $C^*$-algebras and $f:A\rightarrow B$ a $*$-homomorphism. There is an induced homomorphism $f_*:V(A)\rightarrow V(B)$, which descends to a homomorphism $f_*:K_0(A)\rightarrow K_0(B)$. Moreover, homotopic maps $f,g:A\rightarrow B$ induce the same map $f_* = g_*$ in $K$-theory \cite{Bla}. It follows from \eqref{Eq: V(A) stable}, that $K$-theory is \emph{stable}, or \emph{Morita invariant}, meaning that
\begin{equation}\label{Eq: K-theory is Morita invariant} K_0(A\otimes \mathcal{K}) \cong K_0(A) \,. \end{equation}
\begin{proposition}
Let $\mathbf{SC^*}$ be the category of separable $C^*$-algebras. Operator $K$-theory defines a covariant functor $K_0:\mathbf{SC^*}\rightarrow \mathbf{Ab}$.\qed
\end{proposition}

For a $C^*$-algebra $A$, we define the \emph{suspension} $C^*$-algebra
\begin{equation}\label{Eq: Suspension of a C*-algebra} \Sigma A := C_0(\R)\otimes A \,. \end{equation}
\begin{definition}\label{Definition: Operator K-theory}
Higher operator $K$-functors $K_n$, $n\in \N$, are covariant functors $K_n:\mathbf{SC^*}\rightarrow \mathbf{Ab}$. For a $C^*$-algebra $A$, we set
\[ K_n(A) := K_0(\Sigma^n A) = \widetilde{K}_0((\Sigma^n A)^+) \,. \]
\end{definition}

\begin{theorem}[Bott periodicity \cite{Bla}]
There are canonical isomorphisms
\[ K_n(A) \cong K_{n+2}(A) \,. \]
\end{theorem}
\begin{definition}
By Bott periodicity, there are only two inequivalent operator $K$-functors: $K_0$ and $K_1$. Thus, we write
\[ K_\bullet(A) := K_0(A) \oplus K_1(A) \,. \]
\end{definition}

To understand the relation between operator $K$-theory and $K$-theory of vector bundles, we need the following basic result of noncommutative geometry. Let $X$ be compact, $E\rightarrow X$ be a $U(n)$ vector bundle and $\Gamma(E):=C(X,E)$ the group of continuous sections of $E$. It is a module over the $C^*$-algebra $C(X)$ and is compatible with direct sum and tensor product of vector bundles $E,F$ over $X$:
\[ \Gamma(E)\oplus \Gamma(F) \cong \Gamma(E\oplus F) \qquad \text{and} \qquad \Gamma(E)\otimes_{C(X)} \Gamma(F) \cong \Gamma(E\otimes F) \,. \] 
Finally, $\Gamma$ sends dual bundles to dual modules:
\[ \Gamma(E^*) \cong \Hom_{C(X)}(\Gamma(E),C(X)) \,. \]
It is proved in \cite{GVF} that $\Gamma(E)$ is a \emph{projective module of rank $n$}, meaning that there exists a projection $p\in C(X)^n$, such that $\Gamma(E)\cong p C(X)^n$.
\begin{theorem}[Serre-Swan \cite{GVF}]
The functor $\Gamma$, from $\mathbf{Vect}_\C(X)$ to the category of finite rank projective modules over $C(X)$, is an equivalence of categories. Vector bundles of rank $n$ correspond to projective modules of rank $n$.
\end{theorem}
The Serre-Swan theorem provides the noncommutative generalization of finite rank vector bundles as finite rank projective modules of $C^*$-algebras.
\begin{corollary} Let $X^\infty$ be a compact pointed space and $E\rightarrow X^\infty$ a $U(n)$ vector bundle, corresponding to the projective module $pC(X^\infty)^n$, the projection $p$ defines a class $[p]\in K_0(C(X^\infty))$. The Serre-Swan theorem induces the isomorphism
\[ K^0(X^\infty) \cong K_0(C(X^\infty)) \,. \]
Using $C(\Sigma X^\infty)\cong \Sigma(C(X^\infty))$, we get
\[ K^\bullet(X^\infty) \cong K_\bullet(C(X^\infty)) \,. \]
\qed
\end{corollary}
If $X$ is locally compact, there is a similar result, but with the function algebra replaced by the algebra $C_0(X)$ of functions vanishing at infinity:
\[ K^\bullet(X) \cong K_\bullet(C_0(X)) \,. \]
$K$-theory can also be defined for more general algebras, for example \emph{local $C^*$-algebras} \cite{Bla}. Most importantly, the smooth function algebra $C^\infty_0(X)$ is a local $C^*$-algebra. The algebra $C_0(X)$ is the $C^*$-norm completition of $C^\infty_0(X)$. $K$-theory is invariant under such norm completitions \cite{Bla}:
\[ K_\bullet(C^\infty_0(X)) \cong K_\bullet(C_0(X)) \,. \]

Rosenberg \cite{Ros} defines twisted $K$-theory by modifying the algebra $C_0(X)$ as follows. Let $\sigma\in H^3(X,\Z)$ and $P_\sigma\rightarrow X$ a principal $PU(\Hilb)$ bundle\footnote{Recall that as a structure group of a principal $PU(\Hilb)$ bundle, $PU(\Hilb)$ is endowed with the quotient topology of $U(\Hilb)/U(1)$, where $U(\Hilb)$ has the strong operator topology.}, with Dixmier-Douady class $\sigma$. By lemma \ref{Lemma: PU(H) bundles = H^3(X,Z)}, $P_\sigma$ is unique up to isomorphism. Using the action of $PU(\Hilb)$ on $\mathcal{K}$, we can form the associated bundle
\[ \mathcal{E}_\sigma := P_\sigma\times_{PU(\Hilb)} \mathcal{K} \rightarrow X \,, \]
with fibres isomorphic to $\mathcal{K}$.
\begin{definition}[\cite{Ros}]
The \emph{$\sigma$-twisted $K$-theory} groups of $X$ are defined by
\[ K^\bullet(X,\sigma) := K_\bullet(C_0(X,\mathcal{E}_\sigma)) \,. \]
They are independent of the representative of the isomorphism class of $PU(\Hilb)$ bundles with Dixmier-Douady class $\sigma$.

Untwisted $K$-theory can be recovered by setting $\sigma=0$. Then $\mathcal{E}_0 = X\times \mathcal{K}$, $C_0(X,\mathcal{E}_0) = C_0(X,\mathcal{K})\cong C_0(X)\otimes \mathcal{K}$ and
\[ K^\bullet(X,0) = K_\bullet(C_0(X)\otimes \mathcal{K}) \stackrel{\eqref{Eq: K-theory is Morita invariant}}{\cong} K_\bullet(C_0(X)) = K^\bullet(X) \,. \]
\end{definition}

Twisted $K$-theory, with a torsion twisting class $\sigma\in \Tor H^3(X,\Z)$, can be thought of as $K$-theory of $\sigma$-twisted vector bundles. Recall proposition \ref{Proposition: BG modules and projective bundles}, which explained the relationship between finite rank bundle gerbe modules for a bundle gerbe with Dixmier-Douady class $\sigma$, and finite rank projective bundles with Dixmier-Douady class $\sigma$ over the base space $X$. Thus, it is not very surprising that $\sigma$-twisted $K$-theory can be defined in terms of bundle gerbe modules.
\begin{definition}[\cite{BCMMS}]
Let $(L,Y)$ be a bundle gerbe over $X$ with Dixmier-Douady class $\sigma\in \Tor H^3(X,\Z)$. We define $K_\text{bg}(X,\sigma)$ to be the Grothendieck group of $\Mod(L)$. It is, of course, independent of the choice of the bundle gerbe $(L,Y)$.
\end{definition}
We shall not prove the following.
\begin{proposition}[\cite{BCMMS}]\label{Proposition: BG twisted K-theory = Rosenberg twisted K-theory}
For any torsion twisting class $\sigma\in \Tor H^3(X,\Z)$, there is an isomorphism
\[ K_\text{bg}(X,\sigma) \cong K^0(X,\sigma) \,. \]
\end{proposition}
If $\sigma$ is nontorsion, we have to use $U_1$ bundle gerbe modules instead\footnote{The reason why $U(\Hilb)$ bundle gerbe modules are not interesting is that, while such a module would indeed yield a $PU(\Hilb)$ bundle over $X$, it is the only principal $PU(\Hilb)$ bundle with that precise Dixmier-Douady class. Thus, all $U(\Hilb)$ bundle gerbe modules are actually stably isomorphic.}.
\begin{definition}
Let $(L,Y)$ be a bundle gerbe over $X$ with nontorsion Dixmier-Douady class $\sigma\in H^3(X,\Z)$. We define $K_{U_1}(X,\sigma)$ to be the Grothendieck group of $\Mod_{U_1}(L)$. It is, of course, independent of the choice of the bundle gerbe $(L,Y)$.
\end{definition}
\begin{proposition}[\cite{BCMMS}]
For any nontorsion twisting class $\sigma\in H^3(X,\Z)$, there is an isomorphism
\[ K_{U_1}(X,\sigma) \cong K^0(X,\sigma) \,. \]
\end{proposition}

Twisted $K$-theory can also be defined homotopically by modifying the spectrum of $K$-theory. We have no need for this approach so we shall not discuss it. An extensive treatment is given in \cite{AS1}.

Let $\phi:\Sigma\rightarrow X$ be a twisted D-brane with $PU(n)$ Chan-Paton bundle $P$. Let $\{h_{\alpha\beta}\}$ be $U(n)$ lifts of the transition functions of $P$, satisfying
\[ h_{\alpha\beta}h_{\beta\gamma}h^{-1}_{\alpha\gamma} = \zeta_{\alpha\beta\gamma} \,. \]
Freed-Witten anomaly cancellation requires that
\[ \beta([\underline{\zeta}]) = [H]|_\Sigma + W_3(\Sigma) \,. \]
Thus, $P$ is a $[H]|_\Sigma + W_3(\Sigma)$-twisted principal $PU(n)$ bundle. The trivial rank $n$ bundle gerbe module for its lifting bundle gerbe defines a class
\[ [P]\in K_\text{bg}(X,[H]|_\Sigma + W_3(\Sigma)) \cong K^0(X,[H]|_\Sigma + W_3(\Sigma)) \,. \]
It is clear how coincident anti-D-branes can be taken into account.

If $[H]|_\Sigma$ is nontorsion, the Chan-Paton bundle $P$ is described by a $U_1$ bundle gerbe module, which looks locally like a $[H]|_\Sigma+W_3(\Sigma)$-twisted principal $U_1$ bundle $P$, or a $PU(\Hilb)$ bundle, whose transition functions lift to $U_1\subset U(\Hilb)$, satisfying a $[H]|_\Sigma + W_3(\Sigma)$-twisted cocycle condition. We fix a basepoint for the compact space $\Sigma$. Then, the $U_1$ bundle gerbe module determines a class
\[ [P]\in K_{U_1}(\Sigma,[H]|_\Sigma+W_3(\Sigma)) \cong K^0(\Sigma,[H]|_\Sigma+W_3(\Sigma)) \,. \]
Again, coincident anti-D-branes can be taken into account in the obvious way.

The goal is to obtain a generalization of the Minasian-Moore formula \eqref{Eq: Minasian-Moore formula} to the twisted case. What we need, is a twisted $K$-orientability condition to be able to define a Gysin map from $K^\bullet(\Sigma,[H]|_\Sigma+W_3(\Sigma))$ into twisted $K$-theory of spacetime. We shall later discuss $K$-orientability in a much more general context, when we have the powerful machinery of $KK$-theory at our disposal. Hence, we now only present the result of \cite{CW}.
\begin{theorem}[\cite{CW}]\label{Theorem: Gysin map in twisted K-theory}
Let $f:X\rightarrow Y$ be a continuous map\footnote{In \cite{CW} the authors restrict to differentiable maps. However, all differentiable maps between manifolds are homotopic to continuous maps \cite{BT} and the Gysin map depends only on the homotopy class of the map. Hence, the result of \cite{CW} generalizes to continuous maps.}, $\sigma\in H^3(Y,\Z)$,
\[ W_3(f) := \beta(w_2(X)-f^*w_2(Y)) \]
and $d:=\dim(Y)-\dim(X)$. Then, there exists a natural Gysin homomorphism
\[ f_!:K^\bullet(X,\sigma|_\Sigma + W_3(f)) \rightarrow K^{\bullet+d}(Y,\sigma) \,. \]
If both $X$ and $Y$ are both $K$-orientable ($\spinc$) and $\sigma=0$, then $f_!:K^\bullet(X)\rightarrow K^{\bullet+d}(Y)$ coincides with definition \ref{Definition: Gysin map in K-theory}.
\end{theorem}
Theorem \ref{Theorem: Gysin map in twisted K-theory} implies that there exists a Gysin homomorphism
\[ \phi_!:K^0(\Sigma,[H]|_\Sigma+W_3(\Sigma)) \rightarrow K^0(X,[H]) \,. \]
Consider now a configuration of twisted D-branes wrapping $\phi:\Sigma\rightarrow X$, with twisted $K$-theory class $x\in K^0(\Sigma,[H]|_\Sigma+W_3(\Sigma))$. The \emph{quantized RR-charge} of this D-brane configuration is the $[H]$-twisted $K$-theory class
\[ Q_\Z(\Sigma,x) := \phi_!(x)\in K^0(X,[H]) \,, \]
generalizing \eqref{Eq: Quantized RR-charge}.

Next we need a twisted generalization of the Chern character. The ordinary Chern character admits a purely algebraic description, which allows a generalization to noncommutative algebras, thus yielding the \emph{Chern-Connes character}, which takes values in something called \emph{periodic cyclic homology}, $\HP_\bullet(-)$ \cite{Con,GVF,Cun}. Periodic cyclic homology is the noncommutative generalization of (periodized) de Rham cohomology in the sense that for a compact manifold $X$, we have
\[ H^\bullet_\text{dR}(X) \cong \HP_\bullet(C^\infty(X))  \,, \]
where $\bullet=\text{even/odd}$. Similarly, \emph{periodic cyclic cohomology} generalizes de Rham homology of $X$. Periodic cyclic (co)homology degenerates for $C^*$-algebras. That is why we used the smooth function algebra instead. For a local $C^*$-algebra $\mathcal{A}$ the Chern-Connes character is a homomorphism
\begin{equation}\label{Eq: Chern-Connes character} \ch:K_\bullet(\mathcal{A}) \rightarrow \HP_\bullet(\mathcal{A})\,. \end{equation}

Let $P\rightarrow X$ be a principal $PU(\Hilb)$ bundle with Dixmier-Douady class $\sigma$. The algebra $C_0^\infty(X,P\times_{PU(\Hilb)} \mathcal{L}^1)$, where $\mathcal{L}^1\subset \mathcal{K}$ is the ideal of trace class operators of $\Hilb$, is a local $C^*$-algebra, with $C^*$-norm completition $C_0(X,\mathcal{E}_\sigma)$. The \emph{generalized Hochschild-Connes-Kostant-Rosenberg theorem} of Mathai and Stevenson \cite{MS3} states that $\HP_\bullet(C_0^\infty(X,P\times_{PU(\Hilb)} \mathcal{L}^1))$ is isomorphic to \emph{$\sigma$-twisted (de Rham) cohomology}. Luckily, twisted cohomology is very easy to define.
\begin{definition}[\cite{MS2,MS3,BCMMS,AS2}]
Let $\zeta$ be a closed de Rham $3$-form. The \emph{$\zeta$-twisted de Rham differential} is
\[ d_\zeta := d - \zeta \,, \]
which operates on differential forms $\Omega^\text{even/odd}(X)$ by $d_\zeta(\omega) = d\omega - \zeta\wedge \omega$\footnote{Clearly, $d_\zeta$ maps $\Omega^\text{even}(X)$ to $\Omega^\text{odd}(X)$ and \emph{vice versa}.}. It is easy to see that $d_\zeta^2= 0$:
\[ d_\zeta^2\omega = d^2\omega -d(\zeta\wedge \omega) -\zeta \wedge d\omega + \zeta\wedge \zeta \wedge \omega  = 0 -d\zeta \wedge \omega +\zeta \wedge d\omega -\zeta \wedge d\omega + 0 = 0 \,. \]
The (periodized) $\zeta$-twisted de Rham cohomology groups are
\[ H^\bullet(X,\zeta) := \ker(d_\zeta)/\im(d_\zeta) \,, \]
where $\bullet = \text{even/odd}$. Moreover, adding an exact $3$-form $d\eta$ to $\zeta$ leaves the twisted de Rham cohomology group invariant:
\[ H^\bullet(X,\zeta+d\eta) \cong H^\bullet(X,\zeta) \,. \]
On the level of cocycles, the isomorphism is given by
\[ \omega \mapsto e^\eta \wedge \omega \,. \]
It is easy to see to be well defined, for example
\begin{align*} (d-(\zeta+d\eta))(e^\eta\wedge \omega) &= d\eta\wedge e^\eta\wedge \omega + e^\eta \wedge d\omega -(\zeta+d\eta)\wedge(e^\eta\wedge \omega) \\ &= e^\eta \wedge d\omega - \zeta\wedge e^\eta\wedge \omega \\ &= e^\eta \wedge (d\omega - \zeta\wedge \omega) = 0 \,. \end{align*}
Thus, de Rham cohomology is most naturally twisted by the cohomology class $[\zeta]\in H^3_\text{dR}(X)$. More generally, if $\sigma\in H^3(X,\Z)$, \emph{$\sigma$-twisted (de Rham) cohomology} is defined by
\[ H^\bullet(X,\sigma) := H^\bullet(X,\sigma_\text{dR}) \,, \]
where $\sigma_\text{dR}$ is the image of $\sigma$ in $H^3_\text{dR}(X)$.

If $\mu\in \Omega^\bullet(X)$ is $d$-closed and $\omega$ represents a class in $H^\bullet(X,\sigma)$, then $d_{\sigma_\text{dR}}(\omega \wedge \mu)=0$. This action makes $H^\bullet(X,\sigma)$ into a $H^\bullet(X)$ module.
\end{definition}

For $\sigma\in H^3(X,\Z)$, the \emph{twisted Chern character} is a homomorphism
\begin{equation}\label{Eq: Twisted Chern character} \ch_\sigma:K^\bullet(X,\sigma) \rightarrow H_c^\bullet(X,\sigma) \,, \end{equation}
where $H_c^\bullet(X,\sigma)$ denotes \emph{compactly supported $\sigma$-twisted cohomology}, defined by restricting to compactly supported differential forms. It can be defined as a special case of \eqref{Eq: Chern-Connes character}, where $\mathcal{A} = C_0^\infty(X,P\times_{PU(\Hilb)} \mathcal{L}^1)$. Another definition was given in \cite{BCMMS} using bundle gerbe module connections. We take a slightly different approach, using local representations on $X$. Suppose that $\sigma\in \Tor H^3(X,\Z)$. Let $(L,Y)$ be a bundle gerbe with connection and curving, with Dixmier-Douady class $\sigma$. Then $L$ admits a finite rank bundle gerbe module with connection. If $L$ is represented by the Deligne cocycle $(\underline{g},-\underline{\Lambda},\underline{B})$, then the connection had the local representation \eqref{Eq: A local transformation}. By \eqref{Eq: mu, q, Lambda relation}, we can rewrite the transformation as
\[ A_\beta = h_{\alpha\beta}^{-1}A_\alpha h_{\alpha\beta} + h^{-1}_{\alpha\beta}\,dh_{\alpha\beta} + \mu_\beta - \mu_\alpha - d\log q_{\alpha\beta} \,. \]
Then, writing $\breve{A}_\alpha := A_\alpha - \mu_\alpha$ and $\breve{h}_{\alpha\beta}:= (hq^{-1})_{\alpha\beta}$, yields
\[ \breve{A}_\beta = \breve{h}_{\alpha\beta}^{-1}\breve{A}_\alpha \breve{h}_{\alpha\beta} + \breve{h}^{-1}_{\alpha\beta}\,d\breve{h}_{\alpha\beta} \,. \]
Now, $\breve{F}_\alpha := d\breve{A}_\alpha + \breve{A}_\alpha\wedge \breve{A}_\alpha$ behaves precisely like a curvature $2$-form does, even though $\{\breve{h}_{\alpha\beta}\}$ is not quite a cocycle. Hence, we can use it as the curvature in analytic formula \eqref{Eq: Analytic Chern character} for the Chern character.
\begin{definition}\label{Definition: Twisted Chern character, torsion case}
Suppose that $\sigma\in \Tor H^3(X,\Z)$. By proposition \ref{Proposition: BG twisted K-theory = Rosenberg twisted K-theory}, any class $x\in K^\bullet(X,\Sigma)$ can be equivalently represented as the formal difference of two finite rank bundle gerbe modules. Let $A_1$ and $A_2$ be any bundle gerbe module connections for these two modules. Using the above construction, we can form the ''curvatures'' $\breve{F}_1$ and $\breve{F}_2$ from $A_1$ and $A_2$. The $\sigma$-twisted Chern character of $x$ is
\[ ch_\sigma(x) := \left[\Tr\exp\left(-\frac{\breve{F}_1}{2\pi i}\right)\right] - \left[\Tr\exp\left(-\frac{\breve{F}_2}{2\pi i}\right)\right] \in H^\text{even}_\text{dR}(X) \,. \]
\end{definition}
Remark, that since $\sigma$ is torsion in the above, $\sigma_\text{dR}=0$ and so $H^\bullet(X,\sigma) = H_\text{dR}^\bullet(X)$. Thus, definition \ref{Definition: Twisted Chern character, torsion case} is consistent with \eqref{Eq: Twisted Chern character}. The nontorsion case requires more careful thought. The interested reader is referred to \cite{BCMMS,MS2,MS3}.

Now that we have the twisted Chern character at our disposal, generalizing the Minasian-Moore formula \eqref{Eq: Minasian-Moore formula} is easy. Let $x\in K^0(\Sigma,[H]|_\Sigma+W_3(\Sigma))$ describe a configuration of twisted D-branes, wrapping $\phi:\Sigma\rightarrow X$. The real, or classical, RR-charge of the configuration is given by the following \emph{twisted Minasian-Moore formula}.
\begin{definition}[Minasian-Moore]\label{Definition: Twisted Minasian-Moore formula}
\[ Q_\R(\Sigma,x) := \ch_{[H]}(\phi_!(x)) \wedge \sqrt{\Todd(X)} \in H^\text{even}(X,[H]) \,. \]
\end{definition}

The next logical step would be to define a corresponding \emph{twisted $K$-homology} theory and describe D-branes using it. Twisted $K$-homology can be defined in different ways. There is $C^*$-algebraic definition of spectral $K$-homology, often referred to as \emph{analytic $K$-homology} \cite{HR,Kas3,BHS}. More precisely, for a compact manifold $X$, analytic $K$-homology of $C(X)$ is isomorphic to $K_\bullet(X)\cong K_\bullet^\text{geom}(X)$. Then, $\sigma$-twisted $K$-homology is simply defined as analytic $K$-homology of the $\sigma$-twisted function algebra $C(X,\mathcal{E}_\sigma)$. Just like untwisted $K$-homology, twisted $K$-homology admits a geometric description \cite{Wan,Sza1}. We shall not discuss twisted $K$-homology here, but instead jump straight to $KK$-theory, which yields both $K$-theory and $K$-homology as special cases, thus providing a unified treatment of both cases.

\section{$KK$-theory}
In $KK$-theory, we need to impose some restrictions on the $C^*$-algebras we use, namely \emph{$\sigma$-unitality} and separability. A $\sigma$-unital $C^*$ algebra is one with a countable approximate identity. In particular, every separable $C^*$-algebra is $\sigma$-unital. Thus, for simplicity, we restrict to the category of separable $C^*$-algebras. Furthermore, we often need to take tensor products of $C^*$-algebras and it is well-known that, in general, the tensor product algebra admits several inequivalent $C^*$-norms. For \emph{nuclear $C^*$-algebras} the tensor product works particularly well: if either one of the algebras in the tensor product is nuclear, the projective and injective cross norms on the product algebra, and hence all cross norms, are equivalent \cite{Lan1}. In other words, the tensor product of any $C^*$-algebra with a nuclear $C^*$-algebra is a $C^*$-algebra in a unique way. To make life easier, we assume $C^*$-algebras to be nuclear whenever tensor products are discussed. This is not so strict a requirement, as nuclear $\C^*$-algebras are quite abundant. For example, all commutative, finite dimensional and \mbox{type I} $C^*$-algebras are nuclear.

Our main references for this section are the books \cite{Bla,KT}. We also take this opportunity to mention the papers \cite{Kas1,Kas2} by Kasparov and the excellent papers \cite{Hig1,Hig2} by Higson.

\begin{definition}[\cite{Bla,KT}]
Let $B=(B,\norm{\cdot})$ be a $C^*$-algebra. A \emph{pre-Hilbert $B$-module} is a complex vector space $E$, which is also a right $B$-module with a $B$-valued inner product $\ip{\cdot,\cdot}:E\times E\rightarrow B$. The inner product is $B$-antilinear in the first and $B$-linear in the second argument:
\renewcommand{\labelenumi}{\roman{enumi})}
\begin{enumerate}
\item $\ip{x,yb}=\ip{x,y}b$,
\item $\ip{x,y}*=\ip{y,x}$,
\item $\ip{x,x} \ge 0$, that is, $\ip{x,x}=a^*a$ for some $a\in B$,
\item $\ip{x,x}=0$ if and only if $x=0$,
\end{enumerate}
where $x,y\in E$ and $a,b\in B$ are arbitrary.

A \emph{Hilbert $B$-module} is a pre-Hilbert $B$-module which is complete with respect to the norm $\norm{x}=\norm{\ip{x,x}}^{1/2}$.
\end{definition}

Let $E_1$ and $E_2$ be Hilbert $B$-modules and assume $T:E_1\rightarrow E_2$ to be \emph{adjointable}, meaning that there exists a map $T^*:E_2\rightarrow E_1$, satisfying $\ip{x,Ty}=\ip{T^*x,y}$. Then $T$ is automatically a $B$-module map:
\[ \ip{y,T(x_1+x_2 b)}=\ip{T^*y,x_1+x_2 b}=\ip{T^* y,x_1}+\ip{T^*y x_2}b=\ip{y,T x_1 + (T x_2)b} \]
for all $x_1,x_2\in E_1,y\in E_2$ and $b\in B$. It is also easy to show that $T$ and $T^*$ are bounded and, in fact, $\norm{T}=\norm{T^*}$. The set of adjointable maps $E_1\rightarrow E_2$ is denoted by $\L_B(E_1,E_2)$. If $E_1=E_2=E$, we write simply $\L_B(E)=\L_B(E,E)$. It can easily be shown that $\L_B(E)$ is a $C^*$-algebra with respect to the operator norm \cite{KT}. For any $x,y\in E$, the operators
\[ \Theta_{x,y}(z)=x\ip{y,z} \]
in $E$ are elements of $\L_B(E)$. The norm completion of their linear span yields an ideal $\mathcal{K}_B(E)$ of $\L_B(E)$, which generalizes the ideal of compact operators on a separable Hilbert space.

\begin{definition}[\cite{Bla,KT}]
\label{definition: Kasparov module}
Let $A$ and $B$ be $C^*$-algebras. An \emph{odd Kasparov $(A,B)$-module}\footnote{$KK$-theory is often defined using Kasparov modules with both $C^*$-algebras and Hilbert modules $\Z_2$-graded. We have used this ''simplified'' approach for two reasons. First, it is perhaps slightly easier to comprehend and second, it fits more naturally to our purposes.} is a triple $\E=(E,\phi,F)$, where
\renewcommand{\labelenumi}{\roman{enumi})}
\begin{enumerate}
\item $E$ is a countably generated Hilbert $B$-module,
\item $\phi: A\rightarrow \L_B(E)$ is a $*$-representation and
\item $F\in \L_B(E)$ is a self-adjoint operator, satisfying $(F^2-1)\phi(a) \in \mathcal{K}_B(E)$ and $[F,\phi(a)] \in \mathcal{K}_B(E)$.
\end{enumerate}
Likewise, an \emph{even Kasparov $(A,B)$-module} is a triple $\E=(E,\phi,F)$, where
\renewcommand{\labelenumi}{\roman{enumi})}
\begin{enumerate}
\item $E=E^+\oplus E^-$ is a countably generated $\Z_2$-graded Hilbert $B$-module,
\item $\phi: A\rightarrow \L_B(E)$ is a $*$-representation, such that $\phi(a)$ is an even degree map for all $a\in A$, and
\item $F\in \L_B(E)$ is a self-adjoint operator of odd degree, satisfying $(F^2-1)\phi(a) \in \mathcal{K}_B(E)$ and $[F,\phi(a)] \in \mathcal{K}_B(E)$.
\end{enumerate}

A Kasparov module, even or odd, satisfying $(F^2-1)\phi(a)=[F,\phi(a)]=0$ for all $a\in A$ is said to be \emph{degenerate}. Two Kasparov $(A,B)$-modules are \emph{isomorphic}, when there exists an isomorphism of the Hilbert modules, compatible, in the obvious way, with the representations $\phi$ and the operators $F$. The sets of isomorphism classes of even and odd Kasparov $(A,B)$-modules are denoted by $\mathbb{E}_0(A,B)$ and $\mathbb{E}_1(A,B)$, respectively. Similarly, the sets of even and odd degenerate Kasparov $(A,B)$-modules are denoted by $\mathbb{D}_0(A,B)$ and $\mathbb{D}_1(A,B)$.
\end{definition}

Both $\mathbb{E}_0(A,B)$ and $\mathbb{E}_1(A,B)$ are Abelian monoids. The addition is given by the direct sum of Kasparov modules:
\[ (E_1,\phi_1,F_1)\oplus (E_2,\phi_2,F_2) = (E_1\oplus E_2,\phi_1\oplus \phi_2, F_1\oplus F_2)\,, \]
where, in the even case, the grading of $E_1\oplus E_2$ is given by
\[ (E_1\oplus E_2)^+\oplus (E_1\oplus E_2)^- = (E_1^+\oplus E_2^+)\oplus (E_1^-\oplus E_2^-) \,. \]

\begin{definition}[\cite{Bla,KT}]
Let $\E_1=(E_1,\phi_1,F_1)$ and $\E_2=(E_2,\phi_2,F_2)$ be even (resp. odd) Kasparov $(A,B)$-modules. If $U:E_2\rightarrow E_1$ is a graded (resp. ungraded) unitary isomorphism, such that $F_2=U^*F_1U$ and $\phi_2=U^*\phi_1 U$, then we call $\E_1$ and $\E_2$ \emph{unitarily equivalent} and write $\E_1 \approx_u \E_2$. If there exists a family $\F_t=(E,\phi,F_t), \,t\in[0,1]$ of even (resp. odd) Kasparov $(A,B)$-modules, for which $t\mapsto F_t$ is norm-continuous, $\E_1 \approx_u \F_0 = (E,\phi,F_0)$ and $\E_2\approx_u \F_1=(E,\phi,F_1)$, we call $\E_1$ and $\E_2$ \emph{operator homotopic}. Furthermore, if $\E_1$ and $\E_2$ are operator homotopic modulo direct sum by even (resp. odd) degenerate modules, that is, there exist degenerate even (resp. odd) Kasparov $(A,B)$-modules $\F_1$ and $\F_2$ such that $\E_1\oplus \F_1$ is operator homotopic to $\E_2\oplus \F_2$, then we write $\E_1 \sim_\text{oh} \E_2$.
\end{definition}

\begin{definition}[\cite{Bla,KT}]
For $C^*$-algebras $A$ and $B$, the $KK$-theory groups $KK_0(A,B)$ and $KK_1(A,B)$ are defined as
\[ KK_0(A,B) = \mathbb{E}_0(A,B)/\sim_\text{oh}\quad \text{and} \quad KK_1(A,B) = \mathbb{E}_1(A,B)/\sim_\text{oh}\,. \]
\end{definition}
We omit the proof of the following proposition. It is a straightforward application of the definitions, though somewhat tedious, if written out in full detail.
\begin{proposition}
Both $KK_0(A,B)$ and $KK_1(A,B)$ are Abelian groups, with addition induced by the direct sum of Kasparov modules. All degenerate even (resp. odd) Kasparov $(A,B)$-modules belong to same equivalence class, which is the zero element of $KK_0(A,B)$ (resp. $KK_1(A,B)$). If $(E,\phi,F)$ is an even Kasparov $(A,B)$-module, its inverse in $KK_0(A,B)$ is given by $(E^\circ,\phi,-F)$, where $E^\circ$ denotes $E$ with opposite grading. If $(E,\phi,F)$ is odd, the inverse in $KK_1(A,B)$ is given by $(E,\phi,-F)$.
\end{proposition}

\begin{definition}
We write
\[ KK_\bullet(A,B) := KK_0(A,B)\oplus KK_1(A,B) \,. \]
The lower index is always assumed to be in $\Z_2$\footnote{$KK$-theory can actually be defined more generally than what we have presented \cite{Kas1}, in which case this follows from a $KK$-theoretic version of Bott periodicity.}.
\end{definition}

A key aspect of $KK$-theory is the existence of a powerful product structure on $KK_\bullet$. It is notoriously difficult to define, however, and even the proof for existence is far beyond the scope of our treatment. In fact, the original proof by Kasparov was nonconstructive and extremely complicated \cite{Kas1}. For a simpler constructive proof, see \cite{KT} or \cite{Bla}. In any case, the proof is too long to be presented here.
\begin{theorem}[\cite{Bla,KT}]
For separable $C^*$-algebras $A,B,C$, there exists a bilinear associative \emph{composition product}
\[ \otimes_C:KK_i(A,C)\times KK_j(C,B)\rightarrow KK_{i+j}(A,B) \,. \]
\end{theorem}

The composition product is not the most general product structure in $KK$-theory. Let $\alpha\in KK_i(A,B)$ be represented by the Kasparov $(A,B)$-module $(E,\phi,F)$. \emph{The dilation} of $x$ by the $C^*$-algebra $C$ is a class in $KK_i(A\otimes C,B\otimes C)$, represented by the Kasparov $(A\otimes C,B\otimes C)$-module $(E\otimes C,\phi\otimes \id_C,F\otimes \id_C)$. Dilation can be used, together with the composition product, to define \emph{the Kasparov product}
\[ \otimes_D:KK_i(A_1,B_1\otimes D)\times KK_j(D\otimes A_2,B_2)\rightarrow KK_{i+j}(A_1\otimes A_2,B_1\otimes B_2)\,. \]
For $x_1\in KK_i(A_1,B_1\otimes D)$ and $x_2\in KK_j(D\otimes A_2,B_2)$, the Kasparov product is given by
\[ x_1\otimes_D x_2 = (x_1\otimes 1_{A_2})\otimes_{B_2\otimes D \otimes A_2} (x_2\otimes 1_{B_1}) \,. \]
For $D=\C$, it reduces to \emph{the exterior product}
\[ \otimes=\otimes_{\C}:KK_i(A_1,B_1)\times KK_j(A_2,B_2)\rightarrow KK_{i+j}(A_1\otimes A_2,B_1\otimes B_2) \,. \]
\begin{theorem}[\cite{Kas2}]
\label{Theorem: Properties of Kasparov product}
The Kasparov product
\[ \otimes_D:KK_i(A_1,B_1\otimes D)\times KK_j(D\otimes A_2,B_2)\rightarrow KK_{i+j}(A_1\otimes A_2,B_1\otimes B_2) \]
satisfies the following properties:
\begin{enumerate}
\item It is bilinear, contravariantly functorial in $A_1$ and $A_2$, and covariantly functorial in $B_1$ and $B_2$.
\item If $f:D_1\rightarrow D_2$, $x_1\in KK_i(A_1,B_1\otimes D_1)$ and $x_1\in KK_j(D_2\otimes A_2,B_2)$, then
\[ (\id_{B_1}\otimes f)_*(x_1)\otimes_{D_2} x_2 = x_1\otimes_{D_1} (f\otimes \id_{A_2})^*(x_2) \,. \]
\item It is associative in the sense that if $x_1\in KK_i(A_1,B_1\otimes D_1),x_2\in KK_j(D_1\otimes A_2,B_2\otimes D_2)$ and $x_3\in KK_k(D_2\otimes A_3,B_2)$, then
\[ (x_1\otimes_{D_1} x_2)\otimes_{D_2} x_3 = x_1 \otimes_{D_1} (x_2\otimes_{D_2} x_3) \,. \]
\item If $x_1\in KK_i(A_1,B_1\otimes D_1\otimes D)$ and $x_2\in KK_j(D\otimes D_2\otimes A_2,B_2)$, then
\[ x_1\otimes_D x_2 = (x_1\otimes 1_{D_2})\otimes_{D_1\otimes D \otimes D_2} (1_{D_1}\otimes x_2) \,. \]
\item If $x_1\in KK_i(A_1,B_1\otimes D)$ and $x_2\in KK_j(D\otimes A_2,B_2)$, then
\[ (x_1\otimes_D x_2)\otimes 1_{D_1} = (x_1\otimes 1_{D_1})\otimes_{D\otimes D_1} (x_2\otimes 1_{D_1}) \,. \]
\item The exterior product is commutative, that is, for $x_1\in KK_i(A_1,B_1)$ and $x_2\in KK_j(A_2,B_2)$ one has
\[ x_1\otimes x_2 = x_2\otimes x_1 \,. \]
\item The class $1_\C \in KK_0(\C,\C)$ is a unit for the product:
\[ 1_{\C} \otimes x = x\otimes 1_{\C} = x\,, \]
for any $x\in KK_i(A,B)$.
\end{enumerate}
\end{theorem}

Remarkably, $KK$-theory can also be described purely axiomatically \cite{Hig1,Bla} as follows. There exists an additive category\footnote{By an additive category, we mean a category $\mathbf{C}$, where every $\mathbf{Hom}$-set is equipped with an Abelian group structure and the composition is bilinear in the sense that
\[ (g_1+g_2)\circ (f_1+f_2) = g_1\circ f_1 + g_1\circ f_2 + g_2 \circ f_1 + g_2\circ f_2\,, \]
for any $f_1,f_2\in \mathbf{Hom(a,b)}$ and $g_1,g_2\in \mathbf{Hom(b,c)}$.} $\mathbf{KK}$, with objects separable $C^*$-algebras. The groups of morphisms $\mathbf{Hom}(-,-)$ are denoted by $KK(-,-)$. Composition of morphisms in $\mathbf{KK}$ defines a bilinear map
\[ KK(A,D)\times KK(D,B) \rightarrow KK(A,B) \,, \]
called the \emph{composition product}. There exists a functor $KK:\mathbf{SC^*}\rightarrow \mathbf{KK}$, where $\mathbf{SC^*}$ denotes the category of separable $C^*$-algebras and $*$-homomorphisms, sending each object to itself and satisfying the following properties:
\renewcommand{\labelenumi}{\roman{enumi})}
\begin{enumerate}
\item \underline{Stability}: $KK$ sends the stabilization maps $A\rightarrow A\otimes \mathcal{K}$ and $B\rightarrow B\otimes \mathcal{K}$, given by tensoring with a rank-one projection, to isomorphisms $KK(A,A\otimes \mathcal{K})$ and \mbox{$KK(B,B\otimes \mathcal{K})$} in such a way that the composition product yields group isomorphisms
\[ KK(A,B)\cong KK(A,B\otimes \mathcal{K}) \cong KK(A\otimes \mathcal{K},B)\,. \] 
\item \underline{Homotopy invariance}: If $f_0,f_1:A\rightarrow B$ are homotopic homomorphisms, meaning that there exists a mediating homomorphism $g:A\rightarrow C([0,1],B)$ such that $\ev_0\circ g = f_0$ and $\ev_1\circ g = f_1$, then $KK(f_1) = KK(f_2)$. It follows immediately that if $A$ and $C$ are homotopy equivalent, that is, there exist homomorphisms $f:A\rightarrow C$ and $g:C\rightarrow A$, such that $f\circ g$ and $g\circ f$ are homotopic to the identity homomorphisms, then $KK(C,B)\cong KK(A,B)$, where the isomorphism is given by composition with $KK(f)$ and $KK(g)$. Similarly, if $B$ and $D$ are homotopy equivalent, then $KK(A,D)\cong KK(A,B)$.
\item \underline{Split exactness}: $KK$ preserves split exact sequences. It follows that, for the split exact sequence
\[ \xymatrix{ 0 \ar[r] & J \ar[r] & D \ar@<0.5ex>[r] & D/J \ar[r]\ar@<0.5ex>[l] & 0 } \]
of separable $C^*$-algebras and $*$-homomorphisms and arbitrary separable $C^*$-algebras $A$ and $B$, the sequences
\[ \xymatrix{ 0 \ar[r] & KK(A,J) \ar[r] & KK(A,D) \ar@<0.5ex>[r] & KK(A,D/J) \ar[r]\ar@<0.5ex>[l] & 0 \\ 0 \ar[r] & KK(D/J,B) \ar@<0.5ex>[r] & KK(D,B) \ar@<0.5ex>[l]\ar[r] & KK(J,B) \ar[r] & 0 } \]
are split exact.
\end{enumerate}
Suppose now that $\mathbf{A}$ is another additive category and $F:\mathbf{SC^*}\rightarrow \mathbf{A}$ a functor satisfying the above three properties, that is, stabilization maps induce isomorphisms, homotopic $*$-homomorphisms induce the same morphism and split exact sequences induce split exact sequences. The category $\mathbf{KK}$ is then uniquely determined by the universal property that such an $F$ always factors uniquely through $\mathbf{KK}$. More precisely, there exists a unique functor $\widehat{F}:\mathbf{KK}\rightarrow \mathbf{A}$, such that the diagram
\[ \xymatrix@=4pc{{\mathbf{SC^*}} \ar[r]^{KK} \ar[dr]_F & {\mathbf{KK}} \ar@{-->}[d]^{\widehat{F}} \\ & {\mathbf{A}} } \]
commutes.
\begin{theorem} For any separable $C^*$-algebras $A$ and $B$, the group of morphisms $KK(A,B)$ in $\mathbf{KK}$ is given by $KK_0(A,B)$. $KK_1(A,B)$ can then be defined axiomatically by
\[ KK_1(A,B) \cong KK_0(\Sigma A,B) \cong KK_0(A,\Sigma B)\,, \]
where $\Sigma$ is the suspension functor \eqref{Eq: Suspension of a C*-algebra}.
\end{theorem}
\begin{proof} The proof is essentially contained in \cite{Hig1}. See also \cite{Bla} for axiomatic $KK$-theory. \end{proof}

The functor $KK$ sends a $*$-homomorphisms $f:A\rightarrow B$ to the equivalence class $[f]_{KK}=KK(f)\in KK_0(A,B)$, given by the Kasparov $(A,B)$-module $(B,f,0)$, where $B$ itself is considered as a trivially graded Hilbert $B$-module, with inner product $\ip{b_1,b_2}=b_1^*b_2$. Functoriality of $KK$ implies that for $*$-homomorphisms $f:A\rightarrow B$ and $g:B\rightarrow C$, the composition $g\circ f$ satisfies
\[ [g\circ f]_{KK} = [f]_{KK}\otimes_B [g]_{KK} \in KK_0(A,C) \,. \]
In this sense, classes of $KK_0(A,B)$ can be regarded as generalized $*$-homomorphisms between $A$ and $B$. Remark also, that for any separable $C^*$-algebra $A$, the group $KK_0(A,A)$ admits a unital ring structure, with multiplication given by the composition product and multiplicative identity induced by the identity morphism $\id:A\rightarrow A$.

\section{Poincar\'e duality in noncommutative geometry}
From now on, until the end of the text, we are more or less tacitly referring to \cite{BMRS1,BMRS2,Sza2,CS1}. We basically follow \cite{BMRS1,BMRS2}.

In this section our primary goal is to explain the \emph{noncommutative Poincar\'e duality theorem} \cite{Con}, which describes the precise relation between $K$-theory and $K$-homology. This is where $KK$-theory is needed.
\begin{definition} Let $A$ be a separable $C^*$-algebra. $K$-theory and $K$-homology are defined as special cases of $KK$-theory:
\begin{align*} K_\bullet(A) &:= K_\bullet(\C,A)\qquad \text{($K$-theory)} \\ K^\bullet(A) &:= K_\bullet(A,\C)\qquad \text{($K$-homology)} \end{align*}
\end{definition}
It is far from clear that the above definition of $K$-theory is equivalent to that of definition \ref{Definition: Operator K-theory}. The proof can be found in \cite{Bla}.

The composition product specializes to various homomorphisms and isomorphisms in, and between, $K$-theory and $K$-homology, provided that suitable mediating classes exist. For example, any $\alpha \in KK_d(A,B)$ provides homomorphisms
\begin{align*} &(-)\otimes_A \alpha: K_\bullet(A)=KK_\bullet(\C,A)\rightarrow KK_{\bullet + d}(\C,B) = K_{\bullet+d}(B)\,, \\ &\alpha \otimes_B (-): K^\bullet(B)=KK_\bullet(B,\C)\rightarrow KK_{\bullet + d}(A,\C) = K_{\bullet+d}(A) \,. \end{align*}
\begin{definition}
If $\alpha\in KK_d(A,B)$ is invertible with respect to the composition product, that is, if there exists a class $\beta\in KK_{-d}(B,A)$, such that
\[ \alpha\otimes_B \beta = 1_A \in KK_0(A,A) \]
and
\[ \beta\otimes_A \alpha = 1_B \in KK_0(B,B)\,, \]
then we write $\beta = \alpha^{-1}$ and call $A$ and $B$ \emph{(strongly\footnote{There is also a weaker pointwise notion of $KK$-equivalence given in \cite{BMRS2}.}) $KK$-equivalent}.
\end{definition}
Let $A$ and $B$ be $KK$-equivalent separable $C^*$-algebras and $\alpha\in KK_d(A,B)$ invertible. Then
\begin{align*} &KK_\bullet(\C,A)\xrightarrow{(-)\otimes_A \alpha} KK_{\bullet + d}(\C,B) \xrightarrow{(-)\otimes_B \alpha^{-1}} KK_\bullet(\C,A) \,, \\ &KK_\bullet(B,\C)\xrightarrow{\alpha\otimes_B(-)} KK_{\bullet + d}(A,\C) \xrightarrow{\alpha^{-1}\otimes_A(-)} K_\bullet(B,\C) \,, \end{align*}
but on the other, by associativity, we have
\begin{align*} &\big((-)\otimes_A\alpha\big)\otimes_B \alpha^{-1} = (-)\otimes_A\big(\alpha\otimes_B \alpha^{-1}\big) = (-)\otimes_A 1_A\,, \\ &\alpha^{-1}\otimes_A\big(\alpha\otimes_B (-)\big)= \big(\alpha^{-1}\otimes_A \alpha\big)\otimes_B (-) = 1_B\otimes_B (-)\,. \end{align*}
Thus, $K_\bullet(A)\cong K_{\bullet+d}(B)$ and $K^\bullet(B)\cong K^{\bullet+d}(A)$.

There is a \emph{universal coefficient theorem} for $KK$-theory, which we are going to need later. It tells us to what extent $KK_\bullet(A,B)$ can be computed from $K_\bullet(A)$ and $K_\bullet(B)$.
\begin{theorem}[Universal coefficient theorem \cite{Bla}]
The sequence
\[ \xymatrix{ 0 \ar[r] & \Ext^1_{\Z}(K_{\bullet+1}(A),K_\bullet(B)) \ar[r] & KK_\bullet(A,B) \ar[r] & \Hom_{\Z}(K_\bullet(A),K_\bullet(B)) \ar[r] & 0 } \]
is exact for all separable $C^*$-algebras $B$ if and only if $A$ is $KK$-equivalent to a commutative separable $C^*$-algebra. We then say that \emph{$A$ satisfies the universal coefficient theorem (UCT) for $KK$-theory}.
\end{theorem}
The class of algebras satisfying the UCT is quite abundant. In fact, it is stated in \cite{Mey} that, at the time it was written, no example of a separable nuclear $C^*$-algebra \emph{not} satisfying the UCT is known.

A \emph{Hilbert $(A,B)$-bimodule} is a Hilbert $B$-module with a left $A$-module structure and an $A$-valued inner product, which is compatible with the $B$-valued inner product in the obvious sense that $\ip{x,y}_A z = x\ip{y,z}_B$, for all $x,y,z\in E$. Alternatively, one can think of a Hilbert $(A,B)$-bimodule as a Hilbert $A^\circ \otimes B$-module, where $A^\circ$ denotes the opposite algebra for $A$\footnote{The opposite algebra $A^\circ$ of $A$ has the same elements as $A$, but with the product reversed. For $x\in A$, we denote the corresponding element in $A^\circ$ by $x^\circ$. The product of $x^\circ,y^\circ\in A^\circ$ is given by $x^\circ y^\circ = (yx)^\circ$.}.

\begin{definition} Two $C^*$-algebras, $A$ and $B$, are \emph{(strongly) Morita equivalent}, if there exists a full Hilbert $A$-module $E$, such that $\mathcal{K}_A(E) \cong B$. A Hilbert $A$-module $E$ is full, if the norm closure of $\ip{E,E}$ is the whole of $B$. It can be shown that $A$ and $B$ are Morita equivalent if and only if there exists a Hilbert $(A,B)$-bimodule $E$ and a Hilbert $(B,A)$-bimodule $E^\vee$, such that
$E\otimes_B E^\vee \cong A$ and $E^\vee \otimes_A E \cong B$. $E$ and $E^\vee$ are called \emph{equivalence bimodules}. For a proof, see \cite{GVF}. We denote Morita equivalence of $A$ and $B$ by $A\stackrel{\text{Mor}}{\sim} B$. \end{definition}

If $A$ is unital, we associate to the map $\Theta_{a,b}$ the element $\Theta_{a,b}(1) = ab^*\in A$. Conversely, to an element $a\in A$ we associate the map $\Theta_{a,1}$. Taking the norm closure yields the $*$-isomorphism $\mathcal{K}_A(A)\cong A$, which proves that $A$ is Morita equivalent to itself. In the nonunital case one has to work with approximate units \cite{GVF}. Moreover, as the name indicates, Morita equivalence determines an equivalence relation. Another important example of Morita equivalence is the following.
\begin{proposition}[\cite{GVF}] A $C^*$-algebra $A$ is Morita equivalent to its stabilization $A\otimes \mathcal{K}$. \end{proposition}

\begin{proposition} If $A\stackrel{\text{Mor}}{\sim} B$, then $A$ is $KK$-equivalent to $B$. \end{proposition}
\begin{proof} Denote the equivalence bimodules by $E$ (an $(A,B)$-bimodule) and $E^\vee$ (a $(B,A)$-bimodule). Let $\alpha \in KK_0(A,B)$ and $\beta\in KK_0(B,A)$ be classes represented by the Kasparov modules $(E,\id_A,0)$ and $(E^\vee,\id_B,0)$, respectively. Both $E$ and $E^\vee$ are taken trivially graded. The composition product $\alpha\otimes_B \beta$ is represented by the Kasparov $(A,A)$-module
\[ (E\otimes_B E^\vee,\id_A\otimes \id_{E^\vee},0)\stackrel{\text{Mor}}{=}(A,\id_A,0) \,. \]
Likewise, $\beta\otimes_A \alpha$ is represented by
\[ (E^\vee\otimes_A E,\id_B \otimes \id_E,0) \stackrel{\text{Mor}}{=} (B,\id_B,0) \,. \]
Thus, we have $\alpha\otimes_B \beta = 1_A\in KK_0(A,A)$ and $\beta\otimes_A \alpha = 1_B\in KK_0(B,B)$, proving the $KK$-equivalence of $A$ and $B$.
\end{proof}

Suppose that there exist classes
\[ \Delta \in KK_d(A\otimes A^\circ,\C)= K^d(A\otimes A^\circ) \]
and
\[ \Delta^\vee \in KK_{-d}(\C,A\otimes A^\circ)=K_{-d}(A\otimes A^\circ)\,, \]
satisfying
\[ \Delta^\vee \otimes_{A^\circ} \Delta = 1_A \in KK_0(A,A) \quad \text{and} \quad \Delta^\vee \otimes_A \Delta = (-1)^d\,1_{A^\circ} \in KK_0(A^\circ,A^\circ) \,. \]
The $C^*$-algebra $A$ is called a \emph{(strong\footnote{Again, one can give also a weaker pointwise definition.}) Poincar\'e duality algebra (PD algebra)}. The $K$-homology class $\Delta$ is called the \emph{fundamental class} for $A$ and $\Delta^\vee$ its \emph{inverse}. The fundamental class can be used to construct \emph{Poincar\'e duality isomorphisms}
\begin{align*} K_\bullet(A) =KK_\bullet(\C,A) &\xrightarrow{(-)\otimes_A \Delta} KK_{\bullet + d}(A^\circ,\C) = K^{\bullet + d}(A^\circ) = K^{\bullet + d}(A)\,, \\ K^\bullet(A) = K^\bullet(A^\circ)=KK_\bullet(A^\circ,\C) &\xrightarrow{\Delta^\vee \otimes_{A^\circ}(-)} KK_{\bullet - d}(\C,A) = K_{\bullet - d}(A)\,. \end{align*}

Poincar\'e duality algebras are, unfortunately, quite scarce. Instead, we introduce a slightly weaker notion of Poincar\'e duality, where the duality isomorphism is between $K$-theory and $K$-homology of generally different $C^*$-algebras.
\begin{definition}
A pair $(A,B)$ of separable $C^*$-algebras is called a \emph{(strong) Poincar\'e duality pair (PD pair)}, if there exist classes $\Delta \in KK_d(A\otimes B,\C)=K^d(A\otimes B)$ and $\Delta^\vee \in KK_{-d}(\C,A\otimes B)=K_{-d}(A\otimes B)$, satisfying
\[ \Delta^\vee \otimes_B \Delta = 1_A \in KK_0(A,A) \quad \text{and} \quad \Delta^\vee \otimes_A \Delta = (-1)^d\,1_B \in KK_0(B,B) \,. \]
The class $\Delta$ is called a \emph{fundamental class} for the PD pair $(A,B)$ and $\Delta^\vee$ its \emph{inverse}.
\end{definition}
Given a PD pair $(A,B)$ with fundamental class $\Delta\in K^d(A\otimes B)$, there are \emph{Poincar\'e duality isomorphisms}
\begin{align*} K_\bullet(A) =KK_\bullet(\C,A) &\xrightarrow{(-)\otimes_A \Delta} KK_{\bullet + d}(B,\C) = K^{\bullet + d}(B)\,, \\ K^\bullet(B)=KK_\bullet(B,\C) &\xrightarrow{\Delta^\vee \otimes_B(-)} KK_{\bullet - d}(\C,A) = K_{\bullet - d}(A)\,. \end{align*}
The complete form of the Kasparov product can be employed to define even more general Poincar\'e duality isomorphisms. Let $C$ and $D$ be any separable $C^*$-algebras. Then there are isomorphisms
\begin{align*} KK_\bullet(C,A\otimes D) &\xrightarrow{(-)\otimes_A \Delta} KK_{\bullet+d}(C\otimes B,D)\,, \\ KK_\bullet(B\otimes C,D) &\xrightarrow{\Delta^\vee \otimes_B (-)} KK_{\bullet-d}(C,A\otimes D) \,. \end{align*}
It is easy to see that, given a PD pair $(A,B)$, there exists an isomorphism $KK_\bullet(A,A)\cong KK_\bullet(B,B)$. In particular, the multiplicative groups of their invertible classes are isomorphic.

Let $X$ be a compact oriented manifold with continuous function algebra $C(X)$. The pair $(C(X),C(X,\Cliff(T^*X)))$, where $\Cliff(T^*X)$ is the \emph{Clifford algebra bundle} of the cotangent bundle \cite{GVF}, is a PD pair. The fundamental class is the $K$-homology class of the Dirac operator on $\Cliff(T^*X)$, represented by the Kasparov $(C(X)\otimes C(X,\Cliff(T^*X),\C)$-module
\[ \Big(L^2(X,\Cliff(T^*X)),\rho,\frac{\Dirac}{\sqrt{1+\Dirac^2}}\Big)\,, \]
where $L^2(X,\Cliff(T^*X))$ is the Hilbert space of square integrable sections of $\Cliff(T^*X)$ and the action $\rho$ is given by Clifford multiplication \cite{HR,Con}. The Dixmier-Douady class of $\Cliff(T^*X)$ is $W_3(X)$ \cite{GVF}. Thus, Poincar\'e duality yields isomorphisms $K^\bullet(X,W_3(X)) \cong K_\bullet(X)$ and $K_\bullet(X,W_3(X))\cong K^\bullet(X)$. If $X$ is $\spinc$, then $W_3(X)=0$ and we have the usual Poincar\'e duality isomorphisms between (untwisted) $K$-theory and $K$-homology and $C(X)$ becomes a PD algebra. Finally, if $\sigma\in H^3(X,\Z)$, then
\begin{align*} KK(C(X,\mathcal{E}_\sigma)\otimes C(X,\mathcal{E}_{-\sigma}\otimes W_3(X)),\C) &= KK(C(X)\otimes C(X,W_3(X))\otimes \mathcal{K},\C) \\ &\cong KK(C(X)\otimes C(X,W_3(X)),\C) \,. \end{align*}
Thus, also $(C(X,\mathcal{E}_\sigma),C(X,\mathcal{E}_{-\sigma}\otimes W_3(X)))$ are PD pairs.

Generally, a PD pair $(A,B)$ may have several different fundamental classes. For example, if $\Delta\in K^d(A\otimes B)$ is a fundamental class and $\lambda\in KK_0(A,A)$ is invertible, then $\lambda\otimes_A \Delta\in K^d(A\otimes B)$ is another fundamental class, with inverse $\Delta^\vee \otimes_A \lambda^{-1}\in K_{-d}(A\otimes B)$. This is easy to see using the properties of the Kasparov product, listed in theorem \ref{Theorem: Properties of Kasparov product}:
\begin{align*} (\Delta^\vee \otimes_A \lambda^{-1})\otimes_A (\lambda \otimes_A \Delta) &= \Delta^\vee \otimes_A (\lambda^{-1} \otimes_A \lambda) \otimes_A \Delta \\ &= \Delta^\vee \otimes_A 1_A \otimes_A \Delta \\ &= \Delta^\vee \otimes_A \Delta \\ &= (-1)^d\, 1_B\,. \end{align*}
The other direction requires more sophisticated use of associativity:
\begin{align*} (\Delta^\vee \otimes_A \lambda^{-1})\otimes_B (\lambda \otimes_A \Delta) &= (\Delta^\vee \otimes_{A\otimes B} (\lambda^{-1}\otimes 1_B))\otimes_B ((\lambda\otimes 1_B)\otimes_{A\otimes B} \Delta) \\ &= \left((\Delta^\vee \otimes_{A\otimes B} (\lambda^{-1}\otimes 1_B))\otimes 1_A\right) \\ &\qquad \otimes_{B\otimes A\otimes A} \left(((\lambda\otimes 1_B)\otimes_{A\otimes B} \Delta) \otimes 1_A\right) \\ &= \left((\Delta^\vee \otimes 1_A)\otimes_{A\otimes B\otimes A} (\lambda^{-1}\otimes 1_B \otimes 1_A)\right) \\ &\qquad \otimes_{B\otimes A\otimes A} \left(((\lambda\otimes 1_B \otimes 1_A)\otimes_{A\otimes B \otimes A} (\Delta\otimes 1_A))\right)
\\ &= (\Delta^\vee \otimes 1_A)\otimes_{A\otimes B\otimes A} \left((\lambda^{-1}\otimes 1_B \otimes 1_A) \otimes_{B\otimes A\otimes A} ((\lambda\otimes 1_B \otimes 1_A)\right) \\ &\qquad \otimes_{A\otimes B \otimes B} (\Delta\otimes 1_A)) \\ &= (\Delta^\vee \otimes 1_A)\otimes_{A\otimes B\otimes A} \left((\lambda^{-1} \otimes_A \lambda) \otimes 1_B \otimes 1_A\right) \\ &\qquad \otimes_{A\otimes B \otimes A} (\Delta\otimes 1_A)) \\ &= (\Delta^\vee \otimes 1_A)\otimes_{A\otimes B\otimes A} 1_{A\otimes B \otimes A} \otimes_{A\otimes B \otimes A} (\Delta\otimes 1_A)) \\ &= (\Delta^\vee \otimes 1_A)\otimes_{A\otimes B\otimes A} (\Delta\otimes 1_A) \\ &= \Delta^\vee \otimes_B \Delta \\ &= 1_A\,.
\end{align*}
Conversely, if $\Delta_1,\Delta_2\in K^d(A\otimes B)$ are fundamental classes, then $\Delta_1^\vee \otimes_B \Delta_2\in KK_0(A,A)$ is invertible, with inverse $(-1)^d\,\Delta_2^\vee \otimes_B \Delta_1 \in KK_0(A,A)$. The proof is again a straightforward computation, so we omit it.
\begin{proposition} For a PD pair $(A,B)$ the space of fundamental classes is isomorphic to the group of invertible elements in the ring $KK_0(A,A)\cong KK_0(B,B)$. \qed \end{proposition}

\section{Local bivariant cyclic homology}
It would be useful to have some sort of a bivariant Chern character from $KK$-theory to some bivariant cohomology theory. A first guess for the target would be a bivariant extension of periodic cyclic homology. However, $\HP_\bullet$ is not suited for this purpose, since it degenerates on $C^*$-algebras and Kasparov's $KK$-theory works \emph{only} for $C^*$-algebras. There are two possible solutions. Either we replace Kasparov's $KK$-theory with some other bivariant $K$-theory, which is defined for a wider class of algebras, or we use some other kind of cyclic theory as the target. Indeed, the \emph{bivariant $K$-theory of Cuntz} \cite{KT} is defined for local $C^*$-algebras and, thus, is compatible with periodic cyclic homology. Since we want to use Kasparov's theory to describe D-branes, we need to find some replacement for periodic cyclic theory. 

The bivariant cyclic homology theory we are looking for, is the \emph{local bivariant cyclic homology} of Puschnigg \cite{Pus1,Pus2,Cun}. Due to its complexity, we will not give an explicit definition, but rather take a more or less axiomatic approach.

The following theorem is a collection of results proven in \cite{Pus1,Pus2}.
\begin{theorem}
There exist bivariant bifunctors $\HL_\bullet(-,-)\,,\bullet=1,2$ from a suitable category of topological algebras, containing the category of separable $C^*$-algebras as a subcategory, to the category of complex vector spaces, satisfying:
\begin{enumerate}
\item $\HL_\bullet(-,-)$ is contravariant in the first variable and covariant in the second variable.
\item It is split exact, homotopy invariant and satisfies excision in both variables.
\item It is invariant under passing to smooth subalgebras of Banach algebras with the metric approximation property. All nuclear $C^*$-algebras have the metric approximation property \cite{CE}.
\item For (nuclear) separable $C^*$-algebras, there exists a bilinear associative composition product
\[ \otimes_D:\HL_i(A_1,B_1\otimes D)\times \HL_j(D\otimes A_2,B_2) \rightarrow \HL_{i+j}(A_1\otimes A_2,B_1\otimes B_2)\,, \]
which naturally reduces into an exterior product
\[ \otimes:\HL_i(A_1,B_1)\times \HL_j(A_2,B_2)\rightarrow \HL_{i+j}(A_1\otimes A_2,B_1\otimes B_2) \]
and into a composition product
\[ \otimes_D:\HL_i(A,D)\times \HL_j(D,B)\rightarrow \HL_{i+j}(A,B)\,. \]
\item For separable $C^*$-algebras $A,B$, there exists a natural \emph{bivariant Chern character}
\[ \ch:KK_\bullet(A,B)\rightarrow \HL_\bullet(A,B)\,, \]
which is compatible with the product structures. For example, if $x\in KK_i(A,B)$ and $y\in KK_j(B,C)$, then
\begin{equation}\label{Eq: Puschnigg Chern character, multiplicativity} \ch(x)\otimes_B \ch(y) = \ch(x\otimes_B y) \,. \end{equation}
\item Any $*$-homomorphism $f:A\rightarrow B$ of separable $C^*$-algebras induces a class $[f]_{\HL} \in \HL_0(A,B)$. If also $g:B\rightarrow C$ is a $*$-homomorphism, then $[g\circ f]_{\HL} = [f]_{\HL} \otimes_B [g]_{\HL}$ and $\ch([f]_{KK}) = [f]_{\HL}$. In particular, $\HL_0(A,A)$ is a ring with multiplicative identity
\[ 1_A = [\id_A]_{\HL} = \ch([\id_A]_{KK})\in \HL_0(A,A)\,. \]
\item For any smooth manifold $X$ there exists an isomorphism
\[ \HL_\bullet(\C,C_0(X)) \stackrel{\textit{iii)}}{\cong} \HL_\bullet(\C,C_0^\infty(X)) \cong \HP_\bullet(C_0^\infty(X))\,, \]
and therefore
\[ \HL_0(\C,C_0(X)) \cong H^\text{even}_\text{dR}(X) \quad \text{and} \quad \HL_1(\C,C_0(X)) \cong H^\text{odd}_\text{dR}(X)\,. \]
\end{enumerate}
\end{theorem}

It follows immediately from the above properties that invertibility in $KK$-theory carries over to $\HL$-theory under $\ch$: if $x\in KK_d(A,B)$ and $y\in KK_{-d}(B,A)$ are such that
\[ x\otimes_B y = 1_A\in KK_0(A,A)\quad \text{and} \quad y\otimes_A x = 1_B\in KK_0(B,B)\,, \]
then
\[ \ch(x)\otimes_B \ch(y) = 1_A\in \HL_0(A,A) \quad \text{and} \quad \ch(y)\otimes_A \ch(x) = 1_B \in \HL_0(B,B)\,. \]
Furthermore, there is an analogue for $KK$-equivalence in $\HL$-theory, called $\HL$-equivalence. It is easy to see that $KK$-equivalence implies $\HL$-equivalence. The converse is not true, though, basically because $\HL$-theory yields vector spaces over $\C$, which are insensitive to torsion.

\begin{definition} A pair $(A,B)$ of separable $C^*$-algebras is called a \emph{cyclic Poincar\'e duality pair (C-PD pair)}, if there exist $\Xi\in \HL_d(A\otimes B,\C)$ and $\Xi^\vee\in \HL_{-d}(\C,A\otimes B)$, such that
\[ \Xi^\vee \otimes_B \Xi = 1_A \in \HL_0(A,A) \quad \text{and} \quad \Xi^\vee \otimes_A \Xi = (-1)^d\,1_B \in \HL_0(B,B) \,. \]
We call the class $\Xi$ a \emph{cyclic fundamental class} for the pair $(A,B)$ and $\Xi^\vee$ its \emph{inverse}.
\end{definition}
As in $KK$-theory, Poincar\'e duality for a C-PD pair $(A,B)$ with cyclic fundamental class $\Xi$, yields isomorphisms
\begin{align*} \HL_\bullet(C,A\otimes D)&\xrightarrow{(-)\otimes_A \Xi} \HL_{\bullet+d}(C\otimes B,D)\,, \\ \HL_\bullet(B\otimes C,D) &\xrightarrow{\Xi^\vee \otimes_B (-)} \HL_{\bullet-d}(C,A\otimes D)\,. \end{align*}
By choosing $D=C=\C$, we get
\begin{align*} \HL_\bullet(A):=\HL_\bullet(\C,A)&\xrightarrow{(-)\otimes_A \Xi} \HL_{\bullet+d}(B,\C) = \HL^{\bullet+d}(B)\,, \\  \HL^\bullet(B):=\HL_\bullet(B,\C) &\xrightarrow{\Xi^\vee \otimes_B (-)} \HL_{\bullet-d}(\C,A)= \HL_{\bullet-d}(A)\,. \end{align*}
Again, it follows easily that $\HL_\bullet(A,A)\cong \HL_\bullet(B,B)$, which restricts to an isomorphism between the multiplicative groups of their invertible elements.

If $(A,B)$ is a PD pair with fundamental class $\Delta\in K^d(A\otimes B)$, then $(A,B)$ is also a C-PD pair\footnote{The converse is, of course, not true, due to the usual problem with torsion.} with cyclic fundamental class $\ch(\Delta)\in \HL^d(A\otimes B)$ (and inverse $\ch(\Delta)^\vee = \ch(\Delta^\vee)$). It follows from \eqref{Eq: Puschnigg Chern character, multiplicativity} that the diagram
\[ \xymatrix{ K_\bullet(A) \ar[rrr]^{\Delta\otimes_A (-)}\ar[d]_{\ch} & & & K^{\bullet+d}(B) \ar[d]^{\ch} \\ \HL_\bullet(A) \ar[rrr]_{\ch(\Delta)\otimes_A (-)} & & & \HL^{\bullet+d}(B) } \]
commutes. It should be noted that $\ch(\Delta)$ is generally not the only possible choice for a cyclic fundamental class for $(A,B)$. In fact, the space of cyclic fundamental classes is isomorphic to the space of invertible elements of $\HL_0(A,A)\cong \HL_0(B,B)$ \cite{BMRS2}. If $\Xi\in \HL^d(A\otimes B)$ is any cyclic fundamental class for $(A,B)$ with inverse $\Xi^\vee$, then
\[ \Xi = \lambda\otimes_A \ch(\Delta) \quad \text{and} \quad \Xi^\vee = \ch(\Delta^\vee) \otimes_A \lambda^{-1}\,, \]
for some invertible class $\lambda\in \HL_0(A,A)$.

Let $X$ be a compact oriented manifold of dimension $d$. For the C-PD algebra $C(X)$, we obtain the classical cyclic fundamental class $\Xi\in \HL^d(C(X)\otimes C(X))$ as the image of the class $[\varphi_X]\in \HP^\bullet(C^\infty(X))$ of the cyclic $d$-cocycle
\[ \varphi_X(f^0,f^1,\ldots,f^d) = \frac{1}{d!}\int_X f^0\,df^1\wedge\ldots \wedge df^d\,, \]
for $f^i\in C^\infty(X)$, under the homomorphism
\[ m^*:\HP^\bullet(C^\infty(X)) \cong \HL^\bullet(C(X)) \rightarrow \HL^\bullet(C(X)\otimes C(X))\,, \]
where the arrow is induced by the product map $m:C(X)\otimes C(X) \rightarrow C(X)$, $f\otimes g\mapsto fg$. The class $\Xi$, therefore, corresponds to the classical fundamental class $[X]$ and Poincar\'e duality in $\HL$-theory agrees with the classical one. 

We conclude this section with the \emph{universal coefficient theorem for $\HL$-theory}. For the proof, see \cite{Mey}.
\begin{theorem}[Universal coefficient theorem for $\HL$-theory]
Let $A$ be a separable $C^*$-algebra satisfying the UCT. There exist natural isomorphisms
\[ \HL_\bullet(A,B) \cong \Hom_{\C}(\HL_\bullet(A),\HL_\bullet(B)) \cong \Hom_{\C}(K_\bullet(A)\otimes_{\Z} \C, \HL_\bullet(B)) \,. \]
If also $B$ satisfies the UCT, then
\[ KK_\bullet(A,B)\otimes_{\Z} \C \cong \Hom_{\Z}(K_\bullet(A),K_\bullet(B))\otimes_{\Z} \C \,. \]
It also follows easily from the universal coefficient theorems that if both $A$ and $B$ satisfy the UCT and the groups $K_\bullet(A)$ are finitely generated, then
\[ \HL_\bullet(A,B)\cong KK_\bullet(A,B)\otimes_{\Z} \C \,. \]
\end{theorem}

\section{Todd classes and Gysin homomorphisms}
Let $\mathbf{\underline{SC}^*}$ denote the category of separable $C^*$-algebras, which form a PD pair with some other separable $C^*$-algebra. We remind that, within this category, we restrict to nuclear algebras whenever tensor products are involved (practically always), for which we do not have to worry about the topology of the tensor product $C^*$-algebra. Let $A$ be an object of $\mathbf{\underline{SC}^*}$ and denote by $\widetilde{A}$ the algebra, for which $(A,\widetilde{A})$ is a PD pair. For a PD algebra $A$, the canonical choice $\widetilde{A}=A^\circ$ is always assumed.

\begin{definition} Let $A\in \mathbf{\underline{SC}^*}$ and $\Delta\in K^d(A\otimes \widetilde{A})$ the fundamental class for $(A,\widetilde{A})$. Then $(A,\widetilde{A})$ is also a C-PD pair, with cyclic fundamental class $\Xi\in \HL^d(A\otimes \widetilde{A})$. The \emph{Todd class} of $A$ is defined to be the difference between the classes $\Xi$ and $\ch(\Delta)$\footnote{Recall the discussion before theorem \ref{Theorem: A-H-G-R-R}.}:
\[ \Todd(A) = \Todd_{\Delta,\Xi}(A,\widetilde{A}):=\Xi^\vee \otimes_{\widetilde{A}} \ch(\Delta) \in \HL_0(A,A) \,. \]
It is invertible with inverse
\[ \Todd(A)^{-1} = \Todd_{\Delta,\Xi}(A,\widetilde{A})^{-1} := (-1)^d\, \ch(\Delta^\vee)\otimes_{\widetilde{A}} \Xi \in \HL_0(A,A) \,. \]
Remark, that $\Todd(A)=1_A$, when $\Xi=\ch(\Delta)$.
\end{definition}
It is proved in \cite{BMRS1} that $\mathbf{\underline{SC}^*}$ is not too restrictive a category of noncommutative spaces: any algebra in $\mathbf{SC^*}$, which satisfies the universal coefficient theorem for $KK$-theory and has finitely generated $K$-theory, is also in $\mathbf{\underline{SC}^*}$.

The definition of $\Todd(A)$ involves a choice for the fundamental classes. Recall that if $\Delta_1\in K^d(A\otimes \widetilde{A})$ and $\Xi_1\in \HL^d(A\otimes \widetilde{A})$ are any other fundamental classes for the pair $(A,\widetilde{A})$, they are necessarily of the form $\Delta_1 = \lambda_{KK}\otimes_A \Delta$ and $\Xi_1 = \lambda_{\HL}\otimes_A \Xi$, with $\lambda_{KK}\in KK_0(A,A)$ and $\lambda_{\HL} \in \HL_0(A,A)$. In \cite{BMRS2} it is shown, using associativity of the Kasparov product, that
\[ \Todd_{\Delta_1,\Xi_1}(A,\widetilde{A}) = \ch(\lambda_{KK})\otimes_A \Todd_{\Delta,\Xi}(A,\widetilde{A})\otimes_A \lambda_{\HL}^{-1}\,. \]

Let $X$ be a compact complex manifold\footnote{Recall that complex manifolds are always $\spinc$.}. Then $C(X)$ is a PD algebra with fundamental class given by the Dolbeault (Dirac) operator and the cyclic fundamental class given by the orientation homology cycle. It follows from the universal coefficient theorem that $\HL_0(C(X),C(X))\cong \End(H^\bullet(X,\Q))$. Under this isomorphism, $\Todd(C(X))$ corresponds to $(-)\smile \Todd(X)$, the cup product with the complex $\Todd$ class $\Todd_\C(X)\in H^\bullet(X,\Q)$. If $X$ is a real compact oriented manifold, there is an analogous result, with $\Todd(C(X))$ corresponding to the cup product with the usual $\Todd$ class $\Todd(X)$.

Next we want develop a noncommutative generalization for the Gysin ''wrong way'' homomorphisms. Recall that if $A$ and $B$ are separable $C^*$-algebras and $f:A\rightarrow B$ a $*$-homomorphism, then $f$ induces homomorphisms
\[ f_*:K_\bullet(A)\rightarrow K_\bullet(B)\qquad \text{and}\qquad f^*:K^\bullet(B)\rightarrow K^\bullet(A) \,. \]
The corresponding Gysin maps would be
\[ f_!:K_\bullet(B)\rightarrow K_{\bullet+d}(A)\qquad \text{and} \qquad f^!:K^\bullet(A)\rightarrow K^{\bullet+d}(B)\,. \]
If $A$ and $B$ happen to be PD algebras, the Gysin maps can be defined, in the spirit of definitions \ref{Eq: Homology Gysin map} and \ref{Definition: Gysin map in K-theory}, as compositions
\[ f_! := \Pd^{-1}_{KK}\circ (f^\circ)_* \circ \Pd_{KK} \qquad \text{and} \qquad f^! := \Pd_{KK}\circ (f^\circ)^* \circ \Pd^{-1}_{KK}\,. \]
The maps $(f^\circ)_*:K_\bullet(A^\circ)\rightarrow K_\bullet(B^\circ)$ and $(f^\circ)^*:K^\bullet(B^\circ)\rightarrow K^\bullet(A^\circ)$ are induced by $f^\circ:A^\circ\rightarrow B^\circ$, $f^\circ(a^\circ)= f(a)^\circ$, for $a\in A$. This is not very satisfactory however, since PD algebras are too restrictive to be useful. We shall instead consider a certain subcategory of $\mathbf{SC^*}$, in which well-behaving Gysin maps exist.

\begin{definition} Let $\mathbf{SC^*_K}$ be a subcategory of $\mathbf{SC^*}$, consisting of separable $C^*$-algebras and \emph{(strongly)\footnote{There is also a weaker definition of $K$-orientation but we shall not discuss it here. See \cite{BMRS2} for further details.} $K$-oriented} $*$-homomorphisms, such that there exists a contravariant functor $!:\mathbf{SC^*_K}\rightarrow \underline{KK}$, sending
\[ \mathbf{SC^*_K} \ni (A\xrightarrow{f} B) \rightarrow f!\in KK_d(B,A)\,, \]
with the following properties:
\begin{enumerate}
\item For any $C^*$-algebra $A$ in $\mathbf{SC^*_K}$, the identity homomorphism $\id_A:A\rightarrow A$ and the zero homomorphism $0_A:A\rightarrow 0$ are $K$-oriented, with $\id_A! = 1_A\in KK_0(A,A)$ and $0_A! = 0\in KK_0(0,A)$.
\item If $(A\xrightarrow{f} B)\in \mathbf{SC^*_K}$, then also $(A^\circ \xrightarrow{f^\circ} B^\circ)\in \mathbf{SC^*_K}$.
\item If $A$ and $B$ are PD algebras in $\mathbf{SC^*_K}$, with fundamental classes $\Delta_A\in K^{d_A}(A\otimes A^\circ)$ and $\Delta_B\in K^{d_B}(B\otimes B^ \circ)$ and $(A\xrightarrow{f} B)\in \mathbf{SC^*_K}$, the class $f!\in KK_d(B,A)$ is determined by
\[ f! = (-1)^{d_A}\,\Delta_A^\vee\otimes_{A^\circ} [f^\circ]_{KK} \otimes_{B^\circ} \Delta_B \,, \]
with $d=d_A-d_B$.
\end{enumerate}
\end{definition}
Let $A\xrightarrow{f}B\xrightarrow{g}C$ be $K$-oriented $*$-homomorphisms. Functoriality of $!:\mathbf{SC^*_K}\rightarrow \underline{KK}$ requires that $(g\circ f)!=f!\otimes_B g!$. It is also easy to see that the third property above is compatible with functoriality. Suppose $A,B,C$ are PD algebras, with fundamental classes $\Delta_A\in K^{d_A}(A\otimes A^\circ)$, $\Delta_B\in K^{d_B}(B\otimes B^\circ)$ and $\Delta_C\in K^{d_C}(C\otimes C^\circ)$, and $f,g$ $K$-oriented $*$-homomorphisms as above. It is then a straightforward computation to show that
\begin{align*} f!\otimes_B g! &= \left((-1)^{d_A}\,\Delta_A^\vee\otimes_{A^\circ} [f^\circ]_{KK} \otimes_{B^\circ} \Delta_B\right)\otimes_B \left((-1)^{d_B}\,\Delta_B^\vee\otimes_{B^\circ} [g^\circ]_{KK} \otimes_{C^\circ} \Delta_C\right) \\ &= \left((-1)^{d_A}\, \Delta_A^\vee \otimes_{A^\circ} [(g\circ f)^\circ]_{KK} \otimes_{C^\circ} \Delta_C\right)\,. \end{align*}
One simply has to write out the Kasparov products in terms of dilations and composition products and apply associativity.

\begin{definition} Suppose $(A\xrightarrow{f} B)\in \mathbf{SC^*_K}$ and $f!\in KK_d(B,A)$. Noncommutative versions of the Gysin ''wrong way'' homomorphisms in $K$-theory and $K$-homology are defined as follows:
\begin{align*} f_!&:=(-)\otimes_B f!:K_\bullet(B)\rightarrow K_{\bullet+d}(A) \,, \\ f^!&:=f!\otimes_A (-):K^\bullet(A)\rightarrow K^{\bullet+d}(B) \,. \end{align*}
\end{definition}

We could now almost state a noncommutative version of the Grothendieck-Riemann-Roch theorem, but we are still missing Gysin maps in $\HL$-theory. Due to the algebraic similarity of $KK$-theory and $\HL$-theory, the definition of $K$-oriented maps can be directly carried over to $\HL$-theory to define \emph{$\HL$-oriented maps}.
\begin{definition} Let $\mathbf{SC^*_{\HL}}$ be a subcategory of $\mathbf{SC^*}$, consisting of separable $C^*$-algebras and \emph{(strongly)\footnote{There is again a weaker definition, analogous to the $KK$-theoretic one.} $\HL$-oriented} $*$-homomorphisms, such that there exists a contravariant functor $*:\mathbf{SC^*_{\HL}}\rightarrow \underline{\HL}$\footnote{The category $\underline{\HL}$ is analogous to the category $\underline{KK}$.}, sending
\[ \mathbf{SC^*_{\HL}} \ni (A\xrightarrow{f} B) \rightarrow f*\in \HL_d(B,A)\,, \]
with the following properties:
\begin{enumerate}
\item For any $C^*$-algebra $A$ in $\mathbf{SC^*_{\HL}}$ the identity homomorphism $\id_A:A\rightarrow A$ and the zero homomorphism $0_A:A\rightarrow 0$ are $\HL$-oriented with $\id_A* = 1_A\in \HL_0(A,A)$ and $0_A* = 0\in \HL_0(0,A)$.
\item If $(A\xrightarrow{f} B)\in \mathbf{SC^*_{\HL}}$, then also $(A^\circ \xrightarrow{f^\circ} B^\circ)\in \mathbf{SC^*_{\HL}}$.
\item If $A$ and $B$ are PD algebras in $\mathbf{SC^*_{\HL}}$ with fundamental classes $\Xi_A\in \HL^{d_A}(A\otimes A^\circ)$ and $\Xi_B\in \HL^{d_B}(B\otimes B^ \circ)$ and $(A\xrightarrow{f} B)\in \mathbf{SC^*_{\HL}}$, the class $f*\in \HL_d(B,A)$ is determined by
\[ f* = (-1)^{d_A}\,\Xi_A^\vee\otimes_{A^\circ} [f^\circ]_{\HL} \otimes_{B^\circ} \Xi_B \,, \]
with $d=d_A-d_B$.
\end{enumerate}
\end{definition}
\begin{definition} Suppose $(A\xrightarrow{f} B)\in \mathbf{SC^*_{\HL}}$ and $f!\in \HL_d(B,A)$. The noncommutative \emph{Gysin ''wrong way'' homomorphisms} in $\HL$-theory are defined as follows:
\begin{align*} f_!&:=(-)\otimes_B f!:\HL_\bullet(B)\rightarrow \HL_{\bullet+d}(A) \,, \\ f^!&:=f!\otimes_A (-):\HL^\bullet(A)\rightarrow \HL^{\bullet+d}(B) \,. \end{align*}
\end{definition}

Suppose that $f:A\rightarrow B$ is both $K$-oriented and $\HL$-oriented. Then there exists both Gysin classes $f!\in KK_d(B,A)$ and $f*\in \HL_d(B,A)$. Another class in $\HL_d(B,A)$ can be obtained from $f!$ as $\ch(f!)\in \HL_d(B,A)$. By the properties of the Chern character, the diagrams
\[ \xymatrix{ K_\bullet(B) \ar[rr]^{f_!} \ar[d]_{\ch} & & K_{\bullet+d}(A) \ar[d]^{\ch} \\ \HL_\bullet(B) \ar[rr]_{(-)\otimes_B \ch(f!)} & & \HL_{\bullet+d}(A) } \qquad \xymatrix{ K^\bullet(A) \ar[rr]^{f^!} \ar[d]_{\ch} & & K^{\bullet+d}(B) \ar[d]^{\ch} \\ \HL^\bullet(A) \ar[rr]_{\ch(f!)\otimes_A (-)} & & \HL^{\bullet+d}(B) } \]
obviously commute. Replacing $\ch(f!)$ with $f*$, yields the diagrams
\[ \xymatrix{ K_\bullet(B) \ar[rr]^{f_!} \ar[d]_{\ch} & & K_{\bullet+d}(A) \ar[d]^{\ch} \\ \HL_\bullet(B) \ar[rr]_{f_*} & & \HL_{\bullet+d}(A) } \qquad \xymatrix{ K^\bullet(A) \ar[rr]^{f^!} \ar[d]_{\ch} & & K^{\bullet+d}(B) \ar[d]^{\ch} \\ \HL^\bullet(A) \ar[rr]_{f^*} & & \HL^{\bullet+d}(B) } \]
which, in general, do \emph{not} commute. When $A$ and $B$ are PD algebras, we can compare $\ch(f!)$ and $f*$ explicitly and find the necessary corrections to be made to restore commutativity.
\begin{theorem}[Noncommutative Grothendieck-Riemann-Roch theorem] For any PD algebras $A$ and $B$, the relation 
\begin{equation}\label{Eq: NC G-R-R theorem} \ch(f!) = (-1)^{d_B}\,\Todd(B)\otimes_B (f*) \otimes_A \Todd(A)^{-1} \,, \end{equation}
between $\ch(f!)$ and $f*$, holds. Also, the diagrams
\[ \xymatrix{ K_\bullet(B) \ar[rr]^{f_!} \ar[d]_{(-1)^{d_B}\,\ch(-)\otimes_B \Todd(B)} & & K_{\bullet+d}(A) \ar[d]^{\ch(-)\otimes_A \Todd(A)} \\ \HL_\bullet(B) \ar[rr]_{f_*} & & \HL_{\bullet+d}(A) } \]
and
\[ \xymatrix{ K^\bullet(A) \ar[rr]^{f^!} \ar[d]_{(-1)^{d_B}\,\Todd(A)^{-1}\otimes_A \ch(-)} & & K^{\bullet+d}(B) \ar[d]^{\Todd(B)^{-1}\otimes_B \ch(-)} \\ \HL^\bullet(A) \ar[rr]_{f^*} & & \HL^{\bullet+d}(B) } \]
commute.
\end{theorem}
\begin{proof} Equation \eqref{Eq: NC G-R-R theorem} is proven in \cite{BMRS2} by using
\begin{align*} f! &= (-1)^{d_A}\,\Delta_A^\vee\otimes_{A^\circ} [f^\circ]_{KK} \otimes_{B^\circ} \Delta_B \\ f* &= (-1)^{d_A}\,\Xi_A^\vee\otimes_{A^\circ} [f^\circ]_{\HL} \otimes_{B^\circ} \Xi_B \end{align*}
and theorem \ref{Theorem: Properties of Kasparov product}.
The diagrams are obtained easily from \eqref{Eq: NC G-R-R theorem}:
\begin{align*} \ch(f_!(\xi))\otimes_A \Todd(A) &= \ch(\xi\otimes_B f!)\otimes_A \Todd(A) \\ &= \ch(\xi)\otimes_B \ch(f!) \otimes_A \Todd(A) \\ &= \ch(\xi) \otimes_B (-1)^{d_B}\,\Todd(B)\otimes_B f* \\ &= (-1)^{d_B}\, f_*(\ch(\xi)\otimes_B \Todd(B))\,, \end{align*}
for any $\xi\in K_\bullet(B)$, and
\begin{align*} \Todd(B)^{-1}\otimes_B \ch(f^!(\eta)) &= \Todd(B)^{-1}\ch(f!\otimes_A \eta) \\ &= \Todd(B)^{-1}\ch(f!)\otimes_A \ch(\eta) \\ &= (-1)^{d_B} f* \otimes_A \Todd(A)^{-1} \otimes_A \ch(\eta) \\ &= (-1)^{d_B}\, f^*(\Todd(A)^{-1} \otimes_A \ch(\eta))\,, \end{align*}
for any $\eta\in K^\bullet(A)$.
\end{proof}

\section{The Minasian-Moore formula}
We are almost ready to present the noncommutative generalization of the Minasian-Moore formula
\[ Q_{\Q}(\Sigma,x):= \ch(\phi_!(x))\smile \sqrt{\Todd(X)} \in H^\text{even}(X,\Q) \,. \]
We expect to obtain the correct generalization by replacing $\Todd(X)$ with the $\Todd$ class of a noncommutative spacetime $A$, the untwisted $K$-theory class with a class $\xi\in K_0(B)$, where $B$ is the noncommutative worldvolume. Naturally, the algebra $A$ must be in $\mathbf{\underline{SC}^*}$ for the $\Todd$ class to exist and the continuous map $\phi:A\rightarrow B$ must be $K$-oriented\footnote{Note, that in the noncommutative case the direction of the map $\phi:A\rightarrow B$ is from the spacetime algebra to the worldvolume algebra, whereas in the commutative case it points from the worldvolume to spacetime.}.

There is still a slight technical problem, namely it is not clear that a square root of $\Todd(A)$ exists. However, if the algebra $A$ satisfies the UCT, then $\HL_\bullet(A,A)\cong \End_{\C}(\HL_\bullet(A),\HL_\bullet(A))$. If $\dim_{\C}(\HL_\bullet(A))$ is finite, say $n$, then $\Todd(A)$ is in $\GL_n(\C)$.
\begin{lemma} Every matrix $M$ in $\GL_n(\C)$ admits a square root $\sqrt{M}\in \GL_n(\C)$. \end{lemma}
\begin{proof} The matrix $M$ can be brought to the Jordan canonical form $J$ by conjugation with an invertible matrix $P$, that is,
\[ J = \begin{pmatrix} J_1 & & \\ & \ddots & \\ & & J_p \end{pmatrix} = P^{-1}MP\,, \quad \text{with} \quad J_i = \begin{pmatrix} \lambda_i & 1 & & \\ & \lambda_i & \ddots & \\ & & \ddots & 1 \\ & & & \lambda_i \end{pmatrix}\,. \]
The complex numbers $\lambda_i$ are nonzero since $M$ is invertible. On the other hand,
\[ J_i = \lambda_i\begin{pmatrix} 1 & \lambda_i^{-1} & & \\ & 1 & \ddots & \\ & & \ddots & \lambda_1^{-1} \\ & & & 1 \end{pmatrix} = \lambda_i(1+K_i)\,, \]
where the matrices $K_i$ are nilpotent. The square roots $\sqrt{1+K_i}$ are thus given by \emph{finite} Taylor polynomials. Therefore, also the square roots $\sqrt{J_i}=\sqrt{\lambda_i}\,\sqrt{1+K_i}$ of the Jordan blocks exist. The square root of $J$ is given by
\[ \sqrt{J} = \begin{pmatrix} \sqrt{J_1} & & \\ & \ddots & \\ & & \sqrt{J_p} \end{pmatrix} \,. \]
We can now choose $\sqrt{M}=P\sqrt{J}P^{-1}$.
\end{proof}
The lemma establishes that, under the given assumptions, $\sqrt{\Todd(A)}$ exists. Still, the square root matrix constructed above is not unique. It is not known if, or under which circumstances, a canonical choice for $\sqrt{\Todd(A)}$ can be made \cite{BMRS2}. We shall ignore the possible ambiguity and proceed as if a canonical choice has been made.

\begin{definition}[Minasian-Moore]\label{Definition: NC Minasian-Moore formula}
Let $A\in \mathbf{\underline{SC}^*}$ be the noncommutative spacetime and $B\in \mathbf{SC^*}$ the noncommutative D-brane worldvolume. Suppose that $A$ satisfies the UCT and $\dim_\C(\HL_\bullet(A))$ is finite. Let $\phi:A\rightarrow B$ be a $K$-oriented $*$-homomorphism\footnote{In \cite{BMRS2} the theory is worked out also for weakly $K$-oriented maps, which generalizes the definition slightly.}. Let $\xi\in K_\bullet(B)$ classify the configuration of noncommutative D-branes wrapping $B$.
The quantized RR-charge of the noncommutative D-brane is
\[ Q_{\Z}(B,\xi) := \phi_!(\xi) \in K_\bullet(A) \]
and the classical RR-charge by
\[ Q_{\C}(B,\xi) := \ch(\phi_!(\xi))\otimes_A \sqrt{\Todd(A)}  \in \HL_\bullet(A) \,. \]
\end{definition}


One last thing to consider is a noncommutative generalization of the isometric pairing formula, proposition \ref{Proposition: Isometric pairing formula}. Suppose that $A$ and $B$ are as in definition \ref{Definition: NC Minasian-Moore formula}. Poincar\'e duality yields a pairing
\[ (-,-)_K:K_\bullet(A)\otimes_{\Z} K_{-\bullet-d}(\widetilde{A})\rightarrow \Z \,, \]
given by
\[ (\alpha,\beta)_K := (\alpha\otimes \beta)\otimes_{A\otimes \widetilde{A}} \Delta \in KK_0(\C,\C) = \Z\,. \]
If $X$ is $\spinc$ and $A=\widetilde{A}=C(X)$ with fundamental ''Dirac'' class $\Delta$, the noncommutative pairing reduces to the classical index pairing appearing in proposition \ref{Proposition: Isometric pairing formula}. If both $A$ and $\widetilde{A}$ satisfy the UCT and both $K_\bullet(A)$ and $K_\bullet(\widetilde{A})$ are finitely generated, then it follows from the UCT that the pairing $(-,-)_K$ is nondegenerate (modulo torsion). The PD pair $(A,\widetilde{A})$ is also a C-PD pair with cyclic fundamental class $\Xi$, possibly different from $\ch(\Delta)$. There is a pairing
\[ (-,-)_H:\HL_\bullet(A)\otimes_{\C} \HL_{-\bullet-d}(\widetilde{A})\rightarrow \C \,, \]
given by
\[ (x,y)_H := (x\otimes y)\otimes_{A\otimes \widetilde{A}} \Xi \,. \]
It is nondegenerate if the UCT holds and coincides with the classical pairing appearing in proposition \ref{Proposition: Isometric pairing formula} when $A=\widetilde{A}=C(X)$, with $X$ an oriented manifold, and when $\Xi$ is the cyclic ''orientation'' fundamental class. 
\begin{proposition}[Isometric pairing formula \cite{BMRS1}]\label{Proposition: Isometric pairing formula, NC case}
Suppose that $A$ and $\widetilde{A}$ satisfy the UCT and $\HL_\bullet(A)$ and $\HL_\bullet(\widetilde{A})$ are finite dimensional vector spaces. The modified Chern character
\[ \Ch:=\ch(-)\otimes_A\sqrt{\Todd(A)}: K_\bullet(A)\rightarrow \HL_\bullet(A) \,, \]
appearing in the noncommutative Minasian-Moore formula, is an isometry with respect to the inner products $(-,-)_K$ and $(-,-)_H$:
\[ (\alpha,\beta)_K = \Big(\Ch(\alpha),\,\Ch(\beta)\Big)_H = \left(\ch(\alpha)\otimes_A \sqrt{\Todd(A)},\,\ch(\beta)\otimes_{\widetilde{A}} \sqrt{\Todd(\widetilde{A})}\right)_H \,. \]
\end{proposition}

\end{document}